\documentclass[b5paper, english, oneside]{memoir}

\usepackage{asymptotic-macros}

\ifarxiv
\else
\bibliography{asymptotic}
\usepackage[final]{pdfpages}
\fi

\title{\texorpdfstring{$\ohsy$}{O}-notation in Algorithm Analysis \\ (updated)}
\author{Kalle Rutanen}

\ifindex
\makeindex
\fi

\begin{document}

\frontmatter

\ifthesis
\includepdf[pages=-]{thesis-cover.pdf}
\else
\maketitle
\fi

\newpage
\thispagestyle{empty}
\vspace*{.35\textheight}
\begin{center}To my dear cats Kati (1996--2015) and Minari (2001--2016), \newline with whom I shared my life.
\end{center}

\chapter{Abstract}
\setcounter{page}{1}

In computer science, algorithm analysis is concerned with the correctness and complexity of algorithms. \emph{Correctness analysis} is about verifying that an algorithm solves the problem it is claimed to solve, and \emph{complexity analysis} is about counting the amount of resources an algorithm takes to run in an abstract machine. This \manuscript{} is concerned with complexity analysis. 

The amount of resources consumed by an algorithm is captured by a cost function which maps each input of the algorithm to a non-negative real number. The cost functions are ordered by an order-relation $\preceq$, which makes it possible to \emph{simplify}, and to \emph{compare} cost functions. A cost function can be simplified by replacing it with an equivalent, simpler cost function. The smaller the cost function is according to $\preceq$, the better the algorithm --- at least for the measured resource.  

The $\ohsy$-notation $\oh{}{f}$ is the set of functions $g$ which satisfy $g \preceq f$. It contains the same simplification and comparison tools in a slightly different, but equivalent form. 

How should the $\ohsy$-notation be defined? By the above, we may ask an equivalent question: how should the order relation $\preceq$ be defined? \We{} provide \nprim{} intuitive properties for $\preceq$, and then show that there is exactly one definition of $\preceq$ which satisfies these properties: linear dominance. 

\We{} show that Master theorems hold under linear dominance, define the $\ohsy$\-/mappings as a general tool for manipulating the $\ohsy$\-/notation, and abstract the existing definitions of the $\ohsy$-notation under local linear dominance.

\chapter{Preface}

The research presented in this \manuscript{} was carried out at the Department of Mathematics at Tampere University of Technology during years 2013--2016. The \manuscript{} was financially supported by the Finnish National Doctoral Programme in Mathematics and its Applications (2013--2014), and by the science fund of the city of Tampere (Summer 2015).

I am deeply thankful to my supervisors Sirkka-Liisa Eriksson, Karen Egiazarian, and Germán-Gómez Herrero for the support during these years, and the years before them. 

Special thanks go to office assistant Riitta Lahti, financial administrative assistant Tiina Sävilahti, and laboratory engineer Kari Suomela, at the Department of Mathematics. When it came to practical matters, they made my life easy, all the while uplifting the atmosphere.

I want to thank Raphael Reitzig for providing me with extensive feedback on an Arxiv-version of the \manuscript{}, related to the content and typographical concerns. Concerning typography, the idea to draw inside the $\ohsy$-symbols is his; this resulted in $\lohsy$, $\fohsy$, $\cohsy$, $\pohsy$, $\trohsy$, $\rohsy$, and $\aohsy$, which are both suggestive and can be drawn inline without extending line-height.\footnote{I will take the blame for the shapes.} I was led to prove the limit theorems after he asked me whether ratio-limits could still be used to reason about the $\ohsy$-notation. He was also the first external reader to provide me with positive feedback.

Jarkko Kari suggested to study the independence of the primitive properties, and pointed out how some of those results already followed from the study of candidate definitions.

David Wilding created the original \Cref{OrderPreservingDiagram}, which I have adapted here with his permission. In the introduction, I have used \emph{svg-cards} playing cards designed by David Bellot.

I want to thank Heikki Huttunen and Pekka Ruusuvuori for introducing me to academic work as an undergraduate at the Department of Signal Processing. I am grateful to my friends and coworkers at the Department of Signal Processing, and at the Department of Mathematics for all those excellent detours from research. Namely, Tapio Manninen, Francesco Cricri, Pasi Hämäläinen, Henri Riihimäki, Petteri Laakkonen, Marko Järvenpää, Juho Lauri, Vesa Vuojamo, Elina Viro, and Markku Åkerblom.

Lastly, I want to thank my family, friends, and Alina for their invaluable support.

\bigskip
\begin{flushright}
Tampere, October 2016 \\
\vspace{1.5cm}
Kalle Rutanen
\end{flushright}

\newpage 
\tableofcontents
\newpage 

\mainmatter

\chapter{Introduction to the \manuscript{}}
\label{Introduction}

Suppose we want to sort playing cards into increasing order. The kind of cards does not matter, only that for any two cards $A$ and $B$, either $A$ comes before $B$, denoted by $A < B$, or $B$ comes before $A$, denoted by $B < A.$ One way to sort the deck is by the following \define{insertion sort}, which is visualized in \fref{InsertionSortExample}.

\begin{enumbox}
\item\label{InsertionSort-Step2} If there are no cards left in the deck, we are done.
\item\label{InsertionSort-Step3} Otherwise, we pick the top-most card $A$ from the deck.
\item\label{InsertionSort-Step4} Starting from the right end of the cards on the table, we compare $A$ to each card $B$ on the table, until $A < B$ for the first time.
\item\label{InsertionSort-Step5} If there is no such card $B$, we place $A$ on the table as the right-most card, and start again from step \ref{InsertionSort-Step2}.
\item\label{InsertionSort-Step6} Otherwise, we place $A$ on the table immediately left of $B$, and start again from step \ref{InsertionSort-Step2}.
\end{enumbox}

\begin{figure}
\centering
\subfloat[]{\includegraphics[scale=1.4]{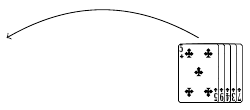}} \hfill
\subfloat[]{\includegraphics[scale=1.4]{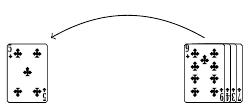}}

\subfloat[]{\includegraphics[scale=1.4]{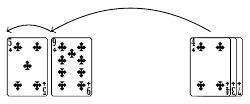}} \hfill
\subfloat[]{\includegraphics[scale=1.4]{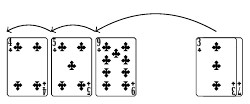}}

\subfloat[]{\includegraphics[scale=1.4]{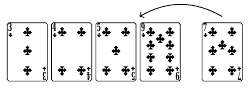}} \hfill
\subfloat[]{\includegraphics[scale=1.4]{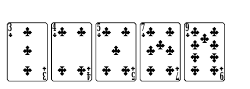}}

\caption{Sorting 5 playing cards with insertion sort. The cards on the table are visualized on the left, while the cards in the deck are visualized on the right. During sorting, 7 comparisons are made between cards.}
\label{InsertionSortExample}
\end{figure}

How many comparisons between cards do we have to perform to sort a given deck with $n$ cards? In the best case, the deck is already sorted in increasing order, requiring us to perform only $n - 1$ comparisons. In the worst case, the deck is already sorted in decreasing order, requiring us to perform $1 + 2 + 3 + ... + (n - 1) = (n - 1) n / 2$ comparisons. Therefore, sorting a deck with $n$ cards using insertion sort always requires something between $n - 1$ and $(n - 1)n / 2$ comparisons.

Another way to sort the deck is by the following \define{merge sort}.

\begin{enumbox}
\item We place all the cards in a row on the table. Every card $A$ is considered to form a sequence $(A)$ consisting of a single card.
\item \label{MergeSortStartOver} If there is at most one card-sequence on the table, we are done.
\item For each two successive card-sequences, $A$ and $B$, we combine $A$ and $B$ into a single sequence $C$ by repeatedly picking the smallest of the cards at the left end of $A$ and the left end of $B$ (e.g., $(3, 7)$ and $(5, 9)$ becomes $(3, 5, 7, 9)$). If one of the sequences runs out before the other, we place the remaining sequence at the right end of $C$ without performing additional comparisons.
\item We start again from step \ref{MergeSortStartOver}.
\end{enumbox}

\begin{figure}
\centering
\hfill\subfloat[]{\includegraphics[scale=1.4]{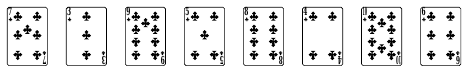}}\hfill\null

\hfill\subfloat[]{\includegraphics[scale=1.4]{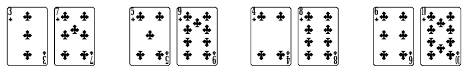}}\hfill\null

\hfill\subfloat[]{\includegraphics[scale=1.4]{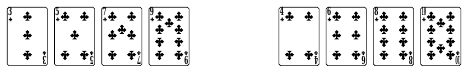}}\hfill\null

\hfill\subfloat[]{\includegraphics[scale=1.4]{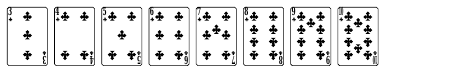}}\hfill\null
\caption{Sorting 8 playing cards with merge sort. During sorting, 17 comparisons are made between cards; the maximum possible for 8 cards. The merging procedure is repeated $\lg{8} = 3$ times.}
\label{MergeSortExample}
\end{figure}

How many comparisons between cards do we now have to perform to sort a given deck with $n$ cards? Suppose $n$ is a power of two, so that $n = 2^d$ for some integer $d$. In the best case, the cards on the table are already sorted, and we need to perform $(n / 2) \lg{n}$ comparisons to combine the sequences into a single deck. In the worst case, we need to perform $n (\lg{n} - 1) + 1$ comparisons. A deck which attains this bound for $n = 8$ is given in \fref{MergeSortExample}. \tref{PlayingCardComparisons} tabulates the number of comparisons required, in the worst case and in the best case, as a function of the number $n$ of cards in the deck for both ways of sorting. It can be seen that merge sort scales better than insertion sort in the worst case. Note however, that there are decks --- especially those which are sorted or almost sorted --- for which insertion sort performs less comparisons than merge sort.

\begin{table}
\begin{tabular}{|l|r|r|r|r|}
\hline 
Sorting procedure / $n$ & $64$ & $256$ & $1024$ & $4096$\\
\hline 
Insertion sort --- worst case & 2016 & 32640 & 523776 & 8386560 \\
Merge sort --- worst case & 321 & 1793 & 9217 & 45057 \\
Insertion sort --- best case & 63 & 255 & 1023 & 4095 \\
Merge sort --- best case & 192 & 1024 & 5120 & 24576 \\
\hline 
\end{tabular}
\centering
\caption{Number of comparisons needed to sort $n$ cards by using either insertion sort, or merge sort. Merge sort scales better than insertion sort in the worst case.}
\label{PlayingCardComparisons}
\end{table}

\section{Algorithms and their analysis}

\begin{figure}
\centering
\includegraphics{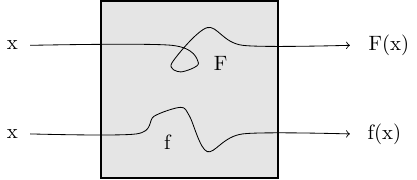}
\caption{A black-box view of an algorithm. The algorithm, visualized here as a gray box, transforms an input $x$ from an input-set $X$ to an output $F(x)$ from an output-set $Y$. The transform is captured by a mapping $\function{F}{X}{Y}$. The abstract machine executing the algorithm consumes resources at each step of the algorithm. The cost function, for given input and resource, is captured by a mapping $\function{f}{X}{\nonn{\TR}}$.}
\label{AlgorithmBlackbox}
\end{figure}

Such step-wise procedures are called algorithms. A \emph{run} of an algorithm is a finite sequence of \emph{atomic operations} --- in the above, picking a card from the deck, comparing two cards, and placing a card on the table. The atomic operations are performed one by one by an \emph{abstract machine} --- in the above the person sorting the cards. The abstract machine, running the algorithm, transforms a given \emph{input} --- such as a deck of cards --- into \emph{output} --- such as the sorted input deck. This correspondence is captured by a function $F$, as visualized in \fref{AlgorithmBlackbox}. The kind of algorithms that can be written depends on which atomic operations the abstract machine supports.

To each atomic operation, we associate a \emph{cost} in resources --- in the above one unit of resource to the task of comparing two cards. We then \emph{analyze} the algorithm for the amount of resources that it uses. The result of the analysis is a \emph{cost function} $f$ which associates each possible input to the resource-cost. However, as for insertion sort and merge sort, the function $f$ can be too hard to study directly --- or to get an intuition to. 

To make the analysis tractable, the input is grouped by some simpler property --- in the above by the number of cards in the deck. The value of the property shared by a group is called the \emph{label} of the group. Since a group may contain more than one input, we must come up with a way to summarize the costs inside each group. 

Such summaries include \emph{best case analysis} and \emph{worst case analysis} --- in the above providing us with the lower and upper bounds, respectively, for the number of comparisons required to sort a deck with a given number of cards. To analyze worst-case behavior under a given grouping, we pick from each group a \emph{representative} input which triggers the worst resource\-/cost. Conceptually, we then analyze a \emph{surrogate algorithm}, which takes a group-label as an input, maps it to the corresponding group-representative, and feeds the representative as an input to the actual algorithm.

Another common summary is \emph{average case analysis}, where we assume a random distribution for the input, and then analyse the mean of the resulting random cost function.

Grouping is often required to make analysis tractable. However, too rough a grouping loses detail without aiding intuition. Sometimes grouping is not necessary at all --- as when the input-set is $\TZ$.

\section{Simplified analysis by \texorpdfstring{$\ohsy$}{O}-notation}

Above we analyzed two algorithms for sorting playing cards, and noticed that merge sort is superior to insertion sort, at least when the measured resource is the number of performed comparisons in the worst case, and the deck contains at least 4 cards. This difference is fundamental, in the sense that the ratio by which merge sort beats insertion sort --- in the worst case --- keeps growing. For example, while for $n = 64$ the ratio is $6.3$, for $n = 4096$ the ratio is already $186$. Such \emph{algorithmic efficiency} is often much more important than the speed at which a given abstract machine executes the atomic operations.

Our analyses above were based on careful counting of the comparisons, with an aim of capturing the worst case and best case bounds exactly. Such analysis is only possible with an algorithm that is sufficiently simple. Numerical analysis is a field whose algorithms (e.g., matrix QR-decomposition) have traditionally been analyzed in this way.

With a complex algorithm, a closed form expression for the cost function may not exist, or otherwise be too complex for a human to make sense of --- not to mention the ingenuity needed to solve the combinatorial problems that occur when deriving those expressions. 

The way out of these problems is to declare that we are not interested in small differences, and to simplify them out as they occur. For example, the relative error of approximating the worst-case resource consumption $(n - 1)n / 2$ of insertion sort by $n^2 / 2$ is only $1/10$ for $n = 11$, and decreases to zero as $n$ grows. 

In addition, we declare that we are only interested in algorithmic efficiency, and so are ready to ignore the speed of the abstract machine up to a constant. It does not matter whether some other machine is 2 times faster to perform the algorithm. What matters is how the algorithm scales --- does it scale like $n^2$ or like $n \lg{n}$? 

We then say that insertion sort takes about $n^2$ comparisons, because this is simple, and does not fundamentally differ from the exact count. However, here we must be careful. Since we aim to simplify expressions whenever possible, we need a guarantee that computing with the simplified expressions provides the same answer as first computing the exact cost function, and then simplifying it. We need \emph{simplification rules} for algorithmic cost functions.

The simplification rules are formalized by a tool called \emph{$\ohsy$-notation}. Intuitively, whenever we want to simplify, we surround the expression with an $\ohsy$, as in $\oh{}{(n - 1)n / 2}$, and then use whatever simplification rules the $\ohsy$-notation provides. With proper assumptions on the domain of $n$, we can then show that
\begin{eqs}
\oh{}{(n - 1)n / 2} & = \oh{}{(n - 1)n} \\
{} & = \oh{}{n^2 - n} \\
{} & = \oh{}{n^2}.
\end{eqs}
The notation $\oh{}{f}$, where $f$ is a real-valued function, stands for the set of real-valued functions which scale at least as well as $f$. Having the same $\ohsy$-set divides the functions into equivalence classes --- in the above, we then say that the functions $(n - 1)n / 2$ and $n^2$ are \emph{equivalent}.

However, simplification is not the whole story. In addition, the $\ohsy$-notation must induce an order, so that the cost functions can be \emph{compared}, and the order must respect the structure that is present in algorithms: looping, branching, and composition.

How do we define such an $\ohsy$-notation, and which rules do we need? 

\section{Definitions of \texorpdfstring{$\ohsy$}{O}-notation}

Here are some of the proposed definitions\footnote{$\posi{\TR}$ is the set of positive real numbers, $\power{X}$ is the set of subsets of $X$, $\restrb{f}{A}$ is the restriction of a function $f$ to a set $A$, $f \leq g$ for functions $f$ and $g$ means $\forall x \in X: f(x) \leq g(x)$, and $\rc{X}$ is the set of functions from $X$ to $\nonn{\TR}$. See \sref{Notation} for more notation.} for an $\ohsy$-notation in algorithm analysis.

\newcommand{\defineasymptotic}{
\begin{definition}[Asymptotic linear dominance]
\defineexp{Asymptotic linear dominance}{linear dominance!asymptotic} $\pohsy$ is defined by $g \in \poh{X}{f}$ if and only if
\begin{equation}
\exists c \in \posi{\TR}, \exists y \in \TR^d: \restry{g}{y} \leq c \restry{f}{y},
\end{equation}
for all $f, g \in \rc{X}$, and all $X \in U$, where $U = \bigcup_{d \in \posi{\TN}} \power{\TR^d}$.
\end{definition}
}

\defineasymptotic

\begin{note}[]
This definition is given in \cite{IntroAlgo} in the univariate form on $\TN$, and generalized to $\TN^2$ in a later exercise \cite[page 50]{IntroAlgo}. 
\end{note}

\newcommand{\definecoasymptotic}{
\begin{definition}[Coasymptotic linear dominance]
\defineexp{Coasymptotic linear dominance}{linear dominance!coasymptotic} $\cohsy$ is defined by $g \in \coh{X}{f}$ if and only if
\begin{equation}
\exists c \in \posi{\TR}, \exists y \in \TR^d: \restryd{g}{y} \leq c \restryd{f}{y}.
\end{equation}
for all $f, g \in \rc{X}$, and all $X \in U$, where $U = \bigcup_{d \in \posi{\TN}} \power{\TR^d}$. 
\end{definition}
}

\definecoasymptotic

\begin{note}[]
This definition is given in \cite{IntroAlgo2009} in the univariate form on $\TN$, and generalized to $\TN^2$ in a later exercise \cite[page 53]{IntroAlgo2009}.
\end{note}

\newcommand{\definecofinite}{
\begin{definition}[Cofinite linear dominance]
\defineexp{Cofinite linear dominance}{linear dominance!cofinite} $\fohsy$ is defined by $g \in \foh{X}{f}$ if and only if\footnote{$\cpower{X} = \setb{A \in \power{X} : \card{X \setminus A} < \infty}$.}
\begin{equation}
\exists c \in \posi{\TR}, \exists A \in \cpower{X}: \restrb{g}{A} \lt c \restrb{f}{A},
\end{equation}
for all $f, g \in \rc{X}$, and all sets $X$.
\end{definition}
}

\definecofinite

\begin{note}[Symbol shapes]
The drawing inside a given version of $\ohsy$ above mimics the shape of its restriction sets in $\TR^2$.
\end{note}

\newcommand{\definelinear}{
\begin{definition}[Full linear dominance]
\defineexp{Full linear dominance}{linear dominance!full} $\lohsy$ is defined by $g \in \loh{X}{f}$ if and only if
\begin{equation}
\exists c \in \posi{\TR}: g \lt c f,
\end{equation}
for all $f, g \in \rc{X}$, and all $X \in U$, where $U$ is the class of all sets.
\end{definition}
}

\definelinear

\begin{note}[Symbol shape]
The drawing inside $\lohsy$ mimics a line.
\end{note}

\begin{note}[Linear dominance]
We will often shorten full linear dominance to linear dominance.
\end{note}

\newcommand{\defineaffine}{
\begin{definition}[Affine dominance]
\emph{Affine dominance} $\aohsy$ is defined by $g \in \aoh{X}{f}$ if and only if
\begin{equation}
\exists c \in \posi{\TR}: g \lt c f + c,
\end{equation}
for all $f, g \in \rc{X}$, and all $X \in U$, where $U$ is the class of all sets.
\end{definition}
}

\defineaffine

\begin{note}[Symbol shape]
The drawing inside $\aohsy$ mimics the plus operator.
\end{note}

\begin{definition}[Related notations]
Given an $\ohsy$-notation, we define the related notations $\omegahs{X}, \smallomegas{X}, \smallohs{X}, \thetahs{X}$ as follows:
\begin{eqs}
g \in \omegah{X}{f} & \iff f \in \ohx{g}, \\
g \in \smallomegax{f} & \iff g \not\in \ohx{f} \land f \in \ohx{g}, \\
g \in \smallohx{f} & \iff g \in \ohx{f} \land f \not\in \ohx{g}, \\
g \in \thetahx{f} & \iff g \in \ohx{f} \land f \in \ohx{g}.
\end{eqs}
\end{definition}

\begin{note}[Equivalence from equality of $\ohsy$-sets]
The expression $\ohx{f} = \ohx{g}$ is equivalent to $f \in \thetahx{g}$; equality between $\ohsy$-sets can be used to establish the equivalence of functions. This is also true for $\omegahx{f} = \omegahx{g}$ and $\thetahx{f} = \thetahx{g}$.
\end{note}

\begin{note}[Study of the $\ohsy$-notation suffices]
It suffices to study the $\ohsy$\-/notation, since the related notations are completely determined by the $\ohsy$\-/notation. 
\end{note}

We will show linear dominance $\lohsy$ to be the correct definition of $\ohsy$-notation for algorithm analysis.

\section{Example analyses}

In the following, we provide some example analyses of algorithms, and demonstrate why most of the definitions of $\ohsy$-notation in the previous section --- except linear dominance --- are not suitable for algorithm analysis.

\begin{note}[Function from an expression]
We will often denote a function in the parameter of $\ohs{X}$ with an expression, as in $\oh{\TN}{n}$, where we actually mean $\oh{\TN}{f}$, with $f \in \rc{\TN}$ such that $f(n) = n$. 

When the expression contains multiple symbols, as in $\oh{\TN^2}{n^2 m}$, we interpret the symbols to be assigned to the input-tuple in alphabetical order, as in $\oh{\TN^2}{(m, n) \mapsto n^2 m}$. This is to acknowledge that $(m, n) \mapsto n^2 m$ and $(n, m) \mapsto n^2 m$ are different functions.
\end{note}

\begin{note}[A cost model for examples]
In the following examples, each addition operation $x + y$ costs one unit, while other operations cost nothing.
\end{note}

\begin{algorithm}
\caption{An algorithm which takes as input $(m, n) \in \TN^2$, and outputs $n$, if $m = 0$, and $0$ otherwise.}
\label{alg:ConstantComplexity}
\begin{algorithmic}[1]
\Procedure {computeOnPlane}{$m, n$}
\State $j \coloneqq 0$
\If {$m = 0$}
  \For {$i \in \icon{0}{n}$}
    \State $j \coloneqq j + 1$
  \EndFor
\EndIf
\State \Return $j$
\EndProcedure
\end{algorithmic}
\end{algorithm}

\begin{algorithm}
\caption{An algorithm which takes as input $n \in \TN$, and returns $n$.}
\label{alg:BasicAnalysis}
\begin{algorithmic}[1]
\Procedure {mapNaturalsToPlane}{$n$}
  \State \Return \Call{computeOnPlane}{$0, n$}
\EndProcedure
\end{algorithmic}
\end{algorithm}

\begin{example}[Algorithm on $\TN^2$]
\label{MultivariateCounterExampleInNSubAlgorithm}
Consider \aref{alg:ConstantComplexity}, which takes as input $(m, n) \in \TN^2$. The cost function for this algorithm is $f \in \rc{\TN^2}$ such that
\begin{eqs}
f(m, n) = 
\begin{cases}
n, & m = 0, \\
0, & m > 0, 
\end{cases}
\end{eqs}
We can show that $f \in \pthetah{\TN^2}{0}$.
\end{example}

\begin{example}[Calling algorithm on $\TN^2$]
\label{MultivariateCounterExampleInN}
Consider \aref{alg:BasicAnalysis}, which takes as input $n \in \TN$, and outputs the result of \aref{alg:ConstantComplexity} at $(0, n)$. The cost function for this algorithm is $g \in \rc{\TN}$ such that
\begin{eqs}
g(n) = n.
\end{eqs}
We can show that $g \in \pthetah{\TN}{n}$. Therefore, by calling an algorithm which is $\pthetah{\TN^2}{0}$, we get an algorithm which is $\pthetah{\TN}{n}$. This shows that $\pohsy$ is not suitable for a definition of $\ohsy$-notation in algorithm analysis.
\end{example}

\begin{algorithm}
\caption{An algorithm which takes as input $z \in \TZ$, and returns $\max(-z, 0)$.}
\label{alg:ConstantComplexityZ}
\begin{algorithmic}[1]
\Procedure {computeOnIntegers}{$z$}
\State $j \coloneqq 0$
\If {$z < 0$}
  \For {$i \in \icon{0}{-z}$}
    \State $j \coloneqq j + 1$
  \EndFor
\EndIf
\State \Return $j$
\EndProcedure
\end{algorithmic}
\end{algorithm}

\begin{algorithm}
\caption{An algorithm which takes as input $n \in \TN$, and returns $\max(n, 0)$.}
\label{alg:BasicAnalysisZ}
\begin{algorithmic}[1]
\Procedure {mapNaturalsToIntegers}{$n$}
  \State \Return \Call{computeOnIntegers}{$-n$}
\EndProcedure
\end{algorithmic}
\end{algorithm}

\begin{example}[Algorithm on $\TZ$]
\label{UnivariateCounterExampleInZSubAlgorithm}
Consider \aref{alg:ConstantComplexityZ}, which takes as input $z \in \TZ$ and returns $\max(-z, 0)$. The cost function of this algorithm is $f \in \rc{\TZ}$ such that
\begin{eqs}
f(z) = \max(-z, 0). 
\end{eqs}
We can show that $f \in \cthetah{\TZ}{0}$.
\end{example}

\begin{example}[Calling algorithm on $\TZ$]
\label{UnivariateCounterExampleInZ}
Consider \aref{alg:BasicAnalysisZ}, which takes as input $n \in \TN$, and outputs the result of \aref{alg:ConstantComplexityZ} at $-n$. The cost function of this algorithm is $g \in \rc{\TN}$ such that
\begin{eqs}
g(n) = n.
\end{eqs}
We can show that $g \in \cthetah{\TN}{n}$. Therefore, by calling an algorithm which is $\cthetah{\TZ}{0}$, we get an algorithm which is $\cthetah{\TN}{n}$. This shows that $\cohsy$ is not suitable for a definition of $\ohsy$-notation in algorithm analysis.
\end{example}

% \begin{example}[$\pohsy$, $\cohsy$, and $\fohsy$ fail \property{SubComp} in $\TN$]
% \label{SubcomposabilityDoesNotHoldExample}
% Let $\ohsy$ be one of $\pohsy$, $\cohsy$, or $\fohsy$, and $\function{s}{\TN}{\TN}$ be such that
% \begin{equation}
% s(n) =
% \begin{cases}
% 0, & n \in 2\TN, \\
% n, & n \in 2\TN + 1.
% \end{cases}
% \end{equation}
% Let $\function{\hat{f}}{\TN}{\TN}$ be such that $\hat{f}(0) = 1$ and $\function{f}{\TN}{\TN}$ be such that $f(0) = 0$ and $\hat{f} \in \oh{\TN}{f}$. Then $\bra{\hat{f} \circ s} \not\in \oh{\TN}{f \circ s}$. The problem is that $(f \circ s)(2k) = 0$ cannot be multiplied to linearly dominate $(\hat{f} \circ s)(2k)$, where $k \in \TN$, and that the number of such points is not finite. Linear dominance $\lohsy$ avoids this problem: from $\hat{f}(0) = 1$ and $f(0) = 0$, it follows that $\hat{f} \not\in \loh{\TN}{f}$.
% \end{example}

\begin{algorithm}
\caption{An algorithm which takes as input $n \in \TN$ and evaluates a sub-algorithm $n$ times at $0$.}
\label{alg:ClassRoom}
\begin{algorithmic}[1]
\Procedure {G}{$n$}
\For{$i \in \icon{0}{n}$}
  \State \Call{F}{$0$}
\EndFor
\EndProcedure
\end{algorithmic}
\end{algorithm}

\begin{example}[Course-exercise]
\label{ZeroCounterExample}
Consider \aref{alg:ClassRoom}, which takes as input $n \in \TN$ and calls a sub-algorithm $F$ at $0$ repeatedly $n$ times. Denote the cost function of $F$ by $f \in \rc{\TN}$, and the cost function of $G$ by $g \in \rc{\TN}$. Suppose $f \in \fthetah{\TN}{1}$. What is $\fthetah{\TN}{g}$? 

The given information is not sufficient to solve this problem. In particular, let $f_1, f_2 \in \rc{\TN}$ be such that $f_1(n) = 1$, and
\begin{eqs}
f_2(n) =
\begin{cases}
0, & n = 0, \\
1, & n > 0.	
\end{cases}
\end{eqs}
Then $f_1, f_2 \in \fthetah{\TN}{1}$, but
\begin{eqs}
g_1 & \in \fthetah{\TN}{n}, \\
g_2 & \in \fthetah{\TN}{0}, \\
\fthetah{\TN}{n} & \cap \fthetah{\TN}{0} = \emptyset.
\end{eqs}
This shows that $\fohsy$ is not suitable for a definition of $\ohsy$-notation in algorithm analysis.

There is a fundamental difference between consuming resources ($f(0) > 0$) and not consuming resources ($f(0) = 0$), which $\fohsy$ ignores here. 

In contrast, $g \in \lthetah{\TN}{n}$ provided $f \in \lthetah{\TN}{1}$, as expected. 
\end{example}

\begin{algorithm}
\caption{An algorithm which takes as input $x \in \posi{\TR}$ and returns $\max(\ceilb{-\lg{x}}, 0) \in \TN$.}
\label{alg:Real}
\begin{algorithmic}[1]
\Procedure {doublesToOne}{$x$}
\State $k \gets 0$
\State $y \gets x$
\While {$y < 1$}
  \State $y \gets 2y$
  \State $k \gets k + 1$
\EndWhile
\State \Return $k$
\EndProcedure
\end{algorithmic}
\end{algorithm}

\begin{algorithm}
\caption{An algorithm which takes as input $n \in \TN$, and returns $n$.}
\label{alg:UnivariateBroken}
\begin{algorithmic}[1]
\Procedure {identity}{$n$}
\State \Return \Call{doublesToOne}{$2^{-n}$}
\EndProcedure
\end{algorithmic}
\end{algorithm}

\begin{example}[Algorithm on $\posi{\TR}$]
\label{UnivariateCounterExampleInRSubAlgorithm}
Consider \aref{alg:Real}, which takes as input $x \in \posi{\TR}$ and outputs $\max(\ceilb{-\lg{x}}, 0) \in \TN$; this is the number of times that $x$ must be doubled to grow $\geq 1$. The cost function of this algorithm is $f \in \rc{\posi{\TR}}$ such that
\begin{eqs}
f(x) = \max(\ceilb{-\lg{x}}, 0).
\end{eqs}
We can show that $f \in \pthetah{\posi{\TR}}{0}$.
\end{example}

\begin{example}[Calling an algorithm on $\posi{\TR}$]
\label{UnivariateCounterExampleInR}
Consider \aref{alg:UnivariateBroken}, which takes in $n \in \TN$, and outputs the result of \aref{alg:Real} at $2^{-n}$. The cost function of this algorithm is $g \in \rc{\TN}$ such that
\begin{eqs}
g(n) = n.
\end{eqs}
We can show that $g \in \pthetah{\TN}{n}$. Therefore, by calling an algorithm which is $\pthetah{\posi{\TR}}{0}$, we get an algorithm which is $\pthetah{\TN}{n}$. This shows that $\pohsy$ is not suitable for a definition of $\ohsy$-notation in algorithm analysis.
\end{example}

\begin{algorithm}
\caption{An algorithm which takes as input $n \in \TN$, and outputs $n$.}
\label{alg:AlmostIdentity}
\begin{algorithmic}[1]
\Procedure {almostIdentity}{$n$}
\State $j \coloneqq 0$
\If {$n = 4$}
\State \Return $4$
\EndIf
\For {$i \in \icon{0}{n}$}
\State $j \coloneqq j + 1$
\EndFor
\State \Return $j$
\EndProcedure
\end{algorithmic}
\end{algorithm}

\begin{example}[Almost identity on $\TN$]
Consider \aref{alg:AlmostIdentity}, which takes as input $n \in \TN$, and outputs $n$. The cost function of this algorithm is $f \in \rc{\TN}$ such that
\begin{eqs}
f(n) = 
\begin{cases}
n, & n \neq 4, \\
0, & n = 4.	
\end{cases}
\end{eqs}
We can show that $f \in \fthetah{\TN}{n}$ and $f \in \lsmalloh{\TN}{n}$.
\end{example}

\begin{algorithm}
\caption{An algorithm which takes as input $(m, n) \in \TN^2$ and outputs $n$.}
\label{alg:SecondComponent}
\begin{algorithmic}[1]
\Procedure {secondComponent}{$m, n$}
\State \Return \Call{almostIdentity}{$n$}
\EndProcedure
\end{algorithmic}
\end{algorithm}

\begin{example}[Calling an almost identity]
\label{ExtensibilityCounterExample}
Consider \aref{alg:SecondComponent}, which takes as input $(m, n) \in \TN^2$, and outputs $n$. The cost function of this algorithm is $g \in \rc{\TN^2}$ such that
\begin{eqs}
g(m, n) = 
\begin{cases}
n, & n \neq 4, \\
0, & n = 4.	
\end{cases}
\end{eqs}
We can show that $g \in \fsmalloh{\TN^2}{n}$. Therefore, by calling an algorithm which is $\fthetah{\TN}{n}$, we get an algorithm which is $\fsmalloh{\TN^2}{n}$. This shows that $\fohsy$ is not suitable for a definition of $\ohsy$-notation in algorithm analysis. 
\end{example}

\begin{note}[Blow-up and zeros]
There are two kinds of problems with those definitions of $\ohsy$-notations --- such as $\pohsy$, $\cohsy$, $\fohsy$ --- which ignore points from the input-set.

If the definition ignores an infinite set --- as in $\pohs{\TN^2}$, $\cohs{\TZ}$, and $\pohs{\posi{\TR}}$ --- then it is possible to construct a cost function to blow up in that ignored set, while not being reflected in the $\ohsy$-notation. This was demonstrated in \eref{MultivariateCounterExampleInN} for $\pohs{\TN^2}$, in \eref{UnivariateCounterExampleInZ} for $\cohs{\TZ}$, and in \eref{UnivariateCounterExampleInR} for $\pohs{\posi{\TR}}$. The last example demonstrates a blow-up in a bounded set.

If the definition ignores even a single point --- as in $\fohsy$ --- then it is possible to construct a cost function with a zero at that point, while not being reflected in the $\ohsy$-notation. This was demonstrated in \eref{ZeroCounterExample} for $\fohsy$.

Such definitions make it impossible to treat the cost functions as black boxes through their $\ohsy$-sets --- something which is essential for cost analysis based on $\ohsy$-notation.
\end{note}

% It can be shown that $\cohx{f} = \fohx{f}$, for all $f \in \rc{X}$ and $X \subset \TN^d$. It can also be shown that \property{ISubComp} holds for cofinite linear dominance $\fohsy$ in any set, and that \property{SubComp} holds for \emph{positive functions}. That is, $\fohx{f} \circ s \subset \foh{Y}{f \circ s}$, for all $f \in \rc{X}$ such that $f > 0$, where $X$ and $Y$ are arbitrary sets. Therefore, the following complexity analyses are guaranteed to be correct:
% \begin{itemize}
% \item Those which assume coasymptotic linear dominance $\cohsy$, remain in subsets of $\TN^d$, and assume positive cost functions.
% \item Those which assume cofinite linear dominance $\fohsy$, and assume positive cost functions.
% \end{itemize}

% To our knowledge, the first item covers the analyses in \cite{IntroAlgo2009}. An analysis of a graph algorithm uses the domain $\TN^2$ --- corresponding to a worst\-/case / best\-/case / average\-/case analysis under the number of vertices and the number of edges.  

\section{Some history of \texorpdfstring{$\ohsy$}{O}-notation}
\label{History}

The first reference to the $\ohsy$-notation seems to be that of Bachmann \cite[page 401]{BachmannOh}, in 1894, who gave a rather brief and informal definition of the $\ohsy$-notation:
\begin{quote}
... wenn wir durch das Zeichen $\oh{}{n}$ eine Gr\"osse ausdr\"ucken, deren Ordnung in Bezug auf $n$ die Ordnung von $n$ nicht \"uberschreitet; ...
\end{quote}
which we translate as: when we use the symbol $\oh{}{n}$ to represent some quantity, its order w.r.t. $n$ does not exceed the order of $n$. 

Landau \cite[page 31]{SmallOh}, in 1909, put this definition on a formal grounding by
\begin{equation}
\label{LandauAsymptotic}
f \in \oh{}{g} \coloniff \exists c \in \posi{\TR}, \exists N \in \TN, \forall n \geq N: f(n) \leq cg(n),
\end{equation}
for all $\function{f, g}{\TN}{\TR}$. This is asymptotic linear dominance on $\TN$. On page 883, Landau credits this definition to \cite{BachmannOh}. 

Both Bachmann and Landau were writing about analytic number theory, which is a branch of mathematics studying the properties of integers using tools from analysis. Correspondingly, the $\ohsy$-notation was constructed to address the needs in this field, such as to bound the error of truncating a series. Landau's definition was adopted by computer science, and remains the most wide-spread definition of $\ohsy$-notation on $\TN$ to this day.

% 147
In \cite{ArtOfProgrammingVol1Ed1}, in 1968, Knuth defined the $\ohsy$-notation as
\begin{equation}
f \in \oh{}{g} \coloniff \exists c \in \posi{\TR}: \abs{f} \leq c \abs{g},
\end{equation}
for all $\function{f, g}{S}{\TR}$, where $S = \TN$, or $S$ is an interval of $\TR$. Knuth credited the definition to Bachmann \cite{BachmannOh}. However, Knuth's definition probably contained an omission, since it does not correspond to Bachmann's definition, at least by Landau's interpretation of asymptotic linear dominance. This view is supported by that in \cite{ArtOfProgrammingVol1Ed2}, in 1973, the definition was replaced with Equation \ref{LandauAsymptotic}. To the best of my knowledge, in both editions of the book Knuth uses the $\ohsy$-notation solely to bound the error of truncating a series, and not for the analysis of algorithms. For analysis of algorithms, he deals with explicit bounds instead.

% Use of O-notation?
% NO: Optimum Binary search trees, 1970
% YES: On the Optimality of Some Set Algorithms, 1972
% YES: On Computing the Transitive Closure of a Relation, 1977

It is unclear to me when exactly the $\ohsy$-notation was first used in the analysis of algorithms. The earliest reference \we{} can find is \cite{EarlyONotation}, from 1972. The book \cite{DesignAndAnalysisOfComputerAlgorithms}, from 1974, uses the $\ohsy$-notation in a modern way in the analysis of algorithms, with the following definition on $\TN$:
\begin{quotation}
A function $g(n)$ is said to be $\oh{}{f(n)}$ if there exists a constant $c$ such that $g(n) < c f(n)$ for all but some finite (possibly empty) set of nonnegative values for $n$.
\end{quotation}
This is cofinite linear dominance on $\TN$. 

On \cite[page 39]{DesignAndAnalysisOfComputerAlgorithms}, an exercise encourages to study the properties of the following definition of $\ohsy$-notation on $\TN$:
\begin{equation}
f \in \oh{}{g} \coloniff \exists c \in \posi{\TR}: f \leq c g.
\end{equation}
This is linear dominance on $\TN$. 

On \cite[page 39]{DesignAndAnalysisOfComputerAlgorithms}, another exercise asks for the equivalence of cofinite affine dominance on $\TN$ and cofinite linear dominance on $\TN$ under certain assumptions, with \emph{cofinite affine dominance} on $\TN$ defined by
\begin{eqs}
f \in \oh{}{g} \coloniff \exists c, d \in \posi{\TR}: f \leq c g + d,
\end{eqs}
for all but a finite number of input-arguments. 

Prior to this \manuscript{}, there were no attempts at studying the $\ohsy$-notation systematically, in order to check whether the definition from analytic number theory was suitable for the analysis of algorithms. Indeed, it seemed as if the definition were simply a matter of taste; to quote \cite{BigOmega},

\begin{quotation}
On the basis of the issues discussed here, I propose that members of SIGACT, and editors of
computer science and mathematics journals, adopt the $\ohsy$, $\omegahsy$, and $\thetahsy$ notations as defined above, unless a better alternative can be found reasonably soon. 
\end{quotation}

The two exercises in \cite{DesignAndAnalysisOfComputerAlgorithms} show that the possibility of using a different definition of the $\ohsy$-notation was certainly noted, but that an argument to favor one over another was missing. 

\subsection{Related notations}

There are several other notations which resemble the way the $\ohsy$-notation works, which perhaps force strictness, reverse the ordering, or otherwise vary the ordering. In \cite[page 61]{SmallOh}, Landau defined the $\smallohsy$-notation as
\begin{eqs}
f \in \smalloh{}{g} \coloniff \lim_{x \to \infty} \frac{f(x)}{g(x)} = 0,
\end{eqs}
for all $\function{f, g}{\TR}{\TR}$. On page 883, Landau states that the $\smallohsy$-notation is his own.

Knuth \cite{BigOmega}, in 1976, defined the $\Omega$-notation as
$f \in \omegah{}{g} \coloniff g \in \oh{}{f}$,
the $\Theta$-notation as 
$f \in \thetah{}{g} \coloniff f \in \oh{}{g} \textrm{ and } g \in \oh{}{f}$,
and the $\omega$-notation as
$f \in \smallomega{}{g} \coloniff g \in \smalloh{}{f}$,
for all $\function{f, g}{\TN}{\TR}$. This was to remedy the occasional misuse where $\ohsy$ was used in place of the now-defined $\thetahsy$ or $\omegahsy$.

Vit\'anyi \cite{BigOmegaVsWild} argued that $\Omega$ should be defined asymmetrically by
\begin{equation}
f \in \omegah{}{g} \coloniff \exists c \in \posi{\TR}, \forall y \in \TN, \exists x \geq y: f(x) \geq cg(x).
\end{equation}

\section{Objective and contributions}

The objective of this \manuscript{} is to provide a rigorous mathematical foundation for the $\ohsy$-notation in algorithm analysis. In the following is a list of problems this \manuscript{} solves. 

\subsection{Problem 1: What is a cost function?}

A cost function of an algorithm has invariably been explained in books --- such as \cite{DesignAndAnalysisOfComputerAlgorithms} and \cite{IntroAlgo2009} --- as a function which maps the \emph{size of the input} to $\nonn{\TR}$. Unfortunately, the term \emph{size} has never been defined formally. Instead, the concept has been demonstrated by examples such as the length of a sequence, the number of bits in a number, or the pair $(n, m)$, where $n$ is the number of vertices in a graph, and $m$ is the number of edges in a graph. 

On Turing machines, input\-/size refers to the number of consecutive bits on the input\-/tape that need to be written in order to set the initial state. If the input is not in binary form already, the input has to be encoded as such, and the encoding specified. This definition of input\-/size \emph{does not} formalize the intuitive concept of input\-/size used in books. For example, an encoding of a graph could have an input\-/size proportional to $m + n$. However, there are graph algorithms whose cost\-/functions are only dependent on $n$, and graph algorithms whose cost\-/functions\footnote{e.g. $m^2 \lg{n + 1}$.} depend non\-/trivially on both $m$ and $n$. Such cost\-/functions cannot be captured as a function of $m + n$.

Here are some examples of the confusion this concept creates. The prototypical examples of \emph{size} are cardinality and volume, with an implied linear order. But how does one define a linear order \emph{sensibly} on, say, the pairs $(n, m)$? If that does not make sense, should the concept of size be extended to a partial order, or even a preorder? When the cost function is a function of input size, what does it mean to remap the input\-/sizes by function composition from the right (i.e. $f \circ s$)? When this \manuscript{} allows the domain of the cost function to be arbitrary, such as $\TZ$, $\TR$ or $\TC$, what exactly does it mean for an input to have size $-5$, $4.5$ or $2 + 4i$?

The present \manuscript{} resolves these questions as follows. The term \emph{input\-/size} does not make sense for an arbitrary model of computation. Instead, this \manuscript{} defines the cost function as a function which maps the input-set to $\nonn{\TR}$ --- the most detailed characterization of the cost function of an algorithm. Function composition from the right is used to access the input in a different order --- as when the algorithm is called as a sub-algorithm of another algorithm. Compatibility with this operation means that merely a different access-order cannot be made to show inconsistencies in the $\ohsy$-notation. 

The grouping of the input-set need only be done if the cost function by itself is too complex to make sense of. In such a case, we choose a grouping property --- such as the pair $(n, m)$ for graphs --- and partition the input-set based on that property. No additional properties --- such as an ordering --- are required from a grouping property.

\subsection{Problem 2: Why should the \texorpdfstring{$\ohsy$}{O}-notation be defined as it is?}

Introductory texts on algorithms, such as \cite{IntroAlgo}, traditionally explain the definition of the $\ohsy$-notation by the merits it has, such as abstracting out differences in the speeds of otherwise identical machines. The asymptotic part -- ignoring parts of the input --- is explained by being essential for the notation to concentrate on the scaling behavior on large inputs, and not on some possible artifacts on small inputs. 

Unfortunately, these explanations fail to pinpoint why we should not pick any other definition with seemingly the same properties, or whether it is necessary at all to ignore part of the input to gain sensitivity to the scaling behavior. Indeed, we show that, in general, no part of the input can be ignored without the notation failing in some way. The first hint towards this fact is seen by making an exhaustive list\footnote{An exhaustive list can be found in \tref{TableOfDesirableProperties}.} of all the rules we need the $\ohsy$-notation to possess; to express each one, none requires any restrictions --- such as an ordering or a sense of convergence --- on the input domain. 

The present \manuscript{} resolves this question as follows. First, the $\ohsy$-notation should provide us with the simplification we wanted. Second, it should impose an order on the functions. Third, it should be compatible with all those operations --- such as addition, multiplication, and scalar multiplication --- that are needed to combine the actual cost functions during explicit analysis. Specifically, the following procedures must yield the same answer:
\begin{itemize}
\item Apply a sequence of operations to cost functions, and then simplify the result with the $\ohsy$-notation.
\item Simplify every cost function with the $\ohsy$-notation, and then apply the same operations as above to the $\ohsy$-sets.
\end{itemize}

Perhaps the most important operation for algorithm analysis --- identified in this \manuscript{} --- is that of \emph{function composition} from the right, which corresponds to reordering or remapping the input. Compatibility with function composition corresponds to the requirement that the cost function of an algorithm must stay consistent when it is called with distorted input as a sub-algorithm of another algorithm. This is a strong requirement: it forces the definitions of $\ohsy$-notations on \emph{different domains} to be consistent with each other. When combined with the other requirements, this turns the definition from a matter of taste to a unique one.

\subsection{Problem 3: How should the multivariate $\ohsy$-notation be defined?}

When analyzing a graph algorithm, the input\-/set is often grouped by a two-dimensional property $(n, m)$, where $n$ is the number of vertices in the graph, and $m$ is the number of edges in the graph. To provide the complexity in the $\ohsy$-notation under such a grouping, the $\ohsy$-notation must also be defined in $\TN^2$.

Textbooks on algorithms approach the definition on $\TN^2$ in two ways. The first way --- exemplified in \cite[page 312]{DesignAndAnalysisOfComputerAlgorithms} --- is to define the $\ohsy$-notation only in $\TN$, and then silently continue using the notation over multiple variables too. That is, with no definition. The second way --- exemplified in \cite{IntroAlgo} and \cite{IntroAlgo2009} --- is to define the $\ohsy$-notation in $\TN$, and later provide some definition for a generalization to $\TN^2$, with the implication that the generalization does not bring anything new. Among all the books \we{} have looked at, \cite{IntroAlgo} and \cite{IntroAlgo2009} were the only ones to provide a definition for $\TN^2$ --- although different ones. The former provided asymptotic linear dominance, while the latter provided coasymptotic linear dominance.

The first documented symptom of a non\-/trivial generalization to $\TN^2$ was given by Howell \cite{OhImpossible}, who demonstrated a flaw\footnote{We will take a closer look at this in \sref{HowellCounter}} in asymptotic linear dominance on $\TN^2$. Perhaps it was for this issue, or a similar issue, that \cite{IntroAlgo2009} changed the definition to coasymptotic linear dominance. 

How can one be sure that \emph{this} definition as coasymptotic linear dominance has the desirable properties? In \cite{OhImpossible}, Howell thought he had shown that a definition with desirable properties is impossible in $\TN^d$. However, we will show in \sref{HowellCounter} that --- due to excessively strong assumptions --- Howell only showed that asymptotic linear dominance does not have the desirable properties. 

The present \manuscript{} resolves this question by showing that the only definition which has the desirable properties is linear dominance; the multivariate definition is merely an instance of the definition on $\TN^d$. 

\subsection{Problem 4: How are the desirable properties connected to each other?}

To make the theory in this \manuscript{} robust to variation, and to self-test it, this \manuscript{} studies the connections between the desirable properties of the $\ohsy$-notation. This study reveals that certain combinations of properties imply other properties, and that certain combinations of properties are equivalent to each other. As a side effect, this \manuscript{} does not merely study the $\ohsy$-notation in algorithm analysis, but also other possible or historical definitions of the $\ohsy$-notation. 

The study on the connections culminates in the discovery of the \nprim{} primitive properties, which imply all the desirable properties. Therefore, the question of whether the desirable properties are reasonable reduces to asking whether the primitive properties are reasonable.

\subsection{Problem 5: Do Master theorems hold for linear dominance?}

Master theorems are useful tools for analyzing the cost function of a recursive algorithm up to an $\ohsy$-equivalence. Therefore, it is important to make sure that the conclusions of Master theorems still hold after adopting linear dominance as the definition. 

The present \manuscript{} provides the necessary proofs in \sref{MasterTheorems}. In summary, Master theorems hold similarly as with asymptotic linear dominance, with the simplification that there is no need for so-called regularity conditions \cite{IntroAlgo2009}.

\subsection{Problem 6: How do existing definitions compare to each other?}

The present \manuscript{} provides in \sref{CandidateDefinitions} a comparison between various definitions of the $\ohsy$-notation, based on which primitive properties each fulfills. The study of how exactly the existing definitions fail the primitive properties is important, because a priori there is a risk that they may have produced incorrect analyses. 

Fortunately, coasymptotic linear dominance on $\TN^d$ --- on positive functions --- is equivalent to linear dominance on $\TN^d$. This covers most of the analyses that have been done in algorithm analysis. In practice, the saving grace is that the $\ohsy$-notation has been manipulated by assuming that the rules on $\TN$ work in general, and by implicitly assuming that \property{SubComp} holds. 

\subsection{Problem 7: How should the related notations be defined?}

There has been some confusion on how to define the related notations $\omegahsy$, and $\smallohsy$, as discussed in \sref{History}. 

The present \manuscript{} argues that these notations are merely different viewpoints of the same underlying order-relation between cost functions; they are fixed by the definition of the $\ohsy$-notation.

\subsection{Problem 8: Can existing techniques of analysis still be used?}

The present \manuscript{} shows in \sref{LimitTheorems} that local linear dominance --- a definition which covers most of the existing definitions  --- has an equivalent definition by limits. In addition, under suitable conditions, the other existing definitions are sometimes equivalent to linear dominance. Together, this allows to continue using existing analysis tools --- such as taking limits --- as before, provided one makes sure that there is a way to transfer the result to linear dominance.

\subsection{Problem 9: Point out misuses}

The present \manuscript{} points out in \sref{Misuses} some misuses which occur in actual publications in computer science. In particular, the $\ohsy$-notation is sometimes used as a general something-like operator to generalize statements, where it actually does not make sense. It is my hope that making such misuses explicit eventually improves communication between computer scientists.

\section{Relation to publications}

The present \manuscript{} extends a peer-reviewed publication in the Bulletin of EATCS \cite{ONotationBeatcs}, a doctoral thesis with the same name, as well as various non-peer-reviewed versions in Arxiv \cite{ONotationArxiv}. 

Most of the proofs have been deferred to appendices. This is to make the \manuscript{} readable to the largest possible audience, emphasizing the intuitive ideas. This does not reflect an ordering in importance --- to \me{}, the proofs are the most important part of this \manuscript{}.

The research and writing related to this \manuscript{}, and the above publications, were done solely by \kr; the co-authors checked.

\section{Automated checking}

This \manuscript{} was written in such a way that the dependencies between the theorems can be --- and have been --- checked by a machine. Each theorem in the \LaTeX{} source is annotated --- in a lightweight, but machine-readable manner --- with its assumed and implied properties. When the proof of Theorem A references Theorem B, it inherits the properties implied by B, from that point on, provided that the assumptions of B are satisfied at that point. For each theorem, the software checks that
\begin{itemize}
\item the assumptions of each referred theorem are satisfied, and that
\item there are no extraneous assumptions.
\end{itemize}
When a property is not proved by referring to a theorem, it is explicitly marked proved by the writer, based on the preceding non-machine-checkable proof. 

Such automatic checking was useful in the research phases to guarantee that theorems were not broken due to changing assumptions, or otherwise to point them out. The software for checking the dependencies, written in Python, can be obtained from 
\ifwe
KR's
\else
my
\fi
homepage.\footnote{http://kaba.hilvi.org}

\section{Outline of the \manuscript{}}

This \manuscript{} is organized as follows. \sref{Introduction} introduces the \manuscript{}, reviews the history of the topic, and lists the contributions of the \manuscript{}. \sref{Preliminaries} provides a more formal introduction to the concepts related to algorithms and their analysis. \sref{DesirableProperties} provides a list of desirable properties for an $\ohsy$-notation. \sref{Characterization} proves that the primitive properties are equivalent to the definition of the $\ohsy$-notation as linear dominance. \sref{WorkingWith} provides tools for working with the $\lohsy$-notation, including Master theorems and $\lohsy$-mapping rules. \sref{LocalLinearDominance} studies local linear dominance, a generalization which covers most of the existing definitions. \sref{Conclusion} concludes the \manuscript{}.

\sref{Notation} provides the notation. \sref{HowellCounter} reviews Howell's counterexample in more detail. \sref{ImpliedProperties} shows that the desirable properties reduce to a set of \nprim{} primitive properties which imply the other properties. \sref{ProofsForLocalLinearDominance} shows properties of local linear dominance. \sref{ProofsOfLimitsTheorems} shows that local linear dominance can be characterized by ratio\-/limits. \sref{MasterTheorems} shows that Master theorems work as before under linear dominance. \sref{CandidateDefinitions} compares several definitions for an $\ohsy$-notation, and shows how each of them fail the desirable properties. \sref{ProofsOfMinimality} provides additional definitions for an $\ohsy$-notation, with an aim to show the minimality of pre-primitive properties. \sref{PartitionedSets} provides some theory of partitioned sets. \sref{PreorderedSets} provides some theory of preordered sets.

\chapter{Preliminaries}
\label{Preliminaries}

In this chapter we provide a brief introduction to algorithms, the computational model, the cost-model, the primitive properties, and the $\ohsy$-notation.

\section{Algorithms}
\label{Algorithms}

What is an algorithm? We adopt an extremely liberal, but completely formalized view: an \emph{algorithm} is an abstract state machine \cite{SequentialAsm, ParallelAsm, ASMBook}. 

A variable in an abstract state machine $M$ is identified with a string, called a \emph{(function) symbol}. Each symbol has an \emph{arity} $n \in \TN$, which gives the number of arguments the symbol accepts as input. A \emph{(ground) term} is defined recursively as follows:
\begin{itemize}
\item a $0$-ary symbol is a term, and
\item if $f$ is an $n$-ary symbol, and $t_1, \dots, t_n$ are terms, then the string $f(t_1, \dots, t_n)$ is a term, and
\item there are no other terms.
\end{itemize}
The set of user-defined symbols, together with a small set of predefined symbols --- such as $\textrm{true}$, $\textrm{false}$, $\textrm{or}$, $\textrm{and}$, $\textrm{not}$, $=$, $\textrm{undef}$ --- is called the \emph{vocabulary} of the abstract state machine $M$.

Each $n$-ary symbol $f$ is associated with a function $\function{\inter{f}}{S^n}{S}$.\footnote{$S^0 = \setb{\emptyset}$.} The function $\inter{f}$ is an \emph{interpretation} of $f$. The set $S$ is the \emph{base-set}, which is common to all interpretations. The \emph{value} of a term $t$, denoted by $[t]$, is defined recursively as follows:
\begin{itemize}
\item if $x$ is a $0$-ary symbol, then $[x] = \inter{x}(\emptyset)$,
\item if $f$ is an $n$-ary symbol, for $n \in \posi{\TN}$, and $t_1, \dots, t_n$ are terms, then
\begin{equation*}
[f(t_1, \dots, t_n)] = \inter{f}([t_1], \dots, [t_n]).
\end{equation*}
\end{itemize}
In the following, by $\function{f}{X}{Y}$ --- where $f$ is a symbol, $X \subset S^n$, and $Y \subset S$ --- we mean that $\inter{f}(X) \subset Y$. In addition, if $x$ is a $0$-ary symbol, then by $x \in X$ we mean that $\function{x}{\setb{\emptyset}}{X}$ in the previous sense.

Compared to an ordinary programming language, a $0$-ary symbol corresponds to a variable, while an $n$-ary symbol, for $n \in \posi{\TN}$, corresponds to an $n$-dimensional array --- however, here the index can be an arbitrary set element. 

The abstract machine specifies how the interpretations of symbols are to be modified at each step. The program driving the abstract machine is a finite sequence of conditional assignments of the form 
\begin{algorithmic}
\If{$\textrm{condition}$}
\State $t_1 \coloneqq s_1$
\State $\quad \vdots$
\State $t_n \coloneqq s_n$
\EndIf
\end{algorithmic}
where $\textrm{condition}$, $t_1$, \dots, $t_n$, and $s_1$, \dots, $s_n$ are terms. The formula $t_i \coloneqq s_i$ can be thought of as copying an element from an array to another, $[t_i] \coloneqq [s_i]$, provided $[\textrm{condition}] = [\textrm{true}]$. This sequence of assignments is repeated until (possible) termination. 

All of the assignments in a single step are carried out in parallel --- not in sequence. For example, $t_1 \coloneqq s_1$ followed by $s_1 \coloneqq t_1$ causes $t_1$ and $s_1$ to swap values in the next step.

The basic definition of abstract state machines is both simple, and extremely general. The generality derives from the virtue of making the whole of set-theory available for modeling variables. Further abstraction-tools --- such as sequential composition, sub-machine calls, local variables, and return values --- are constructed over this basic definition. For example, Turbo-ASMs \cite{ASMBook} provide such features. From now on, we will assume that such abstraction tools have already been defined.

Consider \aref{alg:FindZero}, which is Newton's method for finding local zeros of differentiable functions.\footnote{This example generalizes an example from \cite{RealRam} where Newton's method is formalized as an algorithm for real rational functions under the real-RAM model.} The input-symbols to this algorithm are a $1$-ary continuously differentiable function $\function{f}{\TR}{\TR}$, a $0$-ary initial guess $x^* \in \TR$, and a $0$-ary error threshold $\epsilon \in \nonn{\TR}$; the input-set is $C_1(\TR \to \TR) \times \TR \times \nonn{\TR}$. The output --- provided the algorithm terminates --- is a $0$-ary point $x \in \TR$ such that $\abs{f(x)} \leq \epsilon$; the output-set is $\TR$. Other symbols are a $1$-ary symbol $\function{'}{C_1(\TR \to \TR)}{C_1(\TR \to \TR)}$ --- differentiation  --- , a $1$-ary symbol $\function{\abs{\; \cdot \;}}{\TR}{\nonn{\TR}}$ --- absolute value --- a $2$-ary symbol $> \; : \TR \times \TR \to \setb{0, 1}$ --- greater-than --- and $2$-ary symbols $\function{-}{\TR \times \TR}{\TR}$ and $\function{/}{\TR \times \nonn{\TR}}{\TR}$ --- subtraction and division. We have used infix notation for subtraction, division, and greater-than; postfix notation for differentiation; and midfix notation for the absolute value.

\begin{algorithm}
\caption{Newton's method for finding an element $x \in X$, such that $\abs{f(x)} \leq \epsilon$, for a continuously differentiable function $f \in C_1(\TR \to \TR)$.}
\label{alg:FindZero}
\begin{algorithmic}[1]
\Procedure {findZeroOrHang}{f, $x^*$, $\epsilon$}
\State $x \coloneqq x^*$
\While {$\abs{f(x)} > \epsilon$}
  \State $x \coloneqq x - f(x) / f'(x)$
\EndWhile
\State \Return $x$
\EndProcedure
\end{algorithmic}
\end{algorithm}

\aref{alg:FindZero} reads like pseudo-code, but is a completely formalized abstract state machine. With abstract state machines, the programmer is free to use the most fitting abstraction for the problem at hand. It should be clear how using such an abstract programming language enhances communication between software engineers and domain experts (e.g., physicists). 

Termination is not required for an algorithm; consider for example an operating system. Depending on the input, \aref{alg:FindZero} may terminate, or not. Suppose we require \aref{alg:FindZero} to always terminate. To satisfy this requirement, the programmer can either restrict the input-set to terminating inputs, or to modify the algorithm --- perhaps by limiting the number of iterations. From now on, we assume every analyzed algorithm to terminate.

A software development project utilizing abstract state machines starts by creating the most abstract description of the software as an abstract state machine, called the \emph{ground model} \cite{ASMBook}. This model captures the requirements, but does not provide any additional details on how the goals are to be attained. The project then proceeds to refine the model until it can be implemented in a concrete programming language. The model --- in all stages --- is used for verification and validation, and may even be used to generate code automatically.

This brief introduction to abstract state machines is to encourage the reader to look beyond the Church-Turing \manuscript{}, and to realize the usefulness of even the most abstract algorithms --- not just those computable algorithms which work with natural numbers. These are algorithms which take arbitrary sets as input, and produce arbitrary sets as output. To analyze such abstract algorithms, we need correspondingly abstract tools.

Here are some examples of input-sets that algorithms can have, or by which they can be modeled.

\begin{example}[Algorithms on integers]
An algorithm which takes integers as input could be modeled by $\TZ^d$.
\end{example}

\begin{example}[Algorithms on floating\-/point numbers]
An algorithm which takes floating\-/point numbers as input could be modeled by $\TR^d$.
\end{example}

\begin{example}[Algorithms on complex numbers]
An algorithm which takes complex numbers as input could be modeled by $\TC^d$.
\end{example}

\begin{example}[Algorithms on matrices]
An algorithm which takes a matrix as input could be modeled by $\TC^{m \times n}$.
\end{example}

\begin{example}[Algorithms on sequences]
An algorithm which takes a sequence in $X$ as input could be modeled by
\begin{eqs}
X^* = \bigcup_{d \in \TN} X^d.
\end{eqs}
\end{example}

\begin{example}[Algorithms on combinations of the above]
An algorithm which takes a sequence in $X$ and an integer as input could be modeled by $X^* \times \TZ$.
\end{example}

\begin{example}[Algorithms on graphs]
A graph algorithm is often analyzed under the domain $\TN \times \TN$. This domain comes from the worst\-/case analysis (say), where the input is grouped according to the number of vertices and edges.
\end{example}

\begin{example}[Algorithms in computational geometry]
An algorithm in computational geometry \cite{CGeometry} works directly with points in $\TR^d$, and may not mention floating\-/point numbers at all. The worst\-/case (or direct) analysis may require the domain $\TR^d$.
\end{example}

\begin{example}[Algorithms on different dimensions]
We may be interested in how seeing how a geometric algorithm in $\TR^d$ scales with respect to the dimension $d \in \TN$. In this case the input-set may consist of a union of data structures (e.g. range tree) for different dimensionalities. Dimensionality can then be used as a grouping property.
\end{example}

\begin{example}[Approximation schemes]
When considering an NP-hard problem, we may consider an approximation scheme \cite{IntroAlgo2009}, where $\epsilon \in \nonnb{\TR}{1}$ specifies the quality of the approximation; $\epsilon$ is specified as part of the input.
\end{example}

\section{Computational model and cost-model}
\label{Cost models}

Before an algorithm can be written, the writer must decide on the model of computation. A \define{model of computation} is a mathematical structure, and a set of atomic operations which manipulate that structure. Some models of computation are the Turing machine, the random-access machine (RAM) \cite{RamModel}, the real-RAM \cite{RealRam}, and the abstract state machine.

The result of complexity analysis --- for a given model of computation, a given algorithm, and a given resource --- is a function $\function{f}{X}{\nonn{\TR}}$, which provides for each \emph{input} of the algorithm a non-negative real number. This number tells how much of that resource the algorithm consumes with the given input. As discussed in \sref{Algorithms}, the input-set $X$ can be arbitrary.

Before a complexity analysis can be carried out, the analyst must decide on the cost\-/model. A \define{cost\-/model} specifies the amount of resources that running an atomic operation takes on a given input. A given computational model can assume different cost\-/models. When the cost\-/model is unspecified --- as it often is --- the cost of each atomic operation is assumed to be one unit irrespective of input. 

\begin{example}[Constant cost\-/models]
The most common cost\-/model is the unit\-/cost model, which counts the number of performed atomic operations. Zero costs can be used to concentrate the interest to specific resources, such as order\-/comparisons. 
\end{example}

\begin{example}[Non\-/constant cost\-/models]
An example of a non\-/constant cost\-/model is to assign the addition of natural numbers\footnote{Assuming the computational model supports such an operation.} a cost which is proportional to the logarithms of the arguments, so as to be proportional to the number of bits in their binary representations. 
\end{example}

\begin{example}[Cost\-/models for abstract state machines]
An abstract state machine specifies costs for reading or writing a memory location through a given symbol. Reading an addition symbol $+$ at $(x, y) \in \TN^2$ --- $x + y$ --- could be assigned a logarithmic cost as described above.
\end{example}

\section{Primitive properties}

Complexity analysis aims to classify and compare algorithms based on their resource consumptions; the less an algorithm uses resources to solve a given problem, the better it is compared to other algorithms which solve the same problem. The cost functions of the algorithms to solve a problem $\function{P}{X}{Y}$ are elements of $\rc{X}$.\footnote{$\rc{X} = X \to \nonn{\TR}$; see \sref{Notation}.} The most general way to compare them is to define a preorder\footnote{A preorder on a set $X$ is a reflexive and transitive relation $\relationin{\preleq}{X}$.} $\relationin{\domi}{\rc{X}}$. This relation should capture the intuitive concept of a cost function $f \in \rc{X}$ being either better than or equivalent to the cost function $g \in \rc{X}$ --- in some sense. For brevity, we use the term \emph{dominated by}.

\begin{note}[Not worse]
It is tempting to use the phrase \emph{not worse than} instead of \emph{better or equivalent}. However, the former means better, equivalent, \emph{or incomparable}, which is not what we want.
\end{note}

\begin{note}[Cost function of an algorithm]
Given an algorithm $A \colon \algof{X}{Y}$, we shall denote its cost function by $f_A \in \rc{X}$.
\end{note}

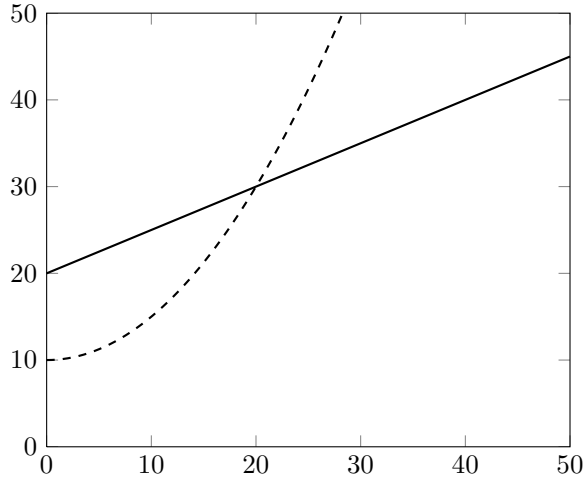
\begin{figure}
\center
\begin{tikzpicture}
\begin{axis}[xmin = 0, ymin = 0, xmax=50, ymax=50, samples=100]
  \addplot[dashed, thick, domain=0:50] (x, x * x / 20 + 10);
  \addplot[thick, domain=0:50] (x, x / 2 + 20);
\end{axis}
\end{tikzpicture}
\caption{Here $f_A, f_B \in \rc{\TN}$ are such that $f_A(n) = n / 2 + 20$ (solid line), and $f_B(n) = n^2 / 20 + 10$ (dashed line). Should $f_A \domi f_B$, $f_B \domi f_A$, $f_A \approx_X f_B$, or should $f_A$ and $f_B$ be incomparable?}
\label{AmbiguousOrder}
\end{figure}

\begin{example}[Essentially better?]
Consider Figure \ref{AmbiguousOrder}, where $f_A, f_B \in \rc{\TN}$ are such that $f_A(n) > f_B(n)$ for $n \in [0, 20)_{\TN}$, and $f_A(n) \leq f_B(n)$, for $n \in [20, \infty)_{\TN}$. We would then be inclined to say that $A$ has a better cost function, since $f_A(n)$ is small anyway on the finite interval $n \in [0, 20)_{\TN}$; $f_A$ is ``essentially'' better than $f_B$. How can this intuition be formalized? 
% There are many possibilities, some of which are:
% \begin{description}
% \item[local dominance] \hfill \\ there exists (an infinite interval / an infinite set / a cofinite set / a cobounded set / \dots) $A \subset X$, such that $\restrb{f_A}{A} \leq \restrb{f_B}{A}$,
% \item[linear dominance] \hfill \\ there exists $c \in \posi{\TR}$, such that $f_A \leq c f_B$,
% \item[translation dominance] \hfill \\ there exists $c \in \posi{\TR}$, such that $f_A \leq f_B + c$,
% \item[combinations of the above] \hfill \\ local linear dominance $\restrb{f_A}{A} \leq c \restrb{f_B}{A}$, affine dominance $f_A \leq c f_B + c$, and local affine dominance $\restrb{f_A}{A} \leq c \restrb{f_B}{A} + c$.
% \end{description}
% How do we choose between these?
\end{example}
To decide between various definitions, we reflect on the fundamental properties that the analyst needs to complete his/her complexity analysis. An obvious \define{dominance property} is 
\begin{description}
\item[\property{Order}] \hfill \\ 
if an algorithm $A$ never uses more resources than algorithm $B$, then $f_A$ is better than or equivalent to $f_B$:
\begin{eqs}
f_A \leq f_B \implies f_A \domi f_B.
\end{eqs}
\end{description}
\begin{note}[Reflexivity]
In particular, \property{Order} implies that $\domi$ is reflexive: $f \domi f$. 
\end{note}
Since we want $\domi$ to be a preorder, we also need
\begin{description}
\item[\property{Trans}] \hfill \\ 
if $f_A$ is better than or equivalent to $f_B$, and $f_B$ is better than or equivalent to $f_C$, then $f_A$ is better than or equivalent to $f_C$:
\begin{eqs}
\bra{f_A \domi f_B \land f_B \domi f_C} \implies f_A \domi f_C.
\end{eqs}
\end{description}
In addition to being a preorder, the dominance relation must preserve the \emph{structure} present in algorithms: conditional branching, calls with transformed input, and looping.

\begin{algorithm}
\caption{An algorithm to demonstrate \property{Local}. Improving sub-algorithm $F$ for even integers improves the whole algorithm. We assume resources are only spent in $F$.}
\label{alg:locality}
\begin{algorithmic}[1]
\Procedure {branching}{$n$}
\If{$\text{mod}(n, 2) = 0$}
\State \Return{$F(n)$}
\EndIf
\State \Return{$0$}
\EndProcedure
\end{algorithmic}
\end{algorithm}

\begin{description}
\item[\property{Local}] \hfill \\ 
substituting a sub-algorithm \emph{for a subset of the input-set} with a better or equivalent sub-algorithm results in a better or equivalent algorithm overall:
\begin{eqs}
\bra{\forall i \in [1, n]_{\TN}: \restrb{f}{A_i} \preceq_{A_i} \restrb{g}{A_i}} \implies f \domi g,
\end{eqs}
for every finite cover $A_1, \dots, A_n \subset X$. \aref{alg:locality} demonstrates \property{Local}.
\end{description}

\begin{algorithm}
\caption{An algorithm to demonstrate \property{SubComp}. Improving sub-algorithm $F$ improves the whole algorithm. We assume resources are only spent in $F$.}
\label{alg:subcomp}
\begin{algorithmic}[1]
\Procedure {calling}{$n$}
\State \Return{\Call{F}{$s(n)$}}
\EndProcedure
\end{algorithmic}
\end{algorithm}

\begin{description}
\item[\property{SubComp}] \hfill \\ 
substituting a sub-algorithm \emph{called with transformed input} with a better or equivalent sub-algorithm results in a better or equivalent algorithm overall:
\begin{eqs}
f_A \preleq f_B \implies f_A \circ s \preleqb f_B \circ s,
\end{eqs}
for all $\function{s}{Y}{X}$. \aref{alg:subcomp} demonstrates \property{SubComp}.
\end{description}

\begin{algorithm}
\caption{An algorithm to demonstrate \property{NSubHom}. Improving sub-algorithm $F$ improves the whole algorithm. We assume resources are only spent in $F$.}
\label{alg:nsubhom}
\begin{algorithmic}[1]
\Procedure {looping}{$n$}
\For {$i \in \icon{0}{u(n)}$}
  \State \Call{F}{$n$}
\EndFor
\EndProcedure
\end{algorithmic}
\end{algorithm}

\begin{description}
\item[\property{NSubHom}] \hfill \\ 
substituting a sub-algorithm \emph{in a loop} with a better or equivalent sub-algorithm results in a better or equivalent algorithm overall:
\begin{eqs}
f_A \preleq f_B \implies uf_A \preleq uf_B,
\end{eqs}
for all $u \in \rc{X}$ such that $\image{u}{X} \subset \TN$. \aref{alg:nsubhom} demonstrates \property{NSubHom}.
\end{description}

\begin{description}
\item[\property{NSubDiv}] \hfill \\ 
if an algorithm that loops sub-algorithm A is better than or equivalent to an algorithm that loops sub-algorithm B, then it is because $f_A$ is better than or equivalent to $f_B$:
\begin{eqs}
f_A \preleq f_B \implies uf_A \preleq uf_B,
\end{eqs}
for all $u \in \rc{X}$ such that $\image{u}{X} \subset 1/\posi{\TN}$.
\end{description}

\begin{note}[Many preorders]
These properties reveal that the preorders in different sets cannot be defined independently of each other; they are tightly connected by \property{Local} and \property{SubComp}. Rather, the problem is to find a consistent class of preorders $\setb{\domi : X \in U}$, where $U$ is a given universe.
\end{note}

\begin{note}[Trivial dominance]
A problem with the listed properties thus far is that the trivial preorder $\domis = \rc{X} \times \rc{X}$, for all $X \in U$, fulfills them all; then the functions in $\rc{X}$ are all equivalent. Therefore, for the comparison to be useful, we need to require $\domi$ to distinguish at least some functions in $\rc{X}$, at least for some $X \in U$. 
\end{note}
The primitive \define{non\-/triviality property} is:
\begin{description}
\item[\property{One}] \hfill \\ $n \not\preceq_{\posi{\TN}} 1$.
\end{description}
Note that $1 \preceq_{\posi{\TN}} n$ already holds by \property{Order}; \property{One} prevents the order from collapsing due to equivalence.

Finally, there is the question of robustness. Writing an algorithm is an iterative process, which causes the cost function of an algorithm to change constantly. If every change is reflected in the ordering, then each change invalidates an existing complexity analysis of the algorithm. Worse, the change invalidates all the analyses of the algorithms which use the changed algorithm as a sub-algorithm. Since an algorithm may face hundreds or thousands of changes before stabilising into something usable in practice, such an approach is infeasible. Therefore, the ordering needs to introduce identification, to make it robust against small changes in the algorithms. But what is a small change?

The key realization is that for a given algorithm $F \in \algo{P}$ it is often easy, with small changes, to produce an algorithm $G \in \algo{P}$, for which the improvement\-/ratio\footnote{Assuming $f_F > 0$.} $f_F/f_G$ stays bounded.\footnote{For example, change to use Single Instruction Multiple Data (SIMD) instructions.} In contrast, obtaining an unbounded improvement\-/ratio often requires considerable insight and fundamental changes to the way an algorithm works --- a new way of structuring data and making use of it. This definition of interesting provides the desired robustness against small changes in algorithms. 

The primitive \define{abstraction property} is:
\begin{description}
\item[\property{Scale}] \hfill \\ if $f_A \domi f_B$, then $f_A \domi \alpha f_B$, for all $\alpha \in \posi{\TR}$.
\end{description}

\begin{note}[Dominance is linear dominance]
We will show in \sref{Characterization} that the \nprim{} primitive properties are equivalent to the definition of dominance as linear dominance:
\begin{eqs}
f \preleq g \iff \exists c \in \posi{\TR}: f \leq cg.
\end{eqs}
\end{note}

\begin{note}[Generators and propagators]
When the primitive properties are used as axioms, every proof of $f \preceq_X g$ must appeal to \property{Order}, and every proof of  $f \not\preceq_X g$ must appeal to \property{One}. All the other primitive properties propagate these results.
\end{note}

\begin{definition}[Related relations]
A preorder $\preleq$ induces the following related relations in $\rc{X}$:
\begin{eqs}
f \approx_X g & \iff f \preleq g \land g \preleq f, \\
f \prec_X g & \iff f \preleq g \land f \not\approx_X g, \\
f \succ_X g & \iff g \prec_X f, \\
f \pregeq g & \iff g \preleq f.
\end{eqs}
\end{definition}

\begin{note}[Strict comparison is weaker]
\label{StrictComparisonIsWeaker}
In a preorder which is not a partial order, the $\prec_X$ alone cannot be used to deduce whether $f \approx_X g$; it is a weaker concept than $\preleq$. However, providing $\prec_X$ and $\approx_X$ suffices. This is in contrast to a partial order, where $\approx_X$ is known to be the set\-/equality in $X$.
\end{note}

\section{Problem complexity}

Apart from analyzing the resource\-/cost of a specific algorithm, complexity analysts also have a greater underlying goal: that of analyzing the resource\-/cost of the underlying problem $\function{P}{X}{Y}$ itself --- for a given computational model, a given cost\-/model, and a given resource. This activity divides into two sub-activities; finding a lower\-/bound and an upper\-/bound for the cost function of $P$.

\begin{definition}[Lower bound]
A cost function $f \in \rc{X}$ is a \define{lower\-/bound} for a problem $\function{P}{X}{Y}$, if $f \domi f_A$, for all $A \in \algo{P}$. 
\end{definition}

\begin{definition}[Optimal cost function]
A cost function $f \in \rc{X} $ is \defineexp{optimal}{optimal!cost function} for a problem $\function{P}{X}{Y}$, if $f$ is a lower\-/bound for $P$, and $g \domi f$ for each $g \in \rc{X}$ a lower\-/bound for $P$.\footnote{That is, $f$ is a greatest lower\-/bound.} 
\end{definition}

\begin{note}[Optimal cost function may not exist]
An optimal cost function may not exist for a problem, as shown in \thref{DominanceIsIncomplete}. When it exists, the optimal cost function for $P$,
\begin{eqs}
\inf_{\domi} \setb{f_A : A \in \algo{P}},
\end{eqs}
is unique up to equivalence. 
\end{note}

\begin{definition}[Optimal algorithm]
An algorithm $A \in \algo{P}$ is \defineexp{optimal}{optimal!algorithm}, if $f_A$ is optimal for $P$.
\end{definition}

\begin{note}[Optimal algorithm may not exist]
An optimal algorithm may not exist for a given problem.
\end{note}

\begin{definition}[Upper bound]
A cost function $f \in \rc{X}$ is an \define{upper\-/bound} for a problem $\function{P}{X}{Y}$, if $f_A \domi f$, for some $A \in \algo{P}$. 
\end{definition}

An upper\-/bound is found by finding an actual algorithm for solving $P$. While there are computational problems which cannot be solved at all,\footnote{e.g., the halting problem under Turing machines.} establishing at least one upper\-/bound for a solvable problem is often easy. These are the \define{brute\-/force} algorithms, which compute or check everything without making any use of the underlying structure in the problem. 

\section{\texorpdfstring{$\ohsy$}{O}-notation}

\begin{definition}[$\ohsy$-notation]
An \define{$\ohsy$-notation} over the universe $U$ is a class of functions
\begin{equation}
\ohsy \coloneqq \setb{\function{\ohs{X}}{\rc{X}}{\power{\rc{X}}} : X \in U},
\end{equation}
where
\begin{eqs}
\ohx{f} = \setb{g \in \rc{X} : g \domi f}.
\end{eqs}
\end{definition}

\begin{note}[Different symbols]
The $\ohsy$ is used as a generic symbol for studying how the different desirable properties of $\ohsy$ interact with each other. We will use drawings inside the $\ohsy$ symbol for specific versions of the $\ohsy$-notation, such as $\pohsy$ for asymptotic linear dominance. 
\end{note}

\begin{note}[Different viewpoints]
It is equivalent to define either the functions $\setb{\ohs{X} : X \in U}$, or the preorders $\setb{\domi : X \in U}$; one can be recovered from the other. We shall give the theorems in terms of $\setb{\ohs{X} : X \in U}$, since this is more familiar to computer scientists. However, \we{} find that intuition works better when working with $\setb{\domi : X \in U}$. 
\end{note}

\begin{note}[Related notations]
The related notations are given in terms of the dominance relation as:
\begin{eqs}
\smallohx{f} & = \setb{g \in \rc{X} : g \prele f}, \\
\omegahx{f} & = \setb{g \in \rc{X} : g \pregeq f}, \\
\smallomegax{f} & = \setb{g \in \rc{X} : g \prege f}, \\
\thetahx{f} & = \setb{g \in \rc{X} : g \approx_X f}.
\end{eqs}
\end{note}

The $\ohsy$-notation extends naturally to sets of functions; such generality is sometimes needed.

\begin{definition}[$\setohsy$-notation]
The \define{$\setohsy$-notation} on a set $X$ is a function $\function{\setohs{X}}{\power{\rc{X}}}{\power{\rc{X}}}$ such that
\begin{eqs}
\setohx{A} = \bigcup_{f \in A} \ohx{f}.
\end{eqs}
\end{definition}

\section{Implicit conventions}
\label{ImplicitConventions}

An \define{implicit convention} is an overload of notation adopted by people working in a given field. Since it is an overload, the reader is required to deduce the correct meaning of such notation from the context. Here are some implicit conventions related to the $\ohsy$-notation --- as commonly used in computer science.

\subsection{Placeholder convention}

\begin{definition}[Placeholder convention]
The \define{placeholder convention} is to use $A \subset \power{\rc{X}}$ as a placeholder for an anonymous function $g \in \setoh{X}{A}$. It then must be guessed from the context whether the author means by $A$ the actual set, or the anonymous function $g$. 
\end{definition}

\begin{example}[Cost of an algorithm]
An algorithm costs $\oh{\TN}{n^2}$ if its cost function is an element of $\setoh{\TN}{\oh{\TN}{n^2}} = \oh{\TN}{n^2}$.
\end{example}

\begin{example}[Exponential of an $\ohsy$-set]
Consider the statement that an algorithm costs $2^{\oh{\TN}{n}}$. As a set, $0.5 \not\in 2^{\oh{\TN}{n}}$. However,
\begin{eqs}
0.5 \in \setoh{\TN}{2^{\oh{\TN}{n}}},
\end{eqs}
since $0 \in \oh{\TN}{n}$ and $0.5 \in \oh{\TN}{2^0}$. Similarly,
\begin{eqs}
2^{2.5n} \in \setoh{\TN}{2^{\oh{\TN}{n}}},
\end{eqs}
since $2.5n \in \oh{\TN}{n}$. 
\end{example}

% \subsection{Evaluation convention}

% \begin{note}[Evaluation convention]
% The \define{evaluation convention} is to first apply the placeholder convention, and then to substitute the variables in the anonymous function with identically named variables in the outer scope. 
% \end{note}

% \begin{example}[Sum with variable interpretation]
% Suppose $n$ is a variable of the $\ohsy$-expression in $\sum_{n = 1}^m \oh{\TN}{n}$, and we do not apply the evaluation convention. Then the sum-index $n$ is independent of the variable $n$ in the expression, and the sum is $\oh{\TN}{mn} = \oh{\TN}{n}$ by \property{Additive} and \property{Scale}.
% \end{example}

% \begin{example}[Sum with constant intrepretation]
% Suppose $n$ is not a variable of the $\ohsy$-expression in $\sum_{n = 1}^m \oh{\TN}{n}$. Then $\oh{\TN}{n} = \oh{\TN}{1}$ by \property{Scale}, and the sum is $\oh{\TN}{m} = \oh{\TN}{1}$ by \property{Additive} and \property{Scale}.
% \end{example}

% \begin{example}[Sum with evaluation convention]
% Suppose $n$ is a variable of the $\ohsy$-expression in $\sum_{n = 1}^m \oh{\TN}{n}$, and we apply the evaluation convention. Then the sum is given by
% \begin{eqs}
% \sum_{n = 1}^m f_n(n),
% \end{eqs}
% where $f_i \in \oh{\TN}{n}$, for all $i \in \posi{\TN}$. By \property{SubsetSum},
% \begin{eqs}
% \oh{\TN}{\sum_{n = 1}^m f_n(n)} & \subset \oh{\TN}{\sum_{n = 1}^m n} \\
% {} & = \oh{\TN}{(m + 1) m / 2}.
% \end{eqs}
% \end{example}

\subsection{Domain convention}

\begin{definition}[Domain convention]
The \define{domain convention} is to leave off the domain of the $\ohsy$-notation, say $\oh{}{n^2}$, and then let the reader guess, for each use, the domain from the context. Sometimes the domain convention leads to a difficult interpretation. 
\end{definition}

\begin{example}[Difficult interpretation]
Consider an algorithm \cite{LocalSearch} under the unit\-/cost $w$-bit RAM model, where $w \in \TN$, which for $I \subset [0, 2^w)_{\TN} = U$ finds a nearest neighbor of $i \in U$ in $I$ in time
\begin{eqs}
\oh{}{\lg{\lg{\Delta + 4}}}, 
\end{eqs}
where $\Delta \in \TN$ is the distance between $i$ and its nearest neighbor in $I$. Intuitively, this sounds reasonable, but what is the domain? \Our{} thinking process went as follows.

Since the expression contains only a single symbol $\Delta$, \we{} assumed it to be a univariate $\ohsy$\-/notation. \Our{} first guess was $\oh{U}{\lg{\lg{\Delta + 4}}}$ --- with $w$ fixed. However, since $U$ is bounded, this is equal to $\oh{U}{1}$. The guess had to be wrong; otherwise the authors would have reported the complexity as $\oh{}{1}$. 

\Our{} second guess was $\oh{\TN}{\lg{\lg{\Delta + 4}}}$ --- again with $w$ fixed. However, this arbitrarily extends the complexity analysis to elements outside the input-set, since $\Delta < 2^w$. In addition, it is not always possible to do such an extension, such as when the function is $\lg{\lg{2^w - \Delta + 4}}$ instead. 

Finally, \we{} observed that the complexity depends both on $w$ and $\Delta$ --- although $\oh{}{\lg{\lg{\Delta + 4}}}$ never mentions $w$. The correct formalization is given by $\oh{D}{\lg{\lg{\Delta + 4}}}$, where $D = \setb{(w, \Delta) \in \TN^2 : \Delta \in [0, 2^w)}$. The corresponding algorithm would then take as input
\begin{eqs}
w & \in \posi{\TN}, \\
I & \in \cup_{k \in \posi{\TN}} \power{[0, 2^k)_{\TN}}, \\
i & \in \cup_{k \in \posi{\TN}} [0, 2^k),
\end{eqs}
subject to $I \subset [0, 2^w)_{\TN}$ and $i \in [0, 2^w)$.
\end{example}

\section{Worst case, best case, average case}
\label{WorstCaseAndOthers}

In this section we will formalize the concepts of worst\-/case, best\-/case, and average\-/case analyses.

\begin{definition}[Grouping]
A \define{grouping} of $X$ is a function $\function{g}{X}{Z}$. 
\end{definition}

\begin{definition}[Case]
A \define{case} over a grouping $\function{g}{X}{Z}$ is a function $\function{s}{\image{g}{X}}{X}$ such that $g \circ s = \iden{Z}$; a right inverse of $g$.
\end{definition}

\begin{definition}[Worst case]
A case $\function{s}{\image{g}{X}}{X}$ over a grouping $\function{g}{X}{Z}$ is called \emph{worst} of $f \in \rc{X}$, if
\begin{equation}
(f \circ s)(z) = \sup \image{f}{\preimage{g}{\setb{z}}},
\end{equation}
for all $z \in \image{g}{X}$.
\end{definition}

\begin{note}
A worst case may not exist.
\end{note}

\begin{definition}[Worst-case analysis]
A \define{worst-case analysis} of $f \in \rc{X}$ over a grouping $\function{g}{X}{Z}$ is the process of finding out
\begin{eqs}
\sup \image{f}{\preimage{g}{\setb{z}}}
\end{eqs}
or some $\ohsy$-set which contains it. 
\end{definition}

\begin{definition}[Best case]
A case $\function{s}{\image{g}{X}}{X}$ over a grouping $\function{g}{X}{Z}$ is called \emph{best} of $f \in \rc{X}$ if
\begin{equation}
(f \circ s)(z) = \inf \image{f}{\preimage{g}{\setb{z}}},
\end{equation}
for all $z \in \image{g}{X}$. 
\end{definition}

\begin{note}
A best case may not exist.
\end{note}

\begin{definition}[Best-case analysis]
A \define{best-case analysis} of $f \in \rc{X}$ over a grouping $\function{g}{X}{Z}$ is the process of finding out
\begin{eqs}
\inf \image{f}{\preimage{g}{\setb{z}}}
\end{eqs}
or some $\omegahsy$-set which contains it.
\end{definition}

\begin{example}[Analysis of insertion sort]
Assume the RAM model, with unit cost for comparison of integers and zero cost for other atomic operations. Let $\TN^* = \bigcup_{d \in \TN} \TN^d$ be the set of all finite sequences over $\TN$. Let $F \colon \algof{\TN^*}{\TN^*}$ be the insertion sort algorithm \cite{IntroAlgo2009}, which sorts a given input sequence $x$ into increasing order. Let $f \in \rc{\TN^*}$ be the number of comparisons made by $F$. Let $\function{g}{\TN^*}{\TN}$ be such that $g(x) = |x|$, the length of the sequence $x$. Let $\function{s}{\TN}{\TN^*}$ be the worst case of $f$ over $g$; each such sequence is decreasing. Then the worst-case complexity of $f$ over $g$ is $f \circ s$, and the worst-case analysis of $f$ provides $\oh{\TN}{f \circ s} = \oh{\TN}{n^2 + 1}$. Let $\function{r}{\TN}{\TN^*}$ be the best case of $f$ over $g$; each such sequence is increasing. Then the best-case complexity of $f$ over $g$ is $f \circ r$, and the best-case analysis of $f$ provides $\omegah{\TN}{f \circ r} = \omegah{\TN}{n + 1}$. For an arbitrary case $\function{p}{\TN}{\TN^*}$ of $f$ over $g$, it holds that $(f \circ p) \in \omegah{\TN}{n + 1} \cap \oh{\TN}{n^2 + 1}$.
\end{example}

\begin{definition}[Average-case analysis]
Let $(X, \Sigma_X, \TP)$ be a probability space, and $(Z, \Sigma_Z)$ be a measurable space. Let $\function{g}{X}{Z}$ be a random element, and $f \in \rc{X}$ be a random variable. An \define{average-case analysis of $f$ over $g$} is the process of finding out $\oh{Z}{\mathbb{E}\brac{f \; | \; g} \circ s}$, where $\function{s}{\image{g}{X}}{X}$ is any case over $g$, and $\mathbb{E}$ stands for (conditional) expectation.
\end{definition}

Since worst-case, best-case, and average-case analyses are the most common forms of complexity analysis in computer science --- with the grouping set almost always $Z \subset \TN^d$, for some $d \in \posi{\TN}$ --- this has led to the often repeated claim that the result of complexity analysis is a function which maps an `input size' to the amount of used resources. For example, \cite[page 25]{IntroAlgo2009} writes as follows (emphasis theirs):
\begin{quotation}
\noindent
The best notion for \emph{input size} depends on the problem being studied. For many problems, such as sorting or computing discrete Fourier transforms, the most natural measure is the \emph{number of items in the input} - for example, the array size $n$ for sorting. For many other problems, such as multiplying two integers, the best measure of input size is the \emph{total number of bits} needed to represent the input in ordinary binary notation. Sometimes, it is more appropriate to describe the size of the input with two numbers rather than one. For instance, if the input to an algorithm is a graph, the input size can be described by the numbers of vertices and edges in the graph. We shall indicate which input size measure is being used with each problem we study.
\end{quotation}

Complexity theorists sometimes study the cost functions of Turing machines with respect to input-tape-size. However, this is not a formalization of the input\-/size as described in the above quotation. 

The term input\-/size --- as it has been used --- is a synonym for a group label. A group label need not have any properties, such as an order. A grouping need only be done if the analysis or its interpretation otherwise seems difficult. It would seem clearer to use the term input\-/size only in those cases where the group labels form a set which is linearly ordered and contains a least element (e.g. $\nonn{\TN}$, $\nonn{\TR}$, or cardinal numbers).

We have shown above how input\-/size-thinking is subsumed by the more general input-set-thinking. In the input-set thinking, a set is used to provide a mathematical model for a data structure, and a cost function is a function of this data. 

\section{Misuses}
\label{Misuses}

The $\ohsy$-notation (and related definitions) is sometimes misused even by experienced researchers in computer science. By a \emph{misuse} we mean to use the notation in a context which does not have a formal meaning --- even after applying the implicit conventions specific to computer science. 

In the following we review some misuses made by experienced researchers in computer science. Let us note that, despite the misuses, the sources we refer to here are both great reading: \cite{IntroAlgo2009} is a classic book about algorithms, data-structures, and their analysis, while \cite{ExponentialTrees} provides an ingenious scheme for converting a static dictionary to a linear-space dynamic dictionary.

\begin{example}[Recurrence equations]
Reference \cite[page 102]{IntroAlgo2009} studies the solutions to the recurrence equation
\begin{equation}
T(n) = 
\begin{cases}
a T(n / b) + F(n), & n > 1, \\
\thetah{}{1}, & n = 1, 
\end{cases}
\end{equation}
where $a \in \nonnb{\TR}{1}$, $b \in \posib{\TR}{1}$, $B = \setb{b^i : i \in \TN}$, and $T, F \in \rc{B}$. The intent here is to study the Master theorem over powers --- as we do in \sref{MasterTheoremOverPowersSection}. Unfortunately, since $\thetah{}{1} \coloneqq \thetah{X}{1} \subset \rc{X}$ is a set of functions, and not a non-negative real number, this equation does not have a formal meaning. In addition, the set $X$ is left undefined. 

Applying the placeholder convention leads to
\begin{equation}
T(n) = 
\begin{cases}
a T(n / b) + F(n), & n > 1, \\
\hat{1}, & n = 1,
\end{cases}
\end{equation}
where $\hat{1} \in \thetah{X}{1}$. This still does not make sense; $\hat{1}$ is a function, not a non-negative real number.

The motivation for this recurrence equation is the analysis of divide\-/and\-/conquer algorithms, where the algorithm recursively solves $b$ -- perhaps overlapping --- sub-problems (when $b$ is an integer), and then combines their results to solve the original problem. The call-graph of such an algorithm is a tree, and in each leaf of this tree we would perhaps like to assign a \emph{different} constant for the amount of resources it takes. The problem is that the function $T$ can have only one value for $T(1)$.

One way to fix the recurrence equation is to simplify it to
\begin{equation}
T(n) = 
\begin{cases}
a T(n / b) + F(n), & n > 1, \\
d, & n = 1,
\end{cases}
\end{equation}
where $d \in \posi{\TR}$, which is the form we study in \sref{MasterTheorems}. We will show that $\loh{B}{T}$ is independent of the choice of $d$. Therefore, the simplified recurrence equation provides a solution for the divide\-/and\-/conquer analysis even when the costs in the leaf nodes vary in a fixed closed interval.\footnote{Assuming the interval does not contain zero.} \We{} believe this is the idea that \cite{IntroAlgo2009} was aiming for; it just is not captured by replacing $d$ with $\thetah{}{1}$.
\end{example}

\begin{example}[Conditional statements]
Let $k \in \posi{\TN}$, and $S, T \in \rc{\posi{\TN}}$. In \cite[page 9]{ExponentialTrees}, there is the following statement:
\begin{quote}
Then for $n = \smallomega{}{1}$, we have
\begin{equation*}
n \geq \oh{}{n^{1/k}} \geq \oh{}{n^{(1 - 1/k)/k}},
\end{equation*}
and
\begin{equation*}
n^{1-1/k} \geq \oh{}{n^{(1-1/k)^2}}.
\end{equation*}
For $n = \oh{}{1}$, we trivially have
\begin{equation*}
T(n) = \oh{}{1} = \oh{}{S(n)}.	
\end{equation*}
\end{quote}
We decode this as follows:
\begin{quote}
Then for $n \in \smallomega{\posi{\TN}}{1}$, we have
\begin{equation*}
\oh{\posi{\TN}}{n} \supset \oh{\posi{\TN}}{n^{1/k}} \supset \oh{\posi{\TN}}{n^{(1 - 1/k)/k}},
\end{equation*}
and
\begin{equation*}
\oh{\posi{\TN}}{n^{1-1/k}} \supset \oh{\posi{\TN}}{n^{(1-1/k)^2}}.
\end{equation*}
For $n \in \oh{\posi{\TN}}{1}$, we trivially have
\begin{equation*}
T \in \oh{\posi{\TN}}{1} = \oh{\posi{\TN}}{S}.	
\end{equation*}
\end{quote}
The expression $n \in \smallomega{\posi{\TN}}{1}$ is always true for $\lohsy$, $\pohsy$, $\cohsy$ and $\fohsy$. Similarly, the expression $n \in \oh{\posi{\TN}}{1}$ is always false. The authors have an intuitive concept which they want to transmit. However, the formalization of the intuition is incorrect, and so the communication fails. 
\end{example}

\section{Completeness}

In this section we study the completeness of the dominance relation $\domi$.

\begin{definition}[Completeness over a family]
A preorder $\relationin{\preleq}{\rc{X}}$ is \define{complete} over $A \subset \power{\rc{X}}$, if every $F \in A$ which has a lower\-/bound (an upper\-/bound) in $\rc{X}$ has a greatest lower\-/bound (a least upper\-/bound) in $\rc{X}$. 
\end{definition}

\begin{table}
\centering
\begin{tabular}{|l|l|}
\hline
Concept & Complete over... \\
\hline
Complete & $\power{\rc{X}}$ \\
Directed-complete & directed subsets of $\rc{X}$ \\
Chain-complete & linearly\-/ordered subsets of $\rc{X}$ \\
Lattice & $\fpower{\rc{X}}$ \\
Algorithm-complete & $\setb{\setb{f_F : F \in \algo{P}} : P \in \bra{X \to P(X)}}$ \\
\hline
\end{tabular}
\caption{Different types of completeness.}
\label{DifferentTypesOfCompleteness}
\end{table}

\begin{note}
Different types of completeness are listed in \tref{DifferentTypesOfCompleteness}.
\end{note}

\begin{proposition}[\uproperty{Lattice} is implied]
\label{DominanceRelationIsALattice}
$\domi$ has \require{Order}, \require{Local}, and \require{ISubComp}. $\implies$ $\domi$ has \prove{Lattice}. 
\end{proposition}

\begin{proof}
Let $f_1, \dots, f_n \in \rc{X}$. Then $f_1, \dots, f_n \leq \max(f_1, \dots, f_n)$, and by \require{Order} $f_1, \dots, f_n \domi \max(f_1, \dots, f_n)$. Suppose $h \in \rc{X}$ is such that $f_1, \dots, f_n \domi h \domi \max(f_1, \dots, f_n)$. Let
\begin{eqs}
F_i = \setb{x \in X: f_i(x) = \max(f_1, \dots, f_n)}. 
\end{eqs}
By \require{ISubComp}, $\restrb{f_i}{F_i} \preceq_{F_i} \restrb{h}{F_i} \preceq_{F_i} \restrb{f_i}{F_i}$. Therefore $\restrb{f_i}{F_i} \approx_{F_i} \restrb{h}{F_i}$. By \require{Local}, $\max(f_1, \dots, f_n) \approx_X h$. Therefore
\begin{eqs}
\sup \setb{f_1, \dots, f_n} \approx_X \max(f_1, \dots, f_n). 
\end{eqs}
Similarly, $\inf \setb{f_1, \dots, f_n} \approx_X \min(f_1, \dots, f_n)$. \sprove{Lattice}
\end{proof}

\begin{note}[Equivalence between types of completeness]
When $\domi$ is a lattice, every subset of $\rc{X}$ is directed, and therefore directed\-/complete is equivalent to complete. Further, chain\-/complete is equivalent to directed\-/complete by Zorn's lemma.\footnote{Zorn's lemma is equivalent to the axiom of choice.}
\end{note}

\begin{theorem}[Incompleteness]
\label{DominanceIsIncomplete}
$\preceq_{\posi{\TN}}$ has \require{Order}, \require{Trans}, \require{NSubHom}, \require{NSubDiv}, \require{Scale}, and \require{One}. $\implies$ $\preceq_{\posi{\TN}}$ is not complete. \sprove{Incomplete}
\end{theorem}

\begin{proof}
$\preceq_{\posi{\TN}}$ has \property{SubHom} by \proveby{RSubhomogenuityIsImplied}.

Let $f_{\alpha} \in \rc{\posi{\TN}}$ be such that $f_{\alpha}(n) = \alpha^n$, for all $\alpha \in \nonn{\TR}$. By \require{Order}, $\alpha \leq \beta \implies f_{\alpha} \preceq_{\posi{\TN}} f_{\beta}$, for all $\alpha, \beta \in \nonn{\TR}$. 

Suppose there exists $\alpha, \beta \in \nonn{\TR}$ such that $\alpha < \beta$ and $f_{\alpha} \succeq_{\posi{\TN}} f_{\beta}$. By \require{SubHom}, $\bra{f_{\beta} / f_{\alpha}} \preceq_{\posi{\TN}} 1$. It can be shown that $n \leq \frac{1}{e \lgbe{\beta / \alpha}} \bra{f_{\beta} / f_{\alpha}}$. By \require{Order} and \require{Scale}, $n \preceq_{\posi{\TN}} \bra{f_{\beta} / f_{\alpha}}$. By \require{Trans}, $n \preceq_{\posi{\TN}} 1$, which contradicts \property{One}. Therefore, $\alpha < \beta \implies f_{\alpha} \prec_{\posi{\TN}} f_{\beta}$, for all $\alpha, \beta \in \nonn{\TR}$.

Let $F = \setb{f_{\alpha} : \alpha \in \posib{\TR}{2}}$.  Then $f_2$ is a lower\-/bound of $F$. Let $\underline{f} \in \rc{\posi{\TN}}$ be a lower\-/bound of $F$ such that $f_2 \preceq_{\posi{\TN}} \underline{f}$. 

Suppose $\underline{f} \approx_{\posi{\TN}} \beta^n$, for some $\beta \in \posib{\TR}{2}$. Then $\bra{\frac{2 + \beta}{2}}^n \prec_{\posi{\TN}} \beta^n \approx_{\posi{\TN}} \underline{f}$, which contradicts $\underline{f}$ being a lower\-/bound of $F$. Therefore $\underline{f} \prec_{\posi{\TN}} \alpha^n$, for all $\alpha \in \posib{\TR}{2}$. 

By \require{Order}, $\underline{f} \preceq_{\posi{\TN}} n \underline{f}$. Suppose $n \underline{f} \preceq_{\posi{\TN}} \underline{f}$. By \require{SubHom}, $n \preceq_{\posi{\TN}} 1$, which contradicts \require{One}. Therefore $\underline{f} \prec_{\posi{\TN}} n \underline{f}$.

By \require{SubHom}, $n \underline{f} \prec_{\posi{\TN}} n \alpha^n$, for all $\alpha \in \posib{\TR}{2}$. It can be shown that $n \alpha^n \leq \frac{\beta / \alpha}{\lgbe{\beta / \alpha} e} \beta^n$, for all $\alpha, \beta \in \posib{\TR}{2}$ such that $\alpha < \beta$. By \require{Order} and \require{Scale}, $n \alpha^n \preceq_{\posi{\TN}} \beta^n$, for all $\alpha, \beta \in \posib{\TR}{2}$ such that $\alpha < \beta$. By \require{Trans}, $n \underline{f} \prec_{\posi{\TN}} \alpha^n$, for all $\alpha \in \posib{\TR}{2}$; $n \underline{f}$ is a lower\-/bound for $F$. Since $\underline{f} \prec_{\posi{\TN}} n \underline{f}$, there is no greatest lower\-/bound for $F$. 

\sprove{Incomplete}
\end{proof}

\chapter{Desirable properties}
\label{DesirableProperties}

In this section we provide an extensive list of desirable properties for an $\ohsy$-notation. We will show that they all hold for linear dominance in \sref{ImpliedProperties}. Each property is given for $\ohs{X}$, where $X \in U$. A given property holds for $\ohsy$ if it holds for $\ohs{X}$, for all $X \in U$. The desirable properties are listed in \taref{TableOfDesirableProperties}.

\begin{table}
\begin{tabular}{|l|l|}
\hline 
Name & Property \\
\hline 
\hline 
\textbf{\uproperty{Order}} & $f \lt g \implies f \in \ohx{g}$  \\
\hline 
\uproperty{Reflex} & $f \in \ohx{f}$ \\
\hline 
\textbf{\uproperty{Trans}} & $\bra{f \in \ohx{g} \textrm{ and } g \in \ohx{h}} \implies f \in \ohx{h}$  \\
\hline 
\uproperty{Orderness} & $f \in \ohx{g} \iff \ohx{f} \subset \ohx{g}$ \\
\hline 
\hline 
\uproperty{Zero} & $1 \not\in \oh{\posi{\TN}}{0}$ \\
\hline 
\textbf{\uproperty{One}} & $n \not\in \oh{\posi{\TN}}{1}$ \\
\hline 
\uproperty{TrivialZero} & $\ohx{0} = \setb{0}$ \\
\hline 
\hline 
\textbf{\uproperty{Scale}} & $\ohx{\alpha f} = \ohx{f}$ \\
\hline 
\usproperty{Translation} & $\ohx{f + \beta + \alpha} = \ohx{f + \beta}$  \\
\hline 
\uproperty{PowerH} & $\ohx{f}^{\alpha} = \ohx{f^{\alpha}}$ \\
\hline 
\uproperty{AddCons} & $u \ohx{f} + v \ohx{f} = (u + v) \ohx{f}$ \\
\hline 
\uproperty{MultiCons} & $\ohx{f}^{u} \cdot \ohx{f}^{v} = \ohx{f}^{u + v}$ \\
\hline 
\uproperty{MaxCons} & $\max(\ohx{f}, \ohx{f}) = \ohx{f}$ \\
\hline 
\textbf{\uproperty{Local}} & $\bra{\forall D \in C: \restrb{f}{D} \in \oh{D}{\restr{g}{D}}} \implies f \in \ohx{g}$  \\
\hline 
\uproperty{ScalarHom} & $\alpha \ohx{f} = \ohx{\alpha f}$ \\
\hline 
\uproperty{SubHom} & $u \ohx{f} \subset \ohx{u f}$ \\
\hline 
\uproperty{QSubHom} & $u \ohx{f} \subset \ohx{u f} \quad (\image{u}{X} \subset \nonn{\TQ})$ \\
\hline 
\textbf{\uproperty{NSubHom}} & $u \ohx{f} \subset \ohx{u f} \quad (\image{u}{X} \subset \TN)$ \\
\hline 
\textbf{\usproperty{NSubDiv}} & $u \ohx{f} \subset \ohx{u f} \quad (\image{u}{X} \subset 1 / \posi{\TN})$ \\
\hline 
\uproperty{SuperHom} & $u \ohx{f} \supset \ohx{u f}$ \\
\hline 
\uproperty{SubMulti} & $\ohx{f} \cdot \ohx{g} \subset \ohx{fg}$ \\
\hline 
\uproperty{SuperMulti} & $\ohx{f} \cdot \ohx{g} \supset \ohx{fg}$ \\
\hline 
\uproperty{SubRestrict} & $\bra{\restr{\ohx{f}}{D}} \subset \oh{D}{\restr{f}{D}}$ \\
\hline 
\uproperty{SuperRestrict} & $\bra{\restr{\ohx{f}}{D}} \supset \oh{D}{\restr{f}{D}}$ \\
\hline 
\uproperty{Additive} & $\ohx{f} + \ohx{g} = \ohx{f + g}$ \\
\hline 
\uproperty{Summation} & $\ohx{f + g} = \ohx{\max(f, g)}$ \\
\hline 
\uproperty{Maximum} & $\max(\ohx{f}, \ohx{g}) = \ohx{\max(f, g)}$ \\
\hline 
\uproperty{MaximumSum} & $\max(\ohx{f}, \ohx{g}) = \ohx{f} + \ohx{g}$ \\
\hline 
\textbf{\uproperty{SubComp}} & $\ohx{f} \circ s \subset \oh{Y}{f \circ s}$ \\
\hline 
\usproperty{ISuperComp} & $\ohx{f} \circ s \supset \oh{Y}{f \circ s}$ \quad ($s$ injective) \\
\hline 
\uproperty{Extend} & $\ohx{f} \circ \projections{X}{Y} \subset \oh{X \times Y}{f \circ \projections{X}{Y}}$ \\
\hline 
\uproperty{SubsetSum} & $\ohx{\sum_{(y, z) \in S_x} a(z) h(y)} \subset $ \\
{} & $\ohx{\sum_{(y, z) \in S_x} a(z) \bar{h}(y)}$ \\
\hline 
\end{tabular}
\centering
\caption{Desirable properties for an $\ohsy$-notation. Here $X, Y, Z \in U$, $f, g, u, v \in \rc{X}$, $\overline{h} \in \rc{Y}$, $h \in \oh{Y}{\overline{h}}$, $\alpha, \beta \in \posi{\TR}$, $D \subset X$, $\function{s}{Y}{X}$, $\function{S}{X}{\fpower{Y \times Z}}$, $a \in \rc{Z}$, and $C \subset \power{X}$ is a finite cover of $X$. Primitive properties marked with a bold face.}
\label{TableOfDesirableProperties}
\end{table}

\begin{note}[A computational model for examples]
\label{ModelForExamples}
When analyzing the cost functions of the example-algorithms in this section, addition costs one unit, while all other operations cost nothing.
\end{note}

\section{Dominance properties}

The \define{dominance properties} are those which mirror the desire for $\ohs{X}$ to represent a down\-/set of a preorder, where the preorder is consistent with the partial order $\leq$ in $\rc{X}$. 

\begin{definition}[\uproperty{Order}]
$\ohs{X}$ has \defineproperty{Order}, if 
\begin{equation}
f \leq g \implies f \in \ohx{g},
\end{equation} 
for all $f, g \in \rc{X}$. 
\end{definition}

\begin{example}
Let $X \in U$ and $f \in \rc{X}$. By \require{Order}, $0 \in \ohx{f}$.
\end{example}

\begin{example}[Powers after $1$]
Suppose $\TR \in U$. Let $\alpha, \beta \in \TR$ be such that $\alpha \leq \beta$. Then 
\begin{equation}
x^{\alpha} \leq x^{\beta},
\end{equation}
for all $x \in \nonnb{\TR}{1}$. By \require{Order},
\begin{equation}
x^{\alpha} \in \oh{\nonnb{\TR}{1}}{x^{\beta}}.
\end{equation}
\end{example}

\begin{example}[Powers before $1$]
Suppose $\TR \in U$. Let $\alpha, \beta \in \TR$ be such that $\alpha \leq \beta$. Then 
\begin{equation}
x^{\beta} \leq x^{\alpha},
\end{equation}
for all $x \in (0, 1] \subset \TR$. By \require{Order},
\begin{equation}
x^{\beta} \in \oh{(0, 1]}{x^{\alpha}}.
\end{equation}
\end{example}

\begin{example}[Positive power dominates a logarithm]
\label{LogGrowsSlowly}
Suppose $\nonnb{\TR}{1} \in U$. Let $\alpha, \gamma \in \posi{\TR}$ and $\beta \in \posib{\TR}{1}$. It can be shown that
\begin{equation}
\lgb{\beta}{x}^{\gamma} \leq \bra{\frac{\gamma}{\alpha e \lgbe{\beta}}}^{\gamma} x^{\alpha},
\end{equation}
for all $x \in \nonnb{\TR}{1}$. By \require{Order},
\begin{equation}
\lgbs{\beta}(x)^{\gamma} \in \oh{\nonnb{\TR}{1}}{\bra{\frac{\gamma}{\alpha e \lgbe{\beta}}}^{\gamma} x^{\alpha}}.
\end{equation}
\end{example}

\begin{example}[Positive power dominates a logarithm, generalized]
\label{LogGrowsSlowlyGeneralized}
Let $X \in U$, $\alpha, \gamma \in \posi{\TR}$, $\beta \in \posib{\TR}{1}$, and $f \in \rc{X}$ be such that $f \geq 1$. Then
\begin{equation}
\lgbs{\beta}^{\gamma} \circ f \leq \bra{\frac{\gamma}{\alpha e \lgbe{\beta}}}^{\gamma} f^{\alpha}.
\end{equation}
By \require{Order}, 
\begin{equation}
\lgbs{\beta}^{\gamma} \circ f \in \ohx{\bra{\frac{\gamma}{\alpha e \lgbe{\beta}}}^{\gamma} f^{\alpha}}.
\end{equation}
Since this holds for all $\alpha \in \posi{\TR}$,
\begin{equation}
\lgbs{\beta}^{\gamma} \circ f \in \bigcap_{\alpha \in \posi{\TR}} \ohx{\bra{\frac{\gamma}{\alpha e \lgbe{\beta}}}^{\gamma} f^{\alpha}}.
\end{equation}
\end{example}

\begin{example}[Functions on a finite set]
\label{FunctionsInFiniteSet}
Let $X \in U$ be finite, and $f \in \rc{X}$. Then $f \leq \max(\image{f}{X})$. By \require{Order},
\begin{eqs}
f \in \ohx{\max(\image{f}{X})}.
\end{eqs}
\end{example}

\begin{definition}[\uproperty{Reflex}]
$\ohs{X}$ has \defineproperty{Reflex}, if 
\begin{equation}
f \in \ohx{f},
\end{equation} 
for all $f \in \rc{X}$. 
\end{definition}

\begin{example}
$(5n + 3) \in \oh{\TN}{5n + 3}$.
\end{example}

\begin{definition}[\uproperty{Trans}]
$\ohs{X}$ has \defineproperty{Trans}, if
\begin{equation}
\bra{f \in \ohx{g} \textrm{ and } g \in \ohx{h}} \implies f \in \ohx{h},
\end{equation}
for all $f, g, h \in \rc{X}$.
\end{definition}

\begin{example}
Let $g \in \oh{\TN}{n^2}$, and $f \in \oh{\TN}{g}$. By \require{Trans}, $f \in \oh{\TN}{n^2}$.
\end{example}

\begin{definition}[\uproperty{Orderness}]
$\ohs{X}$ has \defineproperty{Orderness}, if
\begin{equation}
f \in \ohx{g} \iff \ohx{f} \subset \ohx{g},
\end{equation}
for all $f, g \in \rc{X}$. 
\end{definition}

\begin{example}
$n \in \oh{\TN}{n^2} \iff \oh{\TN}{n} \subset \oh{\TN}{n^2}$.
\end{example}

\section{Non-triviality properties}

The \emph{non\-/triviality properties} are those which require that the $\ohsy$-notation be detailed enough. 

\begin{example}
The class of functions $\function{\ohs{X}}{\rc{X}}{\power{\rc{X}}}$, such that $\ohx{f} \coloneqq \rc{X}$, for all $X \in U$, satisfies all of the desirable properties for an $\ohsy$-notation except those of non\-/triviality.
\end{example}

\begin{definition}[\uproperty{Zero}]
$\ohsy$ has \defineproperty{Zero}, if
\begin{equation}
1 \not\in \oh{\posi{\TN}}{0}.
\end{equation}
\end{definition}

\begin{definition}[\uproperty{One}]
$\ohsy$ has \defineproperty{One}, if
\begin{equation}
n \not\in \oh{\posi{\TN}}{1}.
\end{equation}
\end{definition}

\begin{definition}[\uproperty{TrivialZero}]
$\ohs{X}$ has \defineproperty{TrivialZero}, if
\begin{equation}
\ohx{0} = \setb{0}.
\end{equation}
\end{definition}

\section{Abstraction properties}

The \emph{abstraction properties} are those which define the way in which the $\ohsy$-notation identifies functions. 

\begin{definition}[\uproperty{Scale}]
$\ohs{X}$ has \defineproperty{Scale}, if
\begin{equation}
\ohx{\alpha f} = \ohx{f},
\end{equation}
for all $f \in \rc{X}$, and $\alpha \in \posi{\TR}$.
\end{definition}

\begin{example}
$\oh{\TN}{2n} = \oh{\TN}{n}$.
\end{example}

\begin{example}[Power dominates a logarithm, continued]
Let $X \in U$, $\beta \in \posib{\TR}{1}$, $\gamma \in \posi{\TR}$, and $f \in \rc{X}$ be such that $f \geq 1$. Continuing \eref{LogGrowsSlowlyGeneralized}, by \require{Scale},
\begin{equation}
\ohx{\lgbs{\beta}^{\gamma} \circ f} \subset \bigcap_{\alpha \in \posi{\TR}} \ohx{f^{\alpha}}.
\end{equation}
\end{example}

\begin{example}[Functions on a finite set, continued]
Let $X \in U$ be finite, and $f \in \rc{X}$. By \eref{FunctionsInFiniteSet}, $f \in \ohx{\max(\image{f}{X})}$. By \require{Scale}, $f \in \ohx{1}$, if $f \neq 0$, and $f \in \ohx{0}$, if $f = 0$.
\end{example}

\begin{definition}[\uproperty{Translation}]
$\ohs{X}$ has \defineproperty{Translation}, if
\begin{equation}
\ohx{f + \beta + \alpha} = \ohx{f + \beta},
\end{equation}
for all $f \in \rc{X}$, and $\alpha, \beta \in \posi{\TR}$.
\end{definition}

\begin{example}
$\oh{\posi{\TN}}{n + 1} = \oh{\posi{\TN}}{n}$.
\end{example}

\begin{note}[The role of $\beta$ in \property{Translation}]
Suppose $f \in \rc{X}$ is such that that $0 \in \image{f}{X}$. Then $\beta$ protects against transforming a zero cost to a non-zero cost.

Suppose $f \in \rc{X}$ is such that $f > 0$ and $\inf \image{f}{X} = 0$ --- e.g., $f \in \rc{\posi{\TN}}$ such that $f(n) = 1 / n^2$. Then it is possible to call the corresponding algorithm $F$ arbitrary number of times --- with different inputs --- while spending a bounded amount of resources. The $\beta$ protects against transforming a bounded cost function into an unbounded one.
\end{note}

\begin{example}
Suppose $ \posi{\TR}, X \in U$, and $X \neq \emptyset$. Then the following do \emph{not} follow from \property{Translation}: $\oh{\TN}{n + 1} = \oh{\TN}{n}$, $\ohx{1} = \ohx{0}$, and $\oh{\posi{\TR}}{(1 / x) + 1} = \oh{\posi{\TR}}{1 / x}$.
\end{example}

\begin{definition}[\uproperty{PowerH}]
$\ohs{X}$ has \defineproperty{PowerH}, if
\begin{equation}
\ohx{f}^{\alpha} = \ohx{f^{\alpha}},
\end{equation}
for all $f \in \rc{X}$, and $\alpha \in \posi{\TR}$.
\end{definition}

\begin{example}
$\oh{\TN}{n}^2 = \oh{\TN}{n^2}$.
\end{example}

\begin{definition}[\uproperty{AddCons}]
$\ohs{X}$ has \defineproperty{AddCons}, if
\begin{equation}
u \ohx{f} + v \ohx{f} = (u + v) \ohx{f},
\end{equation}
for all $f, u, v \in \rc{X}$.
\end{definition}

\begin{example}
$3 \oh{\TN}{n} + n^2 \oh{\TN}{n} = (3 + n^2)\oh{\TN}{n}$.
\end{example}

\begin{definition}[\uproperty{MultiCons}]
$\ohs{X}$ has \defineproperty{MultiCons}, if
\begin{equation}
\ohx{f}^{u} \cdot \ohx{f}^{v} = \ohx{f}^{u + v},
\end{equation}
for all $f, u, v \in \rc{X}$.
\end{definition}

\begin{example}
$\oh{\TN}{n}^3 \oh{\TN}{n}^{n^2} = \oh{\TN}{n}^{3 + n^2}$.
\end{example}

\begin{definition}[\uproperty{MaxCons}]
$\ohs{X}$ has \defineproperty{MaxCons}, if
\begin{equation}
\max(\ohx{f}, \ohx{f}) = \ohx{f},
\end{equation}
for all $f \in \rc{X}$.
\end{definition}

\section{Structural properties}

The \emph{structural properties} are those which mirror the structure of algorithms: repetition, conditional branching, and abstraction.

\begin{definition}[\uproperty{Local}]
$\ohsy$ has \defineproperty{Local}, if
\begin{equation}
\bra{\forall D \in C: \restrb{f}{D} \in \oh{D}{\restr{g}{D}}} \implies f \in \ohx{g}.
\end{equation}
for all $X \in U$, $f, g \in \rc{X}$, and $C \subset \power{X}$ a finite cover of $X$.
\end{definition}

\begin{definition}[\uproperty{ScalarHom}]
$\ohs{X}$ has \defineproperty{ScalarHom}, if
\begin{equation}
\alpha \ohx{f} = \ohx{\alpha f},
\end{equation}
for all $f \in \rc{X}$, and $\alpha \in \posi{\TR}$.
\end{definition}

\begin{example}
$2 \oh{\TN}{n} = \oh{\TN}{2n}$.
\end{example}

\begin{definition}[\uproperty{SubHom}]
$\ohs{X}$ has \defineproperty{SubHom}, if
\begin{equation}
u \ohx{f} \subset \ohx{uf},
\end{equation}
for all $f, u \in \rc{X}$.
\end{definition}

\begin{definition}[\uproperty{SuperHom}]
$\ohs{X}$ has \defineproperty{SuperHom}, if
\begin{equation}
u \ohx{f} \supset \ohx{uf},
\end{equation}
for all $f, u \in \rc{X}$.
\end{definition}

\begin{definition}[\uproperty{Hom}]
$\ohs{X}$ has \defineproperty{Hom}, if it has \property{SubHom} and \property{SuperHom}.
\end{definition}

\begin{definition}[\uproperty{SubMulti}]
$\ohs{X}$ has \defineproperty{SubMulti}, if
\begin{equation}
\ohx{f} \cdot \ohx{g} \subset \ohx{fg},
\end{equation}
for all $f, g \in \rc{X}$.
\end{definition}

\begin{definition}[\uproperty{SuperMulti}]
$\ohs{X}$ has \defineproperty{SuperMulti}, if
\begin{equation}
\ohx{f} \cdot \ohx{g} \supset \ohx{fg},
\end{equation}
for all $f, g \in \rc{X}$.
\end{definition}

\begin{definition}[\uproperty{Multi}]
$\ohs{X}$ has \defineproperty{Multi}, if it has \property{SubMulti} and \property{SuperMulti}.
\end{definition}

\begin{algorithm}
\caption{An algorithm to demonstrate \property{Multi}.}
\label{alg:multiply}
\begin{algorithmic}[1]
\Procedure {H}{$x$}
\For{$i \in \icon{0}{G(x)}$}
  \State \Call{F}{$x$}
\EndFor
\EndProcedure
\end{algorithmic}
\end{algorithm}

\begin{example}[Demonstration of \property{Multi}]
Consider \aref{alg:multiply}, $H \colon \algof{X}{Y}$. This algorithm runs the same sub-algorithm, $F \colon \algof{X}{Z}$, repeatedly $G(x) \in \TN$ times, where the algorithm $G \colon \algof{X}{\TN}$ has zero cost (i.e. does not perform additions). Let the cost functions of $H$ and $F$ be $h, f \in \rc{X}$, respectively. Then 
\begin{eqs}
h = fG.
\end{eqs}
By \require{Multi},
\begin{eqs}
\ohx{h} & = \ohx{fG} \\
{} & = \ohx{f} \cdot \ohx{G}.
\end{eqs}
That is, if we have analyzed $F$ to the have complexity $\ohx{f}$, and know $\ohx{G}$, then the complexity of $H$ is given by $\ohx{h} = \ohx{f} \cdot \ohx{G}$.
\end{example}

\begin{example}
$\oh{\TN}{n}  \cdot \oh{\TN}{n^2} = \oh{\TN}{n^3}$.
\end{example}

\begin{example}
$\oh{\TN}{0}  \cdot \oh{\TN}{n} = \oh{\TN}{0}$.
\end{example}

\begin{example}
Suppose $\posi{\TR} \in U$. Then $\oh{\posi{\TR}}{1/x}  \cdot \oh{\posi{\TR}}{x} = \oh{\posi{\TR}}{1}$.
\end{example}

\begin{definition}[\uproperty{SubRestrict}]
$\ohs{X}$ has \defineproperty{SubRestrict}, if
\begin{equation}
\bra{\restr{\ohx{f}}{D}} \subset \oh{D}{\restr{f}{D}},
\end{equation}
for all $f \in \rc{X}$, and $D \subset X$.
\end{definition}

\begin{definition}[\uproperty{SuperRestrict}]
$\ohs{X}$ has \defineproperty{SuperRestrict}, if
\begin{equation}
\bra{\restr{\ohx{f}}{D}} \supset \oh{D}{\restr{f}{D}},
\end{equation}
for all $f \in \rc{X}$, and $D \subset X$.
\end{definition}

\begin{definition}[\uproperty{Restrict}]
$\ohs{X}$ has \defineproperty{Restrict}, if it has \property{SubRestrict} and \property{SuperRestrict}.
\end{definition}

\begin{algorithm}
\caption{An algorithm to demonstrate \property{Restrict}.}
\label{alg:restriction}
\begin{algorithmic}[1]
\Procedure {H}{$x$}
\State \Return \Call{F}{$x$}
\EndProcedure
\end{algorithmic}
\end{algorithm}

\begin{example}[Demonstration of \property{Restrict}]
Let $X \in U$, $D \subset X$, and $Y$ be a set. Consider \aref{alg:restriction}, $H \colon \algof{D}{Y}$. This algorithm passes its argument --- as it is --- to the algorithm $F \colon \algof{X}{Y}$. Let $h \in \rc{D}$ and $f \in \rc{X}$ be the cost functions of algorithms $H$ and $F$, respectively. Then
\begin{eqs}
h = \restrb{f}{D}.
\end{eqs}
By \require{Restrict},
\begin{eqs}
\oh{D}{h} = \restrb{\ohx{f}}{D}.
\end{eqs}
That is, if we have analyzed $F$ to have the complexity $\ohx{f}$, then the complexity of $H$ is given by $\oh{D}{h} = \restrb{\ohx{f}}{D}$.
\end{example}

\begin{example}
$\bra{\restr{\oh{\TN}{n}}{2\TN}} = \oh{2\TN}{n}$
\end{example}

\begin{example}
Let $D \coloneqq \setb{(n, n) \in \TN^2 : n \in \TN}$. Then $\bra{\restr{\oh{\TN^2}{mn}}{D}} = \oh{D}{n^2}$.
\end{example}

\begin{example}
$\bra{\restr{\oh{\TN^2}{mn}}{\setb{(1, 1)}}} = \oh{\setb{(1, 1)}}{1}$
\end{example}

\begin{example}
Suppose $\posi{\TR} \in U$, and $g \in \oh{\posi{\TR}}{1 / x}$. By \require{SubRestrict}, $\restrb{g}{\nonnb{\TR}{1}} \in \oh{\nonnb{\TR}{1}}{1 / x}$. By \require{Order}, $1 / x \in \oh{\nonnb{\TR}{1}}{1}$. By \require{Trans}, $\restrb{g}{\nonnb{\TR}{1}} \in \oh{\nonnb{\TR}{1}}{1}$ 
\end{example}

\begin{definition}[\uproperty{Additive}]
$\ohs{X}$ has \defineproperty{Additive}, if
\begin{equation}
\ohx{f} + \ohx{g} = \ohx{f + g},
\end{equation}
for all $f, g \in \rc{X}$.
\end{definition}

\begin{algorithm}
\caption{An algorithm to demonstrate \property{Additive}.}
\label{alg:sequence}
\begin{algorithmic}[1]
\Procedure {H}{$x$}
\State \Return $($\Call{F}{$x$}, \Call{G}{$x$}$)$
\EndProcedure
\end{algorithmic}
\end{algorithm}

\begin{example}[Demonstration of \property{Additive}]
Consider \aref{alg:sequence}, $H \colon \algof{X}{Y \times Z}$, which decomposes into two sub-algorithms $F \colon \algof{X}{Y}$, and $G \colon \algof{X}{Z}$, such that $H(x) = (F(x), G(x))$. Suppose the cost functions of $H$, $F$, and $G$ are $h, f, g \in \rc{X}$, respectively. Then 
\begin{eqs}
h = f + g.
\end{eqs}
By \require{Additive},
\begin{eqs}
\ohx{h} & = \ohx{f + g} \\
{} & = \ohx{f} + \ohx{g}.
\end{eqs}
That is, if we have analyzed $F$ and $G$ to have complexities $\ohx{f}$ and $\ohx{g}$, respectively, then the complexity of $H$ is given by $\ohx{h} = \ohx{f} + \ohx{g}$.
\end{example}

\begin{example}
$\oh{\TN}{n} + \oh{\TN}{n^2} = \oh{\TN}{n + n^2}$.
\end{example}

\begin{example}
Let $X \in U$, and $f, g \in \rc{X}$. By \require{Order}, and \require{Additive}, 
\begin{eqs}
f + \ohx{g} & \subset \ohx{f} + \ohx{g} \\
{} & = \ohx{f + g}.
\end{eqs}
\end{example}

\begin{definition}[\uproperty{Summation}]
$\ohs{X}$ has \defineproperty{Summation}, if
\begin{equation}
\ohx{f + g} = \ohx{\max(f, g)},
\end{equation}
for all $f, g \in \rc{X}$.
\end{definition}

\begin{note}
\uproperty{Summation} states that, in a finite sequence of algorithm calls, the worst cost function over all calls determines the complexity of the call-sequence.
\end{note}

\begin{example}
$\oh{\TN}{n + n^2} = \oh{\TN}{\max(n, n^2)} = \oh{\TN}{n^2}$.
\end{example}

\begin{example}
Suppose $\nonnb{\TR}{1} \in U$. Then
\begin{eqs}
\oh{\nonnb{\TR}{1}}{x + x^2} & = \oh{\nonnb{\TR}{1}}{\max(x, x^2)} \\
{} & = \oh{\nonnb{\TR}{1}}{x^2}.
\end{eqs}
\end{example}

\begin{example}
Suppose $\nonn{\TR} \in U$. It does not follow from \property{Summation} that $\oh{\nonn{\TR}}{x + x^2} = \oh{\nonn{\TR}}{x^2}$.
\end{example}

\begin{definition}[\uproperty{Maximum}]
$\ohs{X}$ has \defineproperty{Maximum}, if
\begin{equation}
\max(\ohx{f}, \ohx{g}) = \ohx{\max(f, g)},
\end{equation}
for all $f, g \in \rc{X}$.
\end{definition}

\begin{example}
$\max(\oh{\TN}{n}, \oh{\TN}{n^2}) = \oh{\TN}{\max(n, n^2)} = \oh{\TN}{n^2}$.
\end{example}

\begin{definition}[\uproperty{MaximumSum}]
$\ohs{X}$ has \defineproperty{MaximumSum}, if
\begin{equation}
\max(\ohx{f}, \ohx{g}) = \ohx{f} + \ohx{g},
\end{equation}
for all $f, g \in \rc{X}$.
\end{definition}

\begin{example}
$\max(\oh{\TN}{n}, \oh{\TN}{n^2}) = \oh{\TN}{n} + \oh{\TN}{n^2}$.
\end{example}

\begin{definition}[\uproperty{SubComp}]
$\ohs{X}$ has \defineproperty{SubComp}, if
\begin{equation}
\ohx{f} \circ s \subset \oh{Y}{f \circ s},
\end{equation}
for all $f \in \rc{X}$ and $\function{s}{Y}{X}$.
\end{definition}

\begin{definition}[\uproperty{ISubComp}]
$\ohs{X}$ has \defineproperty{ISubComp}, if it has \property{SubComp} for all injective $\function{s}{Y}{X}$.
\end{definition}

\begin{definition}[\uproperty{ISuperComp}]
$\ohs{X}$ has \defineproperty{ISuperComp}, if
\begin{equation}
\ohx{f} \circ s \supset \oh{Y}{f \circ s},
\end{equation}
for all $f \in \rc{X}$ and injective $\function{s}{Y}{X}$.
\end{definition}

\begin{note}[\uproperty{SuperComp} does not make sense without injectivity]
Let $\function{s}{Y}{X}$ be a positive constant; $s(y) = c$ for all $y \in Y$, and some $c \in \posi{\TR}$. Then $\ohx{f} \circ s = \setb{s}$, and $\oh{Y}{f \circ s} = \oh{Y}{c}$. Let $g \in \oh{Y}{c} \setminus \setb{s}$. Then $g \not\in \ohx{f} \circ s = \setb{s}$. 
\end{note}

\begin{algorithm}
\caption{An algorithm to demonstrate \property{IComp}. We assume that $F$ does not consume any resources.}
\label{alg:DependentSequence}
\begin{algorithmic}[1]
\Procedure {$H$}{$x$}
\State \Return \Call{G}{$F(x)$}
\EndProcedure
\end{algorithmic}
\end{algorithm}

\begin{example}[Demonstration of \property{IComp}]
\label{CompositionExample}
Consider \aref{alg:DependentSequence}, $H \colon \algof{X}{Z}$, which decomposes into two sub-algorithms $F \colon \algof{X}{Y}$ and $G \colon \algof{Y}{Z}$ such that $F$ is injective and $H = G \circ F$. Denote the cost functions of $H$, $F$ and $G$ by $h, f \in \rc{X}$ and $g \in \rc{Y}$, respectively. Suppose $f = 0$. Then
\begin{eqs}
h & = g \circ F + f \\
{} & = g \circ F.
\end{eqs}
By \require{IComp},
\begin{eqs}
\ohx{h} & = \ohx{g \circ F} \\
{} & = \oh{Y}{g} \circ F.
\end{eqs}
That is, if $f = 0$, and we know $\oh{Y}{g}$, then the complexity of $H$ is given by $\ohx{h} = \oh{Y}{g} \circ F$. 
\end{example}

\begin{definition}[\uproperty{Extend}]
$\ohs{X}$ has \defineproperty{Extend}, if
\begin{equation}
\ohx{f} \circ \projections{X}{Y} \subset \oh{X \times Y}{f \circ \projections{X}{Y}},
\end{equation}
for all $X, Y \in U$ and $f \in \rc{X}$.
\end{definition}

\begin{note}
\uproperty{Extend} is a special case of \property{SubComp}.
\end{note}

\begin{definition}[\uproperty{SubsetSum}]
$\ohsy$ has \defineproperty{SubsetSum}, if
\begin{equation}
\ohx{\sum_{(y, z) \in S_x} a(z) f(y)} \subset \ohx{\sum_{(y, z) \in S_x} a(z) \bar{f}(y)},
\end{equation}
for all $X, Y, Z \in U$, $\function{S}{X}{\fpower{Y \times Z}}$, $a \in \rc{Z}$, $\bar{f} \in \rc{Y}$, and $f \in \oh{Y}{\bar{f}}$.
\end{definition}

\begin{algorithm}
\caption{An algorithm to demonstrate \property{SubsetSum}.}
\label{alg:SubsetSum}
\begin{algorithmic}[1]
\Procedure {$G$}{$x$}
\State $\alpha \coloneqq 0$
\For {$(y, z) \in S_x$}
\State $\alpha \coloneqq \alpha \cdot \Call{F}{y}$
\EndFor
\State \Return $\alpha$
\EndProcedure
\end{algorithmic}
\end{algorithm}

\begin{example}[Demonstration of \property{SubsetSum}]
\label{SubsetSumExample}
Consider \aref{alg:SubsetSum}, $G \colon \algof{X}{\TZ}$, which for $x \in X$ sums $F \colon \algof{Y}{\TZ}$ over $S_x \in \fpower{Y \times Z}$. Let the algorithms $F$ and $G$ take $f \in \rc{Y}$ and $g \in \rc{X}$ operations, respectively. Then
\begin{eqs}
g(x) = \sum_{y \in S_x} f(y).
\end{eqs}
Suppose we have shown that $f \in \oh{Y}{\bar{f}}$. By \require{SubsetSum},
\begin{eqs}
\ohx{g} & = \ohx{\sum_{y \in S_x} f(y)} \\
{} & \subset \ohx{\sum_{y \in S_x} \bar{f}(y)}.
\end{eqs}
That is, if $f$ has an upper bound of $\bar{f}$, then we can construct an upper bound of $g$ from $\bar{f}$. Suppose we have also shown that $\bar{f} \in \oh{Y}{f}$ --- i.e. $f \in \thetah{Y}{\bar{f}}$. Then
\begin{eqs}
\ohx{g} = \ohx{\sum_{y \in S_x} \bar{f}(y)}.
\end{eqs}
\end{example}

\begin{note}[Factor $a$ in \property{SubsetSum}]
\eref{SubsetSumExample} does not make use of the factor $a$ in \property{SubsetSum}. Such general sums occur naturally in the analysis of recursive algorithms; see \sref{MasterTheorems}.
\end{note}

\chapter{Characterization of the \texorpdfstring{$\ohsy$}{O}-notation}
\label{Characterization}

In this chapter we show that the primitive properties are equivalent to the definition of $\ohsy$-notation as linear dominance.

\section{Linear dominance is sufficient}
\label{SufficientDefinition}

In this section we will show the following theorem:

\begin{theorem}[Linear dominance has primitive properties]
\label{LinearImpliesPrimitive}
Let $\function{\lohs{X}}{\rc{X}}{\power{\rc{X}}}$ be defined by $g \in \lohx{f}$ if and only if
\begin{equation}
\exists c \in \posi{\TR}: g \leq c f,
\end{equation}
for all $f, g \in \rc{X}$, and all $X \in U$, where the universe $U$ is the class of all sets. Then $\lohsy$ satisfies the primitive properties.
\end{theorem}

\begin{proof}
The result follows directly from the Lemmas in this section.
\end{proof}

We shall apply the following lemma repeatedly without mentioning it.

\begin{lemma}[Simplification lemma]
\label{LinearSingleConstantLemma}
Let $X \in U$, $I$ be a finite set, $X_i \subset X$, $f_i \in \rc{X_i}$, and $\hat{f}_i \in \loh{X_i}{f_i}$, for all $i \in I$. Then there exists $c \in \posi{\TR}$, such that $\hat{f}_i \lt c f_i$, for all $i \in I$.
\end{lemma}

\begin{proof}
Since $\hat{f}_i \in \loh{X_i}{f_i}$, there exists $c_i \in \posi{\TR}$, such that $\hat{f}_i \lt c_i f_i$, for all $i \in I$. Let $c = \max \setb{c_i : i \in I}$. Then $\hat{f}_i \lt c f_i$, for all $i \in I$.
\end{proof}

\begin{lemma}[Linear dominance has \property{Order}]
\label{LinearOrderConsistency}
Let $X \in U$, and $f, g \in \rc{X}$. Then \sprove{Order}
\begin{equation}
f \leq g \implies f \in \lohx{g}.
\end{equation}
\end{lemma}

\begin{proof}
Since $f \leq 1g$, it holds that $f \in \lohx{g}$.
\sprove{Order}
\end{proof}

\begin{lemma}[Linear dominance has \property{Trans}]
\label{LinearTransitivity}
Let $X \in U$, and $f, g, h \in \rc{X}$. Then \sprove{Trans}
\begin{equation}
\bra{f \in \lohx{g} \textrm{ and } g \in \lohx{h}} \implies f \in \lohx{h}.
\end{equation}
\end{lemma}

\begin{proof}
Let $f \in \lohx{g}$, and $g \in \lohx{h}$. Then there exists $c \in \posi{\TR}$, such that $f \lt c g$ and $g \lt c h$. It follows that $f \lt c^2 h$. Therefore $f \in \lohx{h}$.
\sprove{Trans}
\end{proof}

\begin{lemma}[Linear dominance has \property{Local}]
\label{LinearLocality}
Let $X \in U$, $f, g \in \rc{X}$, and $C \subset \power{X}$ be a finite cover of $X$. Then \sprove{Local}
\begin{equation}
\bra{\forall D \in C: \restrb{f}{D} \in \loh{D}{\restr{g}{D}}} \implies f \in \lohx{g}.
\end{equation}
\end{lemma}

\begin{proof}
Assume $\bra{\restr{f}{D}} \in \loh{D}{\restr{g}{D}}$, for all $D \in C$. Then there exist $c \in \posi{\TR}$ such that $\restrb{f}{D} \lt c \restrb{g}{D}$, for all $D \in C$. Since $C$ covers $X$, $f \lt c g$. Therefore $f \in \lohx{g}$.
\sprove{Local}
\end{proof}

\begin{lemma}[Linear dominance has \property{One}]
\label{LinearOneSeparation}
\begin{equation}
n \not\in \loh{\posi{\TN}}{1}. 
\end{equation} \sprove{One}
\end{lemma}

\begin{proof}
For all $c \in \posi{\TR}$, there exists $n \in \posi{\TN}$ --- for example $n = \ceilb{c} + 1$ --- such that $n > c 1$. Therefore $n \not\in \loh{\posi{\TN}}{1}$. \sprove{One}
\end{proof}

\begin{lemma}[Linear dominance has \property{Scale}]
\label{LinearScaleInvariance}
Let $X \in U$, $f \in \rc{X}$, and $\alpha \in \posi{\TR}$. Then $\lohx{f} \subset \lohx{\alpha f}$. \sprove{Scale}
\end{lemma}

\begin{proof}
Assume $\hat{f} \in \lohx{f}$. Then there exists $c \in \posi{\TR}$, such that
\begin{eqs}
\hat{f} & \lt c f \\
{} & = (c / \alpha) \bra{\alpha f}.
\end{eqs}
Therefore $\hat{f} \in \lohx{\alpha f}$. 
\sprove{Scale}
\end{proof}

\begin{lemma}[Linear dominance has \property{NSubHom} and \property{NSubDiv}]
\label{LinearSubHomogenuity}
Let $X \in U$, and $f, u \in \rc{X}$. Then \sprove{NSubHom} \sprove{NSubDiv}
\begin{equation}
u \lohx{f} \subset \lohx{uf}. 
\end{equation} 
\end{lemma}

\begin{proof}
Let $\hat{f} \in \lohx{f}$. Then there exists $c \in \posi{\TR}$, such that $\hat{f} \lt c f$. This implies $u\hat{f} \lt c uf$. Therefore $u\hat{f} \in \lohx{uf}$; $\lohs{X}$ has \prove{SubHom}. Since $\TN \subset \nonn{\TR}$, $\lohs{X}$ has \prove{NSubHom}. Since $\frac{1}{\posi{\TN}} \subset \posi{\TR}$, $\lohs{X}$ has \prove{NSubDiv}.
\end{proof}

\begin{lemma}[Linear dominance has \property{SubComp}]
\label{LinearSubComposability}
Let $X \in U$, $f \in \rc{X}$, and $\function{s}{Y}{X}$. Then \sprove{SubComp}
\begin{equation}
\lohx{f} \circ s \subset \loh{Y}{f \circ s}.
\end{equation}
\end{lemma}

\begin{proof}
Let $\hat{f} \in \lohx{f}$. Then there exists $c \in \posi{\TR}$, such that $\hat{f} \lt c f$. This implies $\hat{f} \circ s \lt c (f \circ s)$. Therefore $\hat{f} \circ s \in \lohx{f \circ s}$.
\sprove{SubComp}
\end{proof}

\section{Linear dominance is necessary}
\label{NecessaryDefinition}

In this section we will show the following theorem.

\begin{theorem}[Primitive properties imply linear dominance]
Suppose $\ohsy$ has \require{Order}, \require{Trans}, \require{One}, \require{Local}, \require{Scale},  \require{NSubHom}, \require{NSubDiv}, and \require{SubComp}. Then
\begin{equation}
f \in \ohx{g} \iff \exists c \in \posi{\TR}: f \leq cg.
\end{equation}
\end{theorem}

To prove this result, we will use some of the results from \sref{ImpliedProperties}.

\begin{lemma}[$\ohx{1}$ equals the bounded functions]
\label{OhOneIsExactlyBounded}
Suppose $\ohsy$ has \require{Order}, \require{Trans}, \require{One}, \require{Local}, \require{Scale}, and \require{ISubComp}. Then
\begin{equation}
\ohx{1} = \setb{f \in \rc{X} : \exists c \in \posi{\TR} : f \leq c},
\end{equation}
provided $X \neq \emptyset$.
\sprove{OhForOne}
\end{lemma}

\begin{proof}
\proofpart{Implied properties}
$\ohsy$ has \property{ISuperComp} by \proveby{InjectiveSuperComposabilityIsImplied}, and \property{Orderness} by \proveby{OrdernessIsImplied}.

\proofpart{$\subset$}
Assume $f \in \ohx{1}$ such that $f$ is unbounded. Then for every $n \in \posi{\TN}$ there exists $x_n \in X$ such that $f(x_n) \geq n$. Therefore, let $\function{s}{\posi{\TN}}{X}$ be injective such that $n \leq (f \circ s)(n)$, for all $n \in \posi{\TN}$. By \require{Orderness}, $\ohx{f} \subset \ohx{1}$. By \require{Order}, \require{Order} and \require{Trans}, \require{ISuperComp}, and \require{ISubComp},
\begin{eqs}
(n \mapsto n) & \in \oh{\posi{\TN}}{n}\\
{} & \subset \oh{\posi{\TN}}{f \circ s} \\
{} & \subset \ohx{f} \circ s \\
{} & \subset \ohx{1} \circ s \\
{} & \subset \oh{\posi{\TN}}{1 \circ s} \\
{} & = \oh{\posi{\TN}}{1}.
\end{eqs}
This contradicts $\ohsy$ having \require{One}. Therefore $f$ is bounded, which is equivalent to $\exists c \in \posi{\TR} : f \leq c$.

\proofpart{$\supset$} 
Assume $\exists c \in \posi{\TR}: f \leq c$. By \require{Order}, $f \in \ohx{c}$. By \require{Scale}, $f \in \ohx{1}$.

\sprove{OhForOne}
\end{proof}

\begin{lemma}[$\ohsy$-notation for positive functions]
\label{OhForPositive}
Suppose $\ohs{X}$ has \require{Order}, \require{Trans}, \require{Scale}, \require{NSubHom}, and \require{NSubDiv}. Then
\begin{equation}
f \in \ohx{g} \iff f / g \in \ohx{1},
\end{equation}
for all $f, g \in \rc{X}$ such that $g > 0$.
\sprove{OhForPositive}
\end{lemma}

\begin{proof}
$\ohs{X}$ has \property{SubHom} by \proveby{RSubhomogenuityIsImplied}.

\proofpart{$\implies$}
By \require{SubHom},
\begin{eqs}
{} & f \in \ohx{g} \\
\impliesr & f / g \in \ohx{g} / g \\
\impliesr & f / g \in \ohx{g / g} \\
\impliesr & f / g \in \ohx{1}.
\end{eqs}
\proofpart{$\impliedby$}
By \require{SubHom},
\begin{eqs}
{} & f / g \in \ohx{1} \\
\impliesr & f \in \ohx{1} g \\
\impliesr & f \in \ohx{g}.
\end{eqs}
\sprove{OhForPositive}
\end{proof}

\begin{theorem}[Primitive properties imply linear dominance]
\label{PrimitiveImpliesLinear}
Suppose $\ohsy$ has \require{Order}, \require{Trans}, \require{One}, \require{Local}, \require{Scale}, \require{NSubHom}, \require{NSubDiv}, and \require{SubComp}. Then
\begin{equation}
f \in \ohx{g} \iff \exists c \in \posi{\TR}: f \leq cg.
\end{equation}
\end{theorem}

\begin{proof}
\proofpart{Implied properties}
\srequire{SubComp} \sprove{ISubComp}
$\ohs{X}$ has \property{SubRestrict} by \proveby{SubRestrictabilityIsImplied}, and \property{TrivialZero} by \proveby{ZeroTrivialityIsImplied}. 

\proofpart{Positive subset}
Let $G \coloneqq \preimage{g}{\posi{\TR}}$ and $\compl{G} \coloneqq X \setminus G$. 
Then
\begin{eqs}
{} & \restrb{f}{G} \in \oh{G}{\restr{g}{G}} \\
\iffr & \frac{\restrb{f}{G}}{\restrb{g}{G}} \in \oh{G}{1} \\
\iffr & \exists c \in \posi{\TR}: \frac{\restrb{f}{G}}{\restrb{g}{G}} \leq c \\
\iffr & \exists c \in \posi{\TR}: \restrb{f}{G} \leq c \restrb{g}{G} \\
\end{eqs}
where we used \proveby{OhForPositive} and \proveby{OhOneIsExactlyBounded}. \srequire{OhForPositive} \srequire{OhForOne} 

\proofpart{Zero subset}
By \require{TrivialZero},
\begin{eqs}
{} & \restrb{f}{\compl{G}} \in \oh{\compl{G}}{\restr{g}{\compl{G}}} \\
\iffr & \restrb{f}{\compl{G}} = 0.
\end{eqs}
 
\proofpart{Whole set}
By \require{Local} and \require{SubRestrict},
\begin{eqs}
{} & f \in \ohx{g} \\
\iffr & \restrb{f}{G} \in \oh{G}{\restr{g}{G}} \textrm{ and } \restrb{f}{\compl{G}} \in \oh{\compl{G}}{\restr{g}{\compl{G}}} \\
\iffr & \bra{\exists c \in \posi{\TR}: \restrb{f}{G} \leq c \restrb{g}{G}} \textrm{ and } \restrb{f}{\compl{G}} = 0 \\
\iffr & \exists c \in \posi{\TR}: f \leq c g.
\end{eqs}
\end{proof}

\section{Minimal properties}

In this section we will show that, excluding \property{Local}, a given primitive property can not be deduced from the remaining primitive properties.

\begin{definition}[Minimal set of axioms]
A set $\mathcal{S}$ of axioms is \define{minimal}, if no axiom in $\mathcal{S}$ can be proved from the other axioms in $\mathcal{S}$.
\end{definition}

\begin{note}[Proving minimality]
\label{ProvingMinimality}
The minimality of an axiom set $\mathcal{S}$ can be proved by showing that for any axiom $A \in \mathcal{S}$, there is a model of $\mathcal{S} \setminus \setb{A}$ in which $A$ holds and another model of $\mathcal{S} \setminus \setb{A}$ in which $A$ does not hold.
\end{note}

\begin{note}[\uproperty{Local} is implied]
\uproperty{Local} is implied by the other primitive properties by \proveby{LocalityIsImplied}; primitive properties are not minimal.
\end{note}

\begin{note}[Pre-primitive properties]
The \define{pre-primitive properties} are the primitive properties, with \property{Local} excluded.
\end{note}

\begin{table}
\begin{tabular}{|c|c|c|c|c|c|c|c|c|c|l|}
\hline 
{} & $\leq$ & T & $1$ & $\alpha$ & $*$ & $/$ & $\circ$ \\
\hline 
\hline 
$\mohsy$ & \xmark \ref{MultipleOrderFails} & \checkmark \ref{MultipleTransivity} & \checkmark \ref{MultipleOneSeparation} & \checkmark \ref{MultipleScaleInvariance} & \checkmark \ref{MultipleSubHomogeneity} & \checkmark \ref{MultipleSubHomogeneity} & \checkmark \ref{MultipleSubComposability} \\
\hline
$\nohsy$ & \checkmark \ref{NonTransitiveOrder} & \xmark \ref{NonTransitiveTransivityFails} & \checkmark \ref{NonTransitiveOneSeparation} & \checkmark \ref{NonTransitiveScaleInvariance} & \checkmark \ref{NonTransitiveSubHomogeneity} & \checkmark \ref{NonTransitiveSubHomogeneity} & \checkmark \ref{NonTransitiveSubComposability} \\
\hline
$\trohsy$ & \checkmark \ref{LocalOrderConsistency} & \checkmark \ref{LocalTransitivity} & \xmark \ref{TrivialOneSeparationFails} & \checkmark \ref{LocalScaleInvariance} & \checkmark \ref{LocalSubHomogenuity} & \checkmark \ref{LocalSubHomogenuity} & \checkmark \ref{TrivialSubComposability} \\
\hline
$\eohsy$ & \checkmark \ref{ElementwiseOrder} & \checkmark \ref{ElementwiseTransivity} & \checkmark \ref{ElementwiseOneSeparation} & \xmark \ref{ElementwiseScaleInvarianceFails} & \checkmark \ref{ElementwiseSubHomogeneity} & \checkmark \ref{ElementwiseSubHomogeneity} & \checkmark \ref{ElementwiseSubComposability} \\
\hline
$\aohsy$ & \checkmark \ref{AffineOrderConsistency} & \checkmark \ref{AffineTransitivity} & \checkmark \ref{AffineOneSeparation} & \checkmark \ref{AffinePositiveScaleInvariance} & \xmark \ref{AffineSubHomogenuityFails} & \checkmark \ref{AffineSubHomogeneityNDiv} & \checkmark \ref{AffineSubComposability} \\
\hline
$\wohsy$ & \checkmark \ref{PowerOrder} & \checkmark \ref{PowerTransivity} & \checkmark \ref{PowerOneSeparation} & \checkmark \ref{PowerScaleInvariance} & \checkmark \ref{PowerSubHomogeneity} & \xmark \ref{PowerSubHomogeneityDivNFails} & \checkmark \ref{PowerSubComposability} \\
\hline
$\pohsy$ & \checkmark \ref{LocalOrderConsistency} & \checkmark \ref{LocalTransitivity} & \checkmark \ref{AsymptoticOneSeparation} & \checkmark \ref{LocalScaleInvariance} & \checkmark \ref{LocalSubHomogenuity} & \checkmark \ref{LocalSubHomogenuity} & \xmark \ref{AsymptoticInjectiveSubComposabilityFails} \\
\hline
$\lohsy$ & \checkmark \ref{LocalOrderConsistency} & \checkmark \ref{LocalTransitivity} & \checkmark \ref{LinearOneSeparation} & \checkmark \ref{LocalScaleInvariance} & \checkmark \ref{LocalSubHomogenuity} & \checkmark \ref{LocalSubHomogenuity} & \checkmark \ref{LinearSubComposability} \\
\hline 
\end{tabular}
\centering
\caption{The pre-primitive properties fulfilled by each candidate definition. The abbreviations are: $\leq$ for \property{Order}, T for \property{Trans}, $1$ for \property{One}, $\alpha$ for \property{Scale}, $*$ for \property{NSubHom}, $/$ for \property{NSubDiv}, and $\circ$ for \property{SubComp}}
\label{MinimalProperties}
\end{table}

\begin{theorem}[Pre-primitive properties are minimal]
\label{PreprimitivePropertiesAreMinimal}
Pre-primitive properties form a minimal set of axioms for linear dominance. \srequire{Order} \srequire{Trans} \srequire{One} \srequire{Scale} \srequire{NSubHom} \srequire{NSubDiv} \srequire{SubComp}
\end{theorem}

\begin{proof}
That pre-primitive properties are equivalent to the definition of $\ohsy$-notation as linear dominance follows from \proveby{LocalityIsImplied}, \proveby{PrimitiveImpliesLinear}, and \proveby{LinearImpliesPrimitive}.

That pre-primitive properties are minimal follows from \tref{MinimalProperties}, following \nref{ProvingMinimality}.
\end{proof}

\chapter{Working with the \texorpdfstring{$\ohsy$}{O}-notation}
\label{WorkingWith}

In this section we adopt the linear dominance $\lohsy$-notation as \emph{the} $\ohsy$-notation and develop more refined tools for working with it. These tools are useful in a day-to-day basis for an algorithm analyst, because they provide shortcuts over tedious derivations. A cheat sheet for working with the $\lohsy$-notation is given in \fref{CheatSheet} --- it is a simplified version of \fref{DesirableProperties}.

\begin{table}
\begin{tabular}{|l|l|}
\hline 
Name & Property \\
\hline 
\hline 
\uproperty{Order} & $f \lt g \implies f \in \lohx{g}$  \\
\hline 
\uproperty{Reflex} & $f \in \lohx{f}$ \\
\hline 
\uproperty{Trans} & $\bra{f \in \lohx{g} \textrm{ and } g \in \lohx{h}} \implies f \in \lohx{h}$  \\
\hline 
\uproperty{Orderness} & $f \in \lohx{g} \iff \lohx{f} \subset \lohx{g}$ \\
\hline 
\hline 
\uproperty{Zero} & $1 \not\in \loh{\posi{\TN}}{0}$ \\
\hline 
\uproperty{One} & $n \not\in \loh{\posi{\TN}}{1}$ \\
\hline 
\uproperty{TrivialZero} & $\lohx{0} = \setb{0}$ \\
\hline 
\hline 
\uproperty{Scale} & $\lohx{\alpha f} = \lohx{f}$ \\
\hline 
\usproperty{Translation} & $\lohx{f + \beta + \alpha} = \lohx{f + \beta}$  \\
\hline 
\uproperty{PowerH} & $\lohx{f}^{\alpha} = \lohx{f^{\alpha}}$ \\
\hline 
\uproperty{AddCons} & $u \lohx{f} + v \lohx{f} = (u + v) \lohx{f}$ \\
\hline 
\uproperty{MultiCons} & $\lohx{f}^{u} \cdot \lohx{f}^{v} = \lohx{f}^{u + v}$ \\
\hline 
\uproperty{MaxCons} & $\max(\lohx{f}, \lohx{f}) = \lohx{f}$ \\
\hline 
\uproperty{Local} & $\bra{\forall D \in C: \restrb{f}{D} \in \loh{D}{\restr{g}{D}}} \implies f \in \lohx{g}$  \\
\hline 
\uproperty{Hom} & $u \lohx{f} = \lohx{u f}$ \\
\hline 
\uproperty{Multi} & $\lohx{f} \cdot \lohx{g} = \lohx{fg}$ \\
\hline 
\uproperty{Restrict} & $\bra{\restr{\lohx{f}}{D}} = \loh{D}{\restr{f}{D}}$ \\
\hline 
\uproperty{Additive} & $\lohx{f} + \lohx{g} = \lohx{f + g}$ \\
\hline 
\uproperty{Summation} & $\lohx{f + g} = \lohx{\max(f, g)}$ \\
\hline 
\uproperty{Maximum} & $\max(\lohx{f}, \lohx{g}) = \lohx{\max(f, g)}$ \\
\hline 
\uproperty{MaximumSum} & $\max(\lohx{f}, \lohx{g}) = \lohx{f} + \lohx{g}$ \\
\hline 
\uproperty{SubComp} & $\lohx{f} \circ s \subset \loh{Y}{f \circ s}$ \\
\hline 
\uproperty{IComp} & $\lohx{f} \circ s = \loh{Y}{f \circ s}$ \quad ($s$ injective) \\
\hline 
\uproperty{Extend} & $\lohx{f} \circ \projections{X}{Y} \subset \loh{X \times Y}{f \circ \projections{X}{Y}}$ \\
\hline 
\uproperty{SubsetSum} & $\lohx{\sum_{(y, z) \in S_x} a(z) h(y)} \subset $ \\
{} & $\lohx{\sum_{(y, z) \in S_x} a(z) \bar{h}(y)}$ \\
\hline 
\end{tabular}
\centering
\caption{Cheat sheet for $\protect\lohsy$-notation. Here $X, Y, Z \in U$, $f, g, u, v \in \rc{X}$, $\overline{h} \in \rc{Y}$, $h \in \protect\loh{Y}{\overline{h}}$, $\alpha, \beta \in \posi{\TR}$, $D \subset X$, $\function{s}{Y}{X}$, $\function{S}{X}{\fpower{Y \times Z}}$, $a \in \rc{Z}$, and $C \subset \power{X}$ is a finite cover of $X$.}
\label{CheatSheet}
\end{table}

\section{Surprising aspects}
\label{Faq}

In this section we will look at some aspects of $\lohsy$\-/notation which at first may seem surprising.

\begin{algorithm}
\caption{A family of algorithms, parametrized by $i \in \TN$, which take as input $n \in \TN$, and output $n$.}
\label{alg:InfiniteDescent}
\begin{algorithmic}[1]
\Procedure {$\text{identity}_i$}{$n$}
\If {$n < i$}
\State \Return $n$
\EndIf
\State \Return $n + 0$
\EndProcedure
\end{algorithmic}
\end{algorithm}

\begin{example}[Infinite descent]
\label{InfiniteDescent}
Consider a family of algorithms in \aref{alg:InfiniteDescent}, for which there is a separate implementation for each $i \in \TN$. Suppose we are only interested in the number of performed additions. The cost functions are given by $f_i \in \rc{\TN}$ such that
\begin{eqs}
f_i(n) =
\begin{cases}
1, & n \geq i, \\
0, & \text{otherwise},
\end{cases}
\end{eqs}
where $i \in \TN$. Then
\begin{eqs}
f_0 \in \lthetah{\TN}{1},
\end{eqs}
and
\begin{eqs}
f_{i + 1} \in \lsmalloh{\TN}{f_i},
\end{eqs}
for all $i \in \TN$. That is, with linear dominance, the functions $f_i$ form a decreasing sequence of functions. This is how it should be: it is fundamentally different to use resources --- no matter how small an amount --- than to not use resources at all. 
\end{example}

\begin{example}[Same expressions, different functions]
It may be surprising that $1 / x \in \loh{\posi{\TN}}{1}$, but $1 / x \not\in \loh{\posi{\TR}}{1}$. The function $1 / x$ seems to be the same --- why don't they belong to the same $\lohsy$-sets?

A function is a triple, which consists of the domain, the codomain, and the rule connecting each element of the domain to exactly one element of the codomain. An expression such as $1 / x$ is ambiguous as a function definition, because it does not specify the domain and the codomain. The two occurrences of $1 / x$ here are different functions, and need not share any other property apart from being equal on $\posi{\TN}$.

In our example, $1 / x$ on $\posi{\TN}$ is bounded, while $1 / x$ on $\posi{\TR}$ is unbounded.
\end{example}

\section{Master theorems}

Master theorems are popular for solving recurrence equations arising in the analysis of divide-and-conquer algorithms \cite{IntroAlgo2009} --- up to $\lohsy$-equivalence. We state the theorems here, and prove them in \sref{MasterTheorems}.

\newcommand{\ReMasterFunctionOverIntegers}{
\begin{definition}[Master function over integers]
Let $a \in \nonnb{\TR}{1}$, $b \in \nonnb{\TR}{2}$, $d \in \posi{\TR}$, and $F \in \rc{\nonnb{\TN}{1}}$. 
A \definesub{master function}{over integers} is a function $T \in \rc{\nonnb{\TN}{1}}$ defined by the recurrence equation
\begin{equation}
T(n) =
\begin{cases}
a T(\ceilb{n / b}) + F(n), & n \geq b, \\
d, & n < b.
\end{cases}
\end{equation}
The set of such functions is denoted by $\imasters{a}{b}{d}{F}$.
\end{definition}
}

\ReMasterFunctionOverIntegers

\newcommand{\ReMasterTheoremOverIntegers}[1][]{
\begin{theorem}[Master theorem over integers]
\label{MasterTheoremOverIntegers#1}
\srequire{LinearDominance}
Let $T \in \imasters{a}{b}{d}{F}$ be a Master function over integers, and $F \in \loh{\nonnb{\TN}{1}}{n^c}$, where $c \in \nonn{\TR}$. Then
\begin{eqs}
\lgb{b}{a} < c & \implies T \in \loh{\nonnb{\TN}{1}}{n^c}, \\
\lgb{b}{a} = c & \implies T \in \loh{\nonnb{\TN}{1}}{n^c \lgb{b}{bn}}, \\
\lgb{b}{a} > c & \implies T \in \lthetah{\nonnb{\TN}{1}}{n^{\lgb{b}{a}}}.
\end{eqs}
If $F \in \lthetah{\nonnb{\TN}{1}}{n^c}$, then each $\lohs{\nonnb{\TN}{1}}$ can be replaced with $\lthetahs{\nonnb{\TN}{1}}$.
\end{theorem}
}

\ReMasterTheoremOverIntegers[Text]

\newcommand{\ReMasterFunctionOverReals}{
\begin{definition}[Master function over reals]
Let $a \in \nonnb{\TR}{1}$, $b \in \posib{\TR}{1}$, $d \in \posi{\TR}$, and $f \in \rc{\nonnb{\TR}{1}}$. A \definesub{master function}{over reals} is a function $t \in \rc{\nonnb{\TR}{1}}$ defined by the recurrence equation
\begin{equation}
t(x) =
\begin{cases}
a t(x / b) + f(x), & x \geq b, \\
d, & x < b.
\end{cases}
\end{equation}
The set of such functions is denoted by $\rmasters{a}{b}{d}{f}$.
\end{definition}
}

\ReMasterFunctionOverReals

\newcommand{\ReMasterTheoremOverReals}[1][]{
\begin{theorem}[Master theorem over reals]
\label{MasterTheoremOverReals#1}
\srequire{LinearDominance}
Let $t \in \rmasters{a}{b}{d}{f}$ be a Master function over reals, and $f \in \loh{\nonnb{\TR}{1}}{x^c}$, where $c \in \nonn{\TR}$. Then
\begin{eqs}
\lgb{b}{a} < c & \implies t \in \loh{\nonnb{\TR}{1}}{x^c}, \\
\lgb{b}{a} = c & \implies t \in \loh{\nonnb{\TR}{1}}{x^c \lgb{b}{bx}}, \\
\lgb{b}{a} > c & \implies t \in \lthetah{\nonnb{\TR}{1}}{x^{\lgb{b}{a}}}.
\end{eqs}
If $f \in \lthetah{\nonnb{\TR}{1}}{n^c}$, then each $\lohs{\nonnb{\TR}{1}}$ can be replaced with $\lthetahs{\nonnb{\TR}{1}}$.
\end{theorem}
}

\ReMasterTheoremOverReals[Text]

\section{\texorpdfstring{$\protect\lohsy$}{O}-mappings}
\label{OMappings}

\begin{definition}[$\lohsy$-mapping]
A function $\function{T}{\rc{X}}{\rc{Y}}$ is an \define{$\lohsy$-mapping}, if
\begin{eqs}
\image{T}{\lohx{f}} \subset \loh{Y}{T(f)},
\end{eqs}
for all $f \in \rc{X}$. 
\end{definition}

\begin{theorem}[$\lohsy$-mapping by linear dominance]
\label{LinearDominanceOMapping}
Let $\function{T}{\rc{X}}{\rc{Y}}$. Then $T$ is an $\lohsy$-mapping if and only if
\begin{equation}
\bra{\exists c \in \posi{\TR} : f \lt cg} \implies 
\bra{\exists d \in \posi{\TR} : T(f) \lt dT(g)}
\end{equation}
for all $f, g \in \rc{X}$.
\end{theorem}

\begin{proof}
By definition.
\end{proof}

\begin{example}[Non-negative translation is an $\lohsy$-mapping]
Let $\function{T}{\rc{X}}{\rc{X}}$ be such that $T(f) = f + \alpha$, where $\alpha \in \nonn{\TR}$. Let $f, g \in \rc{X}$ be such that $f \in \lohx{g}$. Then there exists $c \in \posi{\TR}$ such that $f \lt cg$, and
\begin{eqs}
T(f) & = f + \alpha \\
{} & \leq cg + \alpha \\
{} & \leq \max(c, 1) (g + \alpha) \\
{} & = \max(c, 1) T(g).
\end{eqs}
Therefore $T$ is an $\lohsy$-mapping. The inverse of $T$ does not exist, since it does not always hold that $f - \alpha \geq 0$.
\end{example}

\begin{example}[Composition is an $\lohsy$-mapping]
Let $\function{T}{\rc{X}}{\rc{Y}}$ be such that $T(f) = f \circ s$, where $\function{s}{Y}{X}$. Let $f, g \in \rc{X}$ be such that $f \in \lohx{g}$. Then there exists $c \in \posi{\TR}$ such that $f \lt cg$, and
\begin{eqs}
T(f) & = f \circ s \\
{} & \leq (cg) \circ s \\
{} & = c (g \circ s) \\
{} & = c T(g).
\end{eqs}
Therefore $T$ is an $\lohsy$-mapping.
\end{example}

\begin{theorem}[Subset-sum is an $\lohsy$-mapping]
\label{SubsetSumIsAnOMapping}
Let $X, Y, Z \in U$, $\function{S}{X}{\fpower{Y \times Z}}$, $a \in \rc{Z}$, $f \in \rc{Y}$, and $\hat{f} \in \lohx{f}$. Let $\function{T}{\rc{Y}}{\rc{X}}$ be such that
\begin{equation}
T(f) = \bra{x \mapsto \sum_{(y, z) \in S_x} a(z) f(y)}.
\end{equation}
Then $T$ is an $\lohsy$-mapping. \srequire{LinearDominance} \sprove{SubsetSum}
\end{theorem}

\begin{proof}
\srequire{LinearDominance}
There exists $c \in \posi{\TR}$ such that $\hat{f} \lt cf$, and so
\begin{eqs}
T(\hat{f}) & = \bra{x \mapsto \sum_{(y, z) \in S_x} a(z) \hat{f}(y)} \\
{} & \leq \bra{x \mapsto \sum_{(y, z) \in S_x} a(z) (cf)(y)} \\
{} & = c \bra{x \mapsto \sum_{(y, z) \in S_x} a(z) f(y)} \\
{} & = c T(f).
\end{eqs}
Therefore $T$ is an $\lohsy$-mapping. \sprove{SubsetSum}
\end{proof}

\section{\texorpdfstring{$\protect\lohsy$}{O}-equalities}

% \begin{definition}[Right $\lohsy$-inverse]
% A function $\function{\rinvs{T}}{\rc{Y}}{\rc{X}}$ is a \define{right $\lohsy$-inverse} of $\function{T}{\rc{X}}{\rc{Y}}$, if
% \begin{eqs}
% \loh{Y}{g} = \loh{Y}{T(\rinv{T}{g})},
% \end{eqs}
% for all $g \in \rc{Y}$.
% \end{definition}

\begin{definition}[$\lohsy$-residual]
A function $\function{\residuals{T}}{\rc{Y}}{\rc{X}}$ is an \defineexp{$\lohsy$-residual}{$\lohsy$-residual} of $\function{T}{\rc{X}}{\rc{Y}}$, if
\begin{equation}
g \in \loh{Y}{T(f)} \iff \residual{T}{g} \in \loh{X}{f},
\end{equation}
for all $f \in \rc{X}$, and $g \in \rc{Y}$. 
\end{definition}

\begin{definition}[$\lohsy$-residuated function]
A function $\function{f}{X}{Y}$ is \defineexp{$\lohsy$-residuated}{$\lohsy$-residual}, if it has an $\lohsy$-residual.
\end{definition}

\begin{note}[]
The definitions given here are special cases of the theory of partitioned sets --- given in \sref{PartitionedSets} --- and of the theory of preordered sets --- given in \sref{PreorderedSets}, formulated in terms of $\lohsy$-sets.
\end{note}

\begin{definition}[Strong $\lohsy$-equality]
A \define{strong $\lohsy$-equality} is a surjective $\lohsy$-residuated function $\function{T}{\rc{X}}{\rc{Y}}$. 
\end{definition}

\begin{theorem}[Strong $\lohsy$-equality rule]
\label{StrongOEqualityRule}
Let $\function{T}{\rc{X}}{\rc{Y}}$ be a strong $\lohsy$-equality. Then
\begin{equation}
T(\lohx{f}) = \loh{Y}{T(f)},
\end{equation}
for all $f \in \rc{X}$.
\end{theorem}

\begin{proof}
This is proved in \thref{TransposeResiduatedSurjectivePreservesDownSets}.
\end{proof}

\begin{note}[$\lohsy$-equality]
We call the $\lohsy$-equality strong, because we do not know how to neatly characterize the preservation of the $\lohsy$-set under a mapping $T$; such functions would be called $\lohsy$-equalities.
\end{note}

\begin{theorem}[Strong $\lohsy$-equality by linear dominance]
\label{StrongOEqualityByLinearDominance}
Suppose $\function{T}{\rc{X}}{\rc{Y}}$. Then $T$ is a strong $\lohsy$-equality if and only if there exists $\function{\residuals{T}}{\rc{Y}}{\rc{X}}$ such that
\begin{eqs}
T(\residual{T}{g}) = g,
\end{eqs}
and
\begin{equation}
\bra{\exists c \in \posi{\TR} : g \lt c T(f)} \iff 
\bra{\exists d \in \posi{\TR} : \residual{T}{g} \lt d f},
\end{equation}
for all $f \in \rc{X}$, and $g \in \rc{Y}$.
\end{theorem}

\begin{proof}
By definition of linear dominance and \thref{PSurjectivityIsEquivalentToHavingRightInverse}, whose special case says that surjectivity is equivalent to having a right-inverse.
\end{proof}

\begin{example}[Injective composition is a strong $\lohsy$-equality]
Let $\function{s}{Y}{X}$ be injective, and $\function{T}{\rc{X}}{\rc{Y}}$ be such that $T(f) = f \circ s$. Let $\function{\residuals{s}}{Y}{\image{s}{Y}}$ be such that $\residual{s}{y} = s(y)$. Then $\residuals{s}$ is bijective, and $T(f) = f \circ \residuals{s}$. Let $\function{\residuals{T}}{\rc{Y}}{\rc{X}}$ be such that $\residual{T}{g} = \left\langle g \circ \invs{\residuals{s}} \right\rangle$, where $\langle \cdot \rangle$ is a domain extension of a function from $\image{s}{Y}$ to $X$ by mapping the new elements to zeros. Let $f \in \rc{X}$, $g \in \rc{Y}$, and $c \in \posi{\TR}$. Then $T(\residual{T}{g}) = g$, and 
\begin{eqs}
{} & g \lt c T(f) \\
\iffr & g \lt c (f \circ \residuals{s}) \\
\iffr & \left\langle g \circ \invs{\residuals{s}} \right\rangle \lt c f \\
\iffr & \residual{T}{g} \lt c f.
\end{eqs}
Therefore $T$ is a strong $\lohsy$-equality.
\end{example}

\begin{example}[Positive power is a strong $\lohsy$-equality]
Let $\function{T}{\rc{X}}{\rc{X}}$ be such that $T(f) = f^{\alpha}$, where $\alpha \in \posi{\TR}$. Let $\function{\residuals{T}}{\rc{X}}{\rc{X}}$ be such that $\residual{T}{g} = g^{1 / \alpha}$. Let $f, g \in \rc{X}$, and $c \in \posi{\TR}$. Then $T(\residual{T}{g}) = g$, and
\begin{eqs}
{} & g \lt c T(f) \\
\iffr & g \lt c f^{\alpha} \\
\iffr & g^{1 / \alpha} \lt c^{1 / \alpha} f \\
\iffr & \residual{T}{g} \lt c^{1 / \alpha} f.
\end{eqs}
Therefore $T$ is a strong $\lohsy$-equality.
\end{example}

\begin{example}[Positive multiplication is a strong $\lohsy$-equality]
Let $\function{T}{\rc{X}}{\rc{X}}$ be such that $T(f) = \alpha f$, where $\alpha \in \posi{\TR}$. Let $\function{\residuals{T}}{\rc{X}}{\rc{X}}$ be such that $\residual{T}{g} = g / \alpha$. Let $f, g \in \rc{X}$, and $c \in \posi{\TR}$. Then $T(\residual{T}{g}) = g$, and
\begin{eqs}
{} & g \lt c T(f) \\
\iffr & g \lt c \alpha f \\
\iffr & g / \alpha \lt c f \\
\iffr & \residual{T}{g} \lt c f.
\end{eqs}
Therefore $T$ is a strong $\lohsy$-equality.
\end{example}

\chapter{Local linear dominance}
\label{LocalLinearDominance}

\begin{table}
\begin{tabular}{|l|l|l|l|}
\hline 
Name & Universe & $\filterset{X}$ & Reference \\
\hline 
\hline
$\trohsy$ Trivial &
sets &
$\setb{\emptyset}$ &
\cite{ONotationBeatcs} \\
\hline 
$\fohsy$ Cofinite & 
sets &
$\cpower{X}$ & 
\cite{DesignAndAnalysisOfComputerAlgorithms}, \cite{ONotationBeatcs} \\
\hline 
$\pohsy$ Asymptotic &
$\bigcup_{d \in \TN} \power{\TR^d}$ &
$\setb{\nonnb{X}{y} : y \in \TR^d}$ &
\cite{IntroAlgo} \\
\hline 
$\cohsy$ Co-asymptotic &
$\bigcup_{d \in \TN} \power{\TR^d}$ &
$\setb{\nlb{X}{y} : y \in \TR^d}$ &
\cite{IntroAlgo2009} \\
\hline 
$\lohsy$ Full &
sets &
$\setb{X}$ & 
\cite{ONotationBeatcs} \\
\hline 
\end{tabular}
\centering
\caption{Example versions of local linear dominance. Here $U$ is the universe, $X \in U$, and $\filterset{X}$ is the filter basis in $X$. The name is used to replace the word `local', as in asymptotic linear dominance.}
\label{ExamplesOfLocalLinearDominances}
\end{table}

In this section we will study a class of candidate definitions for the $\ohsy$-notation, the \emph{local linear dominances}. These definitions work over all universes --- unless a specific version makes additional assumptions. 

\section{Definition}

In this section we provide the definition of local linear dominance. We prove its properties in \sref{ProofsForLocalLinearDominance}.

\begin{definition}[Filter basis]
A set $\filtersetsym \subset \power{X}$ is a \define{filter basis} in a set $X$ if it is
\begin{description}
\item[non-empty] \hfill \\ 
$\filtersetsym \neq \emptyset$,
\item[$\subset$-directed] \hfill \\
$\forall A, B \in \filtersetsym, \exists C \in \filtersetsym : C \subset A \textrm{ and } C \subset B$.
\end{description}
\end{definition}

\begin{note}[Filter basis may not be proper]
Some authors require a filter basis to be proper --- $\emptyset \not\in \filterset{X}$. We allow a filter basis to be non-proper.
\end{note}

\begin{definition}[Local linear dominance]
\defineexp{Local linear dominance}{linear dominance!local} $\rohsy$ on $X \in U$ is defined by $g \in \rohx{f}$ if and only if
\begin{equation}
\exists c \in \posi{\TR}, \exists A \in \filterset{X}: \restrb{g}{A} \leq c \restrb{f}{A},
\end{equation}
where $\setb{\filterset{X} \subset \power{X} : X \in U}$ is a class of filter bases with \define{induced sub-structure}:
\begin{eqs}
\filterset{D} = \setb{A \cap D : A \in \filterset{X}},
\end{eqs}
for all $D \subset X$.
\end{definition}

\begin{note}[Versions]
Each choice of filter bases corresponds to a version of local linear dominance. Some such versions are given in \taref{ExamplesOfLocalLinearDominances}.
\end{note}

\begin{note}[Filter basis and limits]
A filter basis in $X$ is the minimal amount of structure needed to make sense of the limit of a function $f \in \rc{X}$. In fact, a local linear dominance can be characterized by a ratio\-/limit.
\end{note}

\begin{note}[Motivation]
We have already shown in \sref{Characterization} that only one instance of local linear dominance works for algorithm analysis. Why study local linear dominances? There are two reasons.

First, local linear dominance is commonly used for local function approximation in various fields of mathematics. Such results may indirectly find their way into a complexity analysis, in which case we need a way transfer such results to the algorithmic side. This can be done using the tools provided in \sref{LimitTheorems}.

Second, local linear dominances help to train the intuition behind the primitive properties, because they provide successively better approximations to linear dominance. This is shown in \sref{ProofsForLocalLinearDominance}.
\end{note}

\begin{note}[Eventually non-negative]
A function $\function{f}{X}{\TR}$ is \define{eventually non-negative}, if there exists $A \in \filterset{X}$, such that $\restrb{f}{A} \geq 0$. It is possible to generalize the definition of local linear dominance to functions which are eventually non\-/negative. When $\filterset{X} = \setb{X}$, eventually non\-/negative reduces to non\-/negative.
\end{note}

\newcommand{\ReLinearDominanceFromLocalLinearDominance}[1][]{
\begin{theorem}[$\lohsy$ is almost equal to $\rohsy$ for cofinite filter sets]
\label{LinearDominanceFromLocalLinearDominance#1}
Suppose $\card{X \setminus A} < \infty$, for all $A \in \filterset{X}$. Then
\begin{eqs}
\rohx{g} = \lohx{g},
\end{eqs}
for all $g \in \posi{\rc{X}}$. 
\end{theorem}
}

\ReLinearDominanceFromLocalLinearDominance[Text]

\begin{example}[Linear dominance from asymptotic linear dominance]
\label{LinearDominanceFromAsymptoticLinearDominance}
Let $f \in \rc{\posi{\TN}}$ be such that $f(n) = n^2 \abs{\sin(n)} + n + 3$. It holds that $f \in \poh{\posi{\TN}}{n^2}$ by \eref{AsymptoticLinearDominanceFromALimitInN}. Since $n^2 > 0$, and $\card{\posi{\TN} \setminus \nonnb{\TN}{y}} < \infty$, for all $y \in \TR$, it holds that $f \in \loh{\posi{\TN}}{n^2}$ by \thref{LinearDominanceFromLocalLinearDominance}.
\end{example}

\section{Limit theorems}
\label{LimitTheorems}

In this section we show how to transfer a result from a local linear dominance $\rohsy$ to linear dominance $\ohsy$. First, the filter bases associated with $\rohsy$ make it possible to define the concepts of limit superior, limit inferior, and limit. These limits can then be used to characterize the $\rohsy$-notation. Second, $f \in \rohx{g}$ can sometimes be used to deduce $f \in \lohx{g}$. The proofs are given in \sref{ProofsOfLimitsTheorems}.

\begin{note}[Related notations for local linear dominance]
We shall denote the related notations corresponding to $\rohsy$ by $\romegahsy$, $\romegasy$, $\rthetahsy$, and $\rsmallohsy$.
\end{note}

\begin{definition}[Limit superior and limit inferior under a filter basis]
Let $\mathcal{F} \subset \power{X}$ be a filter basis in a set $X$. Then
\begin{eqs}
\limsup_{\mathcal{F}} f & \coloneqq \inf \gset{\sup \image{f}{A}}{A \in \mathcal{F}}, \\
\liminf_{\mathcal{F}} f & \coloneqq \sup \gset{\inf \image{f}{A}}{A \in \mathcal{F}},
\end{eqs}
for all $f \in \rc{X}$.
\end{definition}

\begin{note}[Existence of limit superior and limit inferior]
The $\limsup$ and $\liminf$ are called the \define{limit superior} and the \define{limit inferior}, respectively. By the completeness of $\TXR = \TR \cup \setb{-\infty, +\infty}$, both of them are well-defined as a number in $\TXR$.
\end{note}

\begin{definition}[Limit under a filter basis]
The \define{limit} of $f \in \rc{X}$ under a filter basis $\mathcal{F} \subset \power{X}$ is
\begin{eqs}
\lim_{\mathcal{F}} f = c,
\end{eqs}
whenever $\limsup_{\mathcal{F}} f = \liminf_{\mathcal{F}} f = c \in \TXR$.
\end{definition}

\begin{note}[Division by zero and infinity]
In this section, we use the conventions that $\alpha / 0 = \infty \in \TXR$, for all $\alpha \in \posi{\TR}$, and that $\beta / \infty = 0$, for all $\beta \in \TR$.
\end{note}

\newcommand{\ReRelationBetweenRatioLimits}[1][]{
\begin{theorem}[Relation between ratio-limits]
\label{RelationBetweenRatioLimits#1}
\begin{eqs}
\limsup_{\filterset{F}} \frac{\restrb{f}{F}}{\restrb{g}{F}} = 1 / \liminf_{\filterset{F}} \frac{\restrb{g}{F}}{\restrb{f}{F}},
\end{eqs}
for all $f, g \in \rc{X}$, where $F = \preimage{f}{\posi{\TR}}$.
\end{theorem}
}

\ReRelationBetweenRatioLimits[Text]

\newcommand{\ReLocalLinearOByALimit}[1][]{
\begin{theorem}[$\rohsy$ by a limit]
\label{LocalLinearOByALimit#1}
\begin{eqs}
\limsup_{\filterset{F}} \frac{\restrb{f}{F}}{\restrb{g}{F}} < \infty & \iff f \in \rohx{g},
\end{eqs}
for all $f, g \in \rc{X}$, where $F = \preimage{f}{\posi{\TR}}$.
\end{theorem}
}

\ReLocalLinearOByALimit[Text]

\begin{example}[Asymptotic linear dominance from a limit in $\posi{\TN}$]
\label{AsymptoticLinearDominanceFromALimitInN}
Consider asymptotic linear dominance $\pohs{\posi{\TN}}$. Let $f \in \rc{\posi{\TN}}$ be such that $f(n) = n^2 \abs{\sin(n)} + n + 3$. Since $f(n) > 0$, and
\begin{eqs}
\limsup_{\filterset{\posi{\TN}}} \frac{f(n)}{n^2} = \limsup_{n \to \infty} \frac{f(n)}{n^2} = 1,
\end{eqs}
it holds that $f \in \poh{\posi{\TN}}{n^2}$ by \thref{LocalLinearOByALimit}
\end{example}

\newcommand{\ReLocalLinearOmegaByALimit}[1][]{
\begin{theorem}[$\romegahsy$ by a limit]
\label{LocalLinearOmegaByALimit#1}
\begin{eqs}
\liminf_{\filterset{G}} \frac{\restrb{f}{G}}{\restrb{g}{G}} > 0 & \iff f \in \romegahx{g},
\end{eqs}
for all $f, g \in \rc{X}$, where $G = \preimage{g}{\posi{\TR}}$.
\end{theorem}
}

\ReLocalLinearOmegaByALimit[Text]

\newcommand{\ReLocalLinearSmallOhByALimit}[1][]{
\begin{theorem}[$\rsmallohsy$ by a limit]
\label{LocalLinearSmallOhByALimit#1}
\begin{eqs}
\bra{\limsup_{\filterset{F}} \frac{\restrb{f}{F}}{\restrb{g}{F}} < \infty \land \liminf_{\filterset{G}} \frac{\restrb{f}{G}}{\restrb{g}{G}} = 0} \iff f \in \rsmallohx{g},
\end{eqs}
for all $f, g \in \rc{X}$, where $F = \preimage{f}{\posi{\TR}}$ and $G = \preimage{g}{\posi{\TR}}$.
\end{theorem}
}

\ReLocalLinearSmallOhByALimit[Text]

\begin{example}[Co-asymptotic linear dominance from a limit in $\TN^2$]
\label{CoasymptoticLinearDominanceFromALimitInPlane}
Consider co\-/asymptotic linear dominance $\cohs{\TN^2}$. Let $f \in \rc{\TN^2}$ be such that $f(m, n) = (m + n + 1)\exp(-mn)$. Since $f(m, n) > 0$, $m + n + 1 > 0$, and
\begin{eqs}
\limsup_{\filterset{\TN^2}} \frac{f(m, n)}{m + n + 1} & = 1, \\
\liminf_{\filterset{\TN^2}} \frac{f(m, n)}{m + n + 1} & = 0,
\end{eqs}
it holds that $f \in \smalloh{\TN^2}{m + n + 1}$ by \thref{LocalLinearSmallOhByALimit}.
\end{example}

\begin{example}[Sufficient limit-condition for $\rsmallohsy$]
In particular,
\begin{eqs}
\limsup_{\filterset{F}} \frac{\restrb{f}{F}}{\restrb{g}{F}} = 0 \implies f \in \rsmallohx{g},
\end{eqs}
but this is only a sufficient condition.
\end{example}

\newcommand{\ReTraditionalSmallOByALimit}[1][]{
\begin{theorem}[Traditional $\smallohs{X}$ by a limit]
\label{TraditionalSmallOByALimit#1}
\begin{eqs}
{} & \forall \epsilon \in \posi{\TR}: \exists A \in \filterset{X}: \restrb{f}{A} \leq \epsilon \restrb{g}{A} \\
\iffr & \limsup_{\filterset{F}} \frac{\restrb{f}{F}}{\restrb{g}{F}} = 0,
\end{eqs}
for all $f, g \in \rc{X}$, where  $F = \preimage{f}{\posi{\TR}}$.
\end{theorem}
}

\ReTraditionalSmallOByALimit[Text]

\begin{example}[Asymptotic linear dominance from a limit in $\TN^2$]
\label{AsymptoticLinearDominanceFromALimitInPlane}
Consider asymptotic linear dominance $\pohs{\TN^2}$. Let $f \in \rc{\TN^2}$ be such that $f(m, n) = (m + n + 1)\exp(-mn)$. Since $f > 0$, and
\begin{eqs}
\limsup_{\filterset{\TN^2}} \frac{f}{1} = 0,
\end{eqs}
it holds that $f \in \psmalloh{\TN^2}{1}$ by \thref{LocalLinearSmallOhByALimit}.
\end{example}

\newcommand{\ReLocalLinearSmallOmegaByALimit}[1][]{
\begin{theorem}[$\romegasy$ by a limit]
\label{LocalLinearSmallOmegaByALimit#1}
\begin{eqs}
\bra{\limsup_{\filterset{F}} \frac{\restrb{f}{F}}{\restrb{g}{F}} = \infty \land \liminf_{\filterset{G}} \frac{\restrb{f}{G}}{\restrb{g}{G}} > 0} \iff f \in \romegax{g},
\end{eqs}
for all $f, g \in \rc{X}$, where $F = \preimage{f}{\posi{\TR}}$ and $G = \preimage{g}{\posi{\TR}}$.
\end{theorem}
}

\ReLocalLinearSmallOmegaByALimit[Text]

\begin{example}[Sufficient limit-condition for $\romegasy$]
In particular,
\begin{eqs}
\liminf_{\filterset{G}} \frac{\restrb{f}{G}}{\restrb{g}{G}} = \infty \implies f \in \romegax{g},
\end{eqs}
but this is only a sufficient condition.
\end{example}

\chapter{Conclusion}
\label{Conclusion}

This \manuscript{} provides a rigorous mathematical foundation for the $\ohsy$-notation and its related notations in algorithm analysis. To find the appropriate definition, an exhaustive list of desirable properties was constructed, and their relations were studied. This revealed that the desirable properties can be reduced to \nprim{} primitive properties which imply the others. It was shown that these primitive properties are equivalent to the definition of the $\ohsy$-notation as linear dominance. Master theorems were shown to hold for linear dominance, and $\lohsy$-mappings were defined for easily proving new rules for the $\lohsy$-notation. Other existing definitions were studied and compared to each other based on the primitive properties. Some misuses of the $\ohsy$-notation from the literature were pointed out.

\We{} hope this \manuscript{} to improve the teaching of the topic, and to improve the communication between computer scientists. Computer scientists have used the $\ohsy$-notation correctly intuitively. That intuition is now backed by a solid mathematical foundation.

\ifarxiv
% Reads the bibliography from asymptotic.bib.
\bibliography{asymptotic}
\bibliographystyle{unsrt}
\else
% The non-arxiv version uses biblatex.
\printbibliography
\fi

\appendix

\chapter{Notation}
\label{Notation}

In this section we provide the definitions for the used notation. This is to avoid ambiguities arising from differing conventions, such as a differing order of composition of functions. We will assume the Zermelo-Fraenkel set-theory with the axiom of choice, abbreviated ZFC.  

\We{} refer to theorems in the form \thref{LinearImpliesPrimitive} --- a short summary followed by a number. \We{} believe this is more useful than \Cref{LinearImpliesPrimitive} when studying the proofs. The numberings in equations, definitions, theorems etc. share the same counter.

\begin{definition}[Numbers]
The set of natural numbers, integers, and real numbers are denoted by $\TN = \setb{0, 1, 2, ...}$, $\TZ$, and $\TR$, respectively. 
\end{definition}

\begin{definition}[Subsets of a set]
The \define{set of subsets} of a set $X$ is denoted by $\power{X}$.
\end{definition}

\begin{definition}[Finite subsets of a set]
The \define{set of finite subsets} of a set $X$ is denoted by
\begin{eqs}
\fpower{X} \coloneqq \setb{D \in \power{X} : \card{D} < \infty}. 
\end{eqs}
\end{definition}

\begin{definition}[Cofinite subsets of a set]
The \define{set of cofinite subsets} of a set $X$ is denoted by
\begin{eqs}
\cpower{X} \coloneqq \setb{D \in \power{X} : \card{X \setminus D} < \infty}. 
\end{eqs}
\end{definition}

\begin{definition}[Cover of a set]
A set $C$ is a \define{cover} of a set $X$, if
\begin{eqs}
X \subset \bigcup C. 
\end{eqs}
\end{definition}

\begin{definition}[Set of functions]
The set of functions from a set $X$ to a set $Y$ is denoted by $\functions{X}{Y}$, or alternatively by $Y^X$. 
\end{definition}

\begin{definition}[Identity function]
The \define{identity function} in a set $X$ is a function $\function{\iden{X}}{X}{X}$ such that
\begin{eqs}
\iden{X}(x) = x. 
\end{eqs}
\end{definition}

\begin{definition}[Composition]
The \define{composition} of $\function{g}{Y}{Z}$ and $\function{f}{X}{Y}$ is $\function{(g \circ f)}{X}{Z}$ such that $(g \circ f)(x) = g(f(x))$. 
\end{definition}

\begin{definition}[Restriction]
The \define{restriction} of $\function{f}{X}{Y}$ to $D \subset X$ is $\function{\bra{\restr{f}{D}}}{D}{Y}$ such that $\bra{\restr{f}{D}}(x) = f(x)$. 
\end{definition}

\begin{definition}[Inverse of a function]
The \define{inverse} of a function $\function{f}{X}{Y}$ is $\function{\invs{f}}{Y}{X}$ such that
\begin{eqs}
f(\inv{f}{y}) = y, \\
\inv{f}{f(x)} = x,
\end{eqs}
for all $x \in X$, $y \in Y$.
\end{definition}

\begin{definition}[Image of a set under a function]
The \define{image} of a set $S \subset X$ under a function $\function{f}{X}{Y}$ is
\begin{eqs}
\image{f}{S} = \setb{f(x) : x \in S}.
\end{eqs}
\end{definition}

\begin{definition}[Pre-image of a set under a function]
The \define{pre-image} of a set $S \subset Y$ under a function $\function{f}{X}{Y}$ is
\begin{eqs}
\preimage{f}{S} = \setb{x \in X : f(x) \in S}.
\end{eqs}
\end{definition}

\begin{definition}[Increasing function]
A function $\function{f}{\nonn{\TR}}{\nonn{\TR}}$ is \define{increasing}, if $x \leq y \implies f(x) \leq f(y)$, for all $x, y \in \nonn{\TR}$. 
\end{definition}

\begin{example}[Non-negative powers are increasing]
Let $\function{f}{\nonn{\TR}}{\nonn{\TR}}$ be such that $f(x) = x^{\alpha}$, where $\alpha \in \nonn{\TR}$. Then $f$ is increasing.
\end{example}

\begin{definition}[Class of sets]
A \define{class of sets} is an unambiguous collection of sets, which may or may not be a set itself. 
\end{definition}

\begin{definition}[Proper class]
A \define{proper class} is a class of sets which is not a set. 
\end{definition}

\begin{example}[Classes of sets]
Every set is a class of sets. The collection of all sets in ZFC is a proper class.
\end{example}

\begin{note}[Formalization of proper classes]
The proper classes are not formalizable in the ZFC set-theory. This is not a problem for two reasons. First, everything in this \manuscript{} can be carried through without forming proper classes, by only referring to the class elements and their relationships. Alternatively, we may adopt the von Neumann-Bernays-G\"odel set theory \cite{IntroToMathematicalLogic} --- a conservative extension of ZFC which formalizes proper classes.
\end{note}

\begin{definition}[Set of relations]
The set of relations between a set $X$ and a set $Y$ is denoted by $\relations{X}{Y}$. 
\end{definition}

\begin{definition}[Reflexive relation]
A relation $\relationin{\sim}{X}$ is \define{reflexive}, if $x \sim x$, for all $x \in X$.
\end{definition}

\begin{definition}[Transitive relation]
A relation $\relationin{\sim}{X}$ is \define{transitive}, if
\begin{eqs}
x \sim y \land y \sim z \implies x \sim z,
\end{eqs}
for all $x, y, z \in X$.
\end{definition}

\begin{definition}[Preorder]
A \define{preorder} in a set $X$ is a reflexive and transitive relation $\relationin{\preleq}{X}$.
\end{definition}

\begin{definition}[Filtered set]
Let $X$ be a set, and $\relationin{\sim}{X}$. Then
\begin{equation}
X^{\sim y} \coloneqq \setb{x \in X : x \sim y},
\end{equation}
for all $y \in X$. 
\end{definition}

\begin{example}
$\nonn{\TR}$ is the set of non-negative real numbers.
\end{example}

\begin{definition}[Indicator function]
The \define{indicator function} of $S \subset X$ in a set $X$ is a function $\function{\idf{S}}{X}{\nonn{\TR}}$ such that 
\begin{equation}
\idf{S}(x) = 
\begin{cases}
1, & x \in S, \\
0, & x \not\in S.
\end{cases}
\end{equation}
\end{definition}

\begin{definition}[Non-negative real-valued functions]
We are specifically interested in non-negative real-valued functions; we define 
\begin{equation}
\rc{X} \coloneqq \bra{X \to \nonn{\TR}},
\end{equation}
where $X$ is a set. 
\end{definition}

\begin{definition}[Extension of a unary operator to functions]
A unary operator $\function{\ominus}{\nonn{\TR}}{\nonn{\TR}}$ is extended to functions $f \in \rc{X}$ by
\begin{equation}
(\ominus f)(x) = \ominus f(x),
\end{equation}
for all $x \in X$.
\end{definition}

\begin{definition}[Extension of a binary operator to functions]
A binary operator $\function{\oplus}{\nonn{\TR} \times \nonn{\TR}}{\nonn{\TR}}$ is extended to functions $f, g \in \rc{X}$ by
\begin{equation}
(f \oplus g)(x) = f(x) \oplus g(x),
\end{equation}
for all $x \in X$.
\end{definition}

\begin{definition}[Extension of a relation to functions]
A relation $\relationin{\sim}{Y}$ is extended to functions $\function{f, g}{X}{Y}$ by
\begin{equation}
f \sim g \coloniff \forall x \in X: f(x) \sim g(x).
\end{equation}
\end{definition}

\begin{example}
It holds that $f \geq 0$, for all $f \in \rc{X}$.
\end{example}

\begin{example}
Let $x, y \in \TR^d$, where $d \in \posi{\TN}$. Then
\begin{eqs}
x \geq y \iff \forall i \in d: x_i \geq y_i.
\end{eqs}
\end{example}

\begin{note}[Relation-lifting pitfall]
We adopted an implicit convention for lifting relations from sets to functions. This introduces a potential for notational ambiguity. 

Consider the formula $f \not< g$, where $\function{f, g}{X}{\TR}$. If the negation refers to the original relation, then the formula is equivalent to
\begin{eqs}
\forall x \in X : f(x) \geq g(x).
\end{eqs}
However, if the negation refers to the lifted relation, then the formula is equivalent to
\begin{eqs}
\exists x \in X : f(x) \geq g(x).
\end{eqs}
We avoid this ambiguity by not using negation for lifted relations. In particular, we denote the filtered sets of co-asymptotic linear dominance by $\nlb{X}{y}$ --- not by $X^{\not<y}$.
\end{note}

\begin{definition}[Unary function for a set of functions]
A unary function $\function{\ominus}{\rc{X}}{\rc{Y}}$ is extended to $\power{\rc{X}} \to \power{\rc{Y}}$ by
\begin{eqs}
\ominus U = \setb{\ominus u : u \in U}.
\end{eqs}
\end{definition}

\begin{definition}[Binary function for sets of functions]
A binary function $\function{\oplus}{\rc{X} \times \rc{X}}{\rc{Y}}$ is extended to $\power{\rc{X}} \times \power{\rc{X}} \to \power{\rc{Y}}$ by
\begin{eqs}
U \oplus V = \setb{u \oplus v : (u, v) \in U \times V}. 
\end{eqs}
\end{definition}

\begin{definition}[Iteration]
The \define{$i$:th iteration} of $\function{f}{X}{X}$, where $i \in \TN$, is $\function{f^{(i)}}{X}{X}$ such that
\begin{equation}
f^{(i)}(x) =
\begin{cases}
x, & i = 0, \\
f^{(i - 1)}(f(x)), & i > 0.
\end{cases} 
\end{equation}
\end{definition}

\begin{definition}[Projection]
A \define{projection} is a function $\function{\projections{X}{Y}}{X \times Y}{X}$ such that
\begin{eqs}
\projection{X}{Y}{x, y} = x.
\end{eqs}
\end{definition}

\begin{definition}[Universe]
A \define{universe} is a class $U$ of sets such that
\begin{enumerate}
\item $\TN^0 \in U$,
\item $\TN^1 \in U$,
\item $\forall X \in U: \power{X} \subset U$,
\item $\forall X, Y \in U: X \times Y \in U$.
\end{enumerate}
\end{definition}

\begin{definition}[Sub-universe]
A \define{sub-universe} of a universe $U$ is a subclass $V \subset U$ which is also a universe. 
\end{definition}

\begin{example}[Examples of universes]
The smallest universe is given by
\begin{eqs}
\bigcup_{d \in \TN} \power{\TN^d} 
\end{eqs}
Every universe contains this set as a sub-universe. The class of all sets is a universe which is a proper class.
\end{example}

\begin{definition}[Computational problem]
A \define{computational problem} is a function $\function{P}{X}{Y}$, where $X, Y \in U$. 
\end{definition}

\begin{definition}[Set of algorithms]
The set of algorithms\footnote{We will define the term \emph{algorithm} formally in \sref{Algorithms}.} which solve a problem $\function{P}{X}{Y}$, under a given model of computation, is denoted by $\algo{P}$. The set of all algorithms from $X$ to $Y$ is
\begin{eqs}
\algof{X}{Y} = \bigcup_{\function{P}{X}{Y}} \algo{P}.
\end{eqs}
\end{definition}

\begin{definition}[Composition of algorithms]
The \define{composition} of algorithms $G \colon \algof{Y}{Z}$ and $F \colon \algof{X}{Y}$ is the algorithm $G \circ F \colon \algof{X}{Z}$, which is obtained by using the output of $F$ as the input of $G$.
\end{definition}

\begin{note}[]
We will sometimes use an algorithm $F \colon \algof{X}{Y}$ as if it were its induced function instead. 
\end{note}

\begin{definition}[Cost function of an algorithm]
The cost function of an algorithm $F \in \algo{P}$ is denoted by $f_F \in \rc{X}$.
\end{definition}

\chapter{Howell's counterexample}
\label{HowellCounter}

In this section, we consider Howell's counterexample \cite{OhImpossible}, which shows that asymptotic linear dominance $\pohsy$ does not satisfy \property{SubsetSum}.\footnote{We have fixed the error of having the sum-index $i$ run only to $m - 1$.} 

\begin{example}[Howell's counterexample]
Let $X = \TN^2$, $\hat{g} \in \rc{X}$ be such that
\begin{equation}
\hat{g}(m, n) = 
\begin{cases}
2^n, & m = 0, \\
mn, & m > 0,
\end{cases}  
\end{equation}
and $g \in \rc{X}$ be such that $g(m, n) = mn$. Then
\begin{eqs}
\sum_{i = 0}^{m} \hat{g}(i, n) & = 2^n + m(m + 1)n / 2 \\
{} & \not\in \pohx{m(m + 1)n / 2} \\
{} & = \pohx{\sum_{i = 0}^{m} g(i, n)}.
\end{eqs}
\end{example}

\begin{note}[Howell's requirements]
To be precise, Howell required the following properties from an $\ohsy$-notation:
\begin{description}
\item[\property{AsymptoticRefinement}] \hfill \\
\begin{equation}
\ohx{f} \subset \pohx{f}
\end{equation}
\item[\property{Reflex}] \hfill \\
\begin{equation}
f \in \ohx{f}
\end{equation}
\item[\property{AsymptoticOrder}] \hfill \\
\begin{equation}
\bra{\exists y \in \TN^d: \restrb{\hat{f}}{\nonnb{X}{y}} \leq \restrb{f}{\nonnb{X}{y}}} \implies \ohx{\hat{f}} \subset \ohx{f}
\end{equation}
\item[\property{SimpleSubsetSum}] \hfill \\
\begin{eqs}
\sum_{i = 0}^{n_k} \hat{g}(n_1, \dots, n_{k - 1}, i, n_{k + 1}, \dots, n_d) \in \\
\ohx{\sum_{i = 0}^{n_k} g(n_1, \dots, n_{k - 1}, i, n_{k + 1}, \dots, n_d)},
\end{eqs}
\end{description}
where $X = \TN^d$, $\hat{f}, f, g \in \rc{X}$, $\hat{g} \in \ohx{g}$, and $k \in [1, d] \subset \TN$. 
\end{note}

While Howell did not do so, we claim that any sensible definition of $\ohsy$-notation must also satisfy \property{Scale}:
\begin{eqs}
\ohx{\alpha f} = \ohx{f},
\end{eqs}
for all $f \in \rc{X}$ and $\alpha \in \posi{\TR}$. The following theorem then shows that Howell's result only concerns the $\pohsy$-notation.

\begin{theorem}[Howell's definition is asymptotic dominance]
\label{HowellsDefinition}
$\ohsy$ has \require{AsymptoticOrder}, \require{Scale}, \require{Reflex}, and \require{AsymptoticRefinement}. $\implies$ $\ohsy = \pohsy$.
\end{theorem}

\begin{proof}
By \require{AsymptoticOrder}, \require{Scale}, and \require{Reflex},
\begin{eqs}
{} \quad & \hat{f} \in \pohx{f} \\
\impliesr & \exists c \in \posi{\TR}, \exists y \in \TN^d : \restrb{\hat{f}}{\nonnb{X}{y}} \leq c \restrb{f}{\nonnb{X}{y}} \\
\impliesr & \exists c \in \posi{\TR} : \ohx{\hat{f}} \subset \ohx{cf} \\
\impliesr & \ohx{\hat{f}} \subset \ohx{f} \\
\impliesr & \hat{f} \in \ohx{f},
\end{eqs}
for all $f, \hat{f} \in \rc{X}$. Therefore $\pohx{f} \subset \ohx{f}$. It follows from \require{AsymptoticRefinement} that $\ohx{f} = \pohx{f}$.
\end{proof}

\chapter{Proofs of implied properties}
\label{ImpliedProperties}

In this section we will show that the primitive properties imply the rest of the properties. 

\begin{definition}[Composite property]
A \define{composite property} is a property which can be equivalently expressed in terms of primitive properties. 
\end{definition}

\begin{table}
\begin{tabular}{|l|c|c|c|c|c|c|c|c|c|c|l|}
\hline 
Property & $\leq$ & T & $1$ & $\alpha$ & L & $*$ & $/$ & $\circ$ & Th \\
\hline 
\hline 
\uproperty{QSubHom} & {} & {} & {} & {} & {} & \checkmark & \checkmark & {} & \ref{QSubhomogenuityIsComposite} \\
\hline 
\uproperty{SubHom} & \checkmark & \checkmark & {} & \checkmark & {} & \checkmark & \checkmark & {} & \ref{RSubhomogenuityIsImplied} \\
\hline 
\uproperty{Reflex} & \checkmark & {} & {} & {} & {} & {} & {} & {} & \ref{ReflexivityIsImplied} \\
\hline 
\uproperty{Zero} & \checkmark & \checkmark & \checkmark & {} & {} & \checkmark & {} & {} & \ref{ZeroSeparationIsImplied} \\
\hline 
\uproperty{Orderness} & \checkmark & \checkmark & {} & {} & {} & {} & {} & {} & \ref{OrdernessIsImplied} \\
\hline 
\uproperty{TrivialZero} & \checkmark & \checkmark & \checkmark & {} & {} & \checkmark & \checkmark & \checkmark & \ref{ZeroTrivialityIsImplied} \\
\hline 
\usproperty{ISuperComp} & \checkmark & {} & {} & {} & \checkmark & {} & {} & \checkmark & \ref{InjectiveSuperComposabilityIsImplied} \\
\hline 
\uproperty{SubRestrict} & {} & {} & {} & {} & {} & {} & {} & \checkmark & \ref{SubRestrictabilityIsImplied} \\
\hline 
\uproperty{SuperRestrict} & \checkmark & {} & {} & {} & \checkmark & {} & {} & {} & \ref{SuperRestrictabilityIsImplied} \\
\hline 
\uproperty{ScalarHom} & \checkmark & \checkmark & {} & \checkmark & {} & {} & {} & {} & \ref{ScalarHomogenuityIsImplied} \\
\hline 
\uproperty{SuperHom} & \checkmark & \checkmark & \checkmark & \checkmark & \checkmark & \checkmark & \checkmark & \checkmark & \ref{SuperHomogenuityIsImplied} \\
\hline 
\uproperty{SubMulti} & \checkmark & \checkmark & {} & \checkmark & {} & \checkmark & \checkmark & {} & \ref{SubMultiplicativityIsImplied} \\
\hline 
\uproperty{SuperMulti} & \checkmark & \checkmark & \checkmark & \checkmark & \checkmark & \checkmark & \checkmark & \checkmark & \ref{SuperMultiplicativityIsImplied} \\
\hline 
\uproperty{AddCons} & \checkmark & \checkmark & {} & {} & \checkmark & {} & {} & \checkmark & \ref{AdditiveConsistencyIsImplied} \\
\hline 
\uproperty{MaxCons} & \checkmark & \checkmark & {} & {} & \checkmark & {} & {} & \checkmark & \ref{MaximumConsistencyIsImplied} \\
\hline 
\usproperty{MultiCons} & \checkmark & \checkmark & {} & {} & \checkmark & {} & {} & \checkmark & \ref{MultiplicativeConsistencyIsImplied} \\
\hline 
% \property{PowerH} & \checkmark & \checkmark & \checkmark & {} & \checkmark & \checkmark & \checkmark & \checkmark & \ref{PositivePowerHomogenuityIsImplied} \\
% \hline 
\uproperty{Maximum} & \checkmark & \checkmark & {} & {} & \checkmark & {} & {} & \checkmark & \ref{MaximumIsImplied} \\
\hline 
\uproperty{Summation} & \checkmark & \checkmark & {} & \checkmark & {} & {} & {} & {} & \ref{SummationIsImplied} \\
\hline 
\uproperty{MaximumSum} & \checkmark & \checkmark & {} & \checkmark & {} & {} & {} & {} & \ref{MaximumSumIsImplied} \\
\hline 
\uproperty{Additive} & \checkmark & \checkmark & {} & \checkmark & \checkmark & {} & {} & \checkmark & \ref{AdditivityIsImplied} \\
\hline 
\usproperty{Translation} & \checkmark & \checkmark & {} & \checkmark & {} & {} & {} & {} & \ref{TranslationInvarianceIsImplied} \\
\hline 
\uproperty{Extend} & {} & {} & {} & {} & {} & {} & {} & \checkmark & \ref{ExtensibilityIsImplied} \\
\hline 
\uproperty{SubsetSum} & \checkmark & \checkmark & \checkmark & \checkmark & \checkmark & \checkmark & \checkmark & \checkmark & \ref{SubsetSumIsImplied} \\
\hline 
\end{tabular}
\centering
\caption{The primitive properties that imply a given non-primitive property. The abbreviations are: $\leq$ for \property{Order}, T for \property{Trans}, $1$ for \property{One}, $\alpha$ for \property{Scale}, L for \property{Local}, $*$ for \property{NSubHom}, $/$ for \property{NSubDiv}, $\circ$ for \property{SubComp}, and Th for theorem.}
\label{PrimitivePropertiesNeededToImplyOtherProperties}
\end{table}

\begin{proposition}[\uproperty{QSubHom} is a composite]
\label{QSubhomogenuityIsComposite}
$\ohs{X}$ has \prove{QSubHom}. $\iff$ $\ohs{X}$ has \require{NSubHom} and \require{NSubDiv}.
\end{proposition}

\begin{proof}
\proofpart{$\implies$}
Suppose $\ohs{X}$ has \property{QSubHom}. Then $\ohs{X}$ has \property{NSubHom}, since $\TN \subset \nonn{\TQ}$, and \property{NSubDiv}, since $(1 / \posi{\TN}) \subset \posi{\TQ}$.

\proofpart{$\impliedby$}
Suppose $\ohs{X}$ has \require{NSubHom} and \require{NSubDiv}. Let $f, g, u \in \rc{X}$ be such that $\image{u}{X} \subset \nonn{\TQ}$. Then there exists $p, q \in \rc{X}$, such that $\image{p}{X} \subset \nonn{\TN}$, $\image{q}{X} \subset \posi{\TN}$, and $u = p / q$. By \property{NSubHom} and \property{NSubDiv},
\begin{eqs}
{} & f \in \ohx{g} \\
\impliesr & pf \in \ohx{pg} \\
\impliesr & \frac{p}{q} f \in \ohx{\frac{p}{q} g} \\
\impliesr & uf \in \ohx{ug}. 
\end{eqs} \sprove{QSubHom}
\end{proof}

\begin{proposition}[\uproperty{SubHom} is implied]
\label{RSubhomogenuityIsImplied}
$\ohs{X}$ has \require{Order}, \require{Trans}, \require{Scale}, \require{NSubHom}, and \require{NSubDiv}. $\implies$ $\ohs{X}$ has \prove{SubHom}.
\end{proposition}

\begin{proof}
$\ohs{X}$ has \property{QSubHom} by \proveby{QSubhomogenuityIsComposite}. Let $f, g, u \in \rc{X}$, and $\function{h}{\nonn{\TR}}{\nonn{\TQ}}$ be such that
\begin{eqs}
x \leq h(x) \leq 2x.
\end{eqs}
By \require{QSubHom}, \require{Order}, \require{Trans}, and \require{Scale},
\begin{eqs}
{} & f \in \ohx{g} \\
\implies & (h \circ u) f \in \ohx{(h \circ u) g} \\
\implies & uf \in \ohx{2ug} \\
\implies & uf \in \ohx{ug}. 
\end{eqs} \sprove{SubHom}
\end{proof}

\begin{proposition}[\uproperty{Reflex} is implied]
\label{ReflexivityIsImplied}
$\ohs{X}$ has \require{Order} $\implies$ $\ohs{X}$ has \prove{Reflex}.
\end{proposition}

\begin{proof}
By \require{Order}, $f \leq f \implies f \in \ohx{f}$, for all $f \in \rc{X}$; $\ohs{X}$ has \prove{Reflex}. 
\end{proof}

\begin{proposition}[\uproperty{Zero} is implied]
\label{ZeroSeparationIsImplied}
$\ohsy$ has \require{Order}, \require{Trans}, \require{One}, and \require{NSubHom} $\implies$ $\ohsy$ has \prove{Zero}.
\end{proposition}

\begin{proof}
Suppose $\ohsy$ does not have \property{Zero}, so that $1 \in \oh{\posi{\TN}}{0}$. By \require{NSubHom}, $n \in \oh{\posi{\TN}}{0}$. By \require{Order}, $0 \in \oh{\posi{\TN}}{1}$. By \require{Trans}, $n \in \oh{\posi{\TN}}{1}$. This contradicts \require{One}. \sprove{Zero}
\end{proof}

\begin{proposition}[\uproperty{Orderness} is a composite]
\label{ReflexiveTransitiveIsOrderness}
$\ohs{X}$ has \require{Reflex} and \require{Trans} $\iff$ $\ohs{X}$ has \prove{Orderness}.
\end{proposition}

\begin{proof}
\proofpart{$\implies$} 
Assume $f \in \ohx{g}$. Let $\hat{f} \in \ohx{f}$. By \require{Trans}, $\hat{f} \in \ohx{g}$, and so $\ohx{f} \subset \ohx{g}$. Assume $\ohx{f} \subset \ohx{g}$. By \require{Reflex}, $f \in \ohx{f}$. Therefore $f \in \ohx{g}$, and so $\ohs{X}$ has \prove{Orderness}. 

\proofpart{$\impliedby$} 
By \require{Orderness}, $f \in \ohx{f} \iff \ohx{f} \subset \ohx{f}$. Therefore $\ohs{X}$ has \prove{Reflex}. Let $f \in \ohx{g}$, and $g \in \ohx{h}$. By \require{Orderness}, $\ohx{g} \subset \ohx{h}$. Therefore $f \in \ohx{h}$, and so $\ohs{X}$ has \prove{Trans}.
\end{proof}

\begin{proposition}[\uproperty{Orderness} is implied]
\label{OrdernessIsImplied}
$\ohs{X}$ has \require{Order} and \require{Trans}. $\implies$ $\ohs{X}$ has \prove{Orderness}.
\end{proposition}

\begin{proof}
$\ohs{X}$ has \property{Reflex} by \proveby{ReflexivityIsImplied}. $\ohs{X}$ has \property{Orderness} by \proveby{ReflexiveTransitiveIsOrderness}.
\end{proof}

\begin{proposition}[\uproperty{TrivialZero} is implied]
\label{ZeroTrivialityIsImplied}
$\ohsy$ has \require{Order}, \require{Trans}, \require{One}, \require{Scale}, \require{NSubHom}, \require{NSubDiv}, and \require{SubComp}. $\implies$ $\ohsy$ has \prove{TrivialZero}.
\end{proposition}

\begin{proof}
$\ohsy$ has \property{Zero} by \proveby{ZeroSeparationIsImplied} and \property{SubHom} by \proveby{RSubhomogenuityIsImplied}. Suppose $\ohsy$ does not have \property{TrivialZero}, so that there exists $f \in \ohx{0}$ such that $f \neq 0$. Then there exists $y \in X$ such that $f(y) = c$, for some $c \in \posi{\TR}$. Let $\function{s}{\posi{\TN}}{X}$ be such that $s(x) = y$. By \require{SubComp} and \require{SubHom},
\begin{eqs}
{} & f \in \ohx{0} \\
\impliesr & f \circ s \in \ohx{0} \circ s \\
\impliesr & c \in \oh{\posi{\TN}}{0 \circ s} \\
\impliesr & 1 \in \oh{\posi{\TN}}{0} / c \\
\impliesr & 1 \in \oh{\posi{\TN}}{0 / c} \\
\impliesr & 1 \in \oh{\posi{\TN}}{0}.
\end{eqs}
This contradicts \require{Zero}. 
\sprove{TrivialZero}
\end{proof}

\begin{proposition}[\uproperty{ISuperComp} is implied]
\label{InjectiveSuperComposabilityIsImplied}
$\ohs{X}$ has \require{Order}, \require{Local}, and \require{ISubComp} for injective $\function{s}{Y}{X}$. $\implies$ $\ohs{X}$ has \prove{ISuperComp} for $s$. \sprove{IComp}
\end{proposition}

\begin{proof}
Let $f \in \rc{X}$, and $\function{\underline{s}}{Y}{\image{s}{Y}}$ be such that $\underline{s}(y) = s(y)$. Then $\underline{s}$ is bijective. Let $\hat{g} \in \oh{Y}{f \circ s}$ and $\hat{f} \in \rc{X}$ be such that
\begin{eqs}
\hat{f}(x) = 
\begin{cases}
\bra{\hat{g} \circ \invs{\underline{s}}}(x), & x \in \image{s}{Y}, \\
0, & x \not\in \image{s}{Y}.
\end{cases}
\end{eqs}
Then $\hat{g} = \hat{f} \circ s$; we show that $\hat{f} \in \ohx{f}$. 
By \require{ISubComp}
\begin{eqs}
\hat{f}|\image{s}{Y} & = \hat{g} \circ \invs{\underline{s}} \\
{} & \in \oh{Y}{f \circ s} \circ \invs{\underline{s}} \\
{} & \subset \oh{\image{s}{Y}}{f \circ s \circ \invs{\underline{s}}} \\
{} & = \oh{\image{s}{Y}}{f|\image{s}{Y}}.
\end{eqs}
By \require{Order},
\begin{equation}
\hat{f}|\bra{X \setminus \image{s}{Y}} = 0 \in \oh{X \setminus \image{s}{Y}}{f|\bra{X \setminus \image{s}{Y}}}.
\end{equation}
By \require{Local},
\begin{equation}
\hat{f} \in \ohx{f}.
\end{equation}
Therefore $\ohs{X}$ has \prove{ISuperComp} for $s$. \sprove{IComp}
\end{proof}

\begin{proposition}[\uproperty{SubRestrict} is implied]
\label{SubRestrictabilityIsImplied}
$\ohs{X}$ has \require{ISubComp}. $\implies$ $\ohs{X}$ has \prove{SubRestrict}.
\end{proposition}

\begin{proof}
Let $D \subset X$, and $\function{s}{D}{X}$ be such that $s(x) = x$. Then $s$ is injective. By \require{ISubComp}
\begin{eqs}
\restr{\ohx{f}}{D} & = \ohx{f} \circ s \\
{} & \subset \oh{D}{f \circ s} \\
{} & = \oh{D}{\restr{f}{D}},
\end{eqs}
for all $f \in \rc{X}$. \sprove{SubRestrict}
\end{proof}

\begin{theorem}[\uproperty{SuperRestrict} is implied]
\label{SuperRestrictabilityIsImplied}
$\ohsy$ has \require{Order} and \require{Local}. $\implies$ $\ohsy$ has \prove{SuperRestrict}.
\end{theorem}

\begin{proof}
Let $D \subset X$, $\compl{D} = X \setminus D$, and $\hat{h} \in \oh{D}{\restr{f}{D}}$. Let $\hat{f} \in \rc{X}$ be such that $\restrb{\hat{f}}{D} = \hat{h}$ and $\restrb{\hat{f}}{\compl{D}} = 0$. Then $\restrb{\hat{f}}{D} \in \oh{D}{\restr{f}{D}}$, and by \require{Order}, $\restrb{\hat{f}}{\compl{D}} \in \oh{\compl{D}}{\restr{f}{\compl{D}}}$. By \require{Local}, $\hat{f} \in \ohx{f}$. \sprove{SuperRestrict}
\end{proof}

\begin{theorem}[\uproperty{ScalarHom} is implied]
\label{ScalarHomogenuityIsImplied}
$\ohs{X}$ has \require{Order}, \require{Trans}, and \require{Scale}. $\implies$ $\ohs{X}$ has \prove{ScalarHom}.
\end{theorem}

\begin{proof}
$\ohs{X}$ has \property{Orderness} by \proveby{OrdernessIsImplied}.

\proofpart{$\subset$}
Let $\hat{f} \in \ohx{f}$, and $\alpha \in \posi{\TR}$. By \require{Order}, \require{Scale}, \require{Orderness}, and \require{Scale} again,
\begin{eqs}
\alpha \hat{f} & \in \ohx{\alpha \hat{f}} \\
{} & = \ohx{\hat{f}} \\
{} & \subset \ohx{f} \\
{} & = \ohx{\alpha f}.
\end{eqs}
Therefore $\alpha \hat{f} \in \ohx{\alpha f}$.

\proofpart{$\supset$}
Let $\hat{f} \in \ohx{\alpha f}$. By \require{Scale}, $\hat{f} \in \ohx{f}$. By \require{Orderness}, $\ohx{\hat{f}} \subset \ohx{f}$. By \require{Scale}, $\ohx{(1 / \alpha) \hat{f}} \subset \ohx{f}$. By \require{Orderness}, $(1 / \alpha) \hat{f} \in \ohx{f}$. Therefore $\hat{f} \in \alpha \ohx{f}$, and so $\ohs{X}$ has \prove{ScalarHom}.
\end{proof}

\begin{theorem}[\uproperty{SuperHom} is implied]
\label{SuperHomogenuityIsImplied}
$\ohs{X}$ has \require{Order}, \require{Local}, \require{Trans}, \require{One}, \require{Scale}, \require{NSubHom}, \require{NSubDiv}, and $\require{SubComp}$. \linebreak $\implies$ $\ohs{X}$ has \prove{SuperHom}.
\end{theorem}

\begin{proof}
\srequire{SubComp} \sprove{ISubComp}
$\ohs{X}$ has \property{SubRestrict} by \proveby{SubRestrictabilityIsImplied}, \property{SubHom} by \proveby{RSubhomogenuityIsImplied}, and \property{TrivialZero} by \proveby{ZeroTrivialityIsImplied}. Let $\hat{h} \in \ohx{fg}$, $G = \preimage{g}{\posi{\TR}}$, and $\compl{G} = X \setminus G$. Let $\hat{f} \in \rc{X}$ be such that $\restrb{\hat{f}}{G} = \restrb{\hat{h}}{G} / \restrb{g}{G}$ and $\restrb{\hat{f}}{\compl{G}} = 0$. By \require{SubRestrict}, and \require{SubHom},
\begin{eqs}
\restrb{\hat{f}}{G} & \in \frac{\restr{\ohx{fg}}{G}}{\restr{g}{G}} \\
{} & \subset \frac{\oh{G}{\restr{fg}{G}}}{\restr{g}{G}} \\
{} & \subset \oh{G}{\restr{f}{G}}.
\end{eqs}
By \require{Order}, $\restrb{\hat{f}}{\compl{G}} \in \oh{\compl{G}}{\restr{f}{\compl{G}}}$. By \require{Local}, $\hat{f} \in \ohx{f}$. It holds that $\restrb{\hat{h}}{\compl{G}} \in \restr{\ohx{fg}}{\compl{G}}$. By \require{SubRestrict}, $\restrb{\hat{h}}{\compl{G}} \in \oh{\compl{G}}{0}$. By \require{TrivialZero}, $\restrb{\hat{h}}{\compl{G}} = 0$. Therefore $\hat{f}g = \hat{h}$.
\sprove{SuperHom}
\end{proof}

\begin{theorem}[\uproperty{SubMulti} is implied]
\label{SubMultiplicativityIsImplied}
$\ohs{X}$ has \require{Order}, \require{Trans}, \require{Scale}, \require{NSubHom}, and \require{NSubDiv}. $\implies$ $\ohs{X}$ has \prove{SubMulti}.
\end{theorem}

\begin{proof}
$\ohs{X}$ has \property{Orderness} by \proveby{OrdernessIsImplied}, and \property{SubHom} by \proveby{RSubhomogenuityIsImplied}. Let $f, g \in \rc{X}$, and $\hat{f} \in \ohx{f}$. By \require{SubHom},
\begin{eqs}
\hat{f} g & \in \ohx{f} g \\
{} & \subset \ohx{fg}.
\end{eqs}
By \require{Orderness},
\begin{eqs}
\ohx{\hat{f} g} \subset \ohx{fg}. 
\end{eqs}
By \require{SubHom},
\begin{eqs}
\hat{f} \ohx{g} & \subset \ohx{\hat{f}g} \\
{} & \subset \ohx{fg}.
\end{eqs}
Therefore $\ohx{f} \cdot \ohx{g} \subset \ohx{fg}$. 
\sprove{SubMulti}
\end{proof}

\begin{theorem}[\uproperty{SuperMulti} is implied]
\label{SuperMultiplicativityIsImplied}
$\ohsy$ has \require{Order}, \require{Trans}, \require{One}, \require{Scale}, \require{NSubHom}, \require{NSubDiv}, \require{Local}, and \require{SubComp}. $\implies$ $\ohsy$ has \prove{SuperMulti}.
\end{theorem}

\begin{proof}
\srequire{SubComp} \sprove{ISubComp}
$\ohs{X}$ has \property{SubRestrict} by \proveby{SubRestrictabilityIsImplied}, \property{SubHom} by \proveby{RSubhomogenuityIsImplied}, and \property{TrivialZero} by \proveby{ZeroTrivialityIsImplied}. Let $f, g \in \rc{X}$, and $\hat{h} \in \ohx{fg}$. Let $G = \preimage{g}{\posi{\TR}}$, and $\compl{G} = X \setminus G$. Let $\hat{f} \in \rc{X}$ be such that $\restrb{\hat{f}}{G} = \frac{\restrb{\hat{h}}{G}}{\restrb{g}{G}}$ and $\restrb{\hat{f}}{\compl{G}} = 0$. It holds that $\restrb{\hat{h}}{G} \in \restr{\ohx{fg}}{G}$. By \require{SubRestrict}, $\restrb{\hat{h}}{G} \in \oh{G}{\restr{fg}{G}}$. Similarly, $\restrb{\hat{h}}{\compl{G}} \in \oh{\compl{G}}{\restr{fg}{\compl{G}}} = \oh{\compl{G}}{0}$. By \require{SubHom}
\begin{eqs}
\restrb{\hat{f}}{G} & \in \frac{\oh{G}{\restr{fg}{G}}}{\restr{g}{G}} \\
{} & \subset \oh{G}{\frac{\restr{fg}{G}}{\restr{g}{G}}} \\
{} & = \oh{G}{\restr{f}{G}}.
\end{eqs}
By \require{Order}, $\restrb{\hat{f}}{\compl{G}} \in \oh{\compl{G}}{\restr{f}{\compl{G}}}$. By \require{Local}, $\hat{f} \in \ohx{f}$. By \require{Order}, $g \in \ohx{g}$. By definition, $\restrb{\hat{f}}{G} \restrb{g}{G} = \restrb{\hat{h}}{G}$. By \require{TrivialZero}, $\restrb{\hat{h}}{\compl{G}} = 0$, and so $\restrb{\hat{h}}{\compl{G}} = \restrb{\hat{f}}{\compl{G}} \restrb{g}{\compl{G}}$. Therefore $\hat{h} = \hat{f} g$; $\ohsy$ has \prove{SuperMulti}.
\end{proof}

\begin{theorem}[\uproperty{Local} is implied]
\label{LocalityIsImplied}
$\ohs{X}$ has \require{Order}, \require{Trans}, \require{One}, \require{Scale}, \require{NSubHom}, \require{NSubDiv}, and \require{SubComp}. $\implies$ $\ohs{X}$ has \prove{Local}.
\end{theorem}

\begin{proof}
\proofpart{Implied properties}
\srequire{SubComp} \sprove{ISubComp} $\ohs{X}$ has \property{TrivialZero} by \proveby{ZeroTrivialityIsImplied}, and \property{SubMulti} by \proveby{SubMultiplicativityIsImplied}.

\proofpart{Reduction from a finite cover to partition}
Let $f, g \in \rc{X}$. We can refine a finite cover $C \subset \power{X}$ of $X$ to a finite partition $\setb{A_1, \dots, A_m} \subset \power{X}$ of $X$. In the definition of \property{Local}, we assume that
\begin{eqs}
\restrb{f}{D} \in \oh{D}{\restr{g}{D}},
\end{eqs}
for all $D \in C$. By \require{ISubComp},
\begin{eqs}
\restrb{f}{A_i} \in \oh{A_i}{\restr{g}{A_i}},
\end{eqs}
for all $i \in \iccn{1}{m}$. Therefore, we may assume the finite cover of $X$ to be a finite partition of $X$.

\proofpart{Notation}
Let $\setb{A_1, \dots, A_m} \subset \power{X}$ be a finite partition of $X$. Suppose that
\begin{eqs}
\restrb{f}{A_i} \in \oh{A_i}{\restr{g}{A_i}},
\end{eqs}
for all $i \in \iccn{1}{m}$. Let
\begin{eqs}
I = \setb{i \in \iccn{1}{m} : \image{f}{A_i} \cap \posi{\TR} \neq \emptyset}.
\end{eqs}
Let $i \in I$. Let $p_i \in X$ be such that $p_i \in A_i$, and
\begin{eqs}
f(p_i) > 0,
\end{eqs}
Let $c_i, d_i \in \nonn{\TR}$ be such that
\begin{eqs}
c_i & = f(p_i) \\
d_i & = g(p_i).
\end{eqs}
Let $\function{s_i}{X}{X}$ be such that
\begin{eqs}
s_i(x) = 
\begin{cases}
x, & x \in A_i, \\
p_i, & \text{otherwise}.
\end{cases}
\end{eqs}

\proofpart{Positivity}
By \require{ISubComp},
\begin{eqs}
\restrb{f}{\setb{p_i}} \in \oh{\setb{p_i}}{\restr{g}{\setb{p_i}}}.
\end{eqs}
By \require{TrivialZero}, and since $c_i > 0$, it holds that $d_i > 0$. 

\proofpart{Lift}
By \require{SubComp},
\begin{eqs}
\restrb{f}{A_i} \circ s_i \in \ohx{\restrb{g}{A_i} \circ s_i}.
\end{eqs}
This is equivalent to
\begin{eqs}
\bra{f \indi{A_i}{X} + c_i \indi{X \setminus A_i}{X}} \in \ohx{g \indi{A_i}{X} + d_i \indi{X \setminus A_i}{X}}.
\end{eqs}
By \require{Order}, and \require{Scale},
\begin{eqs}
\bra{\indi{A_i}{X} + \frac{1}{c_i} \indi{X \setminus A_i}{X}} \in \ohx{\indi{A_i}{X} + \frac{1}{d_i} \indi{X \setminus A_i}{X}}.
\end{eqs}
By \require{SubMulti},
\begin{eqs}
\bra{f \indi{A_i}{X} + \indi{X \setminus A_i}{X}} \in \ohx{g \indi{A_i}{X} + \indi{X \setminus A_i}{X}}.
\end{eqs}

\proofpart{Sum by multiplying}
We have that
\begin{eqs}
\prod_{i \in I} \bra{f \indi{A_i}{X} + \indi{X \setminus A_i}{X}} & = \sum_{i \in I} f \indi{A_i}{X} = f, \\
\prod_{i \in I} \bra{g \indi{A_i}{X} + \indi{X \setminus A_i}{X}} & = \sum_{i \in I} g \indi{A_i}{X} \leq g.
\end{eqs}
By \require{SubMulti},
\begin{eqs}
f \in \ohx{\sum_{i \in I} g \indi{A_i}{X}}.
\end{eqs}
By \require{Order}, and \require{Trans},
\begin{eqs}
f \in \ohx{g}.
\end{eqs}
\sprove{Local}
\end{proof}

\begin{theorem}[\uproperty{AddCons} is implied]
\label{AdditiveConsistencyIsImplied}
$\ohs{X}$ has \require{Order}, \require{Trans}, \require{Local}, and \require{ISubComp}. $\implies$ $\ohs{X}$ has \prove{AddCons}.
\end{theorem}

\begin{proof}
$\ohs{X}$ has \property{SubRestrict} by \proveby{SubRestrictabilityIsImplied}. Let $f, u, v \in \rc{X}$, and $U \coloneqq \setb{x \in X: u(x) + v(x) > 0}$.

\proofpart{$\subset$} 
Let $\hat{f}, \hat{g} \in \ohx{f}$, and $F \coloneqq \setb{x \in X: \hat{g}(x) \leq \hat{f}(x)}$. Let $\hat{h} \in \rc{X}$ be such that
\begin{equation}
\hat{h}(x) =
\begin{cases}
\frac{u(x) \hat{f}(x) + v(x) \hat{g}(x)}{u(x) + v(x)}, & x \in U \\
0, & x \not\in U.
\end{cases}
\end{equation}
Then
\begin{eqs}
\hat{h}(x) & = \frac{u(x) \hat{f}(x) + v(x) \hat{g}(x)}{u(x) + v(x)} \\
{} & \leq \frac{u(x) \hat{f}(x) + v(x) \hat{f}(x)}{u(x) + v(x)} \\
{} & = \hat{f}(x),
\end{eqs}
for all $x \in F \cap U$. Also, $\hat{h}(x) = 0 \leq \hat{f}(x)$, for all $x \in F \cap \compl{U}$. By \require{Order}, $\restrb{\hat{h}}{F} \in \oh{F}{\restr{\hat{f}}{F}}$. By \require{SubRestrict} and \require{Trans}, $\restrb{\hat{h}}{F} \in \oh{F}{\restr{f}{F}}$. Similarly, $\restrb{\hat{h}}{\compl{F}} \in \oh{\compl{F}}{\restr{f}{\compl{F}}}$. By \require{Local}, $\hat{h} \in \ohx{f}$. In addition, $u \hat{f} + v \hat{g} = (u + v)\hat{h}$. Therefore $u \hat{f} + v \hat{g} \in (u + v)\ohx{f}$.

\proofpart{$\supset$}
Let $\hat{f} \in \ohx{f}$. Then $(u + v) \hat{f} = u \hat{f} + v \hat{f} \in u \ohx{f} + v \ohx{f}$.

\sprove{AddCons}
\end{proof}

\begin{theorem}[\uproperty{MaxCons} is implied]
\label{MaximumConsistencyIsImplied}
$\ohs{X}$ has \require{Order}, \require{Trans}, \require{Local}, and \require{ISubComp}. $\implies$ $\ohs{X}$ has \prove{MaxCons}.
\end{theorem}

\begin{proof}
$\ohs{X}$ has \property{SubRestrict} by \proveby{SubRestrictabilityIsImplied}. 

\proofpart{$\subset$} 
Let $\hat{f}, \hat{g} \in \ohx{f}$, and $F \coloneqq \setb{x \in X: \hat{g}(x) \leq \hat{f}(x)}$. Let $\hat{h} \in \rc{X}$ be such that
\begin{equation}
\hat{h} \coloneqq \max(\hat{f}, \hat{g}).
\end{equation}
Then
\begin{eqs}
\restrb{\hat{h}}{F} & = \restrb{\max(\hat{f}, \hat{g})}{F} \\
{} & = \restrb{\hat{f}}{F}.
\end{eqs}
By \require{Order}, $\restrb{\hat{h}}{F} \in \oh{F}{\restr{\hat{f}}{F}}$. By \require{SubRestrict} and \require{Trans}, $\restrb{\hat{h}}{F} \in \oh{F}{\restr{f}{F}}$. Similarly, $\restrb{\hat{h}}{\compl{F}} \in \oh{\compl{F}}{\restr{f}{\compl{F}}}$. By \require{Local}, $\hat{h} \in \ohx{f}$. In addition, $\max(\hat{f}, \hat{g}) = \hat{h}$. Therefore $\max(\hat{f}, \hat{g}) \in \ohx{f}$.

\proofpart{$\supset$}
Let $\hat{f} \in \ohx{f}$. Then $\hat{f} = \max(\hat{f}, \hat{f}) \in \max(\ohx{f}, \ohx{f})$.

\sprove{MaxCons}
\end{proof}

\begin{theorem}[\uproperty{MultiCons} is implied]
\label{MultiplicativeConsistencyIsImplied}
$\ohs{X}$ has \require{Order}, \require{Trans}, \require{Local}, \require{ISubComp}. $\implies$ $\ohs{X}$ has \prove{MultiCons}.
\end{theorem}

\begin{proof}
$\ohs{X}$ has \property{SubRestrict} by \proveby{SubRestrictabilityIsImplied}. Let $f, u, v \in \rc{X}$, and $U \coloneqq \setb{x \in X: u(x) + v(x) > 0}$.

\proofpart{$\subset$} 
Let $\hat{f}, \hat{g} \in \ohx{f}$, and $F \coloneqq \setb{x \in X: \hat{g}(x) \leq \hat{f}(x)}$. Let $\hat{h} \in \rc{X}$ be such that
\begin{equation}
\hat{h}(x) =
\begin{cases}
\bra{\hat{f}(x)^{u(x)} \hat{g}(x)^{v(x)}}^{1 / (u(x) + v(x))}, & x \in U \\
0, & x \not\in U.
\end{cases}
\end{equation}
Then
\begin{eqs}
\hat{h}(x) & = \bra{\hat{f}(x)^{u(x)} \hat{g}(x)^{v(x)}}^{1 / (u(x) + v(x))} \\
{} & \leq \bra{\hat{f}(x)^{u(x)} \hat{f}(x)^{v(x)}}^{1 / (u(x) + v(x))} \\
{} & = \hat{f}(x).
\end{eqs}
for all $x \in F \cap U$. Also, $\hat{h}(x) = 0 \leq \hat{f}(x)$, for all $x \in F \cap \compl{U}$. By \require{Order}, $\restrb{\hat{h}}{F} \in \oh{F}{\restr{\hat{f}}{F}}$. By \require{SubRestrict} and \require{Trans}, $\restrb{\hat{h}}{F} \in \oh{F}{\restr{f}{F}}$. Similarly, $\restrb{\hat{h}}{\compl{F}} \in \oh{\compl{F}}{\restr{f}{\compl{F}}}$. By \require{Local}, $\hat{h} \in \ohx{f}$. In addition, $\hat{f}^u \hat{g}^v = \hat{h}^{u + v}$ (we assume $0^0 = 1$). Therefore $\hat{f}^u \hat{g}^v \in \ohx{f}^{u + v}$.

\proofpart{$\supset$} 
Let $\hat{f} \in \ohx{f}$. Then $\hat{f}^{u + v} = \hat{f}^u \hat{f}^v \in \ohx{f}^u \ohx{f}^v$.

\sprove{MultiCons}
\end{proof}

\begin{theorem}[\uproperty{Maximum} is implied]
\label{MaximumIsImplied}
$\ohs{X}$ has \require{Order}, \require{Trans}, \require{Local}, and \require{ISubComp}. $\implies$ $\ohs{X}$ has \prove{Maximum}.
\end{theorem}

\begin{proof}
$\ohs{X}$ has \property{Orderness} by \proveby{OrdernessIsImplied}, \property{SubRestrict} by \proveby{SubRestrictabilityIsImplied}, and \property{MaxCons} by \proveby{MaximumConsistencyIsImplied}.

\proofpart{$\subset$} 
Let $f, g \in \rc{X}$. By \require{Order}, $f, g \in \ohx{\max(f, g)}$. By \require{Orderness}, $\ohx{f} \subset \ohx{\max(f, g)}$ and $\ohx{g} \subset \ohx{\max(f, g)}$.  By \require{MaxCons},
\begin{eqs}
\max(\ohx{f}, \ohx{g}) & \subset \max(\ohx{\max(f, g)}, \ohx{\max(f, g)}) \\
{} & = \ohx{\max(f, g)}.
\end{eqs}

\proofpart{$\supset$}
Let $\hat{h} \in \ohx{\max(f, g)}$, and $F \coloneqq \setb{x \in X : g(x) \leq f(x)}$. Let $\hat{f} \in \rc{X}$ be such that
\begin{equation}
\hat{f} \coloneqq \hat{h} \idf{F}.
\end{equation}
By \require{Order}, $\hat{f} \in \ohx{\hat{h}}$. By \require{Trans}, $\hat{f} \in \ohx{\max(f, g)}$. By \require{SubRestrict}
\begin{eqs}
\restrb{\hat{f}}{F} & \in \restr{\ohx{\max(f, g)}}{F} \\
{} & \subset \oh{F}{\restr{f}{F}}.
\end{eqs}
Also, $\restrb{\hat{f}}{\compl{F}} = 0 \leq \restrb{f}{\compl{F}}$. By \require{Order}, $\restrb{\hat{f}}{\compl{F}} \in \ohx{\restr{f}{\compl{F}}}$. By \require{Local}, $\hat{f} \in \ohx{f}$. Similarly, let $\hat{g} \in \rc{X}$ be such that $\hat{g} \coloneqq \hat{h} \idf{\compl{F}}$. Then $\hat{g} \in \ohx{g}$. In addition, $\hat{h} = \max(\hat{f}, \hat{g})$. Therefore $\hat{h} \in \max(\ohx{f}, \ohx{g})$.

\sprove{Maximum}
\end{proof}

\begin{theorem}[\uproperty{Summation} is implied]
\label{SummationIsImplied}
$\ohs{X}$ has \require{Order}, \require{Trans}, and \require{Scale} $\implies$ $\ohs{X}$ has \prove{Summation}.
\end{theorem}

\begin{proof}
Let $f, g \in \rc{X}$. 

\proofpart{$\subset$} 
It holds that
\begin{equation}
f + g \leq 2 \max(f, g).
\end{equation}
By \require{Order} and \require{Trans}, $\ohx{f + g} \subset \ohx{2 \max(f, g)}$. By \require{Scale}, $\ohx{f + g} \subset \ohx{\max(f, g)}$.

\proofpart{$\supset$} 
It holds that
\begin{equation}
\max(f, g) \leq f + g.
\end{equation}
By \require{Order} and \require{Trans}, $\ohx{\max(f, g)} \subset \ohx{f + g}$.

\sprove{Summation}
\end{proof}

\begin{theorem}[\uproperty{MaximumSum} is implied]
\label{MaximumSumIsImplied}
$\ohs{X}$ has \require{Order}, \require{Trans}, and \require{Scale}. $\implies$ $\ohs{X}$ has \prove{MaximumSum}.
\end{theorem}

\begin{proof}
\proofpart{$\subset$}
Let $\hat{f} \in \ohx{f}$ and $\hat{g} \in \ohx{g}$. Let $F = \setb{x \in X: \hat{g}(x) \leq \hat{f}(x)}$. Let $u \in \rc{X}$ be such that $u \coloneqq \hat{f} \idf{F}$. By \require{Order}, $u \in \ohx{\hat{f}}$. By \require{Trans}, $u \in \ohx{f}$. Similarly, let $v \in \rc{X}$ be such that $v \coloneqq \hat{g} \idf{\compl{F}}$. Then $v \in \ohx{g}$. In addition, $\max(\hat{f}, \hat{g}) = u + v$. Therefore $\max(\hat{f}, \hat{g}) \in \ohx{f} + \ohx{g}$.

\proofpart{$\supset$} 
Let $f, g \in \rc{X}$, $\hat{f} \in \ohx{f}$, and $\hat{g} \in \ohx{g}$. Let
\begin{eqs}
F = \setb{x \in X: \hat{g}(x) \leq \hat{f}(x)}. 
\end{eqs}
Let $u \in \rc{X}$ be such that $u \coloneqq (\hat{f} + \hat{g}) \idf{F}$. Then $u \leq (2 \hat{f}) \idf{F} \leq 2 \hat{f}$. By \require{Order} and \require{Trans}, $\ohx{u} \subset \ohx{2 \hat{f}}$. By \require{Scale}, $\ohx{u} \subset \ohx{\hat{f}}$. By \require{Order}, $u \in \ohx{\hat{f}}$. By \require{Trans} $u \in \ohx{f}$. Similarly, let $v \in \rc{X}$ be such that
\begin{equation}
v \coloneqq (\hat{f} + \hat{g}) \idf{\compl{F}}.
\end{equation}
Then $v \in \ohx{g}$. In addition, $\max(u, v) = \hat{f} + \hat{g}$. Therefore $\hat{f} + \hat{g} \in \max(\ohx{f}, \ohx{g})$.

\sprove{MaximumSum}
\end{proof}

\begin{theorem}[\uproperty{Additive} is implied]
\label{AdditivityIsImplied}
$\ohs{X}$ has \require{Order}, \require{Trans}, \require{Local}, \require{Scale}, and \require{ISubComp}. $\implies$ $\ohs{X}$ has \prove{Additive}.
\end{theorem}

\begin{proof}
Let $f, g \in \rc{X}$. $\ohs{X}$ has \property{Summation} by \proveby{SummationIsImplied}. Therefore, 
\begin{equation}
\ohx{f + g} = \ohx{\max(f, g)}.
\end{equation}
$\ohs{X}$ has \property{Maximum} by \proveby{MaximumIsImplied}. Therefore, 
\begin{equation}
\ohx{\max(f, g)} = \max(\ohx{f}, \ohx{g}).  
\end{equation}
$\ohs{X}$ has \property{MaximumSum} by \proveby{MaximumSumIsImplied}. Therefore, 
\begin{equation}
\max(\ohx{f}, \ohx{g}) = \ohx{f} + \ohx{g}.   
\end{equation}
Therefore $\ohx{f + g} = \ohx{f} + \ohx{g}$.

\sprove{Additive}
\end{proof}

\begin{theorem}[\uproperty{Translation} is implied]
\label{TranslationInvarianceIsImplied}
$\ohs{X}$ has \require{Order}, \require{Trans}, and \require{Scale} $\implies$ $\ohs{X}$ has \prove{Translation}.
\end{theorem}

\begin{proof}
Let $\alpha, \beta \in \posi{\TR}$, and $f \in \rc{X}$.

\proofpart{$\subset$}
Since $(f + \beta) / \beta \geq 1$,
\begin{eqs}
(f + \beta) + \alpha & \leq (f + \beta) + \alpha ((f + \beta) / \beta) \\
{} & = (1 + \alpha / \beta) (f + \beta).
\end{eqs}
By \require{Order} and \require{Trans}, $\ohx{f + \beta + \alpha} \subset \ohx{(1 + \alpha / \beta)(f + \beta)}$. By \require{Scale} $\ohx{f + \beta + \alpha} \subset \ohx{f + \beta}$

\proofpart{$\supset$}
By \require{Order} and \require{Trans}, $\ohx{f + \beta} \subset \ohx{f + \beta + \alpha}$.

\sprove{Translation}
\end{proof}

\begin{theorem}[\uproperty{Extend} is implied]
\label{ExtensibilityIsImplied}
$\ohs{X}$ has \require{SubComp} $\implies$ $\ohs{X}$ has \prove{Extend}.
\end{theorem}

\begin{proof}
This follows from \require{SubComp} by $X = X^*$, $Y = X^* \times Y^*$, and $s = \projections{X^*}{Y^*}$. \sprove{Extend}
\end{proof}

\begin{theorem}[\uproperty{SubsetSum} is implied]
\label{SubsetSumIsImplied}
$\ohs{X}$ has primitive properties. \srequire{LinearDominance} $\implies$ $\ohs{X}$ has \prove{SubsetSum}.
\end{theorem}

\begin{proof}
This is proved in \proveby{SubsetSumIsAnOMapping}.
\end{proof}

\begin{theorem}[\uproperty{SubsetSum} implies \property{SubComp}]
\label{SubsetSumImpliesSubComposability}
$\ohs{X}$ has \require{Order}, \require{Trans} and \require{SubsetSum}. \srequire{SubsetSum} $\implies$ $\ohs{X}$ has \prove{SubComp}.
\end{theorem}

\begin{proof}
$\ohs{X}$ has \property{Reflex} by \proveby{ReflexivityIsImplied}, and \property{Orderness} by \proveby{OrdernessIsImplied}. 

Let $\function{s}{X}{Y}$, and $S_x = \setb{(s(x), 0) : x \in X}$. Let $\function{a}{\TN}{\nonn{\TR}}$ be such that $a(z) = 1$. Then for $f \in \rc{Y}$,
\begin{eqs}
\ohx{\sum_{(y, z) \in S_x} a(z) f(y)} = \ohx{f \circ s}.
\end{eqs}
Similarly, for $\bar{f} \in \rc{Y}$,
\begin{eqs}
\ohx{\sum_{(y, z) \in S_x} a(z) \bar{f}(y)} = \ohx{\bar{f} \circ s}.
\end{eqs}
By \require{SubsetSum},
\begin{eqs}
\ohx{f \circ s} \subset \ohx{\bar{f} \circ s}.
\end{eqs}
By \require{Orderness},
\begin{eqs}
(f \circ s) \in \ohx{\bar{f} \circ s}.
\end{eqs}
\sprove{SubComp}
\end{proof}

\chapter{Proofs of local linear dominance properties}
\label{ProofsForLocalLinearDominance}

In this section we prove some of the properties of a local linear dominance. We shall apply the following simplification lemma without mentioning it, since it is used in almost every proof.

\begin{lemma}[Simplification lemma for $\rohsy$]
\label{LocalSingleConstantLemma}
Let $X \in U$, $I$ be a finite set, $X_i \subset X$, $f_i \in \rc{X_i}$, and $\hat{f}_i \in \roh{X_i}{f_i}$, for all $i \in I$. Then there exists $c \in \posi{\TR}$ and $A \in \filterset{X}$, such that
\begin{equation}
\restrb{\hat{f}_i}{(X_i \cap A)} \lt c \restrb{f_i}{(X_i \cap A)},
\end{equation}
for all $i \in I$.
\end{lemma}

\begin{proof}
Since $\hat{f}_i \in \roh{X_i}{f_i}$, there exists $c_i \in \posi{\TR}$ and $B_i \in \filterset{X_i}$, such that
\begin{equation}
\restrb{\hat{f}_i}{B_i} \lt c_i \restrb{f}{B_i},
\end{equation}
for all $i \in I$. By induced sub-structure, there exists $A_i \in \filterset{X}$ such that $B_i = X_i \cap A_i$, for all $i \in I$. Since $\filterset{X}$ is $\subset$-directed, and $I$ is finite, there exists $A \in \filterset{X}$ such that $A \subset \bigcap_{i \in I} A_i$. Let $c = \max \setb{c_i : i \in I}$. Since $(X \cap A) \subset (X \cap A_i)$,
\begin{equation}
\restrb{\hat{f}_i}{(X_i \cap A)} \lt c \restrb{f_i}{(X_i \cap A)},
\end{equation}
for all $i \in I$.
\end{proof}

\begin{note}[Filter basis and simplification lemma]
A filter basis in $X$ seems to be the minimal amount of structure needed to prove the simplification lemma. Indeed, we first provided the abstraction of a filter basis solely to prove this lemma in its most general form. It is only later that we noticed that this structure also allows us to define limits.
\end{note}

\begin{theorem}[\uproperty{Order} for $\rohsy$]
\label{LocalOrderConsistency}
Let $X \in U$, and $f, g \in \rc{X}$. Then \sprove{Order}
\begin{equation}
f \leq g \implies \rohx{f} \subset \rohx{g}.
\end{equation}
\end{theorem}

\begin{proof}
Let $\hat{f} \in \rohx{f}$. Then there exists $c \in \posi{\TR}$ and $A \in \filterset{X}$ such that
\begin{equation}
\restrb{\hat{f}}{A} \leq c \restrb{f}{A}.
\end{equation}
Since $f \leq g$,
\begin{equation}
\restrb{\hat{f}}{A} \leq c \restrb{g}{A}.
\end{equation}
Therefore $\hat{f} \in \rohx{g}$.
\sprove{Order}
\end{proof}

\begin{theorem}[\uproperty{Trans} for $\rohsy$]
\label{LocalTransitivity}
Let $X \in U$, and $f, g, h \in \rc{X}$. Then \sprove{Trans}
\begin{equation}
\bra{f \in \rohx{g} \textrm{ and } g \in \rohx{h}} \implies f \in \rohx{h}.
\end{equation}
\end{theorem}

\begin{proof}
Let $f \in \rohx{g}$, and $g \in \rohx{h}$. Then there exists $c \in \posi{\TR}$ and $A \in \filterset{X}$, such that
\begin{eqs}
\restrb{f}{A} & \lt c \restrb{g}{A}, \\
\restrb{g}{A} & \lt c \restrb{h}{A}.
\end{eqs}
It follows that
\begin{eqs}
\restrb{f}{A} & \lt c (c \restrb{h}{A}) \\
{} & = c^2 \restrb{h}{A}.
\end{eqs}
Therefore $f \in \rohx{h}$.
\sprove{Trans}
\end{proof}

\begin{theorem}[\uproperty{Local} for $\rohsy$]
\label{LocalLocality}
Let $X \in U$, $f, g \in \rc{X}$, and $C \subset \power{X}$ be a finite cover of $X$. Then \sprove{Local}
\begin{equation}
\bra{\forall D \in C: \restrb{f}{D} \in \roh{D}{\restr{g}{D}}} \implies f \in \rohx{g}.
\end{equation}
\end{theorem}

\begin{proof}
% $\implies$ Assume $f \in \rohx{g}$. Then there exists $c \in \posi{\TR}$ and $A \in \filterset{X}$, such that
% \begin{equation}
% \restrb{f}{A} \lt c \restrb{g}{A}.
% \end{equation}
% Then
% \begin{equation}
% \begin{split}
% \qquad & \bra{\restr{f}{\bra{D \cap A}}} \lt c \bra{\restr{g}{\bra{D \cap A}}} \\
% \iffr & \restrb{\restrb{f}{D}}{(D \cap A)} \lt c \restrb{\restrb{g}{D}}{(D \cap A)},
% \end{split}
% \end{equation}
% for all $D \in C$. Since $D \cap A \in \filterset{D}$ by induced sub-structure, it holds that $\bra{\restr{f}{D}} \in \roh{D}{\restr{g}{D}}$, for all $D \in C$.

% $\impliedby$ 
Assume $\bra{\restr{f}{D}} \in \roh{D}{\restr{g}{D}}$, for all $D \in C$. Then there exist $c \in \posi{\TR}$ and $A \in \filterset{X}$ such that
\begin{eqs}
\qquad & \restrb{\restrb{f}{D}}{(D \cap A)} \lt c \restrb{\restrb{g}{D}}{(D \cap A)} \\
\iffr & \bra{\restr{f}{\bra{D \cap A}}} \lt c \bra{\restr{g}{\bra{D \cap A}}},
\end{eqs}
for all $D \in C$. Since $C$ covers $X$,
\begin{equation}
\restrb{f}{A} \lt c \restrb{g}{A}.
\end{equation}
Therefore $f \in \rohx{g}$.
\sprove{Local}
\end{proof}

\begin{theorem}[\uproperty{One} for $\rohsy$ characterized]
\label{LocalOneSeparationCharacterized}
\begin{equation}
\bra{\forall A \in \filterset{\posi{\TN}} : |A| \not\in \TN} \iff n \not\in \roh{\posi{\TN}}{1}.
\end{equation}
\end{theorem}

\begin{proof}
\proofpart{$\implies$}
Suppose every set in $\filterset{\posi{\TN}}$ is infinite. Then for all $c \in \posi{\TR}$ and $A \in \filterset{\posi{\TN}}$ there exists $n \in A$ such that $n > c 1$. Therefore $n \not\in \roh{\posi{\TN}}{1}$.

\proofpart{$\impliedby$}
Suppose some set $A \in \filterset{\posi{\TN}}$ is finite. Let $c = \max(A \cup \setb{1})$. Then $\restrb{n}{A} \leq c \restrb{1}{A}$. Therefore $n \in \roh{\posi{\TN}}{1}$.
\end{proof}

\begin{theorem}[\uproperty{Scale} for $\rohsy$]
\label{LocalScaleInvariance}
Let $X \in U$, $f \in \rc{X}$, and $\alpha \in \posi{\TR}$. Then $\rohx{\alpha f} = \rohx{f}$. \sprove{Scale}
\end{theorem}

\begin{proof}
\proofpart{$\subset$}
Assume $\hat{f} \in \rohx{\alpha f}$. Then there exists $c \in \posi{\TR}$ and $A \in \filterset{X}$, such that
\begin{eqs}
\restrb{\hat{f}}{A} & \lt c \restrb{(\alpha f)}{A} \\
{} & = (c \alpha) \restrb{f}{A}.
\end{eqs}
Therefore $\hat{f} \in \rohx{f}$. 

\proofpart{$\supset$}
Assume $\hat{f} \in \rohx{f}$. Then there exists $c \in \posi{\TR}$ and $A \in \filterset{X}$, such that
\begin{eqs}
\restrb{\hat{f}}{A} & \lt c \restrb{f}{A} \\
{} & = (c / \alpha) \restrb{(\alpha f)}{A}.
\end{eqs}
Therefore $\hat{f} \in \rohx{\alpha f}$.

\sprove{Scale}
\end{proof}

\begin{theorem}[\uproperty{SubHom} for $\rohsy$]
\label{LocalSubHomogenuity}
Let $X \in U$, and $f, u \in \rc{X}$. Then \sprove{SubHom}
\begin{equation}
u \rohx{f} \subset \rohx{uf}.
\end{equation}
\end{theorem}

\begin{proof}
Let $\hat{f} \in \rohx{f}$. Then there exists $c \in \posi{\TR}$ and $A \in \filterset{X}$, such that
\begin{eqs}
\restrb{\hat{f}}{A} & \lt c \restrb{f}{A}.
\end{eqs}
This implies
\begin{eqs}
\restrb{u\hat{f}}{A} \lt c \restrb{uf}{A}.
\end{eqs}
Therefore $u\hat{f} \in \rohx{uf}$.
\sprove{SubHom}
\end{proof}

\begin{theorem}[\uproperty{SuperHom} for $\rohsy$ characterized]
\label{LocalSuperHomogenuityCharacterized}
$\ohs{X}$ has \property{SuperHom}. $\iff$ $\filterset{X} = \setb{X}$.
\end{theorem}

\begin{proof}
\proofpart{$\implies$}
Suppose $\filterset{X} \neq \setb{X}$. Then there exists exists $A \in \filterset{X}$ such that $A \neq X$. Let $\compl{A} = X \setminus A$, and $f \in \rc{X}$ be such that $f = \indi{A}{X}$. Let $g \in \rc{X}$, and $\hat{h} \in \rc{X}$ be such that $\hat{h} = fg\indi{A}{X} + \indi{\compl{A}}{X}$. Then $\restrb{\hat{h}}{A} \leq \restrb{fg}{A}$, and so $\hat{h} \in \rohx{fg}$. Let $\hat{g} \in \rohx{g}$. Then $0 =\restrb{f \hat{g}}{\compl{A}} \neq \restrb{\hat{h}}{\compl{A}} = 1$; $\ohs{X}$ does not have \property{SuperHom}.

\proofpart{$\impliedby$}
Suppose $\filterset{X} = \setb{X}$. Let $f, g \in \rc{X}$, and $\hat{h} \in \rohx{fg}$. Then there exists $c \in \posi{\TR}$ such that
\begin{equation}
\hat{h} \leq cfg.
\end{equation}
Let $F = \preimage{f}{\posi{\TR}}$, and $\compl{F} = X \setminus F$. Let $\hat{g} \in \rc{X}$ be such that $\restrb{\hat{g}}{F} = \restrb{\hat{h}}{F} / \restrb{f}{F}$ and $\restrb{\hat{g}}{\compl{F}} = 0$. Then $\hat{g} \leq cg$, and so $\hat{g} \in \rohx{g}$. In addition, $\hat{h} = f \hat{g}$. 
\end{proof}

\begin{theorem}[\uproperty{SuperMulti} for $\rohsy$]
\label{LocalSuperMultiplicativity}
Let $X \in U$, and $f, g \in \rc{X}$. Then \sprove{SuperMulti}
\begin{equation}
\rohx{f} \rohx{g} \supset \rohx{fg}.
\end{equation}
\end{theorem}

\begin{proof}
Let $\hat{h} \in \rohx{fg}$. Then there exists $c \in \posi{\TR}$ and $A \in \filterset{X}$, such that
\begin{equation}
\restrb{\hat{h}}{A} \lt c \restrb{fg}{A}.
\end{equation}
Let $\hat{g} \in \rc{X}$ be such that
\begin{equation}
\hat{g}(x) =
\begin{cases}
g(x), & x \in A, \\
1, & x \not\in A. \\
\end{cases}
\end{equation}
Then $\hat{g} \in \rohx{g}$, since $\restrb{\hat{g}}{A} = \restrb{g}{A}$. Let $\hat{G} = \setb{x \in X : \hat{g}(x) > 0}$, and let $\hat{f} \in \rc{X}$ be such that
\begin{equation}
\hat{f}(x) =
\begin{cases}
\hat{h}(x) / \hat{g}(x), & x \in \hat{G} \\
0, & x \not\in \hat{G}.
\end{cases}
\end{equation}
If $x \in \hat{G}$, then $\hat{h}(x) = \hat{f}(x) \hat{g}(x)$. If $x \not\in \hat{G}$, then $x \in A$, and $g(x) = 0 = \hat{g}(x)$. Since $\hat{h}(x) \leq c f(x) g(x)$, it follows that $\hat{h}(x) = 0$. Therefore $\hat{h}(x) = \hat{f}(x) \hat{g}(x)$, and $\hat{h} = \hat{f} \hat{g}$. 
If $x \in A \setminus \hat{G}$, then $\hat{f}(x) = 0$, and $\hat{f}(x) \leq c f(x)$. If $x \in A \cap \hat{G}$, then $g(x) = \hat{g}(x) > 0$, and $\hat{f}(x) = \hat{h}(x) / g(x) \leq c f(x)$. Therefore $\hat{f} \in \rohx{f}$, and $\hat{h} \in \rohx{f} \rohx{g}$.
\sprove{SuperMulti}
\end{proof}

\begin{note}[\uproperty{SuperMulti} in another way?]
In proving \property{SuperMulti}, we cannot refer to \thref{SuperMultiplicativityIsImplied}, since there are local linear dominances which do not satisfy its assumptions.
\end{note}

\begin{theorem}[\uproperty{SubComp} for $\rohsy$ for fixed $\function{s}{Y}{X}$]
\label{LocalSubComposabilityForFixedS}
Let $X, Y \in U$, and $\function{s}{Y}{X}$. Then
\begin{eqs}
{} & \forall f \in \rc{X}: \roh{X}{f} \circ s \subset \roh{Y}{f \circ s} \\
\iffr & \forall A_X \in \filterset{X}, \exists A_Y \in \filterset{Y}: \image{s}{A_Y} \subset A_X.
\end{eqs}
\end{theorem}

\begin{proof}
First notice that
\begin{eqs}
{} & \image{s}{A_Y} \subset A_X \\
\iffr & A_Y \subset \preimage{s}{A_X} \\
\iffr & \indi{A_Y}{Y} \leq \indi{\preimage{s}{A_X}}{Y} \\
\iffr & \indi{A_Y}{Y} \leq \indi{A_X}{X} \circ s,
\end{eqs}
for all $A_X \in \filterset{X}$, $A_Y \in \filterset{Y}$. We attempt to prove the statement in this form. 

\proofpart{$\implies$}
Let $A_X \in \filterset{X}$, $f = \indi{A_X}{X}$, and $\hat{f} \in \rc{X}$ be such that $\hat{f} = 1$. Then $\restrb{\hat{f}}{A_X} \leq 1 \restrb{f}{A_X}$, and so $\hat{f} \in \rohx{f}$. By assumption, $\hat{f} \circ s = 1 \in \roh{Y}{f \circ s}$, and so there exists $d \in \posi{\TR}$ and $A_Y \in \filterset{Y}$ such that
\begin{eqs}
\indi{A_Y}{Y} & \leq \indi{A_Y}{Y} d (f \circ s) \\
{} & = \indi{A_Y}{Y} d \bra{\indi{A_X}{X} \circ s} \\
{} & \leq d \bra{\indi{A_X}{X} \circ s}.
\end{eqs}
Since the functions are indicator functions, this is equivalent to
\begin{equation}
\indi{A_Y}{Y} \leq \indi{A_X}{X} \circ s.
\end{equation}

\proofpart{$\impliedby$}
Let $f \in \rc{X}$ and $\hat{f} \in \roh{X}{f}$. Then there exists $c \in \posi{\TR}$ and $A_X \in \filterset{X}$ such that $\idf{A_X} \hat{f} \leq \idf{A_X} c f$. By assumption, there exists $A_Y \in \filterset{Y}$ such that $\indi{A_Y}{Y} \leq \indi{A_X}{X} \circ s$. Then
\begin{eqs}
{} & \idf{A_X} \hat{f} \leq \idf{A_X} c f \\
\impliesr & \bra{\idf{A_X} \circ s} \bra{\hat{f} \circ s} \leq 
\bra{\idf{A_X} \circ s} c \bra{f \circ s} \\
\impliesr & \indi{A_Y}{Y} \bra{\hat{f} \circ s} \leq 
\indi{A_Y}{Y} c \bra{f \circ s} \\
\impliesr & \bra{\hat{f} \circ s} \in \roh{Y}{f \circ s}.
\end{eqs}
\end{proof}

\begin{theorem}[Characterization of \property{SubComp} for $\rohsy$]
\label{LocalSubComposabilityCharacterized}
$\rohsy$ has \property{SubComp} if and only if
\begin{eqs}
\bra{\forall X \in U: \filterset{X} = \setb{X}} \lor \bra{\forall X \in U: \emptyset \in \filterset{X}}.
\end{eqs}
\end{theorem}

\begin{proof}
\proofpart{$\implies$}
Suppose there exists $X \in U$, such that $\emptyset \not\in \filterset{X}$ and $\filterset{X} \neq \setb{X}$. Let $A \in \filterset{X}$, $x^* \in X \setminus A$, and $\function{s}{X}{X}$ be such that $s(x) = x^*$. Then $\image{s}{X} \cap A = \emptyset$, and $\rohsy$ does not have \property{SubComp} under $s$ by \thref{LocalSubComposabilityForFixedS}. This shows that either $\filterset{X} = \setb{X}$, or $\emptyset \in \filterset{X}$, for all $X \in U$. 

Suppose there exists $X, Y \in U$, such that $\emptyset \in \filterset{X}$, and $\filterset{Y} = \setb{Y}$. Then \property{SubComp} does not hold for any $\function{s}{Y}{X}$, by \thref{LocalSubComposabilityForFixedS}, because $\image{s}{Y} \not\subset \emptyset$.

\proofpart{$\impliedby$}
Let $X, Y \in U$, and $\function{s}{Y}{X}$. Either $\emptyset \in \filterset{X}$ and $\emptyset \in \filterset{Y}$, or $\filterset{X} = \setb{X}$ and $\filterset{Y} = \setb{Y}$. In either case, \property{SubComp} holds by \thref{LocalSubComposabilityForFixedS}.
\end{proof}

\begin{theorem}[\uproperty{SubComp} for $\rohsy$ for positive functions]
\label{LocalSubComposabilityForPositiveFunctions}
$\forall X \in U, \forall A \in \filterset{X}: \card{X \setminus A} < \infty \implies \rohsy$ has \property{SubComp} for positive functions .
\end{theorem}

\begin{proof}
Let $\hat{f} \in \rohx{f}$, where $f \in \posi{\rc{X}}$. Then there exists $A \in \filterset{X}$ and $c \in \posi{\TR}$, such that
\begin{eqs}
{} & \restrb{\hat{f}}{A} \leq c \restrb{f}{A} \\
\iffr & \frac{\restrb{\hat{f}}{A}}{\restrb{f}{A}} \leq c.
\end{eqs}
Since $\card{X \setminus A} < \infty$, let
\begin{eqs}
d = \max \setb{\frac{\restr{\hat{f}}{(X \setminus A)}}{\restr{f}{(X \setminus A)}}} \cup \setb{c}.
\end{eqs}
Then $d \in \posi{\TR}$, and
\begin{eqs}
{} & \frac{\hat{f}}{f} \leq d \\
\iffr & \hat{f} \leq d f.
\end{eqs}
Let $\function{s}{Y}{X}$, where $Y \in U$. Then
\begin{eqs}
(\hat{f} \circ s) \leq d (f \circ s).
\end{eqs}
Therefore $\bra{\hat{f} \circ s} \in \rohx{f \circ s}$.
\end{proof}

\begin{theorem}[Characterization of \uproperty{Extend} for $\rohsy$]
\label{LocalExtensibilityCharacterized}
$\rohs{X}$ has \property{Extend} if and only if
\begin{eqs}
\forall A \in \filterset{X}: \exists B \in \filterset{X \times Y}: \image{\projections{X}{Y}}{B} \subset A.
\end{eqs}
\end{theorem}

\begin{proof}
This follows directly from \thref{LocalSubComposabilityForFixedS}.
\end{proof}

\begin{theorem}[\uproperty{PowerH} for $\rohsy$]
\label{LocalPowerHomogenuity}
Let $X \in U$, $f \in \rc{X}$, and $\alpha \in \posi{\TR}$. Then \sprove{PowerH}
\begin{eqs}
\rohx{f}^{\alpha} = \rohx{f^{\alpha}}.
\end{eqs}
\end{theorem}

\begin{proof}
\proofpart{$\subset$}
Let $\hat{f} \in \rohx{f}$. Then there exists $A \in \filterset{X}$, and $c \in \posi{\TR}$, such that
\begin{eqs}
\restrb{\hat{f}}{A} \leq c \restrb{f}{A}.
\end{eqs}
Then
\begin{eqs}
\restrb{\hat{f}}{A}^{\alpha} \leq c^{\alpha} \restrb{f}{A}^{\alpha}.
\end{eqs}
Therefore $\hat{f}^{\alpha} \in \rohx{f^{\alpha}}$.

\proofpart{$\supset$}
Let $\hat{f} \in \rohx{f^{\alpha}}$. Then there exists $A \in \filterset{X}$, and $c \in \posi{\TR}$, such that
\begin{eqs}
\restrb{\hat{f}}{A} \leq c \restrb{f^{\alpha}}{A}.
\end{eqs}
Then
\begin{eqs}
\restrb{\hat{f}}{A}^{1 / \alpha} \leq c \restrb{f}{A}.
\end{eqs}
Let $g \in \rc{X}$ be such that $g = \hat{f}^{1 / \alpha}$. Then $g \in \rohx{f}$, and $g^{\alpha} = \hat{f}$. 
\sprove{PowerH}
\end{proof}

\begin{note}[\uproperty{PowerH} from primitive properties]
We would rather want to prove \property{PowerH} from primitive properties, as done for other non-primitive properties in \sref{ImpliedProperties}. However, \we{} \were{} unable to come up with such a proof.
\end{note}

\begin{theorem}[\uproperty{TrivialZero} for $\rohsy$ characterized]
\label{LocalZeroTrivialityCharacterized}
$\rohs{X}$ has \property{TrivialZero} if and only if $\filterset{X} = \setb{X}$.
\end{theorem}

\begin{proof}
\proofpart{$\implies$}
Suppose there exists $A \in \filterset{X}$ such that $A \neq X$. Then $\indi{X \setminus A}{X} \in \ohx{0}$ and \property{TrivialZero} does not hold.

\proofpart{$\impliedby$}
Suppose $\filterset{X} = \setb{X}$, and $f \in \ohx{0}$. Then there exists $c \in \posi{\TR}$, such that
\begin{eqs}
f \leq c 0.
\end{eqs}
Therefore $f = 0$.
\end{proof}

\ReLinearDominanceFromLocalLinearDominance

\begin{proof}
\proofpart{$\subset$}
Suppose $f \in \rohx{g}$. Then there exists $A \in \filterset{X}$ and $c \in \posi{\TR}$, such that
\begin{eqs}
{} & \restrb{f}{A} \leq c \restrb{g}{A} \\
\iffr & \frac{\restrb{f}{A}}{\restrb{g}{A}} \leq c.
\end{eqs}
Since $X \setminus A$ is finite, let
\begin{eqs}
d & = \max \bra{\frac{f}{g}}\bra{X \setminus A} \cup \setb{c}.
\end{eqs}
Then
\begin{eqs}
{} & \frac{f}{g} \leq d \\
\iffr & f \leq dg.
\end{eqs}
Therefore $f \in \lohx{g}$.

\proofpart{$\supset$}
Suppose $f \in \lohx{g}$. Then there exists $c \in \posi{\TR}$, such that
\begin{eqs}
f \leq cg.
\end{eqs}
Let $A \in \filterset{X}$. Then
\begin{eqs}
\restrb{f}{A} \leq c \restrb{g}{A}.
\end{eqs}
Therefore $f \in \rohx{g}$.
\end{proof}

\chapter{Proofs of limit theorems}
\label{ProofsOfLimitsTheorems}

In this section we prove the limit theorems for local linear dominance.

\ReRelationBetweenRatioLimits

\begin{proof}
\proofpart{Notation}
Let
\begin{eqs}
c \coloneqq \limsup_{\filterset{F}} \frac{\restrb{f}{F}}{\restrb{g}{F}},
\end{eqs}
and
\begin{eqs}
d \coloneqq \liminf_{\filterset{F}} \frac{\restrb{g}{F}}{\restrb{f}{F}}.
\end{eqs}

\proofpart{$c = \infty \iff d = 0$}
\begin{eqs}
{} & c = \infty \\
\iffr & \forall A \in \filterset{F}: \sup\image{\bra{\frac{\restr{f}{A}}{\restr{g}{A}}}}{A} = \infty \\
\iffr & \forall A \in \filterset{F}: \exists x \in A: g(x) = 0 \\
\iffr & \forall A \in \filterset{F}: \inf\image{\bra{\frac{\restr{g}{A}}{\restr{f}{A}}}}{A} = 0 \\
\iffr & d = 0.
\end{eqs}
Therefore, if $c = \infty$, or $d = 0$, we have that
\begin{eqs}
c = 1 / d.
\end{eqs}

\proofpart{$c = 0 \implies d = \infty$}
Suppose $c = 0$. For any $\epsilon \in \posi{\TR}$, there exists $A \in \filterset{F}$, such that
\begin{eqs}
{} & \sup\image{\bra{\frac{\restr{f}{A}}{\restr{g}{A}}}}{A} \leq \epsilon \\
\impliesr & \frac{\restrb{f}{A}}{\restrb{g}{A}} \leq \epsilon \\
\impliesr & \frac{\restrb{g}{A}}{\restrb{f}{A}} \geq \frac{1}{\epsilon} \\
\impliesr & \inf\image{\bra{\frac{\restr{g}{A}}{\restr{f}{A}}}}{A} \geq \frac{1}{\epsilon} \\
\end{eqs}
Therefore,
\begin{eqs}
d \geq \frac{1}{\epsilon}.
\end{eqs}
Since this holds for all $\epsilon$,
\begin{eqs}
d = \infty.
\end{eqs}
Therefore, if $c = 0$, then $d = \infty$, and
\begin{eqs}
c = 1 / d.
\end{eqs}

\proofpart{$c = 0 \impliedby d = \infty$}
Suppose $d = \infty$. For $\epsilon \in \posi{\TR}$, there exists $A \in \filterset{F}$, such that
\begin{eqs}
{} & \inf\image{\bra{\frac{\restr{g}{A}}{\restr{f}{A}}}}{A} \geq \epsilon \\
\impliesr & \frac{\restrb{g}{A}}{\restrb{f}{A}} \geq \epsilon \\
\impliesr & \frac{\restrb{f}{A}}{\restrb{g}{A}} \leq \frac{1}{\epsilon} \\
\impliesr & \sup\image{\bra{\frac{\restr{f}{A}}{\restr{g}{A}}}}{A} \leq \frac{1}{\epsilon}.
\end{eqs}
Therefore
\begin{eqs}
c \leq \frac{1}{\epsilon}.
\end{eqs}
Since this holds for all $\epsilon$,
\begin{eqs}
c = 0.
\end{eqs}
Therefore, if $d = \infty$, then $c = 0$, and
\begin{eqs}
c = 1 / d.
\end{eqs}

\proofpart{$c \geq 1 / d$}
Suppose $c, d \in \posi{\TR}$. By definition, for any $\epsilon \in \posi{\TR}$, there exists $A \in \filterset{F}$, such that
\begin{eqs}
{} & \sup\image{\bra{\frac{\restr{f}{A}}{\restr{g}{A}}}}{A} - c \leq \epsilon \\
\impliesr & \frac{\restrb{f}{A}}{\restrb{g}{A}} - c \leq \epsilon \\
\impliesr & \frac{\restrb{f}{A}}{\restrb{g}{A}} \leq c + \epsilon \\
\impliesr & \frac{\restrb{g}{A}}{\restrb{f}{A}} \geq \frac{1}{c + \epsilon} \\
\impliesr & \inf\image{\bra{\frac{\restr{g}{A}}{\restr{f}{A}}}}{A} \geq \frac{1}{c + \epsilon}.
\end{eqs}
Therefore
\begin{eqs}
d \geq \frac{1}{c + \epsilon}.
\end{eqs}
Since this holds for all $\epsilon$,
\begin{eqs}
{} & d \geq \frac{1}{c} \\
\impliesr & c \geq 1 / d.
\end{eqs}

\proofpart{$c \leq 1 / d$}
Suppose $c, d \in \posi{\TR}$. By definition, for any $\epsilon \in \posi{\TR}$ such that $\epsilon < d$, there exists $A \in \filterset{F}$, such that
\begin{eqs}
{} & d - \inf\image{\bra{\frac{\restr{g}{A}}{\restr{f}{A}}}}{A} \leq \epsilon \\
\impliesr & d - \frac{\restrb{g}{A}}{\restrb{f}{A}} \leq \epsilon \\
\impliesr & d - \epsilon \leq \frac{\restrb{g}{A}}{\restrb{f}{A}} \\
\impliesr & \frac{1}{d - \epsilon} \geq \frac{\restrb{f}{A}}{\restrb{g}{A}} \\
\impliesr & \frac{1}{d - \epsilon} \geq \sup\image{\bra{\frac{\restr{f}{A}}{\restr{g}{A}}}}{A}.
\end{eqs}
Therefore
\begin{eqs}
c \leq \frac{1}{d - \epsilon}.
\end{eqs}
Since this holds for all $\epsilon$,
\begin{eqs}
c \leq \frac{1}{d}.
\end{eqs}
\end{proof}

\ReLocalLinearOByALimit

\begin{proof}
\proofpart{Notation}
Let
\begin{eqs}
d \coloneqq \limsup_{\filterset{F}} \frac{\restrb{f}{F}}{\restrb{g}{F}}.
\end{eqs}

\proofpart{$\implies$}
Suppose $d < \infty$. By definition, for any $\epsilon \in \posi{\TR}$, there exists $A \in \filterset{F}$, such that $\restrb{g}{A} > 0$, and 
\begin{eqs}
{} & \sup\image{\bra{\frac{\restr{f}{A}}{\restr{g}{A}}}}{A} - d \leq \epsilon \\
\impliesr & \frac{\restrb{f}{A}}{\restrb{g}{A}} - d \leq \epsilon \\
\impliesr & \restrb{f}{A} \leq (d + \epsilon) \restrb{g}{A}.
\end{eqs}
Therefore $\restrb{f}{F} \in \roh{F}{\restr{g}{F}}$. By \property{Order}, $\restrb{f}{\compl{F}} \in \roh{F}{\restr{g}{\compl{F}}}$, where $\compl{F} = X \setminus F$. By \property{Local}, $f \in \rohx{g}$.

\proofpart{$\impliedby$}
Suppose $f \in \rohx{g}$. By \property{Restrict}, $\restrb{f}{F} \in \roh{F}{\restr{g}{F}}$. Therefore, there exists $A \in \filterset{F}$ and $c \in \posi{\TR}$, such that
\begin{eqs}
{} & \restrb{f}{A} \leq c \restrb{g}{A} \\
\end{eqs}
In particular, $\restrb{g}{A} > 0$. Therefore
\begin{eqs}
{} & \frac{\restrb{f}{A}}{\restrb{g}{A}} \leq c \\
\impliesr & \sup\image{\bra{\frac{\restr{f}{A}}{\restr{g}{A}}}}{A} \leq c.
\end{eqs}
Therefore
\begin{eqs}
d \leq c < \infty.
\end{eqs}
\end{proof}

\ReLocalLinearOmegaByALimit

\begin{proof}
This follows from \thref{LocalLinearOByALimit} and \thref{RelationBetweenRatioLimits}.
\end{proof}

\ReLocalLinearSmallOhByALimit

\begin{proof}
By definition,
\begin{eqs}
{} & f \in \rsmallohx{g} \\
\iffr & f \in \rohx{g} \land \lnot (g \in \rohx{f}) \\
\iffr & f \in \rohx{g} \land \lnot (f \in \romegahx{g}).
\end{eqs}
The first term is obtained by \thref{LocalLinearOByALimit}, and the second term is obtained by \thref{LocalLinearOmegaByALimit}.
\end{proof}

\ReTraditionalSmallOByALimit

\begin{proof}
It holds that
\begin{eqs}
{} & \forall \epsilon \in \posi{\TR}: \exists A \in \filterset{X}: \restrb{f}{A} \leq \epsilon \restrb{g}{A} \\
\iffr & \forall \epsilon \in \posi{\TR}: \exists A \in \filterset{F}: \restrb{f}{A} \leq \epsilon \restrb{g}{A} \\
\iffr & \forall \epsilon \in \posi{\TR}: \exists A \in \filterset{F}: \frac{\restrb{f}{A}}{\restrb{g}{A}} \leq \epsilon \\
\iffr & \forall \epsilon \in \posi{\TR}: \exists A \in \filterset{F}: \sup \image{\frac{\restrb{f}{A}}{\restrb{g}{A}}}{A} \leq \epsilon \\
\iffr & \forall \epsilon \in \posi{\TR}: \limsup_{\filterset{F}} \frac{\restrb{f}{F}}{\restrb{g}{F}} \leq \epsilon \\
\iffr & \limsup_{\filterset{F}} \frac{\restrb{f}{F}}{\restrb{g}{F}} = 0.
\end{eqs}
\end{proof}

\ReLocalLinearSmallOmegaByALimit

\begin{proof}
Similarly to \thref{LocalLinearSmallOhByALimit}.
\end{proof}

\chapter{Proofs of Master theorems}
\label{MasterTheorems}

In this section we will show that various Master theorems hold for the $\lohsy$-notation as defined by linear dominance. The theorems are simpler for linear dominance than for asymptotic linear dominance; there are no ``regularity'' requirements (see \cite{IntroAlgo2009}).

\section{Master theorem over powers}
\label{MasterTheoremOverPowersSection}

\begin{definition}[Master function over powers]
Let $a \in \nonnb{\TR}{1}$, $b \in \posib{\TR}{1}$, $d \in \posi{\TR}$, $B = \setb{b^i : i \in \TN}$, and $F \in \rc{B}$. A \definesub{master function}{over powers} is a function $T \in \rc{B}$ defined by the recurrence equation
\begin{equation}
T(n) =
\begin{cases}
a T(n / b) + F(n), & n \geq b, \\
d, & n < b.
\end{cases}
\end{equation}
The set of such functions is denoted by $\pmasters{a}{b}{d}{F}$.
\end{definition}

\begin{theorem}[Logarithm swap]
\label{LogarithmSwap}
Let $x, y, b \in \posi{\TR}$ be such that $b \neq 1$. Then
\begin{equation}
x^{\lgb{b}{y}} = y^{\lgb{b}{x}}.
\end{equation}
\end{theorem}

\begin{proof}
\begin{eqs}
x^{\lgb{b}{y}} & = \bra{b^{\lgb{b}{x}}}^{\lgb{b}{y}} \\
{} & = b^{\lgb{b}{x} \lgb{b}{y}} \\
{} & = \bra{b^{\lgb{b}{y}}}^{\lgb{b}{x}} \\
{} & = y^{\lgb{b}{x}}.
\end{eqs}
\end{proof}

\begin{theorem}[Explicit form for a Master function over powers]
\label{ExplicitFormForMasterFunctionOverPowers}
Let $T \in \pmasters{a}{b}{d}{F}$ be a Master function over powers. Then
\begin{equation}
T(n) = n^{\lgb{b}{a}} d + \sum_{i = 0}^{\lgb{b}{n} - 1} a^i F(n / b^i).
\end{equation}
\end{theorem}

\begin{proof}
\proofpart{Pattern}
Expanding the recurrence, we find that
\begin{eqs}
T(1) & = d, \\
T(b) & = ad + F(b), \\
T(b^2) & = a (ad + F(b)) + F(b^2) \\
{} & = a^2 d + a F(b) + F(b^2), \\
T(b^3) & = a (a^2 d + a F(b) + F(b^2)) + F(b^3) \\
{} & = a^3 d + a^2 F(b) + a F(b^2) + F(b^3).
\end{eqs}

\proofpart{Induction}
This suggests the pattern
\begin{equation}
T(b^m) = a^m d + \sum_{i = 0}^{m - 1} a^i F(b^{m - i}),
\end{equation}
where $m \in \TN$. We prove this by induction. For the base case, $T(1) = T(b^0) = d$. For the induction step, if $m > 0$, then
\begin{eqs}
T(b^m) & = a T(b^{m - 1}) + F(b^m) \\
{} & = a \bra{a^{m - 1} d + \sum_{i = 0}^{m - 2} a^i F(b^{m - 1 - i})} + F(b^m) \\
{} & = a^m d + \sum_{i = 0}^{m - 2} a^{i + 1} F(b^{m - 1 - i}) + F(b^m) \\
{} & = a^m d + \sum_{i = 1}^{m - 1} a^i F(b^{m - i}) + F(b^m) \\
{} & = a^m d + \sum_{i = 0}^{m - 1} a^i F(b^{m - i}).
\end{eqs}
If $n = b^m$, then $m = \lgb{b}{n}$, and
\begin{eqs}
T(n) & = a^{\lgb{b}{n}} d + \sum_{i = 0}^{\lgb{b}{n} - 1} a^i F(n / b^i) \\
{} & = n^{\lgb{b}{a}} d + \sum_{i = 0}^{\lgb{b}{n} - 1} a^i F(n / b^i),
\end{eqs}
where the last step is by \proveby{LogarithmSwap}.
\end{proof}

\begin{theorem}[Difference of powers]
\label{DifferenceOfPowers}
Suppose $\alpha, \beta \in \nonn{\TR}$ are such that $\alpha < \beta$. Then
\begin{eqs}
(n^{\beta} - n^{\alpha}) \in \lthetah{B}{n^{\beta} - 1}.
\end{eqs}
\end{theorem}

\begin{proof}
Since $\alpha < \beta$, we have that $(1 - b^{\alpha - \beta}) > 0$. Then
\begin{eqs}
(n^{\beta} - 1) (1 - b^{\alpha - \beta}) & \leq (n^{\beta} - 1) (1 - n^{\alpha - \beta}) \\
{} & \leq n^{\beta} (1 - n^{\alpha - \beta}) \\
{} & = n^{\beta} - n^{\alpha} \\
{} & \leq n^{\beta} - 1.
\end{eqs}
\end{proof}

\begin{lemma}[Master summation lemma]
\label{MasterSummationLemma}
\srequire{LinearDominance}
Let $a \in \nonnb{\TR}{1}$, $b \in \posib{\TR}{1}$, $c \in \nonn{\TR}$, and $B = \setb{b^i : i \in \TN}$. Let $S \in \rc{B}$ be such that
\begin{eqs}
S(n) = n^c \sum_{i = 0}^{\lgb{b}{n} - 1} (a / b^c)^i.
\end{eqs}
Then
\begin{eqs}
\lgb{b}{a} < c & \implies S \in \lthetah{B}{n^c - 1}, \\
\lgb{b}{a} = c & \implies S \in \lthetah{B}{n^c \lgb{b}{n}}, \\
\lgb{b}{a} > c & \implies S \in \lthetah{B}{n^{\lgb{b}{a}} - 1}.
\end{eqs}
\end{lemma}

\begin{proof}
\sproveby{LinearDominanceImpliesEverything}

\proofpart{$\lgb{b}{a} = c$}
This implies $a / b^c = 1$. Then
\begin{eqs}
n^c \sum_{i = 0}^{\lgb{b}{n} - 1} (a / b^c)^i = n^c \lgb{b}{n}.
\end{eqs}
Therefore $S \in \lthetah{B}{n^c \lgb{b}{n}}$.

\proofpart{$\lgb{b}{a} \neq c$}
Suppose $\lgb{b}{a} \neq c$, and let $\gamma = \frac{1}{(a / b^c) - 1}$. Then
\begin{eqs}
n^c \sum_{i = 0}^{\lgb{b}{n} - 1} (a / b^c)^i & = n^c \frac{(a / b^c)^{\lgb{b}{n}} - 1}{(a / b^c) - 1} \\
{} & = \gamma n^c \bra{a^{\lgb{b}{n}} / n^c - 1} \\
{} & = \gamma \bra{a^{\lgb{b}{n}} - n^c} \\
{} & = \gamma \bra{n^{\lgb{b}{a}} - n^c}.
\end{eqs}
Therefore
\begin{eqs}
S(n) = \gamma \bra{n^{\lgb{b}{a}} - n^c}.
\end{eqs}

\proofpart{$\lgb{b}{a} > c$}
This implies $\gamma > 0$. Then $S \in \lthetah{B}{n^{\lgb{b}{a}} - 1}$ by \thref{DifferenceOfPowers} and \require{Scale}.

\proofpart{$\lgb{b}{a} < c$}
This implies $\gamma < 0$. Then $S \in \lthetah{B}{n^c - 1}$ by \thref{DifferenceOfPowers} and \require{Scale}.
\end{proof}

\begin{theorem}[Master theorem over powers]
\label{MasterTheoremOverPowers}
\srequire{LinearDominance}
Let $T \in \pmasters{a}{b}{d}{F}$ be a Master function over powers, and $F \in \loh{B}{n^c}$, where $c \in \nonn{\TR}$. Then
\begin{eqs}
\lgb{b}{a} < c & \implies T \in \loh{B}{n^c}, \\
\lgb{b}{a} = c & \implies T \in \loh{B}{n^c \lgb{b}{bn}}, \\
\lgb{b}{a} > c & \implies T \in \lthetah{B}{n^{\lgb{b}{a}}}.
\end{eqs}
If $F \in \lthetah{B}{n^c}$, then each $\lohs{B}$ can be replaced with $\lthetahs{B}$.
\end{theorem}

\begin{proof}
\sproveby{LinearDominanceImpliesEverything}

\proofpart{Explicit forms}
By \proveby{ExplicitFormForMasterFunctionOverPowers},
\begin{equation}
T(n) = n^{\lgb{b}{a}} d + \sum_{i = 0}^{\lgb{b}{n} - 1} a^i F(n / b^i),
\end{equation}
Let $S \in \rc{B}$ be such that
\begin{eqs}
S(n) = n^c \sum_{i = 0}^{\lgb{b}{n} - 1} (a / b^c)^i.
\end{eqs}
Since $F \in \loh{B}{n^c}$, by \require{SubsetSum}
\begin{eqs}
\loh{B}{\sum_{i = 0}^{\lgb{b}{n} - 1} a^i F(n / b^i)} & \subset \loh{B}{\sum_{i = 0}^{\lgb{b}{n} - 1} a^i (n / b^i)^c} \\
{} & = \loh{B}{n^c \sum_{i = 0}^{\lgb{b}{n} - 1} (a / b^c)^i} \\
{} & = \loh{B}{S}.
\end{eqs}

\proofpart{$\lgb{b}{a} < c$}
By \thref{MasterSummationLemma},
\begin{eqs}
\loh{B}{S} = \loh{B}{n^c - 1}.
\end{eqs}
By \require{Additive}, \require{Scale}, and since $\lgb{b}{a} < c$,
\begin{eqs}
\loh{B}{T} & = \loh{B}{n^{\lgb{b}{a}}} + \loh{B}{\sum_{i = 0}^{\lgb{b}{n} - 1} a^i F(n / b^i)} \\
{} & \subset \loh{B}{n^{\lgb{b}{a}}} + \loh{B}{n^c - 1} \\
{} & = \loh{B}{n^{\lgb{b}{a}} + n^c - 1} \\
{} & = \loh{B}{n^c}.
\end{eqs}

\proofpart{$\lgb{b}{a} = c$}
By \thref{MasterSummationLemma},
\begin{eqs}
\loh{B}{S} = \loh{B}{n^c \lgb{b}{n}}.
\end{eqs}
By \require{Additive} and \require{Scale},
\begin{eqs}
\loh{B}{T} & = \loh{B}{n^c} + \loh{B}{\sum_{i = 0}^{\lgb{b}{n} - 1} a^i F(n / b^i)} \\
{} & \subset \loh{B}{n^c} + \loh{B}{n^c \lgb{b}{n}} \\
{} & = \loh{B}{n^c + n^c \lgb{b}{n}} \\
{} & = \loh{B}{n^c \lgb{b}{bn}}.
\end{eqs}

\proofpart{$\lgb{b}{a} > c$}
By \thref{MasterSummationLemma},
\begin{eqs}
\loh{B}{S} = \loh{B}{n^{\lgb{b}{a}} - 1}.
\end{eqs}
By \require{Additive} and \require{Scale},
\begin{eqs}
\loh{B}{T} & = \loh{B}{n^{\lgb{b}{a}}} + \loh{B}{\sum_{i = 0}^{\lgb{b}{n} - 1} a^i F(n / b^i)} \\
{} & \subset \loh{B}{n^{\lgb{b}{a}}}  + \loh{B}{n^{\lgb{b}{a}} - 1} \\
{} & = \loh{B}{2n^{\lgb{b}{a}} - 1} \\
{} & = \loh{B}{n^{\lgb{b}{a}}}.
\end{eqs}
By \require{Order}, $\loh{B}{n^{\lgb{b}{a}}} \subset \loh{B}{T}$. Therefore $\loh{B}{T} = \loh{B}{n^{\lgb{b}{a}}}$.

\proofpart{$F \in \lthetah{B}{n^c}$}
Suppose we know that $F \in \lthetah{B}{n^c}$. Then by \property{SubsetSum}
\begin{eqs}
\loh{B}{\sum_{i = 0}^{\lgb{b}{n} - 1} a^i F(n / b^i)} = \loh{B}{S},
\end{eqs}
and we get rid of subsets in the above proofs.
\end{proof}

\section{Master theorem over reals}

\ReMasterFunctionOverReals

\begin{theorem}[Explicit form for a Master function over reals]
\label{ExplicitFormForMasterFunctionOverReals}
Let $t \in \rmasters{a}{b}{d}{f}$ be a Master function over reals. Then
\begin{equation}
t(x) = a^{\floorb{\lgb{b}{x}}} d + \sum_{i = 0}^{\floorb{\lgb{b}{x}} - 1} a^i f(x / b^i).
\end{equation}
\end{theorem}

\begin{proof}
\proofpart{Pattern}
Suppose $b^0 \leq x < b^1$. Then
\begin{equation}
t(x) = d.
\end{equation}
Suppose $b^1 \leq x < b^2$. Then
\begin{eqs}
t(x) & = ad + f(x).
\end{eqs}
Suppose $b^2 \leq x < b^3$. Then
\begin{eqs}
t(x) & = a \bra{ad + f(x / b)} + f(x) \\
{} & = a^2 d + a f(x / b) + f(x).
\end{eqs}
Suppose $b^3 \leq x < b^4$. Then
\begin{eqs}
t(x) & = a \bra{a^2 d + a f(x / b^2) + f(x / b)} + f(x) \\
{} & = a^3 d + a^2 f(x / b^2) + a f(x / b) + f(x).
\end{eqs}

\proofpart{Induction}
This suggests the pattern
\begin{equation}
t(x) = a^m d + \sum_{i = 0}^{m - 1} a^i f(x / b^i),
\end{equation}
where $m \in \TN$ is such that $b^m \leq x < b^{m + 1}$, which is equivalent to $m = \floorb{\lgb{b}{x}}$. That is,
\begin{equation}
t(x) = a^{\floorb{\lgb{b}{x}}} d + \sum_{i = 0}^{\floorb{\lgb{b}{x}} - 1} a^i f(x / b^i).
\end{equation}
This can be proved by induction, as with the analogous theorem for powers. 
\end{proof}

\begin{theorem}[Master reduction from reals to powers]
\label{MasterReductionFromRealsToPowers}
\srequire{LinearDominance}
Suppose $T \in \pmasters{a}{b}{d}{F}$ is a Master function over powers, and $t \in \rmasters{a}{b}{d}{f}$ is a Master function over reals. Then
\begin{eqs}
f \in \loh{\nonnb{\TR}{1}}{F(b^{\floorb{\lgb{b}{x}}})} & \implies t \in \loh{\nonnb{\TR}{1}}{T(b^{\floorb{\lgb{b}{x}}})}, \\
f \in \lomegah{\nonnb{\TR}{1}}{F(b^{\floorb{\lgb{b}{x}}})} & \implies t \in \lomegah{\nonnb{\TR}{1}}{T(b^{\floorb{\lgb{b}{x}}})}.
\end{eqs}
\end{theorem}

\begin{proof}
\sproveby{LinearDominanceImpliesEverything}
\proofpart{Explicit forms}
We have that
\begin{eqs}
T(b^{\floorb{\lgb{b}{x}}}) & = a^{\floorb{\lgb{b}{x}}} d + \sum_{i = 0}^{\floorb{\lgb{b}{x}} - 1} a^i F(b^{\floorb{\lgb{b}{x}}} / b^i), \\
t(x) & = a^{\floorb{\lgb{b}{x}}} d + \sum_{i = 0}^{\floorb{\lgb{b}{x}} - 1} a^i f(x / b^i),
\end{eqs}
by \proveby{ExplicitFormForMasterFunctionOverPowers} and \proveby{ExplicitFormForMasterFunctionOverReals}. 

\proofpart{$\lohsy$-sets}
Note that
\begin{eqs}
b^{\floorb{\lgb{b}{x}}} / b^i = b^{\floorb{\lgb{b}{x / b^i}}}.
\end{eqs}
By \require{Additive} and \require{SubsetSum},
\begin{eqs}
\loh{\nonnb{\TR}{1}}{t} & = \loh{\nonnb{\TR}{1}}{a^{\floorb{\lgb{b}{x}}} d} + \loh{\nonnb{\TR}{1}}{\sum_{i = 0}^{\floorb{\lgb{b}{x}} - 1} a^i f(x / b^i)} \\
{} & \subset \loh{\nonnb{\TR}{1}}{a^{\floorb{\lgb{b}{x}}} d} + \loh{\nonnb{\TR}{1}}{\sum_{i = 0}^{\floorb{\lgb{b}{x}} - 1} a^i F(b^{\floorb{\lgb{b}{x / b^i}}})} \\
{} & = \loh{\nonnb{\TR}{1}}{T(b^{\floorb{\lgb{b}{x}}})}.
\end{eqs}
The proof for $\lomegahsy$-sets is similar.
\end{proof}

\begin{lemma}[Identity equivalent]
\label{IdentityEquivalent}
\begin{eqs}
x \in \lthetah{\nonnb{\TR}{1}}{b^{\floorb{\lgb{b}{x}}}}.
\end{eqs}
\end{lemma}

\begin{proof}
\begin{eqs}
b^{\floorb{\lgb{b}{x}}} & \leq b^{\lgb{b}{x}} \\
{} & = x \\
{} & \leq b^{\floorb{\lgb{b}{x}}} b,
\end{eqs}
for all $x \in \nonnb{\TR}{1}$.
\end{proof}

\ReMasterTheoremOverReals

\begin{proof}
\sproveby{LinearDominanceImpliesEverything}

\proofpart{Reduction to powers}
Let $T \in \pmasters{a}{b}{d}{F}$ be a Master function over powers, where $F(n) = n^c$. By \thref{IdentityEquivalent},
\begin{eqs}
\loh{\nonnb{\TR}{1}}{x} = \loh{\nonnb{\TR}{1}}{b^{\floorb{\lgb{b}{x}}}}.
\end{eqs}
By \property{PowerH},
\begin{eqs}
\loh{\nonnb{\TR}{1}}{x^c} & = \loh{\nonnb{\TR}{1}}{b^{\floorb{\lgb{b}{x}}c}} \\
{} & = \loh{\nonnb{\TR}{1}}{F(b^{\floorb{\lgb{b}{x}}})}.
\end{eqs}
Since $f \in \loh{\nonnb{\TR}{1}}{x^c}$ by assumption,
\begin{eqs}
f \in \loh{\nonnb{\TR}{1}}{F(b^{\floorb{\lgb{b}{x}}})}.
\end{eqs}
By \proveby{MasterReductionFromRealsToPowers},
\begin{eqs}
\loh{\nonnb{\TR}{1}}{t} \subset \loh{\nonnb{\TR}{1}}{T(b^{\floorb{\lgb{b}{x}}})}.
\end{eqs}

\proofpart{$\lgb{b}{a} < c$}
By \thref{MasterTheoremOverPowers},
\begin{eqs}
\loh{B}{T} \subset \loh{B}{n^c}.
\end{eqs}
By \require{SubsetSum},
\begin{eqs}
\loh{\nonnb{\TR}{1}}{T(b^{\floorb{\lgb{b}{x}}})} & \subset \loh{\nonnb{\TR}{1}}{\bra{b^{\floorb{\lgb{b}{x}}}}^c} \\ 
{} & = \loh{\nonnb{\TR}{1}}{x^c}.
\end{eqs}

\proofpart{$\lgb{b}{a} = c$}
By \thref{MasterTheoremOverPowers},
\begin{eqs}
\loh{B}{T} \subset \loh{B}{n^c \lgb{b}{bn}}.
\end{eqs}
By \require{SubsetSum},
\begin{eqs}
\loh{\nonnb{\TR}{1}}{T(b^{\floorb{\lgb{b}{x}}})} & \subset \loh{\nonnb{\TR}{1}}{\bra{b^{\floorb{\lgb{b}{x}}}}^c \lgb{b}{b b^{\floorb{\lgb{b}{x}}}}} \\
{} & = \loh{\nonnb{\TR}{1}}{x^c \lgb{b}{bx}}.
\end{eqs}

\proofpart{$\lgb{b}{a} > c$}
By \thref{MasterTheoremOverPowers},
\begin{eqs}
\loh{B}{T} = \loh{B}{n^{\lgb{b}{a}}}.
\end{eqs}
By \require{SubsetSum},
\begin{eqs}
\loh{\nonnb{\TR}{1}}{t} & = \loh{\nonnb{\TR}{1}}{\bra{b^{\floorb{\lgb{b}{x}}}}^{\lgb{b}{a}}} \\ 
{} & = \loh{\nonnb{\TR}{1}}{x^{\lgb{b}{a}}}.
\end{eqs}

\proofpart{$f \in \lthetah{\nonnb{\TR}{1}}{n^c}$}
If $f \in \lthetah{\nonnb{\TR}{1}}{n^c}$, then the subsets in the above proofs can be replaced with equalities by \thref{MasterTheoremOverPowers}.
\end{proof}

\section{Master theorem over integers}

\ReMasterFunctionOverIntegers

\begin{note}[Stricter requirement]
The requirement $b \geq 2$ is stricter than the requirement $b > 1$ for the other Master theorems; this is needed to avoid the recursion getting stuck to a fixed-point $\geq b$. 
\end{note}

\begin{definition}[Ceiling division]
The \define{ceiling division} is a function  $\function{N}{\nonnb{\TN}{1}}{\nonnb{\TN}{1}}$ such that $N(n) = \ceilb{n / b}$. 
\end{definition}

\begin{definition}[Ceiling division number]
The \define{ceiling division number} is a function $\function{M}{\nonnb{\TN}{1}}{\TN}$ such that $M(n) = \min\setb{i \in \TN : N^{(i)}(n) < b}$.
\end{definition}

\begin{theorem}[Fixed points of the ceiling division]
\label{FixedPointsOfCeilingDivision}
Let $b \in \posib{\TR}{1}$ and
\begin{equation}
F = \setb{n \in \nonnb{\TN}{1} : \ceilb{n / b} = n}.
\end{equation}
Then
\begin{eqs}
F = \setb{n \in \nonnb{\TN}{1} : 1 \leq n < b / (b - 1)}.
\end{eqs}
In addition, $F = \setb{1} \iff b \geq 2$.
\end{theorem}

\begin{proof}
To characterize the fixed points,
\begin{eqs}
{} & n \in F \\
\iffr & \ceilb{n / b} = n \\
\iffr & \exists v \in \TR : -b < v \leq 0 \textrm{ and } nb + v = n \\
\iffr & \exists v \in \TR : -b < v \leq 0 \textrm{ and } v = n(1 - b) \\
\iffr & -b < n(1 - b) \leq 0 \\
\iffr & b / (b - 1) > n \geq 0 \\
\iffr & 0 \leq n < b / (b - 1) \\
\iffr & 1 \leq n < b / (b - 1),
\end{eqs}
where we use a version of Euclidean division and the last step follows because $n \in \nonnb{\TN}{1}$.
Based on this characterization,
\begin{eqs}
{} & F = \setb{1} \\
\iffr & b / (b - 1) \leq 2 \\
\iffr & b \leq 2 (b - 1) \\
\iffr & 0 \leq b - 2 \\
\iffr & b \geq 2.
\end{eqs}
Let us also note that
\begin{eqs}
{} & b / (b - 1) \leq b \\
\iffr & b \leq b (b - 1) \\
\iffr & 1 \leq b - 1 \\
\iffr & b \geq 2.
\end{eqs}
\end{proof}

\begin{theorem}[Explicit form for a Master function over integers]
\label{ExplicitFormForMasterFunctionOverIntegers}
Let $T \in \imasters{a}{b}{d}{F}$ be a Master function over integers. Then
\begin{equation}
T(n) = a^{M(n)} d + \sum_{i = 0}^{M(n) - 1} a^i F(N^{(i)}(n)),
\end{equation}
\end{theorem}

\begin{proof}
The function $M$ is well-defined by \proveby{FixedPointsOfCeilingDivision}, since $b \geq 2$. Let $n \in \nonnb{\TN}{1}$. 

\proofpart{Pattern}
Suppose $M(n) = 0$. Then
\begin{equation}
T(n) = d.
\end{equation}
Suppose $M(n) = 1$. Then
\begin{eqs}
T(n) & = a T(\ceilb{n / b}) + F(n) \\
{} & = a T(N^{(1)}(n)) + F(n) \\
{} & = a d + F(n).
\end{eqs}
Suppose $M(n) = 2$. Then
\begin{eqs}
T(n) & = a T(N^{(1)}(n)) + F(n) \\
{} & = a \bra{a T(\ceilb{N^{(1)}(n) / b}) + F(N^{(1)}(n))} + F(n) \\
{} & = a^2 T(N^{(2)}(n)) + a F(N^{(1)}(n)) + F(n) \\
{} & = a^2 d + a F(N^{(1)}(n)) + F(n).
\end{eqs}
Suppose $M(n) = 3$. Then
\begin{eqs}
T(n) & = a^2 T(N^{(2)}(n)) + a F(N^{(1)}(n)) + F(n) \\
{} & = a^2 \bra{a T(\ceilb{N^{(2)}(n) / b}) + F(N^{(2)}(n))} + a F(N^{(1)}(n)) + F(n) \\
{} & = a^3 T(N^{(3)}(n)) + a^2 F(N^{(2)}(n)) + a F(N^{(1)}(n)) + F(n) \\
{} & = a^3 d + a^2 F(N^{(2)}(n)) + a F(N^{(1)}(n)) + F(n).
\end{eqs}

\proofpart{Induction}
This suggests the pattern
\begin{equation}
T(n) = a^{M(n)} d + \sum_{i = 0}^{M(n) - 1} a^i F(N^{(i)}(n)).
\end{equation}
This can be proved by induction, as in the explicit form for powers.
\end{proof}

\begin{theorem}[Bounds for $N$]
\label{BoundsForN}
Let $n \in \nonnb{\TN}{1}$, and $b \in \nonnb{\TR}{2}$. Then
\begin{eqs}
n / b^i & \leq N^{(i)}(n) < n / b^i + 2.
\end{eqs}
\end{theorem}

\begin{proof}
\proofpart{Pattern}
It holds that $\ceilb{x} < x + 1$ for all $x \in \TR$. Therefore,
\begin{eqs}
N^{(1)}(n) & = \ceilb{n / b} \\
{} & < (n / b) + 1 \\
N^{(2)}(n) & = \ceilb{N^{(1)}(n) / b} \\
{} & \leq \ceilb{((n / b) + 1) / b} \\
{} & = \ceilb{n / b^2 + 1 / b} \\
{} & < n / b^2 + 1 / b + 1 \\
N^{(3)}(n) & = \ceilb{N^{(2)}(n) / b} \\
{} & \leq \ceilb{n / b^3 + 1 / b^2 + 1 / b} \\
{} & < n / b^3 + 1 / b^2 + 1 / b + 1.
\end{eqs}

\proofpart{Induction}
This suggests the pattern
\begin{equation}
N^{(i)}(n) < n / b^i + \sum_{j = 0}^{i - 1} (1 / b)^j.
\end{equation}
This can be proved by induction. Since $b > 1$,
\begin{eqs}
N^{(i)}(n) & < n / b^i + \frac{1 - (1 / b)^i}{1 - (1 / b)} \\
{} & < n / b^i + \frac{1}{1 - (1 / b)}.
\end{eqs}
Since $b \geq 2$, 
\begin{equation}
N^{(i)}(n) < n / b^i + 2.
\end{equation}
Similarly, since $x \leq \ceilb{x}$, for all $x \in \TR$, it can be proved that $n / b^i \leq N^{(i)}(n)$.
\end{proof}

\begin{note}[Ceiling division pitfall]
It holds that $\ceilb{\ceilb{n / a} / b} = \ceilb{n / (ab)}$, for all $a, b \in \posi{\TN}$ and $n \in \TN$. A potential pitfall in the proof of \proveby{BoundsForN} is to assume that this also holds when $a, b \in \posi{\TR}$. A counterexample is given by $n = 6$ and $a = b = 5/2$.
\end{note}

\begin{theorem}[Bounds for $M$]
\label{BoundsForM}
Let $n \in \nonnb{\TN}{1}$, and $b \in \nonnb{\TR}{2}$. Then
\begin{eqs}
\floorb{\lgb{b}{n}} & \leq M(n) \leq \floorb{\lgb{b}{n}} + 2.
\end{eqs}
\end{theorem}

\begin{proof}
By \thref{BoundsForN},
\begin{eqs}
N^{(\floorb{\lgb{b}{n}} + 1)}(n) & < \frac{n}{b^{\floorb{\lgb{b}{n}} + 1}} + 2 \\
{} & < \frac{n}{b^{\lgb{b}{n}}} + 2 \\
{} & = \frac{n}{n} + 2 \\
{} & = 3.
\end{eqs}
This is equivalent to $N^{(\floorb{\lgb{b}{n}} + 1)}(n) \in \setb{1, 2}$. Since $N^{(1)}(2) = \ceilb{2 / b} = 1$, it follows that $M(n) \leq \floorb{\lgb{b}{n}} + 2$. Similarly, it follows that $\floorb{\lgb{b}{n}} \leq M(n)$.
\end{proof}

\begin{theorem}[Multiplicative bounds for $N^{(i)}(n)$]
\label{MultiplicativeBoundsForN}
Let $n \in \nonnb{\TN}{1}$, $b \in \nonnb{\TR}{2}$, and $i \in [0, \floorb{\lgb{b}{n}} + 1]$. Then
\begin{equation}
n / b^i \leq N^{(i)}(n) < 3 b (n / b^i).
\end{equation}
\end{theorem}

\begin{proof}
We would like to find $\beta \in \posi{\TR}$ such that
\begin{eqs}
{} & n / b^i + 2 \leq \beta (n / b^i) \\
\iff \; & \beta \geq 1 + 2b^i / n,
\end{eqs}
For the given argument-sets,
\begin{eqs}
1 + 2b^i / n & \leq 1 + 2b^{\floorb{\lgb{b}{n}} + 1} / n \\
{} & \leq 1 + 2b^{\lgb{b}{n} + 1} / n \\
{} & = 1 + 2b \\
{} & \leq b + 2b \\
{} & = 3b.
\end{eqs}
Therefore $\beta = 3b$ suffices. Then
\begin{eqs}
n / b^i & \leq N^{(i)}(n) \\
{} & < n / b^i + 2 \\
{} & \leq 3b (n / b^i),
\end{eqs}
by \proveby{BoundsForN}.
\end{proof}

\ReMasterTheoremOverIntegers

\begin{proof}
\sproveby{LinearDominanceImpliesEverything}

\proofpart{Explicit form}
By \thref{ExplicitFormForMasterFunctionOverIntegers},
\begin{equation}
T(n) = a^{M(n)} d + \sum_{i = 0}^{M(n) - 1} a^i F(N^{(i)}(n)),
\end{equation}

\proofpart{First term}
By \proveby{BoundsForM},
\begin{eqs}
{} & \floorb{\lgb{b}{n}} \leq M(n) \leq \floorb{\lgb{b}{n}} + 2 \\
\impliesr & a^{\floorb{\lgb{b}{n}}} d \leq a^{M(n)} d \leq a^{\floorb{\lgb{b}{n}} + 2} d \\
\impliesr & a^{\floorb{\lgb{b}{n}}} d \leq a^{M(n)} d \leq a^{\floorb{\lgb{b}{n}}} a^2 d \\
\end{eqs}
By \require{Order} and \require{Scale},
\begin{eqs}
\loh{\nonnb{\TN}{1}}{a^{M(n)} d} = \loh{\nonnb{\TN}{1}}{a^{\floorb{\lgb{b}{n}}}}.
\end{eqs}

\proofpart{Sum term}
Since $F \in \loh{\nonnb{\TN}{1}}{n^c}$, by \require{SubsetSum}
\begin{eqs}
{} & \loh{\nonnb{\TN}{1}}{\sum_{i = 0}^{M(n) - 1} a^i F(N^{(i)}(n))} \\
\air{\subset} & \loh{\nonnb{\TN}{1}}{\sum_{i = 0}^{M(n) - 1} a^i (N^{(i)}(n))^c}.
\end{eqs}
By \proveby{MultiplicativeBoundsForN},
\begin{eqs}
(n / b^i) \leq N^{(i)}(n) < 3b (n / b^i),
\end{eqs}
for all $i \in [0, \floorb{\lgb{b}{n}} + 1]$. By \require{Order} and \require{Scale},
\begin{eqs}
{} & \loh{\nonnb{\TN}{1}}{\sum_{i = 0}^{M(n) - 1} a^i (N^{(i)}(n))^c} \\
\air{=} & \loh{\nonnb{\TN}{1}}{\sum_{i = 0}^{M(n) - 1} a^i (n/b^i)^c}.
\end{eqs}

\proofpart{Combined terms}
We combine the first term and the sum term by \require{Additive}:
\begin{eqs}
\loh{\nonnb{\TN}{1}}{T} \subset \loh{\nonnb{\TN}{1}}{a^{\floorb{\lgb{b}{n}}} + \sum_{i = 0}^{M(n) - 1} a^i (n/b^i)^c}.
\end{eqs}
By \proveby{BoundsForM},
\begin{eqs}
\floorb{\lgb{b}{n}} \leq M(n) \leq \floorb{\lgb{b}{n}} + 2.
\end{eqs}
Therefore, by \require{Order},
\begin{eqs}
{} & \loh{\nonnb{\TN}{1}}{a^{\floorb{\lgb{b}{n}}} + \sum_{i = 0}^{\floorb{\lgb{b}{n}} - 1} a^i (n/b^i)^c} \\
\air{\subset} & \loh{\nonnb{\TN}{1}}{a^{\floorb{\lgb{b}{n}}} + \sum_{i = 0}^{M(n) - 1} a^i (n/b^i)^c} \\
\air{\subset} & \loh{\nonnb{\TN}{1}}{a^{\floorb{\lgb{b}{n}}} + \sum_{i = 0}^{\floorb{\lgb{b}{n}} + 1} a^i (n/b^i)^c}.
\end{eqs}
The last two terms in the sum are given by
\begin{eqs}
{} & a^{\floorb{\lgb{b}{n}}} (n / b^{\floorb{\lgb{b}{n}}})^c + a^{\floorb{\lgb{b}{n}} + 1} (n / b^{\floorb{\lgb{b}{n}} + 1})^c \\
\air{=} & a^{\floorb{\lgb{b}{n}}} (n / b^{\floorb{\lgb{b}{n}}})^c (1 + a / b^c) \\
\air{\in} & \lthetah{\nonnb{\TN}{1}}{a^{\floorb{\lgb{b}{n}}}}.
\end{eqs}
Therefore,
\begin{eqs}
{} & \loh{\nonnb{\TN}{1}}{a^{\floorb{\lgb{b}{n}}} + \sum_{i = 0}^{M(n) - 1} a^i (n/b^i)^c} \\
\air{=} & \loh{\nonnb{\TN}{1}}{a^{\floorb{\lgb{b}{n}}} + \sum_{i = 0}^{\floorb{\lgb{b}{n}} - 1} a^i (n/b^i)^c}.
\end{eqs}

\proofpart{Reduction to reals}
Let $t \in \rmasters{a}{b}{d}{f}$ be a Master function over reals, where $f \in \rc{\posib{\TR}{1}}$ is such that $f(x) = x^c$. By \thref{ExplicitFormForMasterFunctionOverReals},
\begin{eqs}
t(x) = a^{\floorb{\lgb{b}{x}}} + \sum_{i = 0}^{\floorb{\lgb{b}{x}} - 1} a^i (x/b^i)^c.
\end{eqs}
Therefore,
\begin{eqs}
\loh{\nonnb{\TN}{1}}{T} \subset \loh{\nonnb{\TN}{1}}{\restr{t}{\nonnb{\TN}{1}}}.
\end{eqs}
The results follow from \thref{MasterTheoremOverReals}. 

\proofpart{$\thetahsy$-sets}
If $F \in \thetah{\nonnb{\TN}{1}}{n^c}$, then we can get rid of the subsets in the above proof.
\end{proof}

\chapter{Comparison of definitions}
\label{CandidateDefinitions}

In this section we will study various candidate definitions for the $\ohsy$-notation. This highlights the ways in which some familiar candidate definitions fail the primitive properties. The properties of the candidate definitions, proved in this section, are summarized upfront in \taref{tabcomparison}. 

\begin{table}
\small
\begin{tabular}{|l|l|l|l|l|l|l|}
\hline 
Property &
$\lohsy$ &
$\fohsy$ &
$\cohsy$ &
$\pohsy$ &
$\aohsy$
\\
\hline 
\hline 
\textbf{\uproperty{Order}} & 
\multicolumn{4}{l|}{\checkmark \ref{LocalOrderConsistency}} &
\checkmark \ref{AffineOrderConsistency}
\\
\hline 
\uproperty{Reflex} &
\multicolumn{5}{l|}{\checkmark \ref{ReflexivityIsImplied}}
\\
\hline 
\textbf{\uproperty{Trans}} &
\multicolumn{4}{l|}{\checkmark \ref{LocalTransitivity}} &
\checkmark \ref{AffineTransitivity}
\\
\hline 
\uproperty{Orderness} &
\multicolumn{5}{l|}{\checkmark \ref{OrdernessIsImplied}} \\
\hline 
\hline 
\uproperty{Zero} &
\multicolumn{4}{l|}{\checkmark \ref{ZeroSeparationIsImplied}} & 
\xmark \ref{AffineZeroSeparationFails}
\\
\hline 
\uproperty{TrivialZero} &
\checkmark \ref{LocalZeroTrivialityCharacterized} & 
\xmark \ref{LocalZeroTrivialityCharacterized} &
\xmark \ref{LocalZeroTrivialityCharacterized} &
\xmark \ref{LocalZeroTrivialityCharacterized} &
\xmark \ref{ZeroTrivialityFailsForAffineDominance}
\\
\hline 
\textbf{\uproperty{One}} &
\checkmark \ref{LinearOneSeparation} &
\checkmark \ref{CofiniteOneSeparation} &
\checkmark \ref{CoasymptoticOneSeparation} &
\checkmark \ref{AsymptoticOneSeparation} &
\checkmark \ref{AffineOneSeparation}
\\
\hline 
\hline 
\textbf{\uproperty{Scale}} &
\multicolumn{4}{l|}{\checkmark \ref{LocalScaleInvariance}} &
\checkmark \ref{AffinePositiveScaleInvariance}
\\
\hline 
\uproperty{ScalarHom} &
\multicolumn{5}{l|}{\checkmark \ref{ScalarHomogenuityIsImplied}} \\
\hline 
\textbf{\uproperty{NSubHom}} &
\multicolumn{4}{l|}{\checkmark \ref{LocalSubHomogenuity}} & \xmark \ref{AffineSubHomogenuityFails} \\
\hline 
\textbf{\usproperty{NSubDiv}} &
\multicolumn{4}{l|}{\checkmark \ref{LocalSubHomogenuity}} & \checkmark \ref{AffineSubHomogeneityNDiv} \\
\hline 
\uproperty{QSubHom} &
\multicolumn{4}{l|}{\checkmark \ref{QSubhomogenuityIsComposite}} & \xmark \ref{AffineSubHomogenuityFails} \\
\hline 
\uproperty{SubHom} &
\multicolumn{4}{l|}{\checkmark \ref{RSubhomogenuityIsImplied}} & \xmark \ref{AffineSubHomogenuityFails} \\
\hline 
\uproperty{SuperHom} &
\checkmark \ref{LocalSuperHomogenuityCharacterized} & \xmark \ref{LocalSuperHomogenuityCharacterized} & \xmark \ref{LocalSuperHomogenuityCharacterized} & \xmark \ref{LocalSuperHomogenuityCharacterized} & \xmark \ref{AffineSuperHomogenuityFails} \\
\hline 
\uproperty{PowerH} &
\multicolumn{4}{l|}{\checkmark \ref{LocalPowerHomogenuity}} &
\checkmark \ref{AffinePowerHomogenuity}
\\
\hline 
\uproperty{AddCons} &
\multicolumn{5}{l|}{\checkmark \ref{AdditiveConsistencyIsImplied}} \\
\hline 
\usproperty{MultiCons} &
\multicolumn{5}{l|}{\checkmark \ref{MultiplicativeConsistencyIsImplied}} \\
\hline 
\uproperty{MaxCons} &
\multicolumn{5}{l|}{\checkmark \ref{MaximumConsistencyIsImplied}} \\
\hline 
\hline 
\textbf{\uproperty{Local}} &
\multicolumn{4}{l|}{\checkmark \ref{LocalLocality}} & 
\checkmark \ref{AffineLocality}
\\
\hline 
\uproperty{SubMulti} &
\multicolumn{4}{l|}{\checkmark \ref{SubMultiplicativityIsImplied}} & 
\xmark \ref{AffineSubHomogenuityFails}
\\
\hline 
\uproperty{SuperMulti} &
\multicolumn{4}{l|}{\checkmark \ref{LocalSuperMultiplicativity}} & 
\checkmark \ref{AffineSuperMultiplicativity}
\\
\hline 
\uproperty{SubRestrict} &
\multicolumn{5}{l|}{\checkmark \ref{SubRestrictabilityIsImplied}} \\
\hline 
\uproperty{SuperRestrict} &
\multicolumn{5}{l|}{\checkmark \ref{SuperRestrictabilityIsImplied}} \\
\hline 
\uproperty{Maximum} &
\multicolumn{5}{l|}{\checkmark \ref{MaximumIsImplied}} \\
\hline 
\uproperty{Summation} &
\multicolumn{5}{l|}{\checkmark \ref{SummationIsImplied}} \\
\hline 
\uproperty{MaximumSum} &
\multicolumn{5}{l|}{\checkmark \ref{MaximumSumIsImplied}} \\
\hline 
\uproperty{Additive} &
\multicolumn{5}{l|}{\checkmark \ref{AdditivityIsImplied}} \\
\hline 
\usproperty{Translation} &
\multicolumn{5}{l|}{\checkmark \ref{TranslationInvarianceIsImplied}} \\
\hline 
\textbf{\uproperty{SubComp}} &
\checkmark \ref{LinearSubComposability} &
\xmark \ref{CofiniteSubComposabilityFails} &
\xmark \ref{CoasymptoticInjectiveSubComposabilityFails} &
\xmark \ref{AsymptoticInjectiveSubComposabilityFails} &
\checkmark \ref{AffineSubComposability}
\\
\hline 
\usproperty{ISubComp} &
\checkmark \ref{LinearSubComposability}  &
\checkmark \ref{CofiniteInjectiveComposition} &
\xmark \ref{CoasymptoticInjectiveSubComposabilityFails} &
\xmark \ref{AsymptoticInjectiveSubComposabilityFails} &
\checkmark \ref{AffineSubComposability}
\\
\hline 
\usproperty{ISuperComp} &
\checkmark \ref{InjectiveSuperComposabilityIsImplied} &
\checkmark \ref{InjectiveSuperComposabilityIsImplied} &
? &
? &
\checkmark \ref{InjectiveSuperComposabilityIsImplied}
\\
\hline 
\uproperty{Extend} &
\checkmark \ref{LinearExtensibility} &
\xmark \ref{CofiniteExtensibilityFails} &
\xmark \ref{CoasymptoticExtensibilityFails} &
\checkmark \ref{AsymptoticExtensibility} &
\checkmark \ref{AffineExtensibility}
\\
\hline 
\uproperty{SubsetSum} &
\checkmark \ref{SubsetSumIsAnOMapping} &
\xmark \ref{SubsetSumImpliesSubComposability} &
\xmark \ref{SubsetSumImpliesSubComposability} &
\xmark \ref{SubsetSumImpliesSubComposability} &
?
\\
\hline 
\hline 
Universe &
sets &
sets &
$U$ &
$U$ &
sets
\\
\hline 
\end{tabular}
\centering
\caption{Comparison of $\ohsy$-notations. The \checkmark and \xmark mean that we have proved and disproved, respectively, the property. The following number refers to the corresponding theorem. Primitive properties are marked with a bold face, and $U = \bigcup_{d \in \TN} \power{\TR^d}$.}
\label{tabcomparison}
\end{table}

\section{Trivial linear dominance}

\begin{definition}[Trivial linear dominance]
\defineexp{Trivial linear dominance}{linear dominance!trivial} $\trohsy$ is defined by
\begin{eqs}
\trohx{f} = \rc{X},
\end{eqs}
for all $f \in \rc{X}$, and all $X \in U$, where $U$ is the class of all sets.
\end{definition}

\begin{theorem}[Trivial linear dominance is local linear dominance]
\label{TrivialLinearDominanceIsLocalLinearDominance}
The family $\filtersetsym = \setb{\emptyset : X \in U}$ is a family of filter bases with induced sub-structure.
\end{theorem}

\begin{proof}
\proofpart{Directedness}
Let $A, B \in \filterset{X}$, where $X \in U$. Then $A = \emptyset = B$, and we may choose $C \coloneqq \emptyset \in \filterset{X}$, so that $C \subset A \cap B$. 

\proofpart{Induced sub-structure}
Let $D \subset X$, where $X \in U$. Then $\filterset{D} = \setb{\emptyset} = \setb{A \cap D : A \in \filterset{X}}$.  	
\end{proof}

\begin{theorem}[\uproperty{One} fails for $\trohsy$]
\label{TrivialOneSeparationFails}
$\trohsy$ does not have \property{One}.
\end{theorem}

\begin{proof}
It holds that $n \in \rc{\posi{\TN}} = \troh{\posi{\TN}}{1}$.
\end{proof}

\begin{theorem}[\uproperty{SubComp} for $\trohsy$]
\label{TrivialSubComposability}
$\trohsy$ has \property{SubComp}.
\end{theorem}

\begin{proof}
The right side of the implication
\begin{eqs}
(f \circ s) \in \oh{Y}{g \circ s}
\end{eqs}
holds by the definition of trivial dominance for all $f, g \in \rc{X}$ and $\function{s}{Y}{X}$.
\end{proof}

\begin{note}[Trivial linear dominance has lots of nice properties?]
Trivial linear dominance satisfies all of the desirable properties, except those of non\-/triviality. This underlines the importance of the non\-/triviality properties.
\end{note}

\begin{note}[Equivalent definitions]
Suppose local linear dominance is defined so that $\emptyset \in \filterset{X}$, for all $X \in U$. Then it is equivalent to trivial linear dominance.
\end{note}

\section{Asymptotic linear dominance}
\label{AsymptoticLinearDominance}

Recall the definition of asymptotic linear dominance from \sref{Introduction}:
\defineasymptotic

\begin{theorem}[Asymptotic linear dominance is local linear dominance]
\label{AsymptoticLinearDominanceIsLocalLinearDominance}
The family $\filtersetsym = \setb{\setb{\nonnb{X}{y} : y \in \TR^d} : X \in U}$ is a family of filter bases with induced sub-structure.
\end{theorem}

\begin{proof}
\proofpart{Directedness}
Let $X \subset \TR^d$, and $A, B \in \filterset{X}$. Then there exist $y_1, y_2 \in \TR^d$, such that $A = \nonnb{X}{y_1}$, and $B = \nonnb{X}{y_2}$. Let $y \coloneqq \sup_{\leq}\setb{y_1, y_2}$, and $C \coloneqq \nonnb{X}{y}$. Then $C \in \filterset{X}$, and $C \subset A \cap B$. 

\proofpart{Induced sub-structure}
Let $D \subset X \subset \TR^d$. Then
\begin{eqs}
\filterset{D} & = \setb{\nonnb{D}{y} : y \in \TR^d} \\
{} & = \setb{\nonnb{X}{y} \cap D : y \in \TR^d}.
\end{eqs}
\end{proof}

\begin{note}[Intrinsic pitfall]
\label{AsymptoticIntrinsicPitfall}
Let us consider a variant of asymptotic linear dominance, where $\filterset{X} = \setb{\nonnb{X}{y} : y \in X}$, when $X \subset \TR^d$. We call this an \emph{intrinsic} definition --- in contrast to the extrinsic ($y \in \TR^d$) definition we have given. Then $\filterset{X}$ is not a filter basis, and the definition fails many properties.

An example is given by $X = (-\infty, 1) \times \TR \cup [1, \infty) \times (-\infty, 1)$, $y_1 = (0, 1) \in X$, and $y_2 = (1, 0) \in X$. Then $\nonnb{X}{y_1} \cap \nonnb{X}{y_2} = \emptyset$, but there is no $y_3 \in X$ such that $\nonnb{X}{y_3} = \emptyset$.
\end{note}

\begin{theorem}[\uproperty{One} for $\pohsy$]
\label{AsymptoticOneSeparation}
$n \not\in \poh{\posi{\TN}}{1}$.
\end{theorem}

\begin{proof}
$\pohsy$ has \property{One} if and only if $\bra{\forall A \in \filterset{\posi{\TN}} : |A| \not\in \TN}$ by \proveby{LocalOneSeparationCharacterized}. Let $X \coloneqq \posi{\TN}$ and $A \in \filterset{X}$. Then there exists $y \in \TN$ such that $A = \nonnb{X}{y}$. Let $\function{f}{A}{\TN}$ be such that $f(x) = x - y$. Let $x_1, x_2 \in A$. Then
\begin{eqs}
{} \quad & f(x_1) = f(x_2) \\
{} \iffr &  x_1 - y = x_2 - y \\
{} \iffr &  x_1 = x_2.
\end{eqs}
Therefore $f$ is injective. Let $z \in \TN$. Then $f(z + y) = (z + y) - y = z$. Therefore $f$ is surjective. Since $f$ is a bijection between $A$ and $\TN$, $|A| = |\TN| \not\in \TN$. Therefore \prove{One} holds for $\pohsy$.
\end{proof}

\begin{theorem}[\uproperty{SubComp} for $\pohsy$ for fixed $\function{s}{Y}{X}$]
\label{AsymptoticCompositionForFixedS}
Let $d_x, d_y \in \posi{\TN}$, $X \in U$ be such that $X \subset \TR^{d_x}$, $Y \in U$ be such that $Y \subset \TR^{d_y}$, and $\function{s}{Y}{X}$. Then \property{SubComp} holds for $\pohsy$ and $s$ if and only if
\begin{equation}
\forall x^* \in \TR^{d_x}, \exists y^* \in \TR^{d_y}: \image{s}{\nonnb{Y}{y^*}} \subset \nonnb{X}{x^*}.
\end{equation}
\end{theorem}

\begin{proof}
Substitute the filter bases of $\pohsy$ into \proveby{LocalSubComposabilityForFixedS}.
\end{proof}

\begin{theorem}[\uproperty{ISubComp} fails for $\pohsy$ in $\TN^2$]
\label{AsymptoticInjectiveSubComposabilityFails}
$\pohsy$ does not have \property{ISubComp} from $\TN^2$ to $\TN$.
\end{theorem}

\begin{proof}
Let $X \coloneqq \TN^2$, $x^* \coloneqq (1, 0) \in X$, $Y \coloneqq \TN$, $y^* \in \TN$, and $\function{s}{Y}{X}$ be such that $s(n) = (0, n)$. Then $\image{s}{\nonnb{Y}{y^*}} \cap \nonnb{X}{x^*} = \emptyset$. Since $\image{s}{\nonnb{Y}{y^*}}$ is non-empty, $\image{s}{\nonnb{Y}{y^*}} \not\subset \nonnb{X}{x^*}$. \uproperty{ISubComp} fails by \thref{AsymptoticCompositionForFixedS}.
\end{proof}

\begin{theorem}[\uproperty{ISubComp} fails for $\pohsy$ in $\TZ$]
\label{AsymptoticInjectiveSubComposabilityFailsInZ}
$\pohsy$ does not have \property{ISubComp} from $\TZ$ to $\TN$.
\end{theorem}

\begin{proof}
Let $X \coloneqq \TZ$, $Y = \TN$, $x^* \coloneqq 0 \in X$, and $\function{s}{Y}{X}$ be such that $s(n) = -(n + 1)$. Then $\image{s}{\nonnb{Y}{y^*}} \cap \nonnb{X}{x^*} = \emptyset$, for all $y^* \in Y$. Since $\image{s}{\nonnb{Y}{y^*}}$ is non-empty, $\image{s}{\nonnb{Y}{y^*}} \not\subset \nonnb{X}{x^*}$. \uproperty{ISubComp} fails by \thref{AsymptoticCompositionForFixedS}.
\end{proof}

\begin{theorem}[\uproperty{Extend} for $\pohsy$]
\label{AsymptoticExtensibility}
$\pohsy$ has \prove{Extend}.
\end{theorem}

\begin{proof}
Let $A \in \filterset{X}$ and $B = A \times Y$. Then $B \in \filterset{X \times Y}$, and $\image{\projections{X}{Y}}{B} = A$. The result follows by \thref{LocalExtensibilityCharacterized}. \sprove{Extend}
\end{proof}

\section{Co-asymptotic linear dominance}
\label{CoasymptoticLinearDominance}

Recall the definition of co-asymptotic linear dominance from \sref{Introduction}:
\definecoasymptotic

\begin{theorem}[Co-asymptotic linear dominance is local linear dominance]
\label{CoasymptoticLinearDominanceIsLocalLinearDominance}
The family $\filtersetsym = \setb{\setb{\nlb{X}{y} : y \in \TR^d} : X \in U}$ is a family of filter bases with induced sub-structure.
\end{theorem}

\begin{proof}
\proofpart{Directedness}
Let $X \subset \TR^d$, and $A, B \in \filterset{X}$. Then there exists $y_1, y_2 \in \TR^d$, such that $A = \nlb{X}{y_1}$ and $B = \nlb{X}{y_2}$. Let $y \coloneqq \sup_{\leq}\setb{y_1, y_2}$, and $C \coloneqq \nlb{X}{y}$. Then $C \in \filterset{X}$ and $C \subset A \cap B$. 

\proofpart{Induced sub-structure}
Let $D \subset X \subset \TR^d$. Then
\begin{eqs}
\filterset{D} & = \setb{\nlb{D}{y} : y \in \TR^d} \\
{} & = \setb{\nlb{X}{y} \cap D : y \in \TR^d}. 
\end{eqs}
\end{proof}

\begin{note}[Intrinsic pitfall]
The same intrinsic pitfall as described in \nref{AsymptoticIntrinsicPitfall} for asymptotic linear dominance also holds for co-asymptotic linear dominance.
\end{note}

\begin{theorem}[\uproperty{One} for $\cohsy$]
\label{CoasymptoticOneSeparation}
$n \not\in \coh{\posi{\TN}}{1}$.
\end{theorem}

\begin{proof}
The proof is the same as in \proveby{AsymptoticOneSeparation}, since co-asymptotic linear dominance is equivalent to asymptotic linear dominance in $\posi{\TN}$.
\end{proof}

\begin{theorem}[Composability for $\cohsy$ for fixed $\function{s}{Y}{X}$]
\label{CoasymptoticCompositionForFixedS}
Let $d_x, d_y \in \posi{\TN}$, $X \in U$ be such that $X \subset \TR^{d_x}$, $Y \in U$ be such that $Y \subset \TR^{d_y}$, and $\function{s}{Y}{X}$. Then composability holds for $\cohsy$ and $s$ if and only if
\begin{equation}
\forall x^* \in \TR^{d_x}, \exists y^* \in \TR^{d_y}: \image{s}{\nlb{Y}{y^*}} \subset \nlb{X}{x^*}.
\end{equation}
\end{theorem}

\begin{proof}
Substitute the filter bases of $\cohsy$ into \proveby{LocalSubComposabilityForFixedS}.
\end{proof}

\begin{theorem}[\uproperty{ISubComp} fails for $\cohsy$ in $\TZ$]
\label{CoasymptoticInjectiveSubComposabilityFails}
$\cohsy$ does not have \property{ISubComp}.
\end{theorem}

\begin{proof}
Let $X \coloneqq \TZ$, $Y = \TN$, $x^* \coloneqq 0 \in X$, and $\function{s}{Y}{X}$ be such that $s(n) = -(n + 1)$. Then $\image{s}{\nlb{Y}{y^*}} \cap \nlb{X}{x^*} = \emptyset$, for all $y^* \in Y$. Since $\image{s}{\nlb{Y}{y^*}}$ is non-empty, $\image{s}{\nlb{Y}{y^*}} \not\subset \nlb{X}{x^*}$. \uproperty{ISubComp} fails by \thref{CoasymptoticCompositionForFixedS}.
\end{proof}

\begin{theorem}[\uproperty{Extend} fails for $\cohsy$]
\label{CoasymptoticExtensibilityFails}
$\cohsy$ does not have \property{Extend}.
\end{theorem}

\begin{proof}
Let $A = \posi{\TN}$. For any $B \in \filterset{\TN^2}$, $0 \in \image{\projections{\TN}{\TN}}{B}$, but $0 \not\in A$. \uproperty{Extend} fails by \thref{LocalExtensibilityCharacterized}.
\end{proof}

\begin{theorem}[$\cohsy$ is equal to $\fohsy$ in $\TN^d$]
\label{CoasymptoticIsEquivalentToCofiniteInN}
\begin{eqs}
\cohs{X} = \fohs{X},
\end{eqs}
for all $X \subset \TN^d$.
\end{theorem}

\begin{proof}
Let $d \in \posi{\TN}$, $X \subset \TN^d$, $f \in \rc{X}$, and $x \in \TN^d$. Then $\nlb{X}{x} \in \cpower{X}$, and therefore $\coh{X}{f} \subset \foh{X}{f}$. Let $A \in \cpower{X}$, and $y = 1 + \sup_{\leq} \bra{X \setminus A}$. Then $\nlb{X}{y} \subset A$, and therefore $\foh{X}{f} \subset \coh{X}{f}$.
\end{proof}

\begin{theorem}[\uproperty{IComp} holds for $\cohsy$ in $\TN^d$]
\label{CoasymptoticInjectiveComposability}
$\cohsy$ has \property{IComp} in $\TN^d$.
\end{theorem}

\begin{proof}
This follows from \proveby{CoasymptoticIsEquivalentToCofiniteInN} and \proveby{CofiniteInjectiveComposition}.
\end{proof}

\section{Cofinite linear dominance}
\label{CofiniteLinearDominance}

Recall the definition of cofinite linear dominance from \sref{Introduction}:
\definecofinite

\begin{theorem}[Cofinite linear dominance is local linear dominance]
\label{CofiniteLinearDominanceIsLocalLinearDominance}
The family $\setb{\cpower{X} : X \in U}$ is a family of filter bases with induced sub-structure.
\end{theorem}

\begin{proof}
\proofpart{Directedness}
Let $A, B \in \filterset{X}$, where $X \in U$. Then we may choose $C \coloneqq A \cap B \in \filterset{X}$, so that $C \subset A \cap B$. 

\proofpart{Induced sub-structure}
Let $D \subset X$, where $X \in U$. If $B \in \filterset{D}$, then $B = D \cap (X \setminus (D \setminus B))$, and $(X \setminus (D \setminus B)) \in \filterset{X}$. If $A \in \filterset{X}$, then $A \cap D \in \filterset{D}$. 
\end{proof}

\begin{note}[Fr\'echet filter]
The filter basis for cofinite linear dominance is also known as the Fr\'echet filter.
\end{note}

\begin{theorem}[\uproperty{One} for $\fohsy$]
\label{CofiniteOneSeparation}
$\fohsy$ has \property{One}.
\end{theorem}

\begin{proof}
$\fohsy$ has \property{One} if and only if $\bra{\forall A \in \filterset{\posi{\TN}} : |A| \not\in \TN}$ by \proveby{LocalOneSeparationCharacterized}. Let $X \coloneqq \posi{\TN}$ and $A \in \filterset{X} = \cpower{X}$. Let $\function{f}{A}{\TN}$ be such that $f(n) = |A^{< n}|$. Then $f$ is a bijection between $A$ and $\TN$, and so $|A| = |\TN| \not\in \TN$. Therefore \prove{One} holds for $\fohsy$.
\end{proof}

\begin{theorem}[\uproperty{SubComp} for $\fohsy$ for fixed $\function{s}{Y}{X}$]
\label{CofiniteCompositionForFixedS}
Let $X, Y \in U$, and $\function{s}{Y}{X}$. Then $\fohsy$ has \prove{SubComp} for $s$ if and only if $s$ is finite-to-one:
\begin{equation}
\forall x \in X: |\preimage{s}{\setb{x}}| \in \TN.
\end{equation}
\end{theorem}

\begin{proof}
The $\fohsy$ has \property{SubComp} for $s$ if and only if
\begin{equation}
\forall A_X \in \filterset{X}, \exists A_Y \in \filterset{Y}: \image{s}{A_Y} \subset A_X
\end{equation}
by \proveby{LocalSubComposabilityForFixedS}. Substituting the filter bases of $\fohsy$, this is equivalent to
\begin{equation}
\forall A_X \in \cpower{X}, \exists A_Y \in \cpower{Y}: \image{s}{A_Y} \subset A_X.
\end{equation}
We will show equivalence to this formula.

\proofpart{$\implies$}
Suppose there exists $x \in X$ such that $|\preimage{s}{\setb{x}}| \not\in \TN$. Then $|Y| \not\in \TN$, and therefore $\emptyset \not\in \cpower{Y}$. Let $A_X = X \setminus \setb{x} \in \cpower{X}$, and $A_Y \in \cpower{Y}$. Then $A_Y \neq \emptyset$, and  $A_Y \cap \preimage{s}{\setb{x}} \neq \emptyset$, which is equivalent to $x \in \image{s}{A_Y}$. Therefore $\image{s}{A_Y} \not\subset A_X$; $\fohsy$ does not have \property{SubComp} for $s$.

\proofpart{$\impliedby$}
Let $A_X \in \cpower{X}$ be such that $X \setminus A_X = \setb{x_1, \dots, x_n}$, and $A_Y = \preimage{s}{A_X}$. Then $\image{s}{A_Y} = A_X$, and
\begin{eqs}
|Y \setminus A_Y| & = |Y \setminus \preimage{s}{A_X}| \\
{} & = |\preimage{s}{X \setminus A_X}| \\
{} & = |\preimage{s}{\setb{x_1, \dots, x_n}}| \\
{} & \leq |\preimage{s}{\setb{x_1}}| + \cdots + |\preimage{s}{\setb{x_n}}| \\
{} & \in \TN.
\end{eqs}
Therefore $A_Y \in \cpower{Y}$; $\fohsy$ has \prove{SubComp} for $s$.
\end{proof}

\begin{theorem}[\uproperty{SubComp} for $\fohsy$ for positive $f \in \rc{X}$]
\label{CofiniteCompositionPositive}
Let $X, Y \in U$, $\function{s}{Y}{X}$, and $f \in \posi{\rc{X}}$. Then
\begin{equation}
\foh{X}{f} \circ s \subset \foh{X}{f \circ s}.
\end{equation}
\end{theorem}

\begin{proof}
It holds that $\card{X \setminus A} < \infty$, for all $A \in \filterset{X}$, $X \in U$. The claim follows from \thref{LocalSubComposabilityForPositiveFunctions}.
\end{proof}

\begin{theorem}[\uproperty{SubComp} fails for $\fohsy$]
\label{CofiniteSubComposabilityFails}
$\fohsy$ does not have \property{SubComp}.
\end{theorem}

\begin{proof}
Let $\function{s}{\TN}{\TN}$ be such that
\begin{equation}
s(n) =
\begin{cases}
0, & n \in 2\TN, \\
n, & n \in 2\TN + 1.
\end{cases}
\end{equation}
Then $|\preimage{s}{\setb{0}}| \not\in \TN$, and so \property{SubComp} fails for $\fohsy$ and $s$ by \proveby{CofiniteCompositionForFixedS}. For example, let $\function{\hat{f}}{\TN}{\TN}$ be such that $\hat{f}(0) > 0$ and $\function{f}{\TN}{\TN}$ be such that $f(0) = 0$ and $\hat{f} \in \foh{\TN}{f}$. Then $\bra{\hat{f} \circ s} \not\in \foh{\TN}{f \circ s}$.
\end{proof}

\begin{theorem}[\uproperty{SubsetSum} fails for $\fohsy$]
\label{CofiniteSubsetSumFails}
$\fohsy$ does not have \property{SubsetSum}.
\end{theorem}

\begin{proof}
Since \property{SubsetSum} implies \property{SubComp}, and \property{SubComp} does not hold by \thref{CofiniteSubComposabilityFails}, neither does \property{SubsetSum}.
\end{proof}

\begin{theorem}[\uproperty{IComp} for $\fohsy$]
\label{CofiniteInjectiveComposition}
$\fohsy$ has \prove{IComp}.
\sprove{ISubComp}
\end{theorem}

\begin{proof}
Let $X, Y \in U$, and $\function{s}{Y}{X}$ be injective. Then $|\preimage{s}{\setb{x}}| \leq 1$, for all $x \in X$. Therefore $\fohsy$ has \property{SubComp} for $s$ by \proveby{CofiniteCompositionForFixedS}. 
\srequire{SubComp} \sprove{ISubComp} 
Since $\fohsy$ has \property{Order} by \proveby{LocalOrderConsistency}, and \property{Local} by \proveby{LocalLocality}, $\fohsy$ has \property{ISuperComp} by \proveby{InjectiveSuperComposabilityIsImplied}.
\sprove{IComp}
\end{proof}

\begin{theorem}[\uproperty{Extend} fails for $\fohsy$]
\label{CofiniteExtensibilityFails}
$\fohsy$ does not have \property{Extend}.
\end{theorem}

\begin{proof}
Let $A = \posi{\TN}$. For any $B \in \filterset{\TN^2}$, $0 \in \image{\projections{\TN}{\TN}}{B}$, but $0 \not\in A$. \uproperty{Extend} fails by \thref{LocalExtensibilityCharacterized}.
\end{proof}

\section{Linear dominance}
\label{LinearDominance}

Recall the definition of linear dominance from \sref{Introduction}:
\definelinear

\begin{theorem}[Linear dominance is local linear dominance]
\label{LinearDominanceIsLocalLinearDominance}
The family $\setb{\setb{X} : X \in U}$ is a family of filter bases with induced sub-structure.
\end{theorem}

\begin{proof}
\proofpart{Directedness}
Let $A, B \in \filterset{X}$, where $X \in U$. Then $A = X = B$, and we may choose $C \coloneqq X \in \filterset{X}$, so that $C \subset A \cap B$. 

\proofpart{Induced sub-structure}
Let $D \subset X$, where $X \in U$. Then $\filterset{D} = \setb{D} = \setb{A \cap D : A \in \filterset{X}}$.
\end{proof}

\begin{note}[]
The primitive properties of linear dominance are proved in \sref{SufficientDefinition}.
\end{note}

\begin{theorem}[\uproperty{Extend} for $\lohsy$]
\label{LinearExtensibility}
$\lohsy$ has \prove{Extend}.
\end{theorem}

\begin{proof}
Let $A \in \filterset{X}$ and $B = X \times Y$. Then $A = X$ and $B \in \filterset{X \times Y}$. Since $\image{\projections{X}{Y}}{B} = X = A$, the result follows by \thref{LocalExtensibilityCharacterized}. \sprove{Extend}
\end{proof}

\section{Affine dominance}
\label{AffineDominance}

Recall the definition of affine dominance from \sref{Introduction}:
\defineaffine

\begin{note}[]
Affine dominance is not a local linear dominance. 
\end{note}

We shall apply the following lemma without mentioning it, since it is used in almost every proof.

\begin{theorem}[Simplification lemma for $\aohsy$]
\label{AffineSingleConstant}
Let $I$ be a finite set, $X_i \in U$, $f_i \in \rc{X_i}$, and $\hat{f}_i \in \aoh{X_i}{f_i}$, for all $i \in I$. Then there exists $c \in \posi{\TR}$, such that
\begin{equation}
\hat{f}_i \lt c f_i + c,
\end{equation}
for all $i \in I$.
\end{theorem}

\begin{proof}
There exists $c_i \in \posi{\TR}$, such that $\hat{f}_i \lt c_i f_i + c_i$, for all $i \in I$.
Let $c = \max\setb{c_i : i \in I}$. Then $\hat{f_i} \lt c f_i + c$, for all $i \in I$.
\end{proof}

\begin{theorem}[\uproperty{Order} for $\aohsy$]
\label{AffineOrderConsistency}
Let $X \in U$, and $f, g \in \rc{X}$. Then 
\begin{equation}
f \leq g \implies \aohx{f} \subset \aohx{g}.
\end{equation}
\sprove{Order}
\end{theorem}

\begin{proof}
Let $\hat{f} \in \aohx{f}$. Then there exists $c \in \posi{\TR}$ such that $\hat{f} \leq cf + c$. Since $f \leq g$, it follows that $\hat{f} \leq cg + c$. Therefore $\hat{f} \in \aohx{g}$.
\sprove{Order}
\end{proof}

\begin{theorem}[\uproperty{Trans} for $\aohsy$]
\label{AffineTransitivity}
Let $X \in U$, and $f, g, h \in \rc{X}$. Then 
\begin{equation}
\bra{f \in \aohx{g} \textrm{ and } g \in \aohx{h}} \implies f \in \aohx{h}.
\end{equation}
\sprove{Trans}
\end{theorem}

\begin{proof}
Let $f \in \aohx{g}$, and $g \in \aohx{h}$. Then there exists $c \in \posi{\TR}$, such that $f \lt c g + c$ and $g \lt c h + c$. It follows that
\begin{eqs}
f & \lt c (c h + c) + c \\
{} & = c^2 h + (c^2 + c) \\
{} & \leq (c^2 + c) h + (c^2 + c).
\end{eqs}
Therefore $f \in \aohx{h}$.
\sprove{Trans}
\end{proof}

\begin{theorem}[\uproperty{Local} for $\aohsy$]
\label{AffineLocality}
Let $X \in U$, $f, g \in \rc{X}$, and $C \subset \power{X}$ be a finite cover of $X$. Then
\begin{equation}
\forall D \in C: \restrb{f}{D} \in \aoh{D}{\restr{g}{D}} \implies f \in \aohx{g}.
\end{equation}
\sprove{Local}
\end{theorem}

\begin{proof}
Assume $\bra{\restr{f}{D}} \in \aoh{D}{\restr{g}{D}}$, for all $D \in C$. Since $C$ is finite, there exists $c \in \posi{\TR}$, such that $\bra{\restr{f}{D}} \lt c \bra{\restr{g}{D}} + c$, for all $D \in C$. Since $C$ covers $X$, $f \lt c g + c$. Therefore $f \in \aohx{g}$.
\sprove{Local}
\end{proof}

\begin{theorem}[\uproperty{Zero} fails for $\aohsy$]
\label{AffineZeroSeparationFails}
Let $X \in U$ be non-empty. Then
\begin{equation}
\aohx{1} = \aohx{0}.
\end{equation}
\end{theorem}

\begin{proof}
$\aohs{X}$ has \property{Order} by \proveby{AffineOrderConsistency}, \property{Trans} by \proveby{AffineTransitivity}, and \property{Orderness} by \proveby{OrdernessIsImplied}. It holds that $(x \mapsto 1) \leq 1 (x \mapsto 0) + 1$, for all $x \in X$. Therefore $1 \in \aohx{0}$. By \require{Orderness}, $\aohx{1} \subset \aohx{0}$. By \require{Order}, $0 \in \aohx{1}$. By \require{Orderness}, $\aohx{1} \supset \aohx{0}$.
\end{proof} 

\begin{theorem}[\uproperty{One} for $\aohsy$]
\label{AffineOneSeparation}
\begin{equation}
n \not\in \aoh{\posi{\TN}}{1}.
\end{equation}
\sprove{One}
\end{theorem}

\begin{proof}
It holds that $n > c 1 + c$, for all $c \in \posi{\TR}$, and $n \in \posib{\TN}{2c}$. Therefore $n \not\in \aoh{\posi{\TN}}{1}$. 
\sprove{One}
\end{proof}

\begin{theorem}[\uproperty{Scale} for $\aohsy$]
\label{AffinePositiveScaleInvariance}
Let $X \in U$, $f \in \rc{X}$, and $\alpha \in \posi{\TR}$. Then 
\begin{equation}
\aohx{\alpha f} = \aohx{f}.
\end{equation}
\sprove{Scale}
\end{theorem}

\begin{proof}
\proofpart{$\subset$}
Assume $\hat{f} \in \aohx{\alpha f}$. Then there exists $c \in \posi{\TR}$, such that
\begin{eqs}
\hat{f} & \lt c (\alpha f) + c \\
{} & \lt \max(c \alpha, c) f + \max(c \alpha, c).
\end{eqs}
Therefore $\hat{f} \in \aohx{f}$. 

\proofpart{$\supset$}
Assume $\hat{f} \in \aohx{f}$. Then there exists $c \in \posi{\TR}$, such that
\begin{eqs}
\hat{f} & \lt c f + c \\
{} & = (c / \alpha) (\alpha f) + c \\
{} & \lt \max(c / \alpha, c) (\alpha f) + \max(c / \alpha, c).
\end{eqs}
Therefore $\hat{f} \in \aohx{\alpha f}$.

\sprove{Scale}
\end{proof}

\begin{theorem}[Composability for $\aohsy$]
\label{AffineSubComposability}
Let $X, Y \in U$, $f \in \rc{X}$, and $\function{s}{Y}{X}$. Then
\begin{equation}
\aoh{X}{f} \circ s \subset \aoh{Y}{f \circ s},
\end{equation}
\sprove{SubComp}
\end{theorem}

\begin{proof}
Let $\hat{f} \in \aohx{f}$. Then there exists $c \in \posi{\TR}$ such that $\hat{f} \lt cf + c$. This implies $(\hat{f} \circ s) \lt c(f \circ s) + c$. Therefore $(\hat{f} \circ s) \in \aoh{Y}{f \circ s}$.

\sprove{SubComp}
\end{proof}

\begin{theorem}[\uproperty{SubHom} for $\aohsy$ characterized]
\label{AffineSubHomogenuityCharacterization}
\begin{equation}
f \aohx{g} \subset \aohx{fg} \iff f \in \aohx{fg},
\end{equation}
for all $f, g \in \rc{X}$.
\end{theorem}

\begin{proof}
$\aohs{X}$ has \property{Order} by \proveby{AffineOrderConsistency}.

\proofpart{$\implies$}
Suppose $f \not\in \aohx{fg}$. Let $\hat{g} \in \rc{X}$ be such that $\hat{g} = 1$. Then $f \in \aohx{f}$ by \require{Order}, and $\hat{g} \in \aohx{g}$, since $\hat{g}(x) = 1 \leq 1g(x) + 1$, for all $x \in X$. Then $f \hat{g} = f \not\in \aohx{fg}$. Therefore $f \aohx{g} \not\subset \aohx{fg}$.

\proofpart{$\impliedby$}
Suppose $f \in \aohx{fg}$ and $\hat{g} \in \aohx{g}$. Then there exists $c \in \posi{\TR}$ such that
\begin{eqs}
f & \leq cfg + c \\
\hat{g} & \leq cg + c.
\end{eqs}
Therefore
\begin{eqs}
f\hat{g} & \leq f(cg + c) \\
{} & = cfg + cf \\
{} & \leq cfg + c\bra{cfg + c} \\
{} & \leq (c^2 + c) fg + c^2 \\
{} & \leq (c^2 + c) fg + (c^2 + c).
\end{eqs}
Therefore $f \hat{g} \in \aohx{fg}$.
\end{proof}

\begin{theorem}[\uproperty{NSubHom} fails for $\aohsy$]
\label{AffineSubHomogenuityFails}
There exists $f, g \in \rc{\posi{\TN}}$ such that $f(\posi{\TN}) \subset \posi{\TN}$, and
\begin{equation}
f \aoh{\posi{\TN}}{g} \not\subset \aoh{\posi{\TN}}{fg}.
\end{equation}
\end{theorem}

\begin{proof}
Let $f \in \rc{\posi{\TN}}$ be such that $f(n) = n$, and $g \in \rc{\posi{\TN}}$ be such that $g(n) = 1 / n$. Then
\begin{eqs}
f\bra{\ceilb{3c}} & = \ceilb{3c} \\
{} & > c + c \\
{} & = c f\bra{\ceilb{3c}} g\bra{\ceilb{3c}} + c,
\end{eqs}
for all $c \in \posi{\TR}$. Therefore $f \not\in \aohx{fg}$, and so \property{NSubHom} fails by \thref{AffineSubHomogenuityCharacterization}.
\end{proof}

\begin{theorem}[\uproperty{NSubDiv} for $\aohsy$]
\label{AffineSubHomogeneityNDiv}
Let $X \in U$, and $f, u \in \rc{X}$ be such that $\image{u}{X} \subset 1 / \posi{\TN}$. Then \sprove{NSubDiv}
\begin{equation}
u \rohx{f} \subset \rohx{uf}.
\end{equation}
\end{theorem}

\begin{proof}
Let $\hat{f} \in \aohx{f}$. Then there exists $c \in \posi{\TR}$, such that
\begin{eqs}
\hat{f} \leq c f + c.
\end{eqs}
This implies
\begin{eqs}
u \hat{f} & \leq u(cf + c) \\
{} & = c (uf) + uc \\
{} & \leq c (uf) + c.
\end{eqs}
Therefore $u\hat{f} \in \aohx{uf}$.
\sprove{NSubDiv}
\end{proof}

\begin{theorem}[\uproperty{Extend} for $\aohsy$]
\label{AffineExtensibility}
$\aohsy$ has \prove{Extend}.
\end{theorem}

\begin{proof}
Suppose $f \in \aohx{g}$. Then there exists $c \in \posi{\TR}$ such that
\begin{eqs}
f \leq cg + c.
\end{eqs}
Then
\begin{eqs}
f \circ \projections{X}{Y} \leq c \bra{g \circ \projections{X}{Y}} + c.
\end{eqs}
Therefore $\aohsy$ has \prove{Extend}.
\end{proof}

\begin{theorem}[\uproperty{SuperHom} fails for $\aohsy$]
\label{AffineSuperHomogenuityFails}
There exists $f, g \in \rc{X}$ such that
\begin{equation}
f \aohx{g} \not\supset \aohx{fg}.
\end{equation}
\end{theorem}

\begin{proof}
Let $f = 0$, and $\hat{h} \in \rc{X}$ be such that $\hat{h} = 1$. Then $\hat{h} \leq 1 fg + 1$, and so $\hat{h} \in \aohx{fg}$. For all $\hat{g} \in \aohx{g}$ it holds that $f \hat{g} = 0 \neq 1 = \hat{h}$. Therefore $f \aohx{g} \not\supset \aohx{fg}$.
\end{proof}

\begin{theorem}[\uproperty{SuperMulti} for $\aohsy$]
\label{AffineSuperMultiplicativity}
Let $f, g \in \rc{X}$. Then
\begin{equation}
\aohx{f} \aohx{g} \supset \aohx{fg}.
\end{equation}
\sprove{SuperMulti}
\end{theorem}

\begin{proof}
Let $\hat{h} \in \aohx{fg}$. Then there exists $c \in \posi{\TR}$ such that
\begin{equation}
\hat{h} \leq cfg + c.
\end{equation}
Let $\hat{g} \in \rc{X}$ be such that $\hat{g}(x) = \max(g(x), 1)$. Since $\max(g(x), 1) \leq 1 g(x) + 1$, it holds that $\hat{g} \in \aohx{g}$. Then
\begin{eqs}
\frac{\hat{h}}{\hat{g}} & \leq cf \frac{g}{\hat{g}} + \frac{c}{\hat{g}} \\
{} & \leq cf + c.
\end{eqs}
Let $\hat{f} \in \rc{X}$ be such that $\hat{f} = \hat{h} / \hat{g}$. By the above, $\hat{f} \in \aohx{f}$. In addition, $\hat{f} \hat{g} = \hat{h}$. Therefore $\hat{h} \in \aohx{f} \aohx{g}$.
\sprove{SuperMulti}
\end{proof}

\begin{theorem}[\uproperty{PowerH} for $\aohsy$]
\label{AffinePowerHomogenuity}
Let $X \in U$, $f \in \rc{X}$, and $\alpha \in \posi{\TR}$. Then
\begin{equation}
\aohx{f}^\alpha = \aohx{f^{\alpha}}.
\end{equation}
\end{theorem}

\begin{proof}
$\aohsy$ has \property{Local} by \proveby{AffineLocality}.

\proofpart{$\subset$}
Let $\hat{f} \in \aohx{f}$. Then there exists $c \in \posi{\TR}$ such that
\begin{equation}
\hat{f} \leq cf + c.
\end{equation}
Let $\alpha \in \posi{\TR}$. Since positive powers are increasing,
\begin{eqs}
\hat{f}^{\alpha} & \leq \bra{c f + c}^{\alpha} \\
{} & = c^{\alpha} \bra{f + 1}^{\alpha}.
\end{eqs}
Let $F = \setb{x \in X: f(x) \geq 1}$. Then for $x \in F$,
\begin{eqs}
\bra{f(x) + 1}^{\alpha} & \leq \bra{f(x) + f(x)}^{\alpha} \\
{} & = 2^{\alpha} f(x)^{\alpha}.
\end{eqs}
Let $\compl{F} = X \setminus F$. Then for $x \in \compl{F}$,
\begin{eqs}
\bra{f(x) + 1}^{\alpha} & \leq \bra{1 + 1}^{\alpha} \\
{} & = 2^{\alpha}.
\end{eqs}
Therefore,
\begin{eqs}
f & \leq 2^{\alpha} (f^{\alpha} + 1) \\
{} & = 2^{\alpha} f^{\alpha} + 2^{\alpha}.
\end{eqs}
It follows that
\begin{eqs}
\hat{f}^{\alpha} & \leq c^{\alpha} 2^{\alpha} (f^{\alpha} + 1),
\end{eqs}
which shows that $\hat{f}^{\alpha} \in \aohx{f^{\alpha}}$.

\proofpart{$\supset$}
Let $\alpha \in \posi{\TR}$, and $\hat{g} \in \aohx{f^{\alpha}}$. Then there exists $c \in \posi{\TR}$ such that
\begin{equation}
\hat{g} \leq cf^{\alpha} + c.
\end{equation}
Since positive powers are increasing,
\begin{eqs}
\hat{g}^{1 / \alpha} & \leq \bra{c f^{\alpha} + c}^{1 / \alpha} \\
{} & = c^{1 / \alpha} (f^{\alpha} + 1)^{1 / \alpha}.
\end{eqs}
Similarly to above, we can prove that $\hat{g}^{1 / \alpha} \in \aohx{f}$. Therefore, let $\hat{f} \in \rc{X}$ be such that $\hat{f} = \hat{g}^{1 / \alpha}$. Then $\hat{f} \in \aohx{f}$ and $\hat{f}^{\alpha} = \hat{g}$. Therefore $\hat{g} \in \aohx{f}^{\alpha}$.
\end{proof}

\begin{theorem}[Subset\-/sum rule for $\aohsy$ in some cases]
\label{AffineSubsetSummabilityInSomeCases}
Let $X, Y$, $Z \in U$, $\function{S}{X}{\fpower{Y \times Z}}$, $a \in \rc{Z}$, $f \in \rc{Y}$, and $\hat{f} \in \oh{Y}{f}$. Then \prove{SubsetSum} holds for $\aohsy$ if
\begin{equation}
\exists M \in \nonn{\TR}, \forall x \in X: \sum_{(y, z) \in S_x} a(z) \leq M.
\end{equation}
\end{theorem}

\begin{proof}
Let $\function{T}{\rc{Y}}{\rc{X}}$ such that
\begin{equation}
T(f) = \bra{y \mapsto \sum_{(y, z) \in S_x} a(z) f(y)}.
\end{equation}
There exists $c \in \posi{\TR}$ such that $\hat{f} \leq c f + c$. It follows that
\begin{eqs}
T(\hat{f}) & \leq c T(f) + c \sum_{(y, z) \in S_x} a(z) \\
{} & \leq c T(f) + c M \\
{} & \leq \max(c M, c) T(f) + \max(c M, c).
\end{eqs}
Therefore $T(\hat{f}) \in \aoh{Y}{T(f)}$, and so \prove{SubsetSum} holds.
\end{proof}

\begin{note}[Subset-sum rule for $\aohsy$]
It is open to us to characterize when exactly \property{SubsetSum} holds for $\aohsy$. By \proveby{AffineContainment} we would expect \property{SubsetSum} to fail, since \property{SubsetSum} does not hold for $\fohsy$ either.
\end{note}

\begin{theorem}[\uproperty{TrivialZero} fails for $\aohsy$]
\label{ZeroTrivialityFailsForAffineDominance}
$\aohsy$ does not have \property{TrivialZero}.
\end{theorem}

\begin{proof}
It holds that
\begin{eqs}
1 \leq 1 \cdot 0 + 1.
\end{eqs}
Therefore $1 \in \aoh{\posi{\TN}}{0}$, and $\aoh{\posi{\TN}}{0} \neq \setb{0}$.
\end{proof}

\section{Containment relations}

Analysis of local linear dominance reveals that making the filter basis sets larger increases the number of fulfilled desirable properties --- provided that at least \property{One} is satisfied. 

Trivial linear dominance has the smallest filter basis sets. It satisfies all of the desirable properties except those of non\-/triviality. Its worst defect is the failure of \property{One}. In fact, every function in $\rc{X}$ is equivalent for every $X \in U$. 

Asymptotic linear dominance has the next smallest filter basis sets. Its worst defect is the failure of \property{SubComp} in $\TN^d$. 

Co-asymptotic linear dominance improves upon asymptotic linear dominance by satisfying \property{ISubComp} in $\TN^d$. Its worst defects are the failure of \property{SubComp} in $\TN^d$, and the failure of \property{ISubComp} in $\TZ$. 

Cofinite linear dominance improves upon coasymptotic linear dominance by satisfying \property{ISubComp} in all universes. Its worst defect is the failure of \property{SubComp} in $\TN$.

Linear dominance has the largest filter basis sets, and fulfills all of the desirable properties. We formalize the intuitive size-comparison of the filter basis sets by the following theorem.

\begin{theorem}[Containment in $\TR^d$]
\label{Containment}
\begin{equation}
\lohx{f} \air{\subset} \fohx{f} \air{\subset} \cohx{f} \air{\subset} \pohx{f} \air{\subset} \trohx{f},
\end{equation}
for all $X \subset \TR^d$ and $f \in \rc{X}$.
\end{theorem}

\begin{proof}
Assume $d \in \posi{\TN}$ and $X \subset \TR^d$.

\proofpart{$\lohx{f} \subset \fohx{f}$}
Suppose $\hat{f} \in \lohx{f}$. Then there exists $c \in \posi{\TR}$ such that $\hat{f} \leq cf$. Since $X \in \cpower{X}$, it holds that $\restrb{\hat{f}}{X} \leq c\restrb{f}{X}$. Therefore $\hat{f} \in \fohx{f}$.

\proofpart{$\fohx{f} \subset \cohx{f}$}
Suppose $\hat{f} \in \fohx{f}$. Then there exists $c \in \posi{\TR}$ and $A \in \cpower{X}$ such that $\restrb{\hat{f}}{A} \leq c\restrb{f}{A}$. Let $y = \sup_{\leq} \bra{X \setminus A} + 1$. Then $\nlb{X}{y} \subset A$, and so $\restrb{\hat{f}}{\nlb{X}{y}} \leq c\restrb{f}{\nlb{X}{y}}$. Therefore $\hat{f} \in \cohx{f}$.

\proofpart{$\cohx{f} \subset \pohx{f}$}
Suppose $\hat{f} \in \cohx{f}$. Then there exist $c \in \posi{\TR}$ and $y \in \TR^d$ such that $\restrb{\hat{f}}{\nlb{X}{y}} \leq c\restrb{f}{\nlb{X}{y}}$. Since $\nonnb{X}{y} \subset \nlb{X}{y}$, it holds that $\restrb{\hat{f}}{\nonnb{X}{y}} \leq c\restrb{f}{\nonnb{X}{y}}$. Therefore $\hat{f} \in \pohx{f}$.

\proofpart{$\pohx{f} \subset \trohx{f}$}
Suppose $\hat{f} \in \pohx{f}$. Then $\hat{f} \in \trohx{f}$, since every function is equivalent under trivial linear dominance.
\end{proof}

Affine dominance has surprisingly good properties; \we{} \were{} especially surprised to be able to prove \property{PowerH} for it. Its worst defect is the failure of \property{SubHom}. In terms of containment, affine dominance has smaller filter bases than cofinite linear dominance.

\begin{theorem}[Affine containment]
\label{AffineContainment}
\begin{equation}
\fohx{f} \subset \aohx{f},
\end{equation}
for all $X \in U$ and $f \in \rc{X}$.
\end{theorem}

\begin{proof}
Let $\hat{f} \in \fohx{f}$. Then there exists $A \in \cpower{X}$ and $c \in \posi{\TR}$ such that $\restrb{\hat{f}}{A} \leq c \restrb{f}{A}$. Since $A$ is cofinite, let $d = \max\bra{\hat{f}\bra{X \setminus A} \cup \setb{c}}$. Then $\hat{f} \leq d f + d$. Therefore $\hat{f} \in \aohx{f}$.
\end{proof}

\chapter{Proofs of minimality}
\label{ProofsOfMinimality}

In this section we consider additional candidate definitions. These definitions are used --- together with those in \sref{CandidateDefinitions} --- to prove the minimality of pre-primitive properties in \sref{Characterization}.

\section{Elementwise dominance}
\label{ElementwiseDominance}

\begin{definition}[Elementwise dominance]
\defineexp{Elementwise dominance}{dominance!elementwise} $\eohsy$ is defined by $f \in \eohx{g}$ if and only if
\begin{eqs}
f \leq g,
\end{eqs}
for all $g \in \rc{X}$, and all $X \in U$, where $U$ is the class of all sets.
\end{definition}

\begin{theorem}[\uproperty{Order} for $\eohsy$]
\label{ElementwiseOrder}
$\eohsy$ has \prove{Order}.
\end{theorem}

\begin{proof}
By definition. \sprove{Order}
\end{proof}

\begin{theorem}[\uproperty{Trans} for $\eohsy$]
\label{ElementwiseTransivity}
$\eohsy$ has \prove{Trans}.
\end{theorem}

\begin{proof}
Let $f \in \eohx{g}$, $g \in \eohx{h}$, and $h \in \rc{X}$. Then
\begin{eqs}
{} & f \leq g, \\
{} & g \leq h.
\end{eqs}
Therefore
\begin{eqs}
f \leq h.
\end{eqs}
\sprove{Trans}
\end{proof}

\begin{theorem}[\uproperty{Scale} fails for $\eohsy$]
\label{ElementwiseScaleInvarianceFails}
$\eohsy$ does not have \property{Scale}.
\end{theorem}

\begin{proof}
Let $f \in \eohx{g}$ be such that $f \neq 0$ and $\alpha \in \posi{\TR}$ be such that $\alpha < 1$. Let $p \in X$ be such that $f(p) > 0$. Then
\begin{eqs}
f(p) > \alpha f(p).
\end{eqs}
\end{proof}

\begin{theorem}[\uproperty{Local} for $\eohsy$]
\label{ElementwiseLocality}
$\eohsy$ has \prove{Local}.
\end{theorem}

\begin{proof}
Let $C \subset \power{X}$ be a finite cover of $X$, and suppose $\restrb{f}{D} \in \eohx{\restr{g}{D}}$ for all $D \in C$. Then
\begin{eqs}
\restrb{f}{D} \leq \restrb{g}{D},
\end{eqs}
for all $D \in C$. This implies
\begin{eqs}
f \leq g.
\end{eqs}
\sprove{Local}
\end{proof}

\begin{theorem}[\uproperty{One} for $\eohsy$]
\label{ElementwiseOneSeparation}
$\eohsy$ has \prove{One}.
\end{theorem}

\begin{proof}
It holds that  $n > 1$ for $n = 2$. \sprove{One}
\end{proof}

\begin{theorem}[\uproperty{SubHom} for $\eohsy$]
\label{ElementwiseSubHomogeneity}
$\eohsy$ has \prove{SubHom}. \sprove{SubHomN} \sprove{SubDivN}
\end{theorem}

\begin{proof}
Let $f \in \eohx{g}$, and $g, u \in \rc{X}$. Then
\begin{eqs}
{} & f \leq g.
\end{eqs}
This implies
\begin{eqs}
uf \leq ug.
\end{eqs}
\sprove{SubHom} \sprove{SubHomN} \sprove{SubDivN}
\end{proof}

\begin{theorem}[\uproperty{SubComp} for $\eohsy$]
\label{ElementwiseSubComposability}
$\eohsy$ has \prove{SubComp}.
\end{theorem}

\begin{proof}
Let $f \in \eohx{g}$, and $\function{s}{Y}{X}$. Then
\begin{eqs}
{} & f \leq g.
\end{eqs}
This implies
\begin{eqs}
f \circ s \leq g \circ s.
\end{eqs}
\sprove{SubComp}
\end{proof}

\section{Multiple dominance}
\label{MultipleDominance}

\begin{definition}[Multiple dominance]
\defineexp{Multiple dominance}{dominance!multiple} $\mohsy$ is defined by $f \in \mohx{g}$ if and only if there exists $c \in \posi{\TR}$ such that
\begin{eqs}
f = cg,
\end{eqs}
for all $g \in \rc{X}$, and all $X \in U$, where $U$ is the class of all sets.
\end{definition}

\begin{theorem}[\uproperty{Order} fails for $\mohsy$]
\label{MultipleOrderFails}
$\mohsy$ does not have \property{Order}.
\end{theorem}

\begin{proof}
Let $f, g \in \rc{\posi{\TN}}$ be such that $f = 0$ and $g = 1$. Then $f \leq g$, and
\begin{eqs}
0 = f \neq cg = c,
\end{eqs}
for all $c \in \posi{\TR}$.
\end{proof}

\begin{theorem}[\uproperty{Trans} for $\mohsy$]
\label{MultipleTransivity}
$\mohsy$ has \prove{Trans}.
\end{theorem}

\begin{proof}
Let $f \in \mohx{g}$, $g \in \mohx{h}$, and $h \in \rc{X}$. Then there exists $c, d \in \posi{\TR}$ such that
\begin{eqs}
{} & f = cg, \\
{} & g = dh.
\end{eqs}
Therefore
\begin{eqs}
f = cd h.
\end{eqs}
\sprove{Trans}
\end{proof}

\begin{theorem}[\uproperty{Scale} for $\mohsy$]
\label{MultipleScaleInvariance}
$\mohsy$ has \property{Scale}.
\end{theorem}

\begin{proof}
Let $f \in \mohx{g}$, and $\alpha \in \posi{\TR}$. Then there exists $c \in \posi{\TR}$ such that
\begin{eqs}
f = cg.
\end{eqs}
This implies
\begin{eqs}
f = c \frac{1}{\alpha} \alpha g.
\end{eqs}
\end{proof}

\begin{theorem}[\uproperty{Local} fails for $\mohsy$]
\label{MultipleLocalityFails}
$\mohsy$ does not have \property{Local}.
\end{theorem}

\begin{proof}
Let $f \in \rc{\TN}$ be such that
\begin{eqs}
f(n) = 
\begin{cases}
2, & n > 0, \\
1, & n = 0.
\end{cases}
\end{eqs}
Let $g \in \rc{\TN}$ be such that $g = 4$. Then there is no $c \in \posi{\TR}$ such that $f = cg$.
\end{proof}

\begin{theorem}[\uproperty{One} for $\mohsy$]
\label{MultipleOneSeparation}
$\mohsy$ has \prove{One}.
\end{theorem}

\begin{proof}
It holds that $n > c$ for $n = \ceilb{c} + 1$, for all $c \in \posi{\TR}$. \sprove{One}
\end{proof}

\begin{theorem}[\uproperty{SubHom} for $\mohsy$]
\label{MultipleSubHomogeneity}
$\mohsy$ has \prove{SubHom}. \sprove{SubHomN} \sprove{SubDivN}
\end{theorem}

\begin{proof}
Let $f \in \mohx{g}$, and $g, u \in \rc{X}$. Then there exists $c \in \posi{\TR}$ such that
\begin{eqs}
f = cg.
\end{eqs}
This implies
\begin{eqs}
uf = cug.
\end{eqs}
\sprove{SubHom} \sprove{SubHomN} \sprove{SubDivN}
\end{proof}

\begin{theorem}[\uproperty{SubComp} for $\mohsy$]
\label{MultipleSubComposability}
$\mohsy$ has \prove{SubComp}.
\end{theorem}

\begin{proof}
Let $f \in \mohx{g}$, and $\function{s}{Y}{X}$. Then there exists $c \in \posi{\TR}$ such that
\begin{eqs}
f = cg.
\end{eqs}
This implies
\begin{eqs}
(f \circ s) = c (g \circ s).
\end{eqs}
\sprove{SubComp}
\end{proof}

\section{Non-transitive dominance}
\label{NonTransitiveDominance}

\begin{definition}[Non-transitive dominance]
\defineexp{Non-transitive dominance}{dominance!non-transitive} $\nohsy$ is defined by $f \in \nohx{g}$ if and only if 
\begin{eqs}
\bra{\exists c \in \posi{\TR}: f = cg} \lor \bra{f \leq 2g},
\end{eqs}
for all $g \in \rc{X}$, and all $X \in U$, where $U$ is the class of all sets.
\end{definition}

\begin{theorem}[\uproperty{Order} for $\nohsy$]
\label{NonTransitiveOrder}
$\nohsy$ has \property{Order}.
\end{theorem}

\begin{proof}
Let $f, g \in \rc{\posi{\TN}}$ be such $f \leq g$. Then
\begin{eqs}
f \leq 2g.
\end{eqs}
\end{proof}

\begin{theorem}[\uproperty{Trans} fails for $\nohsy$]
\label{NonTransitiveTransivityFails}
$\nohsy$ does not have \property{Trans}.
\end{theorem}

\begin{proof}
Let $p \in X$, and $f, g, h \in \rc{X}$ be such that $h = 1$, $g = 2$, and
\begin{eqs}
f(x) =
\begin{cases}
4 & x \neq p, \\
3 & \text{otherwise}.
\end{cases}
\end{eqs}
Then
\begin{eqs}
f & \leq 2g = 4 \\
g & \leq 2h = 2,
\end{eqs}
so that $f \in \nohx{g}$ and $g \in \nohx{h}$. However, $f(p) > 2h(p)$, and there is no $c \in \posi{\TR}$ such that $f = ch$. Therefore $f \not\in \nohx{h}$.
\end{proof}

\begin{theorem}[\uproperty{Scale} for $\nohsy$]
\label{NonTransitiveScaleInvariance}
$\nohsy$ has \property{Scale}.
\end{theorem}

\begin{proof}
Let $\alpha \in \posi{\TR}$. Then
\begin{eqs}
f = \frac{1}{\alpha} \alpha f.
\end{eqs}
\end{proof}

\begin{theorem}[\uproperty{Local} fails for $\nohsy$]
\label{NonTransitiveLocalityFails}
$\nohsy$ does not have \property{Local}.
\end{theorem}

\begin{proof}
Let $f, g \in \rc{\TN}$ be such that $g = 1$, and
\begin{eqs}
f(n) = 
\begin{cases}
4, & n > 0, \\
3, & n = 0.
\end{cases}
\end{eqs}
Then $\restrb{f}{\setb{0}} = 3 = 3 \restrb{g}{\setb{0}}$, and $\restrb{f}{\posi{\TN}} = 4 = 4 \restrb{g}{\posi{\TN}}$. Therefore $\restrb{f}{\setb{0}} \in \noh{\setb{0}}{g}$ and $\restrb{f}{\posi{\TN}} \in \oh{\posi{\TN}}{\restr{g}{\posi{\TN}}}$. However, $f(0) > 2g(0)$, and there is no $c \in \posi{\TR}$ such that $f = cg$.
\end{proof}

\begin{theorem}[\uproperty{One} for $\nohsy$]
\label{NonTransitiveOneSeparation}
$\nohsy$ has \prove{One}.
\end{theorem}

\begin{proof}
\begin{eqs}
{} & n \in \noh{\posi{\TN}}{1} \\
\iffr & \bra{\exists c \in \posi{\TR}: \forall n \in \posi{\TN}: n = c} \lor \bra{\forall n \in \posi{\TN}: n \leq 2},
\end{eqs}
which is false.\sprove{One}
\end{proof}

\begin{theorem}[\uproperty{SubHom} for $\nohsy$]
\label{NonTransitiveSubHomogeneity}
$\nohsy$ has \prove{SubHom}. \sprove{SubHomN} \sprove{SubDivN}
\end{theorem}

\begin{proof}
Let $f \in \nohx{g}$, and $g, u \in \rc{X}$. Suppose there exists $c \in \posi{\TR}$ such that
\begin{eqs}
f = cg.
\end{eqs}
Then
\begin{eqs}
uf = cug.
\end{eqs}
Suppose $f \leq 2g$. Then
\begin{eqs}
uf \leq 2(ug).
\end{eqs}
\sprove{SubHom} \sprove{SubHomN} \sprove{SubDivN}
\end{proof}

\begin{theorem}[\uproperty{SubComp} for $\nohsy$]
\label{NonTransitiveSubComposability}
$\nohsy$ has \prove{SubComp}.
\end{theorem}

\begin{proof}
Let $f \in \nohx{g}$, and $g, u \in \rc{X}$. Suppose there exists $c \in \posi{\TR}$ such that $f = cg$. Then
\begin{eqs}
f \circ s = c (g \circ s).
\end{eqs}
Suppose $f \leq 2g$. Then
\begin{eqs}
f \circ s \leq 2(g \circ s).
\end{eqs}
\sprove{SubComp}
\end{proof}

\section{Power dominance}
\label{PowerDominance}

\begin{definition}[Clamped power]
The \define{clamped $k$-power}, where $k \in \posi{\TN}$, is a function $\function{\clpower{\dummy}{k}}{\nonn{\TR}}{\nonn{\TR}}$ such that
\begin{eqs}
\clpower{x}{k} =
\begin{cases}
x^k & x \geq 1, \\
x & x < 1.
\end{cases}
\end{eqs}
\end{definition}

\begin{definition}[Power dominance]
\defineexp{Power dominance}{dominance!power} $\wohsy$ is defined by $f \in \wohx{g}$ if and only if there exists $c \in \posi{\TR}$ and $k \in \posi{\TN}$ such that
\begin{eqs}
f \leq c\clpower{g}{k},
\end{eqs}
for all $g \in \rc{X}$, and all $X \in U$, where $U$ is the class of all sets.
\end{definition}

\begin{theorem}[\uproperty{Order} for $\wohsy$]
\label{PowerOrder}
$\wohsy$ has \prove{Order}.
\end{theorem}

\begin{proof}
Suppose $f \leq g$. Then
\begin{eqs}
f \leq 1 \clpower{g}{1} = g.
\end{eqs}
\sprove{Order}
\end{proof}

\begin{theorem}[\uproperty{Trans} for $\wohsy$]
\label{PowerTransivity}
$\wohsy$ has \prove{Trans}.
\end{theorem}

\begin{proof}
Let $f \in \wohx{g}$, $g \in \wohx{h}$, and $h \in \rc{X}$. Then there exists $c \in \posi{\TR}$ such that
\begin{eqs}
{} & f \leq c \clpower{g}{k}, \\
{} & g \leq c \clpower{h}{k}.
\end{eqs}
We may assume $c \geq 1$. Then
\begin{eqs}
f \leq c \clpower{(c \clpower{h}{k})}{k}.
\end{eqs}

\proofpart{Small values}
Suppose $x \in X$ is such that $c \clpower{h(x)}{k} < 1$. Then
\begin{eqs}
f(x) & \leq c^2 \clpower{h(x)}{k} \\
{} & \leq c^{k + 1} \clpower{h(x)}{k^2}.
\end{eqs}

\proofpart{Large values}
Suppose $x \in X$ is such that $c \clpower{h(x)}{k} \geq 1$ and $h(x) \geq 1$. Then
\begin{eqs}
f(x) & \leq c (c h(x)^k)^k \\
{} & = c^{k + 1} h(x)^{k^2} \\
{} & \leq c^{k + 1} \clpower{h(x)}{k^2}.
\end{eqs}
Suppose $x \in X$ is such that $c \clpower{h(x)}{k} \geq 1$ and $h(x) < 1$. Then
\begin{eqs}
f(x) & \leq c (c h(x))^k \\
{} & = c^{k + 1} h(x)^k \\
{} & \leq c^{k + 1} \clpower{h(x)}{k^2}.
\end{eqs}
\sprove{Trans}
\end{proof}

\begin{theorem}[\uproperty{Scale} for $\wohsy$]
\label{PowerScaleInvariance}
$\wohsy$ has \property{Scale}.
\end{theorem}

\begin{proof}
Let $f \in \rc{X}$ and $\alpha \in \posi{\TR}$. Then
\begin{eqs}
f & = \frac{1}{\alpha} \alpha f \\
{} & = \frac{1}{\alpha} \clpower{(\alpha f)}{1}.
\end{eqs}
\end{proof}

\begin{theorem}[\uproperty{Local} for $\wohsy$]
\label{PowerLocality}
$\wohsy$ has \prove{Local}.
\end{theorem}

\begin{proof}
Let $C \subset \power{X}$ be a finite cover of $X$, and suppose $\restrb{f}{D} \in \wohx{\restr{g}{D}}$ for all $D \in C$. Then there exists $c_D \in \posi{\TR}$ and $k_D \in \posi{\TN}$ such that
\begin{eqs}
\restrb{f}{D} \leq c_D \clpower{\restrb{g}{D}}{k_D},
\end{eqs}
for all $D \in C$. Let $c = \max \setb{c_D : D \in C}$ and $k = \max \setb{k_D : D \in C}$. Then
\begin{eqs}
f \leq c \clpower{g}{k}.
\end{eqs}
\sprove{Local}
\end{proof}

\begin{theorem}[\uproperty{One} for $\wohsy$]
\label{PowerOneSeparation}
$\wohsy$ has \prove{One}.
\end{theorem}

\begin{proof}
Let $c \in \posi{\TR}$. Then $n > c \clpower{1}{k}$ for $n = \ceilb{c} + 1$. \sprove{One}
\end{proof}

\begin{lemma}[Clamped power lemma]
\label{ClampedPowerLemma}
\begin{eqs}
u\clpower{g}{k} \leq \clpower{(ug)}{k},
\end{eqs}
for all $g \in \nonn{\TR}$, $u, k \in \posi{\TN}$.
\end{lemma}

\begin{proof}
\proofpart{$ug < 1$}
Since $u \geq 1$, it holds that $g < 1$. Then
\begin{eqs}
u \clpower{g}{k} & = u g \\
{} & = \clpower{(ug)}{k}.
\end{eqs}

\proofpart{$ug \geq 1$ and $g < 1$}
Then
\begin{eqs}
u \clpower{g}{k} & = ug \\
{} & \leq (ug)^k \\
{} & = \clpower{(ug)}{k}.
\end{eqs}

\proofpart{$ug \geq 1$ and $g \geq 1$}
Then
\begin{eqs}
u \clpower{g}{k} & = u g^k \\
{} & \leq u^k g^k \\
{} & = (ug)^k \\
{} & = \clpower{(ug)}{k}.
\end{eqs}
\end{proof}

\begin{theorem}[\uproperty{NSubHom} for $\wohsy$]
\label{PowerSubHomogeneity}
$\wohsy$ has \prove{NSubHom}. 
\end{theorem}

\begin{proof}
Let $f \in \wohx{g}$, and $g, u \in \rc{X}$ be such that $\image{u}{X} \subset \TN$. Then there exists $c \in \posi{\TR}$ and $k \in \posi{\TN}$ such that
\begin{eqs}
f \leq c \clpower{g}{k}.
\end{eqs}
This implies
\begin{eqs}
uf & \leq cu \clpower{g}{k} \\
{} & \leq c \clpower{(ug)}{k},
\end{eqs}
where we used \proveby{ClampedPowerLemma}.
\sprove{NSubHom} 
\end{proof}

\begin{theorem}[\uproperty{NSubDiv} fails for $\wohsy$]
\label{PowerSubHomogeneityDivNFails}
$\wohsy$ does not have \property{NSubDiv}. 
\end{theorem}

\begin{proof}
Let $f, g, u \in \rc{\posi{\TN}}$ be such that
\begin{eqs}
f(n) & = n^2, \\
g(n) & = n, \\
u(n) & = n.
\end{eqs}
Then $f \leq 1 \clpower{g}{2}$. Let $c \in \posi{\TR}$ and $k \in \posi{\TN}$. Then the inequality
\begin{eqs}
n = f / u \leq d \clpower{\bra{g / u}}{m} = d
\end{eqs}
fails for $n = \ceilb{d} + 1$.
\end{proof}

\begin{theorem}[\uproperty{SubComp} for $\wohsy$]
\label{PowerSubComposability}
$\wohsy$ has \prove{SubComp}.
\end{theorem}

\begin{proof}
Let $f \in \wohx{g}$, and $\function{s}{Y}{X}$. Then there exists $c \in \posi{\TR}$ and $k \in \posi{\TN}$ such that
\begin{eqs}
f \leq c \clpower{g}{k}.
\end{eqs}
This implies
\begin{eqs}
(f \circ s) \leq c \clpower{(g \circ s)}{k}.
\end{eqs}
\sprove{SubComp}
\end{proof}

\chapter{Partitioned sets}
\label{PartitionedSets}

In this section we develop some theory of partitioned sets. Partitioned sets occur in the theory of $\ohsy$-notation, because the equality of $\ohsy$-sets in a set $X$ --- $\ohx{f} = \ohx{g}$ --- is an equivalence relation in $\rc{X}$.\footnote{An equivalence relation is a reflexive, symmetric, and transitive relation.} 

\begin{definition}[Partitioned set]
A \define{partitioned set} is a set $X$ with an associated equivalence $\relationin{\preeq}{X}$ in $X$.
\end{definition}

\begin{note}[Conventions]
Let $X$ and $Y$ be partitioned sets. We will often shorten the word \emph{partition} into a single letter p.
\end{note}

\begin{definition}[Partition\-/preserving]
A function $\function{f}{X}{Y}$ is \define{p-preserving} if
\begin{eqs}
x_1 \preeq x_2 \implies f(x_1) \preeqb f(x_2),
\end{eqs}
for all $x_1, x_2 \in X$.
\end{definition}

\begin{note}[Homomorphisms]
The p-preserving functions are the homomorphisms of partitioned sets; they preserve the partition structure. \end{note}

\begin{definition}[Partition closure]
The \define{p-closure} on $X$ is a function $\function{\bar{\; \cdot \;}}{\power{X}}{\power{X}}$ such that
\begin{equation}
\bar{D} = \setb{x \in X : \exists d \in D: x \preeq d}.
\end{equation}
\end{definition}

\newcommand{\blackdot}[1]{\filldraw (#1) circle (0.2)}
\newcommand{\whitedot}[1]{\filldraw[white, draw=black] (#1) circle (0.2)}

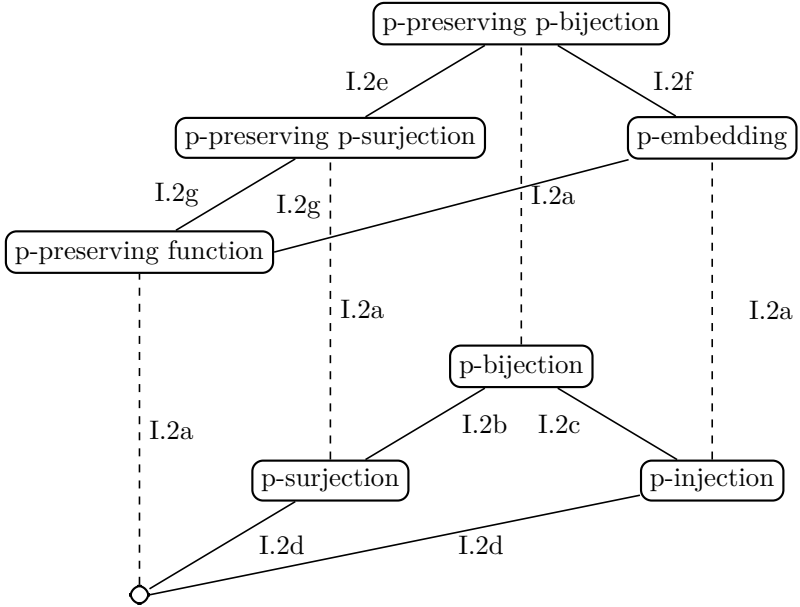
\begin{figure}
\centering
\begin{tikzpicture}[every node/.style={draw,rounded corners,thick}]
        \node (bot) at (0, 0) {};
        \node (sur) at (2.5, 1.5) {p-surjection};
        \node (bij) at (5, 3) {p-bijection};
        \node (inj) at (7.5, 1.5) {p-injection};
        \node (p_inj) at (7.5, 6) {p-embedding};
        \node (p_bij) at (5, 7.5) {p-preserving p-bijection};
        \node (p_sur) at (2.5, 6) {p-preserving p-surjection};
        \node (p) at (0, 4.5) {p-preserving function};
        
        \draw[semithick] (bot) -- node[right=10, draw=none] {\ref{PNothing}} (sur);
        
        \draw[semithick] (sur) -- node[right=10, draw=none] {\ref{PSurjectiveNotPPreservingOrPInjective}} (bij);
        
        \draw[semithick] (p_inj) -- node[right=10, draw=none] {\ref{PInjectiveAndPPreservingNotPSurjective}} (p_bij);

        \draw[semithick] (bot.east) -- node[right=20, draw=none] {\ref{PNothing}} (inj);
        
        \draw[semithick, dashed] (inj) -- node[right=10, draw=none] {\ref{PBijectiveNotPPreserving}} (p_inj);

        \draw[semithick, dashed] (bot) -- node[right, draw=none] {\ref{PBijectiveNotPPreserving}} (p);

        \draw[semithick] (p) -- node[left=10, draw=none] {\ref{PPreservingNotPInjectiveOrPSurjective}} (p_sur);

        \draw[semithick] (p_sur) -- node[left=10, draw=none] {\ref{PSurjectiveAndPPreservingNotPInjective}} (p_bij);

        \draw[semithick] (p.east) -- node[left=45, draw=none] {\ref{PPreservingNotPInjectiveOrPSurjective}} (p_inj);

        \draw[semithick] (inj) -- node[left=10, draw=none] {\ref{PInjectiveNotPPreservingOrPSurjective}} (bij);

        \draw[semithick, dashed] (sur) -- node[right, draw=none] {\ref{PBijectiveNotPPreserving}} (p_sur);

        \draw[semithick, dashed] (bij) -- node[right, draw=none] {\ref{PBijectiveNotPPreserving}} (p_bij);
\end{tikzpicture}
\caption{A Hasse diagram of functions between partitioned sets, ordered by the `generalizes' partial order. The edge labels --- which refer to \fref{PartitionPreservingProper} --- show that the generalizations are proper.}
\label{PartitionPreservingDiagram}
\end{figure}

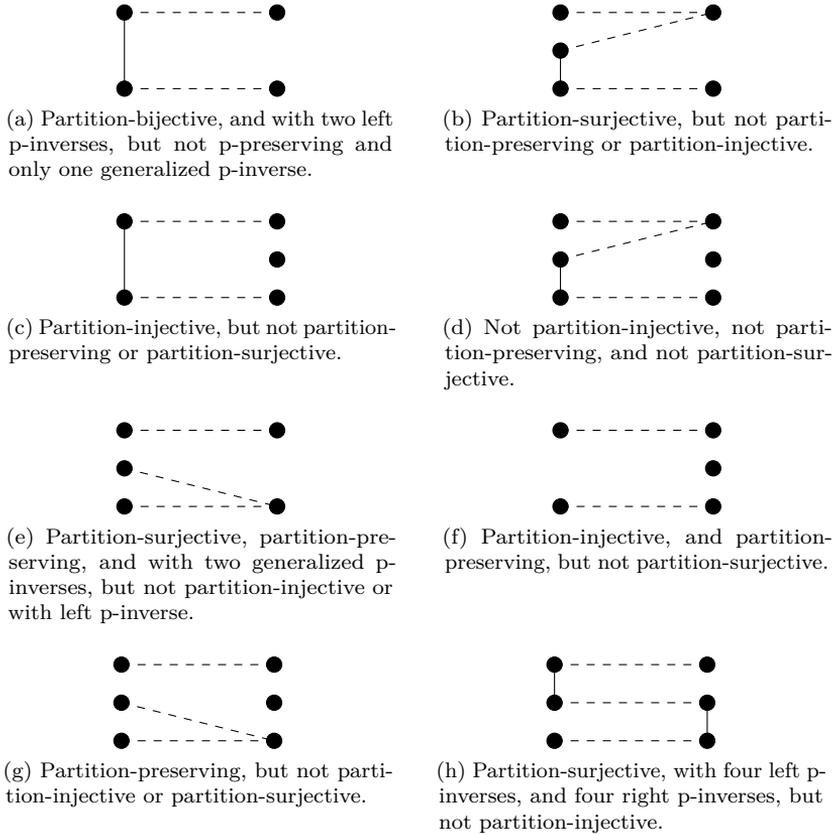
\begin{figure}
\centering
\hfill
\subfloat[Partition\-/bijective, and with two left p-inverses, but not p\-/preserving and only one generalized p-inverse.]
{
\makebox[.4\textwidth]{
\begin{tikzpicture}[scale = 0.5]
\blackdot{0, 2};
\blackdot{0, 0};
\blackdot{4, 0};
\blackdot{4, 2};
\draw[dashed] (0, 2) -- (4, 2);
\draw[dashed] (0, 0) -- (4, 0);
\draw (0, 0) -- (0, 2);
\end{tikzpicture}
}
\label{PBijectiveNotPPreserving}
}
\hfill
\subfloat[Partition\-/surjective, but not partition\-/preserving or partition\-/injective.]
{
\makebox[.4\textwidth]{
\begin{tikzpicture}[scale = 0.5]
\blackdot{0, 1};
\blackdot{0, 0};
\blackdot{4, 0};
\blackdot{4, 2};
\blackdot{0, 2};
\draw[dashed] (0, 1) -- (4, 2);
\draw[dashed] (0, 0) -- (4, 0);
\draw[dashed] (0, 2) -- (4, 2);
\draw (0, 1) -- (0, 0);
\end{tikzpicture}
}
\label{PSurjectiveNotPPreservingOrPInjective}
}
\hfill\null

\hfill
\subfloat[Partition\-/injective, but not partition\-/preserving or partition\-/surjective.]
{
\makebox[.4\textwidth]{
\begin{tikzpicture}[scale = 0.5]
\blackdot{0, 2};
\blackdot{0, 0};
\blackdot{4, 0};
\blackdot{4, 1};
\blackdot{4, 2};
\draw[dashed] (0, 2) -- (4, 2);
\draw[dashed] (0, 0) -- (4, 0);
\draw (0, 0) -- (0, 2);
\end{tikzpicture}
}
\label{PInjectiveNotPPreservingOrPSurjective}
}
\hfill
\subfloat[Not partition\-/injective, not partition\-/preserving, and not partition\-/surjective.]
{
\makebox[.4\textwidth]{
\begin{tikzpicture}[scale = 0.5]
\blackdot{0, 1};
\blackdot{0, 2};
\blackdot{0, 0};
\blackdot{4, 0};
\blackdot{4, 1};
\blackdot{4, 2};
\draw[dashed] (0, 1) -- (4, 2);
\draw[dashed] (0, 0) -- (4, 0);
\draw[dashed] (0, 2) -- (4, 2);
\draw (0, 0) -- (0, 1);
\end{tikzpicture}
}
\label{PNothing}
}
\hfill\null

\hfill
\subfloat[Partition\-/surjective, partition\-/preserving, and with two generalized p-inverses, but not partition\-/injective or with left p-inverse.]
{
\makebox[.4\textwidth]{
\begin{tikzpicture}[scale = 0.5]
\blackdot{0, 0};
\blackdot{0, 1};
\blackdot{4, 0};
\blackdot{4, 2};
\blackdot{0, 2};
\draw[dashed] (0, 2) -- (4, 2);
\draw[dashed] (0, 0) -- (4, 0);
\draw[dashed] (0, 1) -- (4, 0);
\end{tikzpicture}
}
\label{PSurjectiveAndPPreservingNotPInjective}
}
\hfill
\subfloat[Partition\-/injective, and partition\-/preserving, but not partition\-/surjective.]
{
\makebox[.4\textwidth]{
\begin{tikzpicture}[scale = 0.5]
\blackdot{0, 0};
\blackdot{0, 2};
\blackdot{4, 0};
\blackdot{4, 1};
\blackdot{4, 2};
\draw[dashed] (0, 2) -- (4, 2);
\draw[dashed] (0, 0) -- (4, 0);
\end{tikzpicture}
}
\label{PInjectiveAndPPreservingNotPSurjective}
}
\hfill\null

\hfill
\subfloat[Partition\-/preserving, but not partition\-/injective or partition\-/surjective.]
{
\makebox[.4\textwidth]{
\begin{tikzpicture}[scale = 0.5]
\blackdot{0, 0};
\blackdot{0, 2};
\blackdot{0, 1};
\blackdot{4, 0};
\blackdot{4, 1};
\blackdot{4, 2};
\draw[dashed] (0, 2) -- (4, 2);
\draw[dashed] (0, 0) -- (4, 0);
\draw[dashed] (0, 1) -- (4, 0);
\end{tikzpicture}
}
\label{PPreservingNotPInjectiveOrPSurjective}
}
\hfill
\subfloat[Partition\-/surjective, with four left p-inverses, and four right p-inverses, but not partition\-/injective.]
{
\makebox[.4\textwidth]{
\begin{tikzpicture}[scale = 0.5]
\blackdot{0, 0};
\blackdot{0, 1};
\blackdot{0, 2};
\blackdot{4, 0};
\blackdot{4, 1};
\blackdot{4, 2};
\draw[dashed] (0, 2) -- (4, 2);
\draw[dashed] (0, 1) -- (4, 1);
\draw[dashed] (0, 0) -- (4, 0);
\draw (4, 0) -- (4, 1);
\draw (0, 1) -- (0, 2);
\end{tikzpicture}
}
\label{PSurjectiveWithPInversesNotPInjective}
}
\hfill\null
\caption{Diagrams to show that the various definitions of functions on partitioned sets are not equivalent. Equivalent dots are connected with a solid line.}
\label{PartitionPreservingProper}
\end{figure}

\section{Generalized partition-inverse}

\begin{definition}[Generalized partition\-/inverse]
A \define{generalized p-inverse} of $\function{f}{X}{Y}$ is $\function{\ginvs{f}}{Z}{X}$ such that $\image{f}{X} \subset Z \subset Y$, and
\begin{eqs}
f(\ginv{f}{f(x)}) \preeqb f(x),
\end{eqs}
for all $x \in X$.
\end{definition}

\begin{theorem}[Construction for a generalized p-inverse]
\label{ConstructionForGeneralizedInverse}
Let \hfill \\ $\function{f}{X}{Y}$. Then $\function{\hat{f}}{Z}{X}$ is a generalized p-inverse of $f$ if and only if
\begin{eqs}
\hat{f}(y) \in \preimage{f}{\pclosure{\setb{y}}},
\end{eqs}
for all $y \in \image{f}{X}$.
\end{theorem}

\begin{proof}
It holds that
\begin{eqs}
{} & \hat{f}(f(x)) \in \preimage{f}{\pclosure{\setb{f(x)}}} \\
\iffr & f(\hat{f}(f(x))) \in \pclosure{\setb{f(x)}} \\
\iffr & f(\hat{f}(f(x))) \preeqb f(x),
\end{eqs}
for all $x \in X$.
\end{proof}

\begin{note}[Generalized partition\-/inverse exists]
\thref{ConstructionForGeneralizedInverse} shows that a generalized p-inverse always exists.
\end{note}

\begin{note}[Generalized partition\-/inverse may not be unique]
\fref{PSurjectiveWithPInversesNotPInjective} shows that a function $\function{f}{X}{Y}$ can have many non-equivalent generalized p-inverses.
\end{note}

\section{Partition-injectivity and left partition-inverse}

\begin{definition}[Partition\-/injectivity]
A function $\function{f}{X}{Y}$ is \define{p-injective}, if 
\begin{eqs}
f(x_1) \preeqb f(x_2) \implies x_1 \preeq x_2, 
\end{eqs}
for all $x_1, x_2 \in X$.
\end{definition}

\begin{definition}[Left partition\-/inverse]
A \define{left p-inverse} of $\function{f}{X}{Y}$ is a function $\function{\linvs{f}}{Z}{X}$ such that $\image{f}{X} \subset Z \subset Y$ and
\begin{eqs}
\linv{f}{f(x)} \preeq x,
\end{eqs}
for all $x \in X$.
\end{definition}

\begin{theorem}[Generalized partition\-/inverse is a left partition\-/inverse for a partition\-/injective function]
\label{GeneralizedPInverseIsLeftPInverseForPInjective}
Let $\function{f}{X}{Y}$ be p-injective. Then a generalized p-inverse of $f$ is a left p-inverse of $f$.
\end{theorem}

\begin{proof}
Let $\function{\ginvs{f}}{Z}{X}$ be a generalized p-inverse of $f$. By definition,
\begin{eqs}
f(\ginv{f}{f(x)}) \preeqb f(x),
\end{eqs}
for all $x \in X$. Since $f$ is p-injective,
\begin{eqs}
\ginv{f}{f(x)} \preeq x,
\end{eqs}
for all $x \in X$. Therefore $\ginvs{f}$ is a left p-inverse of $f$.
\end{proof}

\begin{note}[Generalized partition\-/inverse may not be a left partition\-/inverse]
\fref{PSurjectiveAndPPreservingNotPInjective} shows that a generalized p-inverse may not be a left p-inverse.
\end{note}

\begin{theorem}[Left partition\-/inverse is a generalized partition\-/inverse for a partition\-/preserving function]
\label{LeftPInverseIsGeneralizedPInverseForPPreserving}
Let $\function{f}{X}{Y}$ be p-preserving. Then a left p-inverse of $f$ is a generalized p-inverse of $f$.
\end{theorem}

\begin{proof}
Let $\function{\linvs{f}}{Z}{X}$ be a left p-inverse of $f$. By definition,
\begin{eqs}
\linv{f}{f(x)} \preeq x,
\end{eqs}
for all $x \in X$. Since $f$ is p-preserving,
\begin{eqs}
f(\linv{f}{f(x)}) \preeqb f(x),
\end{eqs}
for all $x \in X$. Therefore $\linvs{f}$ is a generalized p-inverse of $f$.
\end{proof}

\begin{note}[Left partition\-/inverse may not be a generalized partition\-/inverse]
\fref{PBijectiveNotPPreserving} shows a non-p-preserving function which has two left p-inverses, but only one generalized p-inverse; a left p-inverse may not be a generalized p-inverse.
\end{note}

\begin{theorem}[Partition\-/injectivity is equivalent to having a p-preserving left p-inverse]
\label{PInjectivityIsEquivalentToPPreservingLeftInverse}
Let $\function{f}{X}{Y}$. Then $f$ is p-injective if and only if $f$ has a p-preserving left p-inverse.
\end{theorem}

\begin{proof}
\proofpart{$\implies$}
Let $\function{\hat{f}}{\image{f}{X}}{X}$ be such that
\begin{eqs}
\hat{f}(y) \in \preimage{f}{\pclosure{\setb{y}}}.
\end{eqs}
Then $\hat{f}$ is a generalized p-inverse of $f$ by \thref{ConstructionForGeneralizedInverse}. Since $f$ is p-injective, $\hat{f}$ is a left p-inverse by \thref{GeneralizedPInverseIsLeftPInverseForPInjective}. It holds that
\begin{eqs}
{} & x_1, x_2 \in \preimage{f}{\pclosure{\setb{y}}} \\
\impliesr & f(x_1), f(x_2) \in \pclosure{\setb{y})} \\
\impliesr & f(x_1) \preeqb y \preeqb f(x_2) \\
\impliesr & x_1 \preeq x_2, && \why{$f$ p-injective}
\end{eqs}
for all $x_1, x_2 \in X$ and $y \in \image{f}{X}$. That is, all elements in $\preimage{f}{\pclosure{\setb{y}}}$ are equivalent to each other, for all $y \in \image{f}{X}$. Let $y_1, y_2 \in \image{f}{X}$ be such that $y_1 \preeqb y_2$. Then
\begin{eqs}
\pclosure{\setb{y_1}} = \pclosure{\setb{y_2}}.
\end{eqs}
This implies that
\begin{eqs}
{} & \hat{f}(y_1), \hat{f}(y_2) \in \preimage{f}{\pclosure{\setb{y_1}}}.
\end{eqs}
However, since all elements in $\preimage{f}{\pclosure{\setb{y_1}}}$ are equivalent to each other,
\begin{eqs}
\hat{f}(y_1) \preeq \hat{f}(y_2).
\end{eqs}
Therefore $\hat{f}$ is p-preserving. 

\proofpart{$\impliedby$}
Let $\function{\linvs{f}}{Z}{X}$ be a p-preserving left p-inverse of $f$. Then
\begin{eqs}
{} & f(x_1) \preeqb f(x_2) \\
\impliesr & \linv{f}{f(x_1)} \preeq \linv{f}{f(x_2)} && \why{$\linvs{f}$ p-preserving} \\
\impliesr & x_1 \preeq x_2, && \why{$\linvs{f}$ left p-inverse of $f$}
\end{eqs}
for all $x_1, x_2 \in X$. Therefore $f$ is p-injective.
\end{proof}

\begin{note}[Non\-/partition\-/preserving left p-inverse]
\fref{PSurjectiveWithPInversesNotPInjective} shows that the existence of a non-p-preserving left p-inverse does not imply p-injectivity.
\end{note}

\begin{theorem}[Left p-inverses are equivalent on the image]
\label{LeftPInversesAreEquivalent}
The left p-inverses of $\function{f}{X}{Y}$ are equivalent to each other on $\image{f}{X}$.
\end{theorem}

\begin{proof}
Suppose $\function{g}{Z}{X}$ and $\function{g}{W}{X}$ are left p-inverses of $f$. For each $y \in \image{f}{X}$, there exists $x_y \in X$ such that $y = f(x_y)$. Since $g$ and $h$ are left p-inverses,
\begin{eqs}
g(y) = g(f(x_y)) \preeq x_y \preeq h(f(x_y)) = h(y),
\end{eqs}
for all $y \in \image{f}{X}$.
\end{proof}

\section{Partition-surjectivity and right partition-inverse}

\begin{definition}[Partition\-/surjectivity]
A function $\function{f}{X}{Y}$ is \define{p-surjective}, if 
\begin{eqs}
\pclosure{\image{f}{X}} = Y.
\end{eqs}
\end{definition}

\begin{note}[Right partition\-/inverse]
A \define{right p\-/inverse} of $\function{f}{X}{Y}$ is $\function{\rinvs{f}}{Y}{X}$ such that
\begin{eqs}
f(\rinv{f}{y}) \preeqb y,
\end{eqs}
for all $y \in Y$.
\end{note}

\begin{theorem}[Partition\-/surjectivity is equivalent to having a right p-inverse]
\label{PSurjectivityIsEquivalentToHavingRightInverse}
Let $\function{f}{X}{Y}$. Then $f$ is p-surjective if and only if $f$ has a right p-inverse.
\end{theorem}

\begin{proof}
\proofpart{$\implies$}
Since $f$ is p-surjective,
\begin{eqs}
\preimage{f}{\pclosure{\setb{y}}} \neq \emptyset,
\end{eqs}
for all $y \in Y$. Let $\function{\hat{f}}{Y}{X}$ be such that
\begin{eqs}
\hat{f}(y) \in \preimage{f}{\pclosure{\setb{y}}}.
\end{eqs}
Then
\begin{eqs}
{} & f(\hat{f}(y)) \in \pclosure{\setb{y}} \\
\iffr & f(\hat{f}(y)) \preeqb y,
\end{eqs}
for all $y \in Y$. Therefore $\hat{f}$ is a right p-inverse of $f$.

\proofpart{$\impliedby$}
Let $\function{\rinvs{f}}{Y}{X}$ be a right p-inverse of $f$. Then
\begin{eqs}
f(\rinv{f}{y}) \preeqb y,
\end{eqs}
for all $y \in Y$. Therefore $f$ is p-surjective.
\end{proof}

\begin{theorem}[Right p-inverse is order-reflecting for order-preserving]
\label{RightPInverseIsOrderReflectingForOrderPreserving}
Let $\function{f}{X}{Y}$ be order-preserving. Then a right p-inverse of $f$ is order-reflecting.
\end{theorem}

\begin{proof}
Let $\function{\rinvs{f}}{Y}{X}$ be a right p-inverse of $f$. Then
\begin{eqs}
{} & \rinv{f}{y_1} \preleq \rinv{f}{y_2} \\
{} \impliesr & f(\rinv{f}{y_1}) \preleqb f(\rinv{f}{y_2}) && \why{$f$ order-preserving} \\
{} \impliesr & y_1 \preleqb y_2, && \why{$\rinvs{f}$ right p-inverse of $f$}
\end{eqs}
for all $y_1, y_2 \in Y$. Therefore $\rinvs{f}$ is order-reflecting.
\end{proof}

\section{Partition-bijectivity and partition-inverse}

\begin{definition}[Partition\-/bijectivity]
A function $\function{f}{X}{Y}$ is \define{p-bijective}, if it is both p-injective and p-surjective.
\end{definition}

\begin{definition}[Partition\-/inverse]
A \define{p-inverse of $f$} is a function $\function{\hat{f}}{Y}{X}$ which is both a left p-inverse and a right p-inverse of $f$. 
\end{definition}

\begin{theorem}[Partition\-/bijectivity is equivalent to having a p-preserving p-inverse]
\label{PBijectivityIsEquivalentToHavingPPreservingInverse}
Let $\function{f}{X}{Y}$. Then $f$ is p-bijective if and only if $f$ has a p-preserving p-inverse.
\end{theorem}

\begin{proof}
\proofpart{$\implies$}
A p-preserving left p-inverse $\function{\linvs{f}}{\image{f}{X}}{X}$ of $f$ exists by \thref{PInjectivityIsEquivalentToPPreservingLeftInverse}. We may extend $\linvs{f}$ to a p-preserving left p-inverse of $f$ on $\pclosure{\image{f}{X}}$. Since $f$ is p-surjective, $\pclosure{\image{f}{X}} = Y$. We then notice that the extended $\linvs{f}$ satisfies the construction of the right p-inverse $\rinvs{f}$ in \thref{PSurjectivityIsEquivalentToHavingRightInverse}.

\proofpart{$\impliedby$}
This follows from \thref{PInjectivityIsEquivalentToPPreservingLeftInverse} and \thref{PSurjectivityIsEquivalentToHavingRightInverse}.
\end{proof}

\begin{theorem}[Partition\-/bijectivity implies a unique p-inverse]
Let $\function{f}{X}{Y}$ be p-bijective. Then $f$ has a unique p-inverse up to an equivalence.
\label{PBijectivityImpliesUniquePInverse}
\end{theorem}

\begin{proof}
Since $f$ is p-bijective, there exists a p-inverse $\function{g}{Y}{X}$ of $f$ by \thref{PBijectivityIsEquivalentToHavingPPreservingInverse}. Suppose there exists another p-inverse $\function{h}{Y}{X}$ of $f$. Then
\begin{eqs}
{} & f(g(y)) \preeqb y \preeqb f(h(y)) \\
\impliesr & g(y) \preeq h(y), && \why{$f$ p-injective}
\end{eqs}
for all $y \in Y$.
\end{proof}

\section{Partition-embedding}

\begin{definition}[Partition-embedding]
A function $\function{f}{X}{Y}$ is \define{p-embedding}, if it is both p-preserving and p-injective.
\end{definition}

\begin{theorem}[Partition-embedding by preimages]
\label{PEmbeddingByPreimages}
Let $\function{f}{X}{Y}$. Then $f$ is a p-embedding if and only if
\begin{eqs}
\preimage{f}{\ppclosure{f(x)}} = \ppclosure{x},
\end{eqs}
for all $x \in X$.
\end{theorem}

\begin{proof}
\proofpart{$\implies$}
It holds that
\begin{eqs}
{} & z \in \preimage{f}{\ppclosure{f(x)}} \\
\iffr & f(z) \in \ppclosure{f(x)} \\
\iffr & f(z) \preeqb f(x) \\
\iffr & z \preeq x && \why{$f$ p-embedding} \\
\iffr & z \in \ppclosure{x},
\end{eqs}
for all $x \in X$.

\proofpart{$\impliedby$}
To show that $f$ is p-embedding,
\begin{eqs}
{} & x_1 \preeq x_2 \\
\iffr & \ppclosure{x_1} = \ppclosure{x_2} \\
\iffr & \preimage{f}{\ppclosure{f(x_1)}} = \preimage{f}{\ppclosure{f(x_2)}} && \why{assumption} \\
\iffr & \ppclosure{f(x_1)} = \ppclosure{f(x_2)} && \why{$f$ p-embedding} \\
\iffr & f(x_1) \preeqb f(x_2),
\end{eqs}
for all $x_1, x_2 \in X$. 
\end{proof}

\chapter{Preordered sets}
\label{PreorderedSets}

In this section we develop some theory of preorders, generalizing the theory of partial orders. Preordered sets occur in the theory of $\ohsy$-notation, because $\ohx{f}$ is a principal down-set of a preorder in $\rc{X}$. 

\begin{definition}[Preordered set]
A \define{preordered set} is a set $X$ with an associated preorder $\relationin{\preleq}{X}$.
\end{definition}

\begin{note}[Conventions]
Let $X$ and $Y$ be preordered sets. 
\end{note}

\begin{definition}[Order-preserving]
A function $\function{f}{X}{Y}$ is \define{order-preserving}, or \define{monotone}, if
\begin{eqs}
x_1 \preleq x_2 \implies f(x_1) \preleqb f(x_2),
\end{eqs}
for all $x_1, x_2 \in X$.
\end{definition}

\begin{note}[Homomorphisms]
The order-preserving functions are the homomorphisms of preordered sets; they preserve the preorder structure. 
\end{note}

\begin{definition}[Order-reflecting]
A function $\function{f}{X}{Y}$ is \define{order-reflecting}, if
\begin{eqs}
f(x_1) \preleqb f(x_2) \implies x_1 \preleq x_2,
\end{eqs}
for all $x_1, x_2 \in X$.
\end{definition}

\begin{definition}[Order-embedding]
A function $\function{f}{X}{Y}$ is \define{order-embedding}, if it is both order-preserving and order-reflecting.
\end{definition}

\begin{definition}[Order-isomorphism]
A function $\function{f}{X}{Y}$ is an \define{order-isomorphism}, if it is p-surjective and order-embedding.
\end{definition}

The relationships between such functions is given in \figref{OrderPreservingDiagram}, with \figref{OrderPreservingProper} showing that the inclusions are proper.

\section{Order and partitions}

\begin{definition}[Induced equivalence]
Given a preorder $\relationin{\preleq}{X}$, the \define{induced equivalence} is $\relationin{\preeq}{X}$ such that 
\begin{eqs}
x_1 \preeq x_2 \iff x_1 \preleq x_2 \textrm{ and } x_2 \preleq x_1.
\end{eqs}
\end{definition}

\begin{note}[Relation to partitioned sets]
The induced equivalence $\preeq$ on a preordered set $(X, \preleq)$ partitions $X$; the theory of partitioned sets interacts with the theory of preordered sets.
\end{note}

\begin{theorem}[Order-preserving is partition\-/preserving]
\label{OrderPreservingIsPartitionPreserving}
Let \hfill \\ $\function{f}{X}{Y}$ be order-preserving. Then $f$ is p-preserving.
\end{theorem}

\begin{proof}
Let $x_1, x_2 \in X$ be such that $x_1 \preeq x_2$. Then
\begin{eqs}
x_1 \preleq x_2 \implies f(x_1) \preleqb f(x_2).
\end{eqs}
Similarly,
\begin{eqs}
x_2 \preleq x_1 \implies f(x_2) \preleqb f(x_1).
\end{eqs}
It follows that
\begin{eqs}
f(x_1) \preeqb f(x_2).
\end{eqs}
Therefore $f$ is p-preserving.
\end{proof}

\begin{theorem}[Order-reflecting implies partition\-/injective]
\label{OrderReflectingImpliesInjective}
Let $\function{f}{X}{Y}$ be order-reflecting. Then $f$ is partition\-/injective.
\end{theorem}

\begin{proof}
It holds that
\begin{eqs}
{} & f(x_1) \preeqb f(x_2) \\
\impliesr & f(x_1) \preleqb f(x_2) \textrm{ and } f(x_2) \preleqb f(x_1) \\
\impliesr & x_1 \preleq x_2 \textrm{ and } x_2 \preleq x_1 \\
\impliesr & x_1 \preeq x_2,
\end{eqs}
for all $x_1, x_2 \in X$. 
\end{proof}

\section{Order-preserving functions and down-sets}

\begin{definition}[Generated down\-/set]
The \define{generated down\-/set} in $X$ is a function $\function{\downset{X}{\dummy}}{\power{X}}{\power{X}}$ such that
\begin{equation}
\downset{X}{D} = \setb{x \in X : \exists d \in D : x \preleq d}
\end{equation}
\end{definition}

\begin{definition}[Down\-/set]
A subset $D \subset X$ is a \define{down\-/set} of $X$, if $D = \downset{X}{D}$.
\end{definition}

\begin{definition}[Principal down\-/set]
A subset $D \subset X$ is a \define{principal down\-/set} of $X$, if there exists $d \in X$ such that $D = \pdownset{X}{d}$.
\end{definition}

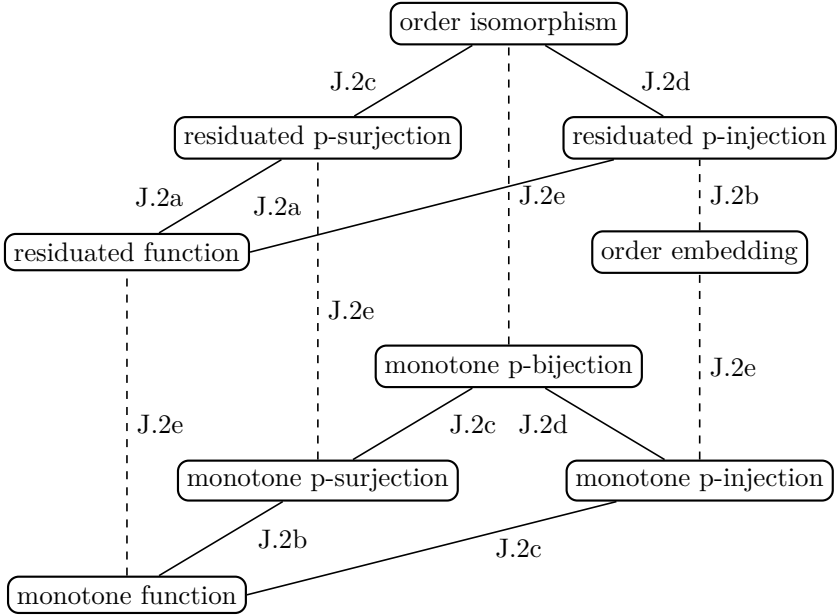
\begin{figure}
\centering
\begin{tikzpicture}[every node/.style={draw,rounded corners,thick}]
        \node (mon) at (0, 0) {monotone function};
        \node (mon_sur) at (2.5, 1.5) {monotone p-surjection};
        \node (mon_bij) at (5, 3) {monotone p-bijection};
        \node (mon_inj) at (7.5, 1.5) {monotone p-injection};
        \node (ord_emb) at (7.5, 4.5) {order embedding};
        \node (res_inj) at (7.5, 6) {residuated p-injection};
        \node (iso) at (5, 7.5) {order isomorphism};
        \node (res_sur) at (2.5, 6) {residuated p-surjection};
        \node (res) at (0, 4.5) {residuated function};
        
        \draw[semithick] (mon) -- node[right=10, draw=none] {\ref{OrderEmbeddingNotResiduatedOrSurjective}} (mon_sur);
        
        \draw[semithick] (mon_sur) -- node[right=10, draw=none] {\ref{ResiduatedAndSurjectiveNotInjective}} (mon_bij);
        
        \draw[semithick] (mon.east) -- node[right=20, draw=none] {\ref{ResiduatedAndSurjectiveNotInjective}} (mon_inj);
        
        \draw[semithick, dashed] (mon_inj) -- node[right, draw=none] {\ref{OrderPreservingAndBijectiveNotOrderReflectingOrResiduated}} (ord_emb);

        \draw[semithick, dashed] (ord_emb) -- node[right, draw=none] {\ref{OrderEmbeddingNotResiduatedOrSurjective}} (res_inj);

        \draw[semithick] (res_inj) -- node[right=10, draw=none] {\ref{ResiduatedAndInjectiveNotSurjective}} (iso);

        \draw[semithick, dashed] (mon) -- node[right, draw=none] {\ref{OrderPreservingAndBijectiveNotOrderReflectingOrResiduated}} (res);

        \draw[semithick] (res) -- node[left=10, draw=none] {\ref{ResiduatedNotSurjectiveOrInjective}} (res_sur);

        \draw[semithick] (res_sur) -- node[left=10, draw=none] {\ref{ResiduatedAndSurjectiveNotInjective}} (iso);

        \draw[semithick] (res.east) -- node[left=45, draw=none] {\ref{ResiduatedNotSurjectiveOrInjective}} (res_inj);

        \draw[semithick] (mon_inj) -- node[left=10, draw=none] {\ref{ResiduatedAndInjectiveNotSurjective}} (mon_bij);

        \draw[semithick, dashed] (mon_sur) -- node[right, draw=none] {\ref{OrderPreservingAndBijectiveNotOrderReflectingOrResiduated}} (res_sur);

        \draw[semithick, dashed] (mon_bij) -- node[right, draw=none] {\ref{OrderPreservingAndBijectiveNotOrderReflectingOrResiduated}} (iso);
\end{tikzpicture}
\caption{A Hasse diagram of order-preserving functions between preordered sets, ordered by the `generalizes' partial order. Original figure was created by David Wilding, and is used here with his permission. The edge labels --- which refer to \fref{OrderPreservingProper} --- have been added to show that the generalizations are proper.
}
\label{OrderPreservingDiagram}
\end{figure}

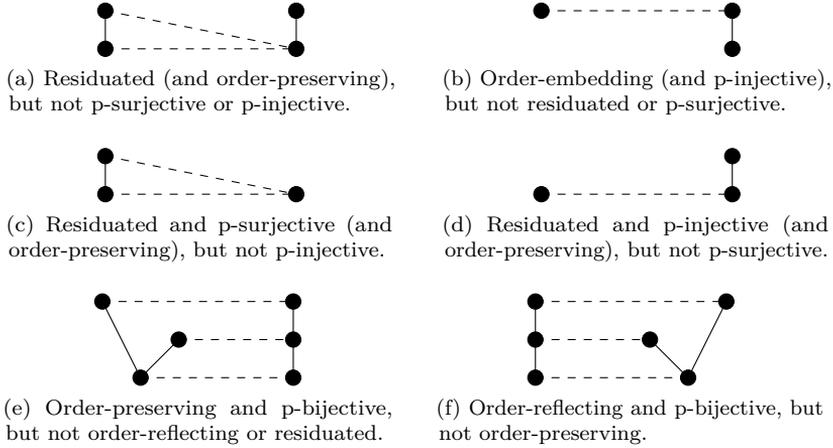
\begin{figure}
\hfill
\subfloat[Residuated (and order\-/preserving), but not p\-/surjective or p\-/injective.]
{
\makebox[.4\textwidth]{
\begin{tikzpicture}[scale = 0.5]
\blackdot{0, 0};
\blackdot{0, 1};
\blackdot{5, 0};
\blackdot{5, 1};
\draw[dashed] (0, 0) -- (5, 0);
\draw[dashed] (0, 1) -- (5, 0);
\draw (0, 0) -- (0, 1);
\draw (5, 0) -- (5, 1);
\end{tikzpicture}
}
\label{ResiduatedNotSurjectiveOrInjective}
}
\hfill
\subfloat[Order\-/embedding (and p\-/injective), but not residuated or p\-/surjective.]
{
\makebox[.4\textwidth]{
\begin{tikzpicture}[scale = 0.5]
\blackdot{0, 1};
\blackdot{5, 0};
\blackdot{5, 1};
\draw[dashed] (0, 1) -- (5, 1);
\draw (5, 0) -- (5, 1);
\end{tikzpicture}
}
\label{OrderEmbeddingNotResiduatedOrSurjective}
}
\hfill\null

\hfill
\subfloat[Residuated and p-surjective (and order-preserving), but not p-injective.]
{
\makebox[.4\textwidth]{
\begin{tikzpicture}[scale = 0.5]
\blackdot{0, 0};
\blackdot{0, 1};
\blackdot{5, 0};
\draw[dashed] (0, 0) -- (5, 0);
\draw[dashed] (0, 1) -- (5, 0);
\draw (0, 0) -- (0, 1);
\end{tikzpicture}
}
\label{ResiduatedAndSurjectiveNotInjective}
}
\hfill
\subfloat[Residuated and p-injective (and order-preserving), but not p-surjective.]
{
\makebox[.4\textwidth]{
\begin{tikzpicture}[scale = 0.5]
\blackdot{0, 0};
\blackdot{5, 0};
\blackdot{5, 1};
\draw[dashed] (0, 0) -- (5, 0);
\draw (5, 0) -- (5, 1);
\end{tikzpicture}
}
\label{ResiduatedAndInjectiveNotSurjective}
}
\hfill\null

\hfill
\subfloat[Order\-/preserving and p\-/bijective, but not order\-/reflecting or residuated.]
{
\makebox[.4\textwidth]{
\begin{tikzpicture}[scale = 0.5]
\blackdot{5, 0};
\blackdot{5, 1};
\blackdot{5, 2};
\blackdot{0, 2};
\blackdot{1, 0};
\blackdot{2, 1};
\draw[dashed] (0, 2) -- (5, 2);
\draw[dashed] (1, 0) -- (5, 0);
\draw[dashed] (2, 1) -- (5, 1);
\draw (5, 0) -- (5, 1);
\draw (5, 1) -- (5, 2);
\draw (1, 0) -- (0, 2);
\draw (1, 0) -- (2, 1);
\end{tikzpicture}
}
\label{OrderPreservingAndBijectiveNotOrderReflectingOrResiduated}
}
\hfill
\subfloat[Order\-/reflecting and p\-/bijective, but not order\-/preserving.]
{
\makebox[.4\textwidth]{
\begin{tikzpicture}[scale = 0.5]
\blackdot{0, 0};
\blackdot{0, 1};
\blackdot{0, 2};
\blackdot{3, 1};
\blackdot{4, 0};
\blackdot{5, 2};
\draw[dashed] (0, 0) -- (4, 0);
\draw[dashed] (0, 1) -- (3, 1);
\draw[dashed] (0, 2) -- (5, 2);
\draw (0, 0) -- (0, 1);
\draw (0, 1) -- (0, 2);
\draw (4, 0) -- (3, 1);
\draw (4, 0) -- (5, 2);
\end{tikzpicture}
}
\label{OrderReflectingAndBijectiveNotOrderPreserving}
}
\hfill\null
\caption{Hasse diagrams to show that the various definitions of order-preserving functions are not equivalent. Since all orders here are partial, the p-prefix is redundant.}
\label{OrderPreservingProper}
\end{figure}

\begin{theorem}[Alternative definition of down\-/sets]
\label{AlternativeDownSet}
Let $D \subset X$. Then $D$ is a down\-/set if and only if
\begin{equation}
\forall x \in X, \forall d \in D: x \preleq d \implies x \in D.
\end{equation}
\end{theorem}

\begin{proof}
It always holds that $D \subset \setb{x \in X : \exists d \in D: x \preleq d}$. Then
\begin{eqs}
{} & D = \setb{x \in X : \exists d \in D: x \preleq d} \\
\iffr & D \supset \setb{x \in X : \exists d \in D: x \preleq d} \\
\iffr & \forall x \in X: \bra{\exists d \in D : x \preleq d} \implies x \in D \\
\iffr & \forall x \in X: \bra{\forall d \in D : x \not\preleq d} \textrm{ or } x \in D \\
\iffr & \forall x \in X, \forall d \in D : \bra{x \not\preleq d \textrm { or } x \in D} \\
\iffr & \forall x \in X, \forall d \in D : x \preleq d \implies x \in D.
\end{eqs}
\end{proof}

\begin{theorem}[Order-preservation by preimages]
\label{OrderPreservationByPreimages}
Let $\function{f}{X}{Y}$. Then $f$ is order-preserving if and only if each preimage of a down\-/set of $Y$ is a down\-/set of $X$.
\end{theorem}

\begin{proof}
\proofpart{$\implies$}
Let $S \subset Y$, $D_Y = \downset{Y}{S}$, and $D_X = \preimage{f}{D_Y}$. Let $x_1 \in X$, and $x_2 \in D_X$ be such that $x_1 \preleq x_2$. Since $f$ is order-preserving, $f(x_1) \preleqb f(x_2)$. It holds that $f(x) \in D_Y$, for all $x \in D_X$. Since $f(x_1) \preleqb f(x_2) \in D_Y$, and $D_Y$ is a down\-/set, $f(x_1) \in D_Y$. This implies $x_1 \in D_X$. Therefore $D_X$ is a down\-/set.

\proofpart{$\impliedby$}
Let $x_1, x_2 \in X$. Then there exists $D \subset X$ such that
\begin{equation}
\preimage{f}{\downset{Y}{\image{f}{\setb{x_2}}}}  = \downset{X}{D}.
\end{equation}
Let $x \in X$. Then
\begin{eqs}
{} & x \in \preimage{f}{\pdownset{Y}{f(x_2)}} \\
\iffr & f(x) \in \pdownset{Y}{f(x_2)} \\
\iffr & f(x) \preleqb f(x_2).
\end{eqs}
In particular, $x_2 \in \preimage{f}{\downset{Y}{f(x_2)}} = \downset{X}{D}$. Then
\begin{eqs}
{} & x_1 \preleq x_2 \\
\impliesr & x_1 \in \downset{X}{D} \\
\impliesr & x_1 \in \preimage{f}{\pdownset{Y}{f(x_2)}} \\
\impliesr & f(x_1) \in \pdownset{Y}{f(x_2)} \\
\impliesr & f(x_1) \preleqb f(x_2).
\end{eqs}
Therefore $f$ is order-preserving.
\end{proof}

\begin{theorem}[Order-preservation by images]
\label{OrderPreservationByImages}
Let $\function{f}{X}{Y}$. Then $f$ is order-preserving if and only if
\begin{eqs}
\image{f}{\downset{X}{S}} \subset \downset{Y}{\image{f}{S}},
\end{eqs}
for all $S \subset X$.
\end{theorem}

\begin{proof}
\proofpart{$\implies$}
It holds that
\begin{eqs}
{} & x \in \downset{X}{S} \\
\impliesr & \exists s \in S: x \preleq s \\
\impliesr & \exists s \in S: f(x) \preleqb f(s) && \why{$f$ order-preserving} \\
\impliesr & f(x) \in \downset{Y}{\image{f}{S}},
\end{eqs}
for all $x \in X$.

\proofpart{$\impliedby$}
It holds that
\begin{eqs}
{} & x_1 \preleq x_2 \\
\impliesr & \pdownset{X}{x_1} \subset \pdownset{X}{x_2} \\
\impliesr & \image{f}{\pdownset{X}{x_1}} \subset \image{f}{\pdownset{X}{x_2}} \\
\impliesr & \image{f}{\pdownset{X}{x_1}} \subset \pdownset{Y}{f(x_2)} \\
\impliesr & f(x_1) \in \pdownset{Y}{f(x_2)} \\
\impliesr & f(x_2) \preleqb f(x_2),
\end{eqs}
for all $x_1, x_2 \in X$. Therefore $f$ is order-preserving.
\end{proof}

\section{Residuated functions}

\begin{definition}[Residual]
A function $\function{\hat{f}}{Y}{X}$ is a \define{residual} of $\function{f}{X}{Y}$, if
\begin{eqs}
f(x) \preleqb y \iff x \preleq \hat{f}(y),
\end{eqs}
for all $x \in X$ and $y \in Y$.
\end{definition}

\begin{definition}[Residuated]
A function $\function{f}{X}{Y}$ is \define{residuated}, if it has a residual.
\end{definition}

\begin{theorem}[Preimage of a down\-/set under a residuated function]
\label{PreimageOfDownSetUnderResiduatedFunction}
Let $\function{f}{X}{Y}$ be residuated, with a residual $\function{\hat{f}}{Y}{X}$. Then
\begin{equation}
 \preimage{f}{\downset{Y}{D}} = \downset{X}{\hat{f}(D)},
\end{equation}
for all $D \subset Y$.
\end{theorem}

\begin{proof}
\begin{eqs}
\preimage{f}{\downset{Y}{D}} & = \setb{x \in X : f(x) \in \downset{Y}{D}} \\
{} & = \setb{x \in X : \exists d \in D : f(x) \preleqb d} \\
{} & = \setb{x \in X : \exists d \in D : x \preleq \hat{f}(d)} \\
{} & = \setb{x \in X : \exists x' \in \hat{f}(D) : x \preleq x'} \\
{} & = \downset{X}{\hat{f}(D)}.
\end{eqs}
\end{proof}

\begin{theorem}[Residuated property by preimages]
\label{ResiduatedByPreimages}
Let $\function{f}{X}{Y}$. Then $f$ is residuated if and only if for each $y \in Y$, there exists $x_y \in X$ such that
\begin{eqs}
\preimage{f}{\pdownset{Y}{y}} = \pdownset{X}{x_y}.
\end{eqs}
\end{theorem}

\begin{proof}
\proofpart{$\implies$}
Let $f$ be residuated, and $y \in Y$. Then 
\begin{equation}
\preimage{f}{\pdownset{Y}{y}} = \pdownset{X}{\hat{f}(y)},
\end{equation}
by \proveby{PreimageOfDownSetUnderResiduatedFunction}. Therefore each preimage of a principal down\-/set of $Y$ is a principal down\-/set of $X$.

\proofpart{$\impliedby$}
For each $y \in Y$, there exists $x_y \in X$ such that
\begin{eqs}
{} & \preimage{f}{\pdownset{Y}{y}} = \pdownset{X}{x_y} \\
\iffr & \setb{x \in X : f(x) \in \pdownset{Y}{y}} = \setb{x \in X : x \preleq x_y} \\
\iffr & \setb{x \in X : f(x) \preleqb y} = \setb{x \in X : x \preleq x_y} \\
\iffr & \forall x \in X: \left[ f(x) \preleqb y \iff x \preleq x_y \right].
\end{eqs}
Let $\function{\hat{f}}{Y}{X}$ be such that $\hat{f}(y) = x_y$. By the above, $\hat{f}$ is a residual of $f$. Therefore $f$ is residuated.
\end{proof}

\begin{theorem}[Properties of a residuated function] 
\label{PropertiesOfResiduatedFunction}
Let $\function{f}{X}{Y}$. Then $f$ is residuated if and only if $f$ is order-preserving, and there exists order-preserving $\function{\residuals{f}}{X}{Y}$, such that
\begin{eqs}
{} & \residual{f}{f(x)} \pregeq x, \\
{} & f(\residual{f}{y}) \preleqb y,
\end{eqs}
for all $x \in X$, $y \in Y$.
\end{theorem}

\begin{proof}
\proofpart{$\implies$}
For the first relation,
\begin{eqs}
{} & f(x) \preleqb f(x) \\
\iffr & x \preleq \residual{f}{f(x)}, && \why{$\residuals{f}$ residual of $f$}
\end{eqs}
for all $x \in X$. 
Then
\begin{eqs}
{} & x_1 \preleq x_2 \\
\impliesr & x_1 \preleq \residual{f}{f(x_2)} && \why{first relation} \\
\impliesr & f(x_1) \preleqb f(x_2), && \why{$\residuals{f}$ residual of $f$}
\end{eqs}
for all $x_1, x_2 \in X$. Therefore $f$ is order-preserving.

For the second relation,
\begin{eqs}
{} & \residual{f}{y} \preleq \residual{f}{y} \\
\iffr & f(\residual{f}{y}) \preleqb y, && \why{$\residuals{f}$ residual of $f$}
\end{eqs}
for all $y \in Y$.
Then
\begin{eqs}
{} & y_1 \preleqb y_2 \\
\impliesr & f(\residual{f}{y_1}) \preleqb y_2 && \why{second relation} \\
\impliesr & \residual{f}{y_1} \preleq \residual{f}{y_2}, && \why{$\residuals{f}$ residual of $f$}
\end{eqs}
for all $y_1, y_2 \in Y$. Therefore $\residuals{f}$ is order-preserving.

\proofpart{$\impliedby$}
It holds that
\begin{eqs}
{} & f(x) \preleqb y \\
\impliesr & \residual{f}{f(x)} \preleq \residual{f}{y} && \why{$\residuals{f}$ order-preserving} \\
\impliesr & x \preleq \residual{f}{y} && \why{first relation} \\
\impliesr & f(x) \preleqb f(\residual{f}{y}) && \why{$f$ order-preserving} \\
\impliesr & f(x) \preleqb y, && \why{second relation}
\end{eqs}
for all $x \in X$, $y \in Y$. Therefore $\residuals{f}$ is a residual of $f$.
\end{proof}

\begin{theorem}[Residual is essentially unique]
\label{ResidualIsEssentiallyUnique}
Let $\function{g, h}{Y}{X}$ both be residuals of $\function{f}{X}{Y}$. Then
\begin{eqs}
g(y) \preeq h(y),
\end{eqs}
for all $y \in Y$.
\end{theorem}

\begin{proof}
By definition,
\begin{eqs}
{} & g(y) \preleq h(y) \\
\iffr & f(g(y)) \preleqb y,
\end{eqs}
for all $y \in Y$. The latter holds by \thref{PropertiesOfResiduatedFunction}. Similarly,
\begin{eqs}
h(y) \preleq g(y),
\end{eqs}
for all $y \in Y$. 
\end{proof}

\section{Residuated functions and supremum}

\begin{definition}[Upper-bound]
An element $s \in X$ is an \define{upper\-/bound} of $S \subset X$, if
\begin{eqs}
x \preleq s,
\end{eqs}
for all $x \in S$.
\end{definition}

\begin{definition}[Least upper\-/bound]
An element $s \in X$ is a \define{least upper\-/bound} of $S \subset X$, if $s$ is an upper\-/bound of $S$, and for any upper\-/bound $t \in X$ of $S$,
\begin{eqs}
s \preleq t.
\end{eqs}
\end{definition}

\begin{definition}[Supremum]
The \define{supremum} of a set in a set $X$ is a function $\function{\supremums}{\power{X}}{\power{X}}$, such that
\begin{eqs}
\supremum{S} = \setb{s \in X : s \text{ is a least upper\-/bound of $S$}}.
\end{eqs}
\end{definition}

\begin{theorem}[Residuated property by suprema]
\label{ResiduatedPropertyBySupremum}
Let $\function{f}{X}{Y}$. Then $f$ is residuated if and only if $f$ is order-preserving,
\begin{eqs}
\supremum[X]{\preimage{f}{\pdownset{Y}{y}}} \neq \emptyset,
\end{eqs}
for all $y \in Y$, and
\begin{eqs}
\image{f}{\supremum[X]{S}} \subset \supremum[Y]{\image{f}{S}},
\end{eqs}
for all $S \subset X$.
\end{theorem}

\begin{proof}
\proofpart{$\implies$ Order-preservation}
The function $f$ is order-preserving by \thref{PropertiesOfResiduatedFunction}.

\proofpart{$\implies$ Non-emptiness}
By \thref{ResiduatedByPreimages}, for each $y \in Y$, there exists $x_y \in X$, such that
\begin{eqs}
\preimage{f}{\pdownset{Y}{y}} = \pdownset{X}{x_y}.
\end{eqs}
It follows that
\begin{eqs}
\supremum[X]{\preimage{f}{\pdownset{Y}{y}}} = \ppclosure{x_y} \neq \emptyset.
\end{eqs}

\proofpart{$\implies$ Upper-bound}
Let $s \in \supremum{S}$. Since $s$ is an upper\-/bound of $S$,
\begin{eqs}
x \preleq s,
\end{eqs}
for all $x \in S$. Since $f$ is order-preserving,
\begin{eqs}
f(x) \preleqb f(s),
\end{eqs}
for all $x \in S$. Therefore $f(s)$ is an upper\-/bound of $\image{f}{S}$.

\proofpart{$\implies$ Least upper\-/bound}
Let $\function{\residuals{f}}{Y}{X}$ be a residual of $f$, and $y \in Y$ be an upper\-/bound of $\image{f}{S}$. Then
\begin{eqs}
f(x) \preleqb y,
\end{eqs}
for all $x \in S$. Since $f$ is residuated,
\begin{eqs}
x \preleq \residual{f}{y},
\end{eqs}
for all $x \in S$. That is, $\residual{f}{y}$ is an upper\-/bound of $S$. Since $s$ is a \emph{least} upper\-/bound of $S$,
\begin{eqs}
s \preleq \residual{f}{y}.
\end{eqs}
Since $f$ is residuated,
\begin{eqs}
f(s) \preleqb y.
\end{eqs}
Therefore $f(s)$ is a least upper\-/bound of $\image{f}{S}$;
\begin{eqs}
\image{f}{\supremum{S}} \subset \supremum{\image{f}{S}}.
\end{eqs}

\proofpart{$\impliedby$ Order-preservation}
Let $\function{\residuals{f}}{Y}{X}$ be such that
\begin{eqs}
\residual{f}{y} \in \supremum{\preimage{f}{\pdownset{Y}{y}}}.
\end{eqs}
The function $\residuals{f}$ is well-defined, since by assumption
\begin{eqs}
\supremum{\preimage{f}{\pdownset{Y}{y}}} \neq \emptyset,
\end{eqs}
for all $y \in Y$. Let $y_1, y_2 \in Y$ be such that $y_1 \preleqb y_2$. Since supremum is an upper\-/bound,
\begin{eqs}
{} & \forall x \in \preimage{f}{\pdownset{Y}{y_2}}: x \preleq \residual{f}{y_2} \\
\impliesr & \forall x \in X: \brac{f(x) \preleqb y_2 \implies x \preleq \residual{f}{y_2}} \\
\impliesr & \forall x \in X: \brac{f(x) \preleqb y_1 \implies x \preleq \residual{f}{y_2}} \\
\impliesr & \forall x \in \preimage{f}{\pdownset{Y}{y_1}}: x \preleq \residual{f}{y_2}.
\end{eqs}
That is, $\residual{f}{y_2}$ is also an upper\-/bound of $\preimage{f}{\pdownset{Y}{y_1}}$. Since $\residual{f}{y_1}$ is a \emph{least} upper\-/bound of $\preimage{f}{\pdownset{Y}{y_1}}$,
\begin{eqs}
\residual{f}{y_1} \preleq \residual{f}{y_2}.
\end{eqs}
Therefore $\residuals{f}$ is order-preserving. 

\proofpart{$\impliedby$ Deflation}
Let $y \in Y$. By assumption,
\begin{eqs}
f(\residual{f}{y}) & \in \image{f}{\supremum{\preimage{f}{\pdownset{Y}{y}}}} \\
{} & \subset \supremum[Y]{\image{f}{\preimage{f}{\pdownset{Y}{y}}}} \\
{} & \subset \supremum[Y]{\pdownset{Y}{y}} \\
{} & = \pclosure{\setb{y}}.
\end{eqs}
Therefore 
\begin{eqs}
f(\residual{f}{y}) \preleqb y.
\end{eqs}

\proofpart{$\impliedby$ Inflation}
Since $x \in \preimage{f}{\pdownset{Y}{f(x)}}$,
\begin{eqs}
x \preleq \residual{f}{f(x)},
\end{eqs}
for all $x \in X$. Therefore $f$ is residuated with residual $\residuals{f}$ by \thref{PropertiesOfResiduatedFunction}.
\end{proof}

\begin{theorem}[Residual by supremum]
\label{SupremumResidual}
Let $\function{f}{X}{Y}$ be residuated. Then $\function{\residuals{f}}{Y}{X}$ is a residual of $f$ if and only if 
\begin{eqs}
\residual{f}{y} \in \supremum{\preimage{f}{\pdownset{Y}{y}}},
\end{eqs}
for all $y \in Y$.
\end{theorem}

\begin{proof}
This is shown by \thref{ResiduatedPropertyBySupremum} and \thref{ResidualIsEssentiallyUnique}.
\end{proof}

\section{Residuated functions and p-inverses}

\begin{theorem}[Residual is a generalized p-inverse]
\label{ResidualIsGeneralizedPInverse}
Let $\function{f}{X}{Y}$ be residuated. Then
\begin{eqs}
f(\residual{f}{f(x)}) & \preeqb f(x), \\
\residual{f}{f(\residual{f}{y})} & \preeq \residual{f}{y},
\end{eqs}
for all $x \in X$, $y \in Y$.
\end{theorem}

\begin{proof}
Let $x \in X$. It holds that
\begin{eqs}
f(\residual{f}{f(x)}) & \preleqb f(x),
\end{eqs}
by \thref{PropertiesOfResiduatedFunction}. By the same theorem,
\begin{eqs}
\residual{f}{f(x)} \pregeq x.
\end{eqs}
Since $f$ is order-preserving by \thref{PropertiesOfResiduatedFunction},
\begin{eqs}
f(\residual{f}{f(x)}) \pregeqb f(x).
\end{eqs}
Similarly for the second equivalence.
\end{proof}

\begin{theorem}[Residual is a right p-inverse for p-surjection]
\label{ResidualIsRightPInverseForPSurjection}
Let $\function{f}{X}{Y}$ be a p-surjection. Then a residual of $f$ is a right p-inverse of $f$.
\end{theorem}

\begin{proof}
The function $f$ is p-preserving by \thref{PropertiesOfResiduatedFunction}. A right p-inverse $\rinvs{f}$ of $f$ exists by \thref{PSurjectivityIsEquivalentToHavingRightInverse}. Let $\function{\residuals{f}}{Y}{X}$ be a residual of $f$. Then
\begin{eqs}
{} \quad & y \preleqb y \\
{} \impliesr & f(\rinv{f}{y}) \preleqb y && \why{right p-inverse of $f$} \\
{} \impliesr & \rinv{f}{y} \preleq \hat{f}(y) && \why{$f$ residuated} \\
{} \impliesr & f(\rinv{f}{y}) \preleqb f(\hat{f}(y)) && \why{$f$ order-preserving} \\
{} \impliesr & y \preleqb f(\hat{f}(y)) && \why{right p-inverse of $f$},
\end{eqs}
for all $y \in Y$. By \thref{PropertiesOfResiduatedFunction},
\begin{eqs}
f(\hat{f}(y)) \preleqb y,
\end{eqs}
for all $y \in Y$. Therefore,
\begin{eqs}
f(\hat{f}(y)) \preeqb y,
\end{eqs}
for all $y \in Y$. That is, $\hat{f}$ is a right p-inverse of $f$.
\end{proof}

\begin{note}[]
A residuated p-surjection may have many non-equivalent right p-inverses; a residual is a specific version of a right p-inverse.
\end{note} 

\begin{theorem}[Residual is a left p-inverse for residuated p-injection]
\label{ResidualIsLeftPInverseForResiduatedPInjection}
Let $\function{f}{X}{Y}$ be a residuated p-injection. Then a residual of $f$ is a left p-inverse of $f$.
\end{theorem}

\begin{proof}
This follows from \thref{ResidualIsGeneralizedPInverse} and \thref{GeneralizedPInverseIsLeftPInverseForPInjective}.
\end{proof}

\begin{theorem}[Residual is a p-inverse for residuated p-bijection]
\label{ResidualIsPInverseForResiduatedPBijection}
Let $\function{f}{X}{Y}$ be a residuated p-bijection. Then a residual of $f$ is a p-inverse of $f$.
\end{theorem}

\begin{proof}
This follows from \thref{ResidualIsLeftPInverseForResiduatedPInjection} and \thref{ResidualIsRightPInverseForPSurjection}.
\end{proof}

\begin{theorem}[Transpose-residuated p-surjective function preserves down-sets]
\label{TransposeResiduatedPSurjectivePreservesDownSets}
Let $\function{f}{X}{Y}$ be $\pregeq$-residuated and p-surjective. Then
\begin{eqs}
\pclosure{\image{f}{\pdownset{X}{S}}} = \downset{Y}{\image{f}{S}},
\end{eqs}
for all $S \subset X$.
\end{theorem}

\begin{proof}
\proofpart{$\subset$}
The function $f$ is $\pregeq$-order-preserving by \thref{PropertiesOfResiduatedFunction}, which is the same as $\preleq$-order-preserving. The result follows from \thref{OrderPreservationByImages}.

\proofpart{$\supset$}
Let $\function{\hat{f}}{Y}{X}$ be a $\pregeq$-residual of $f$, which is a right p-inverse of $f$ by \thref{ResidualIsRightPInverseForPSurjection}. Then
\begin{eqs}
{} & y \in \downset{Y}{\image{f}{S}} \\
\impliesr & \exists x \in S: y \preleqb f(x) && \why{definition of down-set} \\
\impliesr & \exists x \in S: \residual{f}{y} \preleq x && \why{$\hat{f}$ is $\pregeq$-residual of $f$} \\
\impliesr & \residual{f}{y} \in \downset{X}{S}  && \why{definition of down-set} \\
\impliesr & f(\residual{f}{y)} \in \image{f}{\downset{X}{S}} \\
\impliesr & y \in \pclosure{\image{f}{\downset{X}{S}}}, && \why{$\hat{f}$ right p-inverse of $f$}
\end{eqs}
for all $S \subset X$.
\end{proof}

\begin{theorem}[Transpose-residuated surjective function preserves down-sets]
\label{TransposeResiduatedSurjectivePreservesDownSets}
Let $\function{f}{X}{Y}$ be $\pregeq$-residuated and surjective. Then
\begin{eqs}
\image{f}{\pdownset{X}{S}} = \downset{Y}{\image{f}{S}},
\end{eqs}
for all $S \subset X$.
\end{theorem}

\begin{proof}
The proof is the same as \thref{TransposeResiduatedPSurjectivePreservesDownSets}, except that now we can choose the residual to be both a right inverse and right p-inverse, and use the fact that
\begin{eqs}
y = f(\residual{f}{y}).
\end{eqs}
\end{proof}

\section{Residuated functions and order-embeddings}

\begin{theorem}[Residuated p-injection is order-embedding]
\label{ResiduatedPInjectionIsOrderEmbedding}
Let $\function{f}{X}{Y}$ be a residuated p-injection. Then $f$ is order-embedding.
\end{theorem}

\begin{proof}
Let $\function{\hat{f}}{Y}{X}$ be a residual of $f$. Since $f$ is p-injective, $\hat{f}$ is a left p-inverse of $f$ by \thref{ResidualIsLeftPInverseForResiduatedPInjection}. Then
\begin{eqs}
{} & f(x_1) \preleqb f(x_2) \\
\impliesr & x_1 \preleq \residual{f}{f(x_2)} && \why{$\hat{f}$ residual of $f$} \\
\impliesr & x_1 \preleq x_2, && \why{$\hat{f}$ left p-inverse of $f$}
\end{eqs}
for all $x_1, x_2 \in X$. Therefore $f$ is order-reflecting. The function $f$ is order-preserving by \proveby{PropertiesOfResiduatedFunction}. Therefore $f$ is order-embedding. 
\end{proof}

\begin{theorem}[Residuated p-bijection is an order isomorphism]
\label{ResiduatedPBijectionIsOrderIsomorphism}
Let $\function{f}{X}{Y}$. Then $f$ is an order isomorphism if and only if $f$ is a residuated p-bijection. 
\end{theorem}

\begin{proof}
\proofpart{$\implies$}
Since $f$ is order-embedding, $f$ is p-injective by \thref{OrderReflectingImpliesInjective}. Therefore $f$ is p-bijective. Let $\function{\pinvs{f}}{Y}{X}$ be a p-preserving p-inverse of $f$, which exists by \thref{PBijectivityIsEquivalentToHavingPPreservingInverse}. Then
\begin{eqs}
{} & f(x) \preleqb y \\
\iffr & f(x) \preleqb f(\pinv{f}{y}) && \why{$\pinvs{f}$ p-inverse of $f$} \\
\iffr & x \preleqb \pinv{f}{y}, && \why{$f$ p-reflecting}
\end{eqs}
for all $x \in X$ and $y \in Y$. Therefore $f$ is a residuated p-bijection.

\proofpart{$\impliedby$}
This follows from \thref{ResiduatedPInjectionIsOrderEmbedding} and because $f$ is p-surjective.
\end{proof}

\ifindex
\printindex
\fi

\end{document}